\documentclass[10pt,reqno]{amsart} 
\usepackage[margin=0.8in]{geometry} 
\usepackage{amsthm, amsmath,amsfonts,amssymb,euscript,hyperref,graphics,color,slashed}
\usepackage{graphicx}
\usepackage{mathrsfs}
\usepackage{comment}
\usepackage{import}
\usepackage{tikz}
\usepackage{latexsym}
\usepackage[makeroom]{cancel}
\usepackage[font=normalsize]{caption}
\usepackage{subfigure}
\usepackage[title]{appendix}
\usepackage{enumerate}
\usepackage{mathrsfs}
\usepackage{empheq}    
\usepackage{bm}
\usepackage{bookmark}
\usepackage{arydshln}
\usepackage{lineno,hyperref}
\usepackage{cite}
\usepackage{float}
\usepackage[utf8]{inputenc}
\usepackage{dashbox}
\usepackage{ulem}

\usepackage{geometry}
\usepackage{epsfig}
\usepackage{caption}
\usepackage{slashed}  
\newcommand{\Red}[1]{\begingroup \color{black} #1\endgroup}
\usepackage{sansmath}
\usetikzlibrary{shadings,intersections} 
\usepackage{enumitem}

\title[Naked singularity censoring with anisotropic apparent horizon]{Naked Singularity Censoring\\ with Anisotropic Apparent Horizon} 
\date{\today}

\author{Xinliang An}
\address{\small Department of Mathematics, National University of Singapore, 
10 Lower Kent Ridge Road, Singapore, 119076 }
\email{matax@nus.edu.sg}

\theoremstyle{definition}
\newtheorem{lemma}{Lemma}[section]

\newtheorem{proposition}[lemma]{Proposition}
\newtheorem{theorem}[lemma]{Theorem}
\newtheorem{corollary}[lemma]{Corollary}

\newtheorem{remark}{Remark}

\numberwithin{equation}{section}
\newtheorem*{theorem*}{Theorem}
\newtheorem{conjecture}{Conjecture}

\begin{document}

\newcommand{\ub}{\underline{u}}
\newcommand{\Cb}{\underline{C}}
\newcommand{\Lb}{\underline{L}}
\newcommand{\Lh}{\hat{L}}
\newcommand{\Lbh}{\hat{\Lb}}
\newcommand{\phib}{\underline{\phi}}
\newcommand{\Phib}{\underline{\Phi}}
\newcommand{\Db}{\underline{D}}
\newcommand{\Dh}{\hat{D}}
\newcommand{\Dbh}{\hat{\Db}}
\newcommand{\omb}{\underline{\omega}}
\newcommand{\omh}{\hat{\omega}}
\newcommand{\ombh}{\hat{\omb}}
\newcommand{\Pb}{\underline{P}}
\newcommand{\chib}{\underline{\chi}}
\newcommand{\chih}{\hat{\chi}}
\newcommand{\chibh}{\hat{\chib}}

\newcommand{\alb}{\underline{\alpha}}
\newcommand{\zeb}{\underline{\zeta}}
\newcommand{\beb}{\underline{\beta}}
\newcommand{\etb}{\underline{\eta}}
\newcommand{\Mb}{\underline{M}}
\newcommand{\oth}{\hat{\otimes}}

\gdef\dxd{\displaystyle}
\gdef\x{\xi}
\def\al{\aligned}
\def\eal{\endaligned}
\def\be{\begin{equation}}
\def\ee{\end{equation}}

\def\a {\alpha}
\def\b {\beta}
\def\ab {\alphab}
\def\bb {\betab}
\def\nab {\nabla}

\def\ub {\underline{u}}
\def\th {\theta}
\def\Lb {\underline{L}}
\def\Hb {\underline{H}}
\def\chib {\underline{\chi}}
\def\chih {\hat{\chi}}
\def\chibh {\hat{\underline{\chi}}}
\def\omegab {\underline{\omega}}
\def\etab {\underline{\eta}}
\def\betab {\underline{\beta}}
\def\alphab {\underline{\alpha}}
\def\Psib {\underline{\Psi}}
\def\hot{\widehat{\otimes}}
\def\Phib {\underline{\Phi}}
\def\thb {\underline{\theta}}
\def\t {\tilde}
\def\st {\tilde{s}}
\def\h {\hat}

\def\rhoc{\check{\rho}}
\def\sigmac{\check{\sigma}}
\def\Psic{\check{\Psi}}
\def\kappab{\underline{\kappa}}
\def\betabc {\check{\underline{\beta}}}

\def\d {\delta}
\def\f {\frac}
\def\i {\infty}
\def\l {\bigg(}
\def\r {\bigg)}
\def\S {S_{u,\underline{u}}}
\def\o{\omega}
\def\O{\Omega}\
\def\be{\begin{equation}\begin{split}}
\def\en{\end{split}\end{equation}}
\def\at{a^{\frac{1}{2}}}
\def\af{a^{\frac{1}{4}}}
\def\od{\omega^{\dagger}}
\def\ombd{\underline{\omega}^{\dagger}}
\def\K{K-\frac{1}{|u|^2}}
\def\ut{\frac{1}{|u|^2}}
\def\s{\frac{\delta a^{\frac{1}{2}}}{|u|}}
\def\Kb{K-\frac{1}{(u+\underline{u})^2}}
\def\ut{\frac{1}{|u|^2}}
\def\s{\frac{\delta a^{\frac{1}{2}}}{|u|}}
\def\Kb{K-\frac{1}{(u+\underline{u})^2}}
\def\Kt{K-\frac{1}{|u|^2}-\frac14 \nab^A\phi\nab_A\phi} 
\def\bf{b^{\frac{1}{4}}}
\def\bt{b^{\frac{1}{2}}}
\def\lo{\lambda_1}
\def\lt{\lambda_2}
\def\phib{\bar{\phi}}
\def\bR{\bar{R}}
\def\tR{\tilde{R}}
\def\S{S_{u,\ub}}

\def\de{\delta}
\def\ls{\lesssim}
\def\om{\omega}
\def\Om{\Omega}
\def\O{\Omega} 
\def\nab{\nabla}
\def\tp{\widetilde{\phi}}
\def\tO{\tilde{O}}
\def\cp{\check{\phi}}
\newcommand{\ms}{\mu^{*}}
\def\F{\big(1+\frac{\ub^{\frac12}a^{\frac14}}{|u|^{\frac12}\O} \big)}
\def\ho{\hat{\omega}}

\newcommand{\e}{\epsilon}
\newcommand{\et} {\frac{\epsilon}{2}}
\newcommand{\ef} {\frac{\epsilon}{4}}
\newcommand{\LH} {L^2(H_u)}
\newcommand{\LHb} {L^2(\underline{H}_{\underline{u}})}
\newcommand{\M} {\mathcal}
\newcommand{\TM} {\tilde{\mathcal}}
\newcommand{\p}{\psi\hspace{1pt}}
\newcommand{\q}{\underline{\psi}\hspace{1pt}}
\newcommand{\Li}{_{L^{\infty}(S_{u,\underline{u}})}}
\newcommand{\Lt}{_{L^{2}(S)}}
\newcommand{\da}{\delta^{-\frac{\epsilon}{2}}}
\newcommand{\db}{\delta^{1-\frac{\epsilon}{2}}}
\newcommand{\D}{\Delta}


\renewcommand{\div}{\mbox{div }}
\newcommand{\curl}{\mbox{curl }}
\newcommand{\trchb}{\mbox{tr} \chib}
\def\trch{\mbox{tr}\chi}
\newcommand{\tr}{\mbox{tr}}

\newcommand{\Ls}{{\mathcal L} \mkern-10mu /\,}
\newcommand{\eps}{{\epsilon} \mkern-8mu /\,}

\newcommand{\xib}{\underline{\xi}}
\newcommand{\psib}{\underline{\psi}}
\newcommand{\rhob}{\underline{\rho}}
\newcommand{\thetab}{\underline{\theta}}
\newcommand{\gammab}{\underline{\gamma}}
\newcommand{\nub}{\underline{\nu}}
\newcommand{\lb}{\underline{l}}
\newcommand{\mub}{\underline{\mu}}
\newcommand{\Xib}{\underline{\Xi}}
\newcommand{\Thetab}{\underline{\Theta}}
\newcommand{\Lambdab}{\underline{\Lambda}}
\newcommand{\vphb}{\underline{\varphi}}

\newcommand{\ih}{\hat{i}}

\newcommand{\tcL}{\widetilde{\mathscr{L}}}

\newcommand{\sRic}{Ric\mkern-19mu /\,\,\,\,}
\newcommand{\sL}{{\cal L}\mkern-10mu /}
\newcommand{\sLh}{\hat{\sL}}
\newcommand{\sg}{g\mkern-9mu /}
\newcommand{\seps}{\epsilon\mkern-8mu /}
\newcommand{\sd}{d\mkern-10mu /}
\newcommand{\sR}{R\mkern-10mu /}
\newcommand{\snab}{\nabla\mkern-13mu /}
\newcommand{\sdiv}{\mbox{div}\mkern-19mu /\,\,\,\,}
\newcommand{\scurl}{\mbox{curl}\mkern-19mu /\,\,\,\,}
\newcommand{\slap}{\mbox{$\triangle  \mkern-13mu / \,$}}
\newcommand{\sGamma}{\Gamma\mkern-10mu /}
\newcommand{\somega}{\omega\mkern-10mu /}
\newcommand{\somb}{\omb\mkern-10mu /}
\newcommand{\spi}{\pi\mkern-10mu /}
\newcommand{\sJ}{J\mkern-10mu /}
\renewcommand{\sp}{p\mkern-9mu /}
\newcommand{\su}{u\mkern-8mu /}

\begin{abstract} 

Employing the Einstein-scalar field system, we demonstrate an approach for proving high co-dimensional nonlinear instability of naked-singularity solutions as constructed by Christodoulou in \cite{Chr.2}. We further investigate the censorship of Christodoulou's naked singularity and show that a tiny anisotropic perturbation arising from the outgoing characteristic initial data would lead to the emergence of an anisotropic apparent horizon, which covers and censors the naked singularity. Our approach advances the hyperbolic short-pulse method by not requiring the aid of additional large parameters, by permitting the use of initial perturbations for the shear tensor and the derivative of scalar field {\color{black}to be with finite $BV$ and $C^0$ norms}, and by allowing the initial perturbation to be arbitrarily small in scale-critical norms. New elliptic arguments based on non-perturbative methods are also developed. 

\end{abstract} 
\maketitle 
 
\section{Introduction}
In 1994 Christodoulou \cite{Chr.2} constructed his famous naked-singularity solution in $3+1$ dimensions for the Einstein-scalar field system 
\begin{equation}\label{1.1}
\begin{split}
R_{\mu\nu}-&\f12 Rg_{\mu\nu}=T_{\mu\nu},\\
T_{\mu\nu}=&D_{\mu}\phi D_{\nu}\phi-\f12g_{\mu\nu}D^{\lambda}\phi D_{\lambda}\phi.
\end{split}
\end{equation}
Here $\phi: \mathbb{R}^{3+1} \rightarrow \mathbb{R}$ is a real-valued scalar function and $(\mathcal{M}, g)$ is the $3+1$ dimensional spacetime. The stability or instability of these naked singularities is closely related to the \textit{weak cosmic censorship} \cite{Penrose69, Chr.99}: 
\begin{conjecture}\label{conjecture1}  (Weak cosmic censorship conjecture).
For generic asymptotically flat initial data, the maximal development of Einstein's field equations possesses a complete future null infinity $\mathcal{I}^+$ and hides the possibly formed singularities in a (black hole) region causally disconnected from $\mathcal{I}^+$. 
\end{conjecture}  
The above is one of the greatest open problems in classic general relativity. When Penrose formulated the original weak cosmic censorship conjecture in \cite{Penrose69}, the caveat generic was not included. Christodoulou's example in \cite{Chr.2} surprisingly shows that naked singularity could form in the evolution of Einstein's field equations. Within spherical symmetry, for the Einstein-scalar field system \eqref{1.1}, by a celebrated series of works \cite{Chr.1, Chr.1.5, Chr.3} Christodoulou proved that the initial datum leading to the naked-singularity formation is rare, and for generic initial data the possibly formed singularities are indeed hidden in the trapped (black hole) region.  This motivates him to add the term \textit{generic} for Conjecture \ref{conjecture1} in \cite{Chr.99}. 

Christodoulou's proofs in \cite{Chr.1, Chr.1.5, Chr.3} reveal novel mathematical and physical properties of the spherically symmetric gravitational collapse. For the initial datum leading to naked singularity, in \cite{Chr.3} he identified $2$ instability mechanisms and deduced co-dimensional $2$ instability. These instability theorems show that each rare initial datum, which leads to naked singularity formation, is always associated with $2$-dimensional perturbations and each of the perturbations would give rise to a trapped surface and hence a black hole formation. Here a trapped surface is a $2$-sphere in $\mathcal{M}$ with both null expansions negative.

In above arguments, the spherically symmetric perturbation required in \cite{Chr.3} is a bit restrictive, since it imposes a global condition in all angles. Can we relax this requirement by only placing a small anisotropic-in-angle perturbation? If it is possible, this may also significantly raise the co-dimensions of the instability theorems. In this article, we give such an anisotropic result. In particular, we prove that \textit{for Christodoulou's constructed naked singularity in \cite{Chr.2}, by imposing a tiny perturbation from any angle, we can trigger a trapped region and an apparent horizon to censor the naked singularity.} Compared to Christodoulou's instability argument that the singularity is covered by later-formed trapped surfaces, here we further show that the naked singularity is censored by an {\color{black}achronal} apparent horizon emerging from it. Subjected to a general anisotropic perturbation, the naked singularity in \cite{Chr.2} is not naked even locally. It is completely censored for observers outside the apparent horizon.\\   

\subsection{Main Theorem} Here we use the double-null foliation for the spacetime $(\mathcal{M}, g)$. 
Let $u$ and $\ub$ be optical functions increasing towards the future and satisfy the eikonal equations 
$$g^{\mu\nu}\partial_\mu u\partial_\nu u=0,\quad g^{\mu\nu}\partial_\mu\ub\partial_\nu \ub=0.$$ 
Set $H_u$ to be the level sets of $u$ and $\Hb_{\ub}$ to be the level sets of $\ub$. And we require that $H_u$ and $\Hb_{\ub}$ correspond to the outgoing and incoming null cones, respectively. The intersections of $H_u$ and $\Hb_{\ub}$ are topologically 2-spheres and are denoted by $S_{u,\ub}$. We further fix a frame $e_1,e_2$ tangent to $\S$ and denote $\nab$ to be the angular derivative on it. {\color{black}In this paper, we employ the double null coordinates $(u, \ub, \theta^1, \theta^2)$ for $(\mathcal{M}, g)$, where $\partial/\partial \theta^1$ and $\partial/\partial \theta^2$ are tangential to $S_{u, \ub}$.\footnote{See Chapter 1 and Chapter 2 in \cite{Chr:book} for a detailed discussion of the double null coordinates.} For $A, B=1, 2$, the spacetime metric $g$ takes the form}
\begin{equation*}
g=-2\Omega^2(du\otimes d\ub+d\ub\otimes du)+\sg_{AB}(d\th^A-d^Ad\ub)\otimes (d\th^B-d^Bd\ub). 
\end{equation*}
Here $\O(u, \ub, \theta_1, \theta_2)$ is the lapse function, $\sg_{AB}$ is the induced metric on $\S$, and $A, B=1, 2$. For notational simplicity, in this article we also use $\o\in\mathbb{S}^2$ to denote $(\theta_1, \theta_2)$. With the help of $\O$, we choose the null pair $e_3, e_4$ in this article as
$$e_3=\Omega^{-1}\frac{\partial}{\partial u}, \quad \quad e_4=\Omega^{-1}\left(\frac{\partial}{\partial \ub}+d^A\frac{\partial}{\partial \th^A} \right)$$
with $d^A$ obeying $d^A=0$ on {\color{black}$H_{-1}$}. And it holds
$$g(e_3, e_3)=0, \quad g(e_4, e_4)=0, \quad g(e_3, e_4)=-2.$$ 
Define 
$$\chi_{AB}:=g(D_{e_A}e_4, e_B)$$ 
and set $\chih$ to be its traceless part, i.e., the shear tensor. In this article, we first prove 

\begin{theorem} \label{main thm}
Consider the characteristic initial value problem for the Einstein-scalar field system \eqref{1.1}. Assigning Christodoulou's naked-singularity initial data in \cite{Chr.2} along  $\Hb_0$ with $-1\leq u \leq 0$ and prescribing perturbed initial data along $H_{-1}$ satisfying
\begin{equation}\label{upper bound main thm a=1}
\sum_{i\leq 5, j\leq 3}\ub^{j} \|(\partial_{\ub})^j\nab^{i}\chih\|_{L^{\infty}_{\ub}L^2(S_{-1,\ub})}+\ub^{j} \|(\partial_{\ub})^j\nab^{i}\nab_4\phi\|_{L^{\infty}_{\ub}L^2(S_{-1,\ub})}\leq B, 
\end{equation}
then for each $B$ there exist a sufficient small $\delta=\delta(B)$ and the Einstein-scalar field system admits a unique regular solution in the spacetime region $(u,\ub, \o)$ satisfying $0\leq\ub\leq \delta$ and $\ub\leq |u|\O^{2-\t\delta}(u,0)\leq 1$ with $0<\t\d\ll 1$ a fixed constant. 

Moreover, if we further assume 
\begin{equation}\label{lower bound main thm a=1}
 \int_0^{\ub}|\chih|^2 (-1, \ub', \o)+[\nab_4\phi(-1, \ub', \o)-\nab_4\phi(-1, 0, \o)]^2 d\ub'=f(\o, \ub)\ub,
\end{equation}
with $f(\o,\ub)$ obeying
\begin{equation}\label{eq-assumption-lower-bound-f}
0\leq f\leq 1 \mbox{ on } \mathbb S^2\times (0,\delta] \quad \mbox{ and } \quad f(\cdot, \ub)\ge m \mbox{ on } B_{p}(\e)  
\end{equation}
for a point  $p\in \mathbb S^2$  and constants  $m\in (0,1)$  and  $\e\in (0,\pi/2)$, then within the solved hyperbolic region, there exists a unique MOTS (marginally outer trapped surface) $M_{\ub}$ on each $\Hb_{\ub}$ with $0< \ub \leq \d$ and the collection of MOTS $\{M_{\ub}\}$ emerges and censors the central singularity and forms an achronal apparent horizon. 
\end{theorem}  

{\color{black}
For the shear tensor and the derivative of scalar field, the above theorem prescribes {\color{black}initial perturbation along $H_{-1}^{[0,\delta]}$ with a finite $BV$ norm satisfying \eqref{lower bound main thm a=1}.} Our next theorem further allows {\color{black} initial perturbation along $H_{-1}^{[0,\delta]}$ to be with finite $C^0$ norm.}  

\begin{theorem}\label{main theorem 2}
The above statements in Theorem \ref{main thm} still hold, if we replace the requirements for $f(\o, \ub)$ in \eqref{lower bound main thm a=1} and \eqref{eq-assumption-lower-bound-f} by 
$$f(\o, \ub)=g(\ub)\t f(\o, \ub) \,\, \mbox{ with } \,\,  g(\ub)\approx [\ln(\ln\f{1}{\ub})]^{-\f12} \,\, \mbox{ obeying } \,\, \lim_{\ub\rightarrow 0^+}g(\ub)=0,$$
and  $\t f(\o,\ub)$ satisfying $0\leq \t f\leq 1 \mbox{ on } \mathbb S^2\times (0,\delta], \,\, \t f(\cdot, \ub)\ge m \mbox{ on } B_{p}(\e)$ for a point  $p\in \mathbb S^2$ with constants  $m\in (0,1)$ and  $\e\in (0,\pi/2)$. 

\end{theorem}

}

\begin{figure}[h]
\begin{center}
\begin{tikzpicture}[scale=0.86]
\draw [white](0,2.3)--node [midway,sloped, below,black] {$O$}(-0.5,2.3);
\draw [white](1.9, 2.5)--node [midway,sloped, above,black] {$S_{-(\tilde{\ub} a)^{1-\ms}, \tilde{\ub}}$}(2.4, 3);
\fill[white!70!black](0,2)--(1,1)--(2,2)--(1.5, 1.9)--(1,1.8)--(0.8, 1.8)--(0,2);
\draw [white](3,-1)-- node[midway, sloped, below,black]{$H_{-1}(u=-1)$}(4,0);
\draw [white](3.3, 2.6)--node [midway,sloped, above,black] {$S_{-(\delta a)^{1-\ms}, \delta}$}(3.5, 2.8);
\draw [white](2,2)--node [midway,sloped,above,black] {$\Hb_{\delta}(\ub=\delta)$}(4,0);
\draw [white](1,1)--node [midway,sloped, below,black] {$\Hb_{0}(\ub=0)$}(3,-1);
\draw [dashed] (0, 2)--(0, -4);
\draw [dashed] (0, -4)--(4,0)--(2,2);
\draw [dashed] (0,0)--(2,2);
\fill[white!70!black] (1,1)--(3,-1)--(4,0)--(2,2)--(1,1);
\fill[white!50!black] (0.8, 1.7)--(3.25,-0.75)--(3,-1)--(0.55, 1.45)--(0.8, 1.7);
\draw [white] (1.2, 1.3)--node [midway,sloped,above,black] {$\Hb_{\tilde{\ub}}(\ub=\tilde{\ub})$}(3.25,-0.75); 
\draw [thick] (0,2)--(3,-1);
\draw [thick] (1,1)--(3,-1)--(4,0)--(2,2)--(1,1);
\draw [thick] (0.8, 1.7)--(3.25,-0.75); 
\draw [thick] (0.55, 1.45)--(0.8, 1.7); 
\draw[dashed] (0, 2)--(0.5, 1.75);
\draw[dashed] (0.5, 1.75)--(0.8, 1.7);
\draw[dashed] (0.8, 1.7)--(1.5, 1.8);
\draw[dashed] (1.5, 1.8)--(2, 2);
\filldraw[black] (0.8,1.7) circle (2pt);
\filldraw[black] (2,2) circle (2pt);
\end{tikzpicture}
\end{center}
\caption{}\label{Figure 4} 
\end{figure} 

{\color{black}
\begin{remark}\label{tilde f requirement}
To set the above-mentioned $C^0$ perturbation, we can let 
$$|\chih|(-1, \ub, \o)+|\nab_4\phi(-1, \ub, \o)-\nab_4\phi(-1, 0, \o)|=[\ln(\ln \f{1}{\ub})]^{-\f14}\cdot\tilde{f}(\o)^{\f12}$$
with $\tilde{f}(\o)$ being a smooth function depending only on $\o$. By Proposition \ref{ln ln initial data size}, we have
$$\|[\ln(\ln \f{1}{\ub})]^{-\f14}\|_{\dot{H}^{\f12}([0,\delta])}\lesssim [\ln\f{1}{\delta}]^{-\f12}\ll 1.$$
Hence, with Theorem \ref{main theorem 2}, to trigger the censoring of naked-singularity, we only need 
the scale-critical perturbation of $g$ in $\|\cdot\|_{\dot{H}^{\f32}{(H_{-1}^{[0, \delta]}})}$ norm, i.e., 
$$\mbox{the } \|\cdot\|_{\dot{H}^{\f12}{(H_{-1}^{[0, \delta]}})}=\|\cdot\|_{L^2_{\o}\dot{H}^{\f12}_{\ub}([0,\delta])}+\|\cdot\|_{L^2_{\ub}\dot{H}^{\f12}_{\o}(S_{-1,\ub})} \mbox{ norm} $$ 
of $\chih(-1, \ub, \o)$ and $\nab_4\phi(-1, \ub, \o)-\nab_4\phi(-1, 0, \o)$ to be of order $[\ln\f{1}{\delta}]^{-\f12}\ll 1$. And 
as $\delta\rightarrow 0^+$, we allow the size of the perturbation to shrink to $0$.
\end{remark} 
}

\begin{remark}\label{remark a=1}
Compared with the previous short-pulse results in the scale-critical regime, here we don't need the additional large parameter $a$ and we can view that the parameter $a$ is set to be 
$a=a(B)=1$. This is an important \textit{advance} for the short-pulse method invented in \cite{Chr:book} by Christodoulou with $a=\delta^{-\f12}\gg 1$ and generalized in \cite{AL} and \cite{An19} with $a=a(B)\gg 1$ by the author and Luk, and by the author, respectively. With the aid of this parameter $a$, we could consider the strength of the short pulse to be of order $\at$. The larger the short-pulse perturbation, the easier to form a trapped surface. A novel insight of this paper is that, via employing the instability nature of the interior naked-singularity solution, to form an anisotropic trapped surface and the corresponding apparent horizon, we can set the upper bound of the short-pulse perturbation to be of order $a(B)^{\f12}=1$, and according to \eqref{eq-assumption-lower-bound-f} we could even set the size of short pulse to be a universal small number $m$. {\color{black}As pointed out in Remark \ref{tilde f requirement}, we can further choose its size to be $g(\ub)$ with $g(\ub)\rightarrow 0$ as $\ub\rightarrow 0^+$. The perturbation allowed can be arbitrarily small measured at the scale-critical level.} This is markedly different from \cite{AL} by the author and Luk,  and \cite{AH} by the author and Han, where the constant $a=a(B)$ has to be a large parameter and 
the fact $1/a\ll 1$ is utilized crucially throughout \cite{AL, AH} to control the nonlinear terms. This paper contains other hyperbolic and elliptic advancements. To capture the key features of the anisotropic instability mechanism, we also develop quite a few new ingredients. These are summarized in Section \ref{new ingredients}.
\end{remark}

\begin{remark}
In this paper, all the estimates still hold if we restore sharp $a$-weights and set $a=a(B)$ to be a parameter greater than or equal to $1$. In particular, we allow $a=a(B)\gg 1$, but we do not need $a\gg 1$ to close the arguments. To be consistent with \cite{AL, AH} and to contain more analytic information for further use, in below sections we work under the following assumptions for initial perturbation: 

\begin{equation}\label{upper bound main thm}
\sum_{i\leq 5, j\leq 3}\ub^{j} a^{-\frac12}\|(\partial_{\ub})^j\nab^{i}\chih\|_{L^{\infty}_{\ub}L^2(S_{-1,\ub})}+\ub^{j}a^{-\f12} \|(\partial_{\ub})^j\nab^{i}\nab_4\phi\|_{L^{\infty}_{\ub}L^2(S_{-1,\ub})}\leq B, 
\end{equation}
and
\begin{equation}\label{lower bound main thm}
 \int_0^{\ub}|\chih|^2 (-1, \ub', \o)+[\nab_4\phi(-1, \ub', \o)-\nab_4\phi(-1, 0, \o)]^2 d\ub'=f(\o, \ub)\ub a.
\end{equation}
Setting $a=a(B)=1$, we then retrieve assumptions \eqref{upper bound main thm a=1}, \eqref{lower bound main thm a=1}. \end{remark}

\begin{remark}

In below, we also give a stereoscopic perspective, which demonstrates naked-singularity censoring subject to an anisotropic perturbation. Note that, as portrayed in the following picture, the anisotropic apparent horizon censors the naked singularity even locally. No local observer before the apparent horizon can detect the naked singularity. What we prove in this paper is a more precise statement than showing tiny trapped surfaces formed around $O$. 

\begin{figure}[h]
\begin{center}
\begin{tikzpicture}[scale=0.7]
\draw [white](-0.1, 0.1)-- node[midway, sloped, above,black]{$O$}(0.1, 0.1);
\draw [white](3, -0.5)-- node[midway, sloped, above,black]{{\small Anisotropic Apparent Horizon}}(5, -0.5);
\draw [white](5.8, -3.85)-- node[midway, sloped, above,black]{{\small Anisotropic Perturbation}}(6.8, -3.85);
\draw [white](5.8, -3.85)-- node[midway, sloped, below,black]{{\small along} $u=-1$}(6.8, -3.85);

\draw [white](3, -5.5)-- node[midway, sloped, above,black]{$S_{-1, 0}$}(3.5, -5.5);
\draw (0,-4) ellipse (3cm and 1.2cm);
\draw (0,-3) ellipse (3.5cm and 1.2cm);
\draw[thin] (-3, -4)--(0, 0);
\draw[thin] (3, -4)--(0, 0);
\draw[thin] (3, -4)--(3.5, -3);
\draw[thin] (-3, -4)--(-3.5, -3);
\fill[white!30!black] (2.8,-3.7)--(3.4, -3.3)--(3, -4.2)--(2.8,-3.7);

\draw [very thick] (0, 0) to [out=230,in=30] (-1.2, -1);
\draw [very thick] (0, 0) to [out=315,in=150] (2, -1.5);
\draw [very thick] (-1.2,-1) to [out=10,in=150] (2, -1.5);
\draw [very thick] (2, -1.5) to [out=190,in=-35] (-1.2,-1);

\draw[fill] (0,0) circle [radius=0.08];
\end{tikzpicture}
\end{center}
\caption{}
	\label{Figure Naked Singularity Censoring}
\end{figure}

\end{remark}

Our next result further exhibits a high co-dimensional nonlinear instability:

\begin{theorem}\label{main thm 1.2} 
Consider the space of prescribed initial data for $\chih(-1, \ub, \o)$ and $\partial_{\ub}\phi(-1, \ub, \o)$ along $u=-1$ with $\ub\geq 0$. With $|\chih|(-1, 0, \o)=0, \,\, \partial_{\ub}\phi(-1, 0, \o)=\sqrt{2}\cdot\sqrt{\f{1-\ms}{\ms}}\cdot\f{1-\ms}{4}$ corresponding to their values of Christodoulou's naked-singularity solution in \cite{Chr.2}, then for any $k\in \mathbb{Z}^+$, we can find functions $f_1, f_2, ... , f_k$ and $g_1, g_2, ... , g_k$, which are continuous with respect to the $\ub$ variable and smooth with respect to the $\o$ variable, such that the below prescribed initial data for $|\chih|$ and $\partial_{\ub}\phi$ with $\ub\geq 0$ 
\begin{equation*}
\begin{split}
&\bigg(|\chih|(-1, \ub, \o), \quad \partial_{\ub}\phi(-1, \ub, \o)\bigg)\\
=&\bigg(0+\lambda_1 f_1+\lambda_2 f_2+\cdot\cdot\cdot+\lambda_k f_k, \quad \partial_{\ub}\phi(-1, 0, \o)+\lambda_1' g_1+\lambda_2' g_2+\cdot\cdot\cdot+\lambda_k' g_k\bigg)
\end{split}
\end{equation*}
would lead to the trapped surface and the apparent horizon formation when 
$$\sum_{1}^k \lambda_i^2+\sum_1^k {\lambda_i'}^2\neq 0 \quad \mbox{ with } \quad \lambda_k, \lambda_k'\in \mathbb{R}.$$ 

Moreover, if it holds
\begin{equation*}
\begin{split}
&\bigg(0+\lambda_1 f_1+\lambda_2 f_2+\cdot\cdot\cdot+\lambda_k f_k, \quad \partial_{\ub}\phi(-1, 0, \o)+\lambda_1' g_1+\lambda_2' g_2+\cdot\cdot\cdot+\lambda_k' g_k\bigg)\\
=&\bigg(0+\t\lambda_1 f_1+\t\lambda_2 f_2+\cdot\cdot\cdot+\t\lambda_k f_k, \quad \partial_{\ub}\phi(-1, 0, \o)+\t\lambda_1' g_1+\t\lambda_2' g_2+\cdot\cdot\cdot+\t\lambda_k' g_k\bigg),
\end{split}
\end{equation*}
then it renders
$$\lambda_i=\t\lambda_i \quad \mbox{ and } \quad \lambda_i'=\t\lambda_i \quad \mbox{ with } \quad 1\leq i \leq k.$$
Hence, we may say that, for any $k\in\mathbb{Z}^+$, Christodoulou's naked-singularity solution has at least co-dimension $2k$ nonlinear instability subject to outgoing continuous characteristic perturbations at $(-1, 0, \o)$.  
\end{theorem}

\begin{remark}
In Section \ref{section instability theorems}, we obtain \textit{two} instability theorems. With respect to the shear tensor and the derivative of the scalar field, Theorem \ref{main thm section 14} proves instability subject to {\color{black}perturbations with finite $BV$ norms} and Theorem \ref{main thm 1.2} proves instability subject to perturbations {\color{black}with finite $C^0$ norms}. In \cite{AL, An17, AH}, to trigger the short pulse, BV perturbations for $\partial_{\ub}\sg|_{S_{-1, 0}}$ are prescribed. In \cite{AL, An17, AH}, when $\ub\leq 0$, the shear tensor $\chih$ vanishes, and when $\ub>0$, the shear tensor $\chih$ obeys a large lower bound related to the large parameter $a$.  For this article, in Section \ref{BV instability} we first prove Theorem \ref{main thm section 14}, which is consistent with  \cite{AL, An17, AH} by imposing the BV perturbations. In Section \ref{C1 instability}, we further allow the perturbations to be $C^0$. Additional gains in both hyperbolic and elliptic parts enable us to do so. Detailed discussions are provided in Section \ref{section instability theorems}.  

\end{remark}

\begin{remark}
{\color{black}Our above theorems prove arbitrarily large co-dimensional nonlinear instability for a naked-singularity solution. This is not reported before and we spot the importance of the anisotropic arguments. Together with Remark \ref{tilde f requirement}, we can see that this paper contains a desired scale-critical nonlinear instability result.}  
\end{remark}

\begin{remark}
{\color{black}The proofs of above two theorems and the approach we demonstrate here have the potential to be generalized to other Einstein field equations. For the hyperbolic part, we employ $\O(u, 0)\rightarrow 0$ as $u\rightarrow 0$. For the elliptic part, we utilize $\O\omb(u, 0)=-\partial_u \log\O(u, 0)/2>0$, which also holds for Rodnianski-Shlapentokh-Rothman's naked-singularity in \cite{R-S} for the Einstein vacuum equations. }
\end{remark}

Our proofs of above two theorems are rooted in recent progress on trapped surface formation and apparent horizon emergence. In below, we review these developments.
\subsection{Formation of Trapped Surface}

With a 589-page proof, Christodoulou \cite{Chr:book} proved the first trapped-surface-formation result for Einstein's equations with no symmetry. As depicted in Figure \ref{Figure 1} below, Christodoulou considered a characteristic initial-value problem for the Einstein vacuum equations in the shadowed region. Minkowskian initial data are prescribed along the incoming null cone $\Hb_0$ and short-pulse initial data are prescribed along the outgoing null cone $H_{-1}$. Along $H_{-1}$, denoting $\chi$ to be its associated outgoing null second fundamental form and setting $\chih$ to be its traceless part, in 2008 Christodoulou proved the following\footnote{In \cite{Chr:book} Christodoulou's short-pulse initial data are prescribed at past null infinity. Here we focus on its correspondence in a finite region.}:

\begin{figure}[h]
\begin{center}
\begin{tikzpicture}[scale=0.75]

\draw [white](3,-1)-- node[midway, sloped, below,black]{$H_{-1}(u=-1)$}(4,0);

\draw [white](2,2)--node [midway,sloped,above,black] {$\Hb_{\delta}(\ub=\delta)$}(4,0);
\draw [white](1,1)--node [midway,sloped, below,black] {$\Hb_{0}(\ub=0)$}(3,-1);
\draw [dashed] (0, 2)--(0, -4);
\draw [dashed] (0, -4)--(4,0)--(2,2);
\draw [dashed] (0,0)--(2,2);

\draw [dashed] (0,-4)--(2,-2);
\draw [dashed] (0,2)--(3,-1);
\draw [very thick] (1,1)--(3,-1)--(4,0)--(2,2)--(1,1);
\fill[white!70!black] (1,1)--(3,-1)--(4,0)--(2,2)--(1,1);
\draw [white](1,1)-- node[midway,sloped,above,black]{$H_{u_*}$}(2,2);
\end{tikzpicture}
\end{center}
\caption{}
	\label{Figure 1}
\end{figure}

\begin{theorem*}[Christodoulou \cite{Chr:book}]\label{Chr.thm}
Set $\Hb_0$ to coincide with a backward light cone in the Minkowskian spacetime for $-1\leq u\leq 0$. For every $B>0$ and $-1<u_*<0$, there exists a $\de=\de(B,u_*)>0$ sufficiently small such that the following holds: for $0\leq\ub\leq\delta$, if the prescribed $\chih$ on $H_{-1}$ satisfies
\begin{equation}\label{Chr.upper.bound}
\sum_{i\leq 5,\,j\leq 3}\de^{\frac 12+j}\|(\partial_{\ub})^j \nab^i\chih\|_{L^\infty_{\ub}L^2(S_{-1,\ub})}\leq B 
\end{equation}
with $\nab$ being the angular derivative on $S_{-1,\ub}$, then the solution to Einstein vacuum equations remains regular in the spacetime region $-1\leq u\leq u_*$, $0\leq \ub\leq \de$. Furthermore, if the below lower-bound assumption also satisfies
\begin{equation}\label{Chr.lower.bound}
\inf_{\o\in \mathbb{S}^2} \int_0^{\de}  |\chih|^2(-1, \ub', \o)\,d\ub' \geq M_* > 2|u_*|,
\end{equation}
then the $2$-sphere $S_{u_*,\de}$ is a trapped surface.
\end{theorem*}

\noindent Using similar upper and lower bounds as in \cite{Chr:book},  extensions of Theorem \ref{Chr.thm} via introducing the signatures for short pulse and for decay rates were obtained in \cite{KR:Trapped} and \cite{An12}, respectively. The trapped surfaces formed are of radius close to $1$.

\smallskip 

An important result was later obtained in \cite{KLR}. Via deforming the double null foliation and solving a quasilinear elliptic inequality, the lower bound assumption in \cite{Chr:book} is vastly relaxed to allow fully anisotropic initial data.

\begin{theorem*}[Klainerman-Luk-Rodnianski \cite{KLR}]\label{K-L-R theorem}
Suppose that the initial data for Einstein vacuum equations satisfy the condition \eqref{Chr.upper.bound} in Theorem \ref{Chr.thm} and the lower bound assumption
\begin{equation*}
\sup_{\o\in \mathbb{S}^2} \int_0^{\de} |\chih|^2(-1, \ub', \o)\,d\ub' \geq M_* > 0.
\end{equation*}
Then, for sufficiently small $\de$, a trapped surface forms in the evolution. 
\end{theorem*}
{\color{black}\noindent We also remark that, with a more geometric approach, Le in \cite{Le} gave a different proof of the main theorem in \cite{KLR}.}

\smallskip 

Via introducing two more large parameters $a, b$ satisfying $1\ll b\leq \at\leq \d^{-\f12}$ and designing a new scale-critical short-pulse ansatz, a work \cite{AL} relaxed the upper bound assumption in \cite{Chr:book}.

\begin{theorem*}[An-Luk \cite{AL}] \label{thm1.3}
Prescribe Minkowskian data along  $\Hb_0$ with $-1\leq u \leq 0$. For a fixed $\delta$, for every $B$, there exist $a_0=a_0 (B)$ and $b_0=b_0(B)$ sufficiently large such that the following hold: for any $a$ and $b$ verifying $a_0 \leq a \leq \delta^{-1}$ and $b_0\leq b \leq a^{\frac12} \leq \delta^{-\frac12}$, if prescribed $\chih$ along $H_{-1}$ with $0\leq \ub \leq \d$ satisfies
\begin{equation}\label{AL.upper.bound}
\sum_{i\leq 5, j\leq 3}\delta^{j} a^{-\frac12}\|\nab^{j}_{e_4}\nab^{i}\chih\|_{L^{\infty}_{\ub}L^2(S_{-1,\ub})}\leq B, 
\end{equation}
then the Einstein vacuum equations admit a unique regular solution in the spacetime region $-1\leq u \leq -b\delta \at$ and $0\leq \ub \leq \d$. Moreover, if we further require
\begin{equation*}
\inf_{\o\in \mathbb{S}^2}\int_0^{\delta}|\chih|^2(-1, \ub', \o) d\ub'\geq 4\delta a,
\end{equation*}
then $S_{-\d a, \d}$ is the formed $2$-dimensional trapped surface.  
\end{theorem*}

Later, via employing only the signature for decay rates and a spacetime rescaling, a reproof and an extension of Theorem \ref{thm1.3} to the past null infinity was obtained in \cite{An19} for the case $b=\at$ by the author.  

\begin{remark}
The scale-critical bounds in \cite{AL} inspire the hyperbolic estimates of this article. However, there are at least three major differences:
\begin{enumerate} 
\item In \cite{AL} the formed trapped surface is $S_{-\delta a, \delta}$ and it is in the approximately self-similar regime. The designed short-pulse hierarchy and the hyperbolic estimates are all consistent with the considerations of scaling criticality. This is the underlying reason that in \cite{An19} the author can give an alternative proof in the far field regime and very few borderline terms appear. In this article,  the trapped surface is located around $S_{-(\d a)^{1-\ms}, \d}$ with $0<\ms<1$ and it is not in the approximately self-similar regime. Hence, designing a suitable hierarchy and closing the bootstrap arguments are more delicate. We need to divide the spacetime into two regions and deal with many new borderline terms. Moreover, finding a proper order to conduct the estimates for recovering the bootstrap constants is crucial. The newly designed hierarchy is given in Section \ref{subsection norms} and the order of estimates is explained in Section \ref{order of bootstrap}.

\item To carry out the hyperbolic estimates in \cite{AL} and \cite{An19}, in the spacetime region considered, we apply the below inequality ubiquitously
$$\f{\ub\at}{|u|}\leq \f{1}{\at}\ll 1.$$ 
The above inequality hinges on the largeness of the additional parameter $a$. While in this article our conclusions do not depend on the largeness of $a$ and we can set $a=1$, which advances the short-pulse method.  

\item {\color{black}In this paper, we can even set the short-pulse to be small in scale-critical norms. To fulfill the requirement \eqref{lower bound for C0 perturbation} in Theorem \ref{main thm 2 section 14}, we can choose
$$|\chih(-1, \ub, \o)|+|\nab_4\phi(-1, \ub, \o)-\nab_4\phi(-1, 0, \o)|\approx [\ln(\ln\f{1}{\ub})]^{-\f14}.$$
Note that the function $[\ln(\ln 1/\ub)]^{-\f14}\rightarrow 0$ as $\ub\rightarrow 0^+$. Hence, we can let $|\chih(-1, \ub, \o)|$ and $\nab_4\phi(-1, \ub, \o)-\nab_4\phi(-1, 0, \o)$ to be absolutely continuous functions. Their BV norms in $\ub$ variable (and other scale-critical norms) can be set small. While in all previous trapped surface formation works outside of spherical symmetry, the BV perturbation of $|\chih(-1, \ub, \o)|$ in $\ub$ variable requires to be large.  
}

\end{enumerate}

\end{remark}

\subsection{Emergence of Apparent Horizon}
 
Besides trapped surface formation, the rise of the apparent horizon in the gravitational collapse of Einstein's field equations was proved in \cite{An17} by the author and in \cite{AH} by the author and Han. In \cite{An17}, the following result was proved 

\begin{theorem*}[An \cite{An17}]\label{thm1.4} 
Under the same initial-data assumption as in Theorem \ref{thm1.3}, the Einstein vacuum equations admit a unique regular solution in the region $0< \ub \leq \d$ and $-1\leq u \leq -b\ub \at$. In addition, assuming
\begin{equation}\label{eq-definition-f}
 \int_0^{\ub}|\chih|^2 (-1,\ub', \o)d\ub'=f(\o, \ub)\ub a \quad  \mbox{for each }  0<\ub\leq \d,
\end{equation}
with $f(\o,\ub)$ being a smooth function  in $\mathbb S^2\times (0,\delta]$ satisfying $20/21\leq f(\o, \ub)\leq 22/21$.
Then there \textit{exists a unique} MOTS $M_{\ub}$ along each $\Hb_{\ub}$ $(0< \ub \leq \d)$. Furthermore, requiring
\begin{equation*}
a^{-\frac12}\|{{\color{black}\ub}}^{j}\nab^{j}_{e_4}\nab^{i}\chih\|_{L^{\infty}_{\ub}L^2(S_{-1,\ub})}\leq B \mbox{ for and } i,j\in\mathbb{Z}^+,
\end{equation*}
then the MOTSs $\{M_{\ub}\}$ together form a smooth apparent horizon.  
\end{theorem*}

\begin{figure}[h]
\begin{center}
\begin{tikzpicture}[scale=0.86]
\draw [white](0,2.3)--node [midway,sloped, below,black] {$O$}(-0.5,2.3);
\draw [white](0.6,1.9)--node [midway,sloped, above,black] {$S_{-\tilde{\ub} a, \tilde{\ub}}$}(1.5, 2.8);
\fill[white!70!black](0,2)--(1,1)--(2,2)--(0,2);
\draw [white](3,-1)-- node[midway, sloped, below,black]{$H_{-1}(u=-1)$}(4,0);
\draw [white](2.1,2.1)--node [midway,sloped, above,black] {$S_{-\delta a, \delta}$}(2.8,2.8);
\draw [white](2,2)--node [midway,sloped,above,black] {$\Hb_{\delta}(\ub=\delta)$}(4,0);
\draw [white](1,1)--node [midway,sloped, below,black] {$\Hb_{0}(\ub=0)$}(3,-1);
\draw [dashed] (0, 2)--(0, -4);
\draw [dashed] (0, -4)--(4,0)--(2,2);
\draw [dashed] (0,0)--(2,2);
\draw [dashed] (0,-4)--(2,-2);
\fill[white!70!black] (1,1)--(3,-1)--(4,0)--(2,2)--(1,1);
\fill[white!50!black] (0.5, 2)--(3.25,-0.75)--(3,-1)--(0.25, 1.75)--(0.5, 2);
\draw [white] (1.2, 1.3)--node [midway,sloped,above,black] {$\Hb_{\tilde{\ub}}(\ub=\tilde{\ub})$}(3.25,-0.75); 
\draw [thick] (0,2)--(3,-1);
\draw [thick] (1,1)--(3,-1)--(4,0)--(2,2)--(1,1);
\draw [thick] (0.5, 2)--(3.25,-0.75); 
\draw [thick] (0.25, 1.75)--(0.5, 2); 
\draw[dashed] (0, 2)--(2.1, 1.9);
\end{tikzpicture}
\end{center}
\caption{}\label{Figure 2} 
\end{figure}

The next work by An-Han generalized the above result to the fully anisotropic scenario.  

\begin{theorem*}[An-Han \cite{AH}]\label{thm1.5} With the same prescribed initial data as in Theorem \ref{thm1.3}, for each $0<\ub\leq \delta$, if we require $a^{-\frac12}\|{{\color{black}\ub}}^{j}\nab^{j}_{e_4}\nab^{i}\chih\|_{L^{\infty}_{\ub}L^2(S_{-1,\ub})}\leq B$ for any $i,j\in\mathbb{Z}^+$ and
\begin{equation}\label{lower bound thm1.5}
 \int_0^{\ub}|\chih|^2 (-1, \ub', \o) d\ub'=f(\o, \ub)\ub a,
\end{equation}
with $f(\o,\ub)$ satisfying 
\begin{equation*}\label{eq-assumption-f-0-1}
0\leq f\leq 1\quad\text{on }\mathbb S^2\times (0,\delta]\mbox{ and }  f(\cdot, \ub)\ge m\quad\text{on }B_{p}(\e)
\end{equation*}
for a point $p\in \mathbb S^2$ and constants $m\in (0,1)$ and $\e\in (0,\pi/2)$. Then in a large region of each $\Hb_{\ub}$ with $0< \ub \leq \d$,  there \textit{exists a unique} MOTS $M_{\ub}$ and the collection of MOTS $\{M_{\ub}\}$ emerge from the central singularity and form a smooth apparent horizon. Furthermore, if we require $|\ub\partial_{\ub}f(\o,\ub)|\leq a^{-\f13}$,  then the apparent horizon is a spacelike hypersurface, hence a dynamical horizon.
\end{theorem*}

In \cite{AH}, this theorem was further extended to the multi-valley scenario and concrete examples of dynamically formed apparent horizon with single and multi valleys are provided. In Figure \ref{Figure 3}, one can see the shape of MOTS along $\Hb_{\ub}$ with two valleys, which are bounded by the upper and lower barriers.

\begin{figure}[h]
\begin{center}
\begin{tikzpicture}[scale=0.85]
\draw (0,0.5) ellipse (0.875cm and 0.29cm);
\draw (0,-4) ellipse (2cm and 0.6cm);
\draw [white](1.2, -0.3)-- node[midway, sloped, above,black]{$\Hb_{\ub}$}(1.5, -0.3);
\draw [white](2.8, -2.2)-- node[midway, sloped, above,black]{upper barrier}(2.9, -2.2);
\draw [white](3.1, -3.3)-- node[midway, sloped, above,black]{MOTS $M_{\ub}(\o)$}(3.3, -3.3);
\draw [white](3.1, -3.6)-- node[midway, sloped, above,black]{lower barrier}(3.2, -3.6);
\draw [white] (1.85, -2.2)-- node[midway, sloped, below, black]{{\footnotesize A}} (1.9, -2.2);
\draw [white](2.2, -4.3)-- node[midway, sloped, above,black]{$p_1$}(2.4, -4.3);
\draw [white](-2.2, -4.3)-- node[midway, sloped, above,black]{$p_2$}(-2.4, -4.3);
\draw [white](-0.4, -5.2)-- node[midway, sloped, above,black]{$\mathbb{S}^2$}(0.4, -5.2);
\draw[thin] (-0.875,0.5)--(-2,-4);
\draw[thin] (0.875,0.5)--(2,-4);
\draw[thin] (-1,0)--(-2,-4);
\draw[thin] (1,0)--(2,-4);

\fill[white!90!black] (-0.5, -0.5) to [out=50,in=140] (0.6, -0.5) to [out=315,in=130] (1.6,-2.55) to [out=0, in=-60]  (1.625, -2.5) to [out=115, in=315] (0.6, -0.3) to [out=140, in=50] (-0.5, -0.3) to [out=230,in=65] (-1.5,-2) to [out=-20, in=180] (-1.45, -2.02) to [out=45, in=230] (-0.5, -0.5);
\draw [thick] (-0.5, -0.5) to [out=50,in=140] (0.6, -0.5); 
\draw [thick] (0.6,-0.5) to [out=315,in=130] (1.6, -2.55); 
\draw [thick] (1.6,-2.55) to [out=0, in=-60] (1.625, -2.5); 
\draw [dashed] (1.625, -2.5) to [out=115, in=315] (0.6, -0.3);
\draw [dashed] (0.6, -0.3) to [out=140, in=50] (-0.5, -0.3);
\draw [dashed] (-0.5,-0.3) to [out=230,in=65] (-1.5,-2); 
\draw [thick] (-1.5, -2) to [out=-20, in=180] (-1.45, -2.02); 
\draw [thick] (-1.45, -2.02) to [out=45, in=230] (-0.5, -0.5);

\fill[white!70!black] (0.4, -1) to [out=120,in=50] (-0.4, -1) to [out=230,in=15] (-1.575, -2.55) to [out=180, in=-20]  (-1.625,-2.5) to [out=50, in=230] (-0.5, -0.7)  to [out=50,in=120] (0.5, -0.7) to [out=310,in=135] (1.75, -3) to [out=30, in=-30] (1.7, -3.05) to [out=160,in=300] (0.4, -1);
\draw [thick] (0.4, -1) to [out=120,in=50] (-0.4, -1); 
\draw [thick] (-0.4,-1) to [out=230,in=15] (-1.575,-2.55); 
\draw [thick] (-1.575, -2.55) to [out=180, in=-20] (-1.625, -2.5); 
\draw[dashed] (-1.625, -2.5) to [out=50, in=230] (-0.5, -0.7);
\draw [dashed] (-0.5, -0.7) to [out=50,in=120] (0.5, -0.7); 
\draw [dashed] (0.5,-0.7) to [out=310,in=135] (1.75, -3); 
\draw [thick] (1.75,-3) to [out=30, in=-30] (1.7, -3.05); 
\draw [thick] (1.7,-3.05) to [out=160,in=300] (0.4, -1); 

\fill[white!90!black] (0.2, -1.5) to [out=120,in=50] (-0.2, -1.5) to [out=230,in=10] (-1.8,-3.55) to [out=180, in=-20] (-1.875,-3.5) to [out=40, in=230] (-0.2, -1.1) to [out=50,in=120] (0.2, -1.1) to [out=310,in=155] (1.875, -3.5) to [out=30, in=-30] (1.8, -3.55) to [out=170,in=300] (0.2, -1.5);
\draw [thick] (0.2, -1.5) to [out=120,in=50] (-0.2, -1.5); 
\draw [thick] (-0.2,-1.5) to [out=230,in=10] (-1.8,-3.55); 
\draw [thick] (-1.8, -3.55) to [out=180, in=-20] (-1.875, -3.5); 
\draw [dashed] (-1.875, -3.5) to [out=40, in=230] (-0.2, -1.1);
\draw [dashed] (-0.2, -1.1) to [out=50,in=120] (0.2, -1.1); 
\draw [dashed] (0.2,-1.1) to [out=310,in=155] (1.875, -3.5); 
\draw [thick] (1.875,-3.5) to [out=30, in=-30] (1.8, -3.55); 
\draw [thick] (1.8, -3.55) to [out=170,in=300] (0.2, -1.5); 

\draw [white](-0.05, -0.95)-- node[midway, sloped, below,black]{{\footnotesize B}}(0.05, -0.95);
\draw [white](-0.05, -1.91)-- node[midway, sloped, above,black]{{\footnotesize C}}(0.05, -1.91);

\draw[fill] (0,-1.02) circle [radius=0.05];
\draw[fill] (0,-1.43) circle [radius=0.05];
\draw[fill] (1.635, -2.55) circle [radius=0.05];
\end{tikzpicture}
\end{center}
\caption{}\label{Figure 3}
\end{figure}

\begin{remark}
The proof in \cite{AH} takes two steps:
\begin{enumerate}
\item[-] deriving apriori estimates via working through details of Moser's iteration and Schauder's estimates in the order of
$$C^0\rightarrow C^{0, \alpha} \rightarrow C^{1, \f13} \rightarrow C^{2, \f13} \mbox{ with } 0<\alpha\ll 1;$$
\item[-] conducting the method of continuity based on the explicit structure of the linearized operator.
\end{enumerate}
Compared with \cite{AH}, in this article the geometry along the initial incoming cone is completely different and the main terms in the elliptic equation for $R(\ub, \o)$ (with $u=R(\ub, \o)$ being the MOTS along $\Hb_{\ub}$) differentiate, which raises extra challenges. In \cite{AH} we employ the ansatz $R=\ub a e^{-\tp}$ to translate the elliptic equation for $R(\ub, \o)$ into an elliptic equation for $\tp$.  In this article, we design and utilize two forms of ansatz: 

\begin{enumerate}
\item $R=(\ub a)^{1-\ms}e^{-(1-\ms)\tp}$. We apply this form of $R(\ub, \o)$ to perform the apriori estimates via Moser's iteration and Schauder's estimates.
\item $R=(\ub a)^{1-\ms}\cp^{\f{1-\ms}{\ms}}$. We employ this ansatz to eliminate the additional nonlinear borderline term (rising from naked-singularity initial data) in the elliptic equation of $R(\ub, \o)$ and to conduct the method of continuity for the corresponding new linearized equation.
\end{enumerate} 
 
\end{remark}

\subsection{Difficulties and New Ingredients in the Proof}\label{new ingredients}
To prove Theorem \ref{main thm}, we divide the considered spacetime region into two parts:
\begin{itemize}
\item[-] the exterior $\{(u, \ub, \o)| \, -1 \leq u \leq u_1 \, \mbox{ and }\,  0\leq\ub\at\leq |u_1|\O^{2-\t\d}(u_1,0)\}$, 
\item[-] the interior $\{(u, \ub, \o)| \,\,\,\,\,\, u_1\leq u\leq 0 \, \mbox{ and }\, 0\leq\ub\at\leq |u|\O^{2-\t\delta}(u,0)\}$.
\end{itemize}
Here $|u_1|$ is a small constant to be fixed later. For notational simplicity, in the rest of this article, we set constant $\t\d$ to obey $0<\t\d\ll1$. In reality, if needed we can choose $\t\d$ to be any number satisfying $0<\t\d<1/2$ and the below proof still works.

\subsubsection{Hyperbolic Part} Compared with \cite{AL} by An-Luk, here more challenges arise. We list some of them and state several strategies:
\begin{itemize}

\item[-] For the position of MOTS along each $\Hb_{\ub}$, in contrast to the results in \cite{AL, An19, AH}, where the MOTS lies around $u=-\ub a$, in this paper the MOTS is around $u=-(\ub a)^{1-\ms}$ with $0<\ms<1$. The hyperbolic estimates in \cite{AL, An19, AH} are consistent with asymptotically self-similar bounds and crucially rely on 
$$\f{\ub\at}{|u|}\lesssim \f{1}{\at}\ll 1$$
with $a$ being a large parameter. While here the MOTS is not located in the asymptotically self-similar regime, and to obtain a desired instability result we would allow general (tiny) anisotropic perturbations and would set the parameter $a$ to be of size $1$. These bring many new difficulties, especially on how to design and prove the short-pulse hierarchy when $|u|$ is small.\\

\item[-] In this paper, we generalize our constructed ansatz in \cite{AL, An19} by inserting extra $|u|$ weights to the hierarchy. Employing the fact 
$$\f12\O(u,0) \leq \O(u,\ub, \theta_1, \theta_2) \leq 2\O(u,0) \quad \mbox{ and } \quad \O^2(u,0)=|u|^{\f{\ms}{1-\ms}},$$
the suitable $|u|$ weights are added via imposing the corresponding $\O$ weights. Especially, the $\O$-weight top-order energy estimates and top-order elliptic estimates are critical for our proof. In the energy estimates, because of the additional $\O$ weights, basic integration-by-part and pairing arguments are revised. In comparison to \cite{AL}, the top-order elliptic estimates are more subtle. As shown in Section \ref{secRicci} and Section \ref{elliptic estimates}, enhanced estimates for $\sum_{1\leq i\leq 4}\|u^i\nab^i(\O\tr\chi)\|_{L^2(\S)}$, $\|u^6\nab^5(\O\tr\chi)\|_{L^2_{\ub}L^2(\S)}$, {\color{black}precise estimates involving borderline terms for} $\|u^5\nab^5\etb\|_{L^2_{\ub}L^2(\S)}$, $\|u^6\nab^5\eta\|_{L^2_{\ub}L^2(\S)}$,  $\|u^5\nab^5(\tr\chib, \chibh)\|_{L^2_{\ub}L^2(\S)}$ are derived and are employed in an essential way. Subtle structures of the Einstein-scalar field system are explored. \\

\item[-] When bounding the nonlinear terms in the interior region where $u_1\leq u \leq 0$, we utilize the inequality
\begin{equation}\label{quotient small}
\f{\ub\at}{|u|\O^{\f32}}\leq \O^{\f12-\t\delta}(u,0)=|u|^{\f{\ms}{1-\ms}\cdot\f12-\t\delta}\leq |u_1|^{\f{\ms}{1-\ms}\cdot\f12-\t\delta} \ll 1,
\end{equation}
where $0<\t\d\ll1$ and $|u_1|\ll 1$. Notice that the largeness of $a$ is no longer required, but the power of $\O$ has to be discreetly tracked.\\

\item[-] To conduct energy estimates for the scalar field, to derive energy estimates for curvature components and to obtain elliptic estimates, we encounter terms, which only obey the upper bound
$$\f{\ub^{\f12}\af}{|u|^{\f12}\O}=\bigg(\f{\ub\at}{|u|\O^2}\bigg)^{\f12}.$$ 
By contract to \eqref{quotient small}, this quotient is no longer $\lesssim 1$. The terms satisfying only this bound serve as the new \textit{borderline} terms to be estimated. We then add the weights $\F$ to the bootstrap assumptions. Finding the correct order to retrieve the constants (which do not depend on other bootstrap assumptions) in these borderline estimates is decisive.
Additional $\F$ factors are incorporated in the top-order bootstrap assumptions for $\t{\M O}_{5,2}(\etb)$, $\t{\M O}_{5,2}(\O\tr\chib, \tr\chib)$, $\t{\M O}_{5,2}(\O\chibh, \chibh)$, $\M S(\nab_A\phi)$, $\underline{\M S}(\nab_3\phi)$, $\M R(K, \sigmac)$ and $\underline{\M R}(\beb)$ listed in Section \ref{subsection norms}. To improve the bootstrap assumption for each of these terms, the correctly designed order of arguments and enhanced estimates are vitally employed. \\

\item[-] Compared with An-Luk \cite{AL}, there are extra borderline terms arising from $\O\omb$ and $\partial_u\phi$. Along $\ub=0$, they satisfy
$$\O \omb(u,0)=\f{\ms}{4(1-\ms)}\cdot\f{1}{|u|} \quad \mbox{ and } \quad \partial_u\phi(u,0)=\sqrt{2}\sqrt{\f{\ms}{1-\ms}}\f{1}{|u|}.$$
They behave similarly to $\O\tr\chib$ and cause multiple new borderline terms. \\

\item[-] To carry out energy estimates, besides $\sigmac$ defined in \cite{AL}, we also use renormalized quantities $\b_A-\f12\nab_4\phi\nab_A\phi$ {\color{black} (equivalent to $\b_A-\f12 R_{4A}$)} and $\beb_A+\f12\nab_3\phi\nab_A\phi$ {\color{black} (equivalent to $\beb_A+\f12R_{3A}$)} {\color{black}to combine contributions of $\nab_3\b_A$, $\nab_3 R_{4A}$ or $\nab_4\beb_A$, $\nab_4 R_{3A}$ in null Bianchi equations \eqref{eq:null.Bianchi} together}.\footnote{A similar idea was used in \cite{AA} by An-Athanasiou for the Einstein-Maxwell system.} Furthermore, we also introduce and employ the renormalized Gaussian curvature $K-\f{1}{|u|^2}-\f14\nab^A\phi\nab_A\phi$. {\color{black}This new renormalization comes from the below identity for the Weyl curvature 
$$\nab^A W_{A434}=-\f12\nab_4 R_{43}+\f12\nab_3 R_{44}-\f{1}{12}\nab^A R(g_{A3}g_{44}-g_{A4}g_{43}).$$
It gives
$$\nab_4(\rho-\f14 R_{43}-\f{1}{12}R)=\f12\nab^A W_{A434}-\f14\nab_3 R_{44}+\mbox{non-top-order terms}$$
with $\rho=W_{3434}/4$. Further calculations via applying the constraint equation imply that
$$\rho-\f14 R_{43}-\f{1}{12}R=-(K-\f14\nab^A\phi\nab_A\phi)+\f12\chih\cdot\chibh-\f14\tr\chi\tr\chib.$$
For using vanishing initial data along $\Hb_0$, we then utilize the renormalized Gaussian curvature $K-\f{1}{|u|^2}-\f14\nab^A\phi\nab_A\phi$ to carry out the energy estimates.} In our estimates, some hidden structures of the Einstein-scalar field system are revealed. \\ 
 
\end{itemize}

\subsubsection{Elliptic Part} Because of the naked-singularity initial data, here the key structure of the main elliptic equation for MOTS is quite different from its counterpart in \cite{AH}. In this article, we design and utilize two forms of ansatz:

\begin{itemize}
\item[-] $R=(\ub a)^{1-\ms}e^{-(1-\ms)\tp}$. This is a generalization of \cite{AH}, where $R=\ub a e^{-\tp}$. We use this form to construct the anisotropic trapped surfaces and to conduct Moser's iteration and Schauder's estimates as in \cite{AH}.
\item[-] $R=(\ub a)^{1-\ms}\cp^{\f{1-\ms}{\ms}}$. The ansatz here is \textit{new} and is \textit{indispensable} for this paper.  We employ this form to carry out the key continuity method for the new linearized equation to solve for the MOTS.
\end{itemize} 
In below, we will give a more comprehensive explanation on the utilization of these two forms.\\

To construct the trapped and untrapped barriers and to derive the regularity estimates, we employ the ansatz $R=(\ub a)^{1-\ms}e^{-(1-\ms)\tp}$. We observe that 

\begin{itemize}
\item[-] Owning to Christodoulou's initial data, the constructions of the multi- and single-valley anisotropic trapped surfaces as in \cite{AH} and \cite{KLR} still hold. One can see the details in Section \ref{sec-Barriers}.  
\item[-] For the apriori regularity estimates, we find that the approach developed in \cite{AH} by An-Han is robust enough and can be adapted to this new setting. Consider the main elliptic equation \eqref{MOTS eqn}:  

\begin{equation*}
\begin{split}
0=&\Delta_{\gamma}\tp+1-\f12 f(\ub,\o)e^{\tp}-4(1-\ms)R\O\omb|\nab_{\gamma}\tp|^2\\
\end{split}
\end{equation*}
\begin{equation*}
\begin{split}
&-(1+\f{R}{2}\O\tr_g\chib)(1-\ms)|\nab_{\gamma}\tp|^2+2\gamma^{ij}\eta_i\partial_j\tp\\
&+[\f{1}{2(1-\ms)}R\O^{-1}\tr_g\chi-1+\f12 f(\ub,\o)e^{\tp}]\\
&-2(1-\ms)R\O\chibh\cdot\nab_{\gamma}\tp\cdot\nab_{\gamma}\tp.
\end{split}
\end{equation*}
On the right side, besides the main terms $\Delta_{\gamma}\tp+1-\f12 f(\ub,\o)e^{\tp}$ we denote
\begin{align*}
F_1:=&-4(1-\ms)R\O\omb|\nab_{\gamma}\tp|^2,\\
F_2:=&-(1+\f{R}{2}\O\tr_g\chib)(1-\ms)|\nab_{\gamma}\tp|^2+2\gamma^{ij}\eta_i\partial_j\tp\\
&+[\f{1}{2(1-\ms)}R\O^{-1}\tr_g\chi-1+\f12 f(\ub,\o)e^{\tp}]\\
&-2(1-\ms)R\O\chibh\cdot\nab_{\gamma}\tp\cdot\nab_{\gamma}\tp.
\end{align*}
In \cite{AH} we deal with the elliptic equation in the same form as above. But in \cite{AH} we get $|R\O\omb|\leq |u|\ub\at/|u|^2= \ub\at/|u|\leq 1/\at\ll1$. Hence the $F_1$ and $F_2$ terms obey $|F_1|+|F_2|\ll |D\tilde{\phi}|^2+1$. While in this paper, we only have $|R\O\omb|\lesssim |u|\cdot 1/|u|=1$ and it only holds $|F_1|+|F_2|\leq |D\tilde{\phi}|^2+1$. This $-4(1-\ms)R\O\omb|\nab_{\gamma}\tp|^2$ later serves as the main term in the linearization arguments. Luckily, as pointed out one line below $(4.3)$ in \cite{AH}, the approach developed in \cite{AH} via Moser's iteration is robust enough to cover the case 
$$|F_1|+|F_2|\leq c(|D \tp|^2+1) \quad \mbox{ with } \quad c=1.$$
Hence it can still be applied here to obtain the apriori regularity estimates.\\
\end{itemize}

The new ansatz $R=(\ub a)^{1-\ms}\cp^{\f{1-\ms}{\ms}}$ plays a vital role in linearization arguments. 

\begin{itemize}

\item[-] If we proceed with the ansatz $R=(\ub a)^{1-\ms}e^{-(1-\ms)\tp}$. The main elliptic equation for MOTS is then reduced to 
$$0=\Delta_{\gamma}\tp+1-\f12 f(\o,\ub)e^{\tp}-4(1-\ms)R\O\omb|\nab_{\gamma}\tp|^2+\mbox{small terms}.$$ 
In \cite{AH}, since $R\O\omb$ is small, the main terms in above equation are the first three terms on the right and the continuity argument via linearization can be carried out accordingly (without troubles from the $|\nab_{\gamma}\tp|^2$ term). While in this paper, due to the naked-singularity initial data, we have
$$-4(1-\ms)R\O \omb(u, \ub, \o)\approx -4(1-\ms)R\cdot\f{\ms}{4(1-\ms)}\cdot\f{1}{|u|}=-4\ms$$
being of size $1$. It is no longer negligible. Here, by introducing and employing the new ansatz $R=(\ub a)^{1-\ms}\cp^{\f{1-\ms}{\ms}}$, we eliminate this quadratic gradient borderline term and the equation of MOTS becomes
$$0=\Delta_{\gamma}\cp-\ms\cp+\f{\ms}{2}f(\o,\ub)\cp^{1-\f{1}{\ms}}+\mbox{ small terms}.$$
This enables us to complete the method of continuity and to construct the unique MOTS along each $\Hb_{\ub}$ in the considered spacetime region. \\
\end{itemize}

We then consider the position of the MOTS.
\begin{itemize}
\item[-] By the elliptic arguments in Section \ref{sec-Barriers}, along each $\Hb_{\ub}$ we have that the MOTS is bounded in between by an untrapped surface located at $u=-(2\ub a/3)^{1-\ms}$ and by a trapped surface lying within 
$$-(\ub a\cdot m\epsilon^2)^{1-\ms}\leq u \leq -(\ub a\cdot m \epsilon^{\tau+2})^{1-\ms}$$
with $\tau>2.2$. Note that, along the inner hypersurface $u=-(\ub a\cdot m\epsilon^{\tau+2})^{1-\ms}$, it holds 
\begin{equation}\label{star0}
\O(u,0)^{\f{2}{\ms}}=|u|\O^2(u,0)= |u|^{\frac{1}{1-\ms}}=\ub a\cdot m\epsilon^{\tau+2}.
\end{equation}
To guarantee that this hypersurface is located within the hyperbolic existence region, we require $\ub\leq |u|\O^{2-\t\delta}(u,0)$, which is equivalent to $|u|\O^2(u,0)\geq \ub \O^{\t\delta}(u,0)$. Thus, we need 
\begin{equation}\label{star1}
\ub a\cdot m \epsilon^{\tau+2}=|u|\O^2(u, 0)\geq \ub\O^{\t\delta}(u,0).
\end{equation}
The above inequality is equivalent to $m\epsilon^{\tau+2}\geq a^{-1}\O^{\t\delta}(u,0)$. Back to (\ref{star1}), together with \eqref{star0} we need
\begin{equation}\label{star2}
m \epsilon^{\tau+2}\geq a^{-1}(\ub a\cdot m\epsilon^{\tau+2})^{\f{\ms \t\delta}{2}}.
\end{equation}
Hence, if we let $\ub$ satisfy \eqref{star2}, then the apparent horizon constructed in this article stays in the obtained hyperbolic existence region.  Moreover, to fix the small parameter $|u_1|$, we only need to require 
\begin{equation}\label{choose u1}
|2\ub a/3|^{1-\ms}\leq |u_1| \quad \mbox{ and } \quad \O(u_1,0)=|u_1|^{\f{\ms}{1-\ms}}\ll 1. 
\end{equation}
By setting $\ub$ to be sufficiently small, clearly, we can choose $|u_1|$ verifying \eqref{choose u1}. In addition, if we set $a=1$ as in the statement of Theorem \ref{main thm}, we would deduce all the elliptic arguments in the interior region.\\

\end{itemize}

\subsubsection{Instability Results.} Our proven anisotropic trapped surface and apparent horizon formation criteria enable us to obtain instability theorems of high co-dimensions. 

\begin{itemize}
\item Allowing perturbations {\color{black}with finite $BV$ norms} along $u=-1$ with $\ub>0$, in Theorem \ref{co dimension instability theorem v1} we consider the space of initial data for $\chih(-1, \ub, \o)$ and $\partial_{\ub}\phi(-1, \ub, \o)$, that can be prescribed along $u=-1$ with $\ub> 0$. Given Christodoulou's naked-singularity data in \cite{Chr.2} along $\ub=0$, then for any $k\in \mathbb{Z}^+$, we can find smooth functions $f_1, f_2, ... , f_k$ and $g_1, g_2, ... , g_k$ with variables $\ub>0$ and $\o$, such that the below prescribed initial data for $|\chih|$ and $\partial_{\ub}\phi$ with $\ub>0$
\begin{equation*}
\begin{split}
&\bigg(|\chih|(-1, \ub, \o), \quad \partial_{\ub}\phi(-1, \ub, \o)\bigg)\\
=&\bigg(0+\lambda_1 f_1+\lambda_2 f_2+\cdot\cdot\cdot+\lambda_k f_k, \quad \partial_{\ub}\phi(-1, 0, \o)+\lambda_1' g_1+\lambda_2' g_2+\cdot\cdot\cdot+\lambda_k' g_k\bigg)
\end{split}
\end{equation*}
would lead to the trapped surface and the apparent horizon formation, when 
$$\sum_{1}^k \lambda_i^2+\sum_1^k {\lambda_i'}^2\neq 0 \quad \mbox{ with } \quad \lambda_k, \lambda_k'\in \mathbb{R}.$$ 

Moreover, if it holds 
\begin{equation*}
\begin{split}
&\bigg(0+\lambda_1 f_1+\lambda_2 f_2+\cdot\cdot\cdot+\lambda_k f_k, \quad \partial_{\ub}\phi(-1, 0, \o)+\lambda_1' g_1+\lambda_2' g_2+\cdot\cdot\cdot+\lambda_k' g_k\bigg)\\
=&\bigg(0+\t\lambda_1 f_1+\t\lambda_2 f_2+\cdot\cdot\cdot+\t\lambda_k f_k, \quad \partial_{\ub}\phi(-1, 0, \o)+\t\lambda_1' g_1+\t\lambda_2' g_2+\cdot\cdot\cdot+\t\lambda_k' g_k\bigg),
\end{split}
\end{equation*}
then it renders
$$\lambda_i=\t\lambda_i \quad \mbox{ and } \quad \lambda_i'=\t\lambda_i \quad \mbox{ with } \quad 1\leq i \leq k.$$
We may therefore say that, for any $k\in\mathbb{Z}^+$, Christodoulou's naked-singularity solution in \cite{Chr.2} has at least co-dimensional $2k$ nonlinear instability subject to outgoing BV characteristic perturbations at $(-1, 0, \o)$. \\

\item In Theorem \ref{co dimension instability theorem v1}, with $m_i$ and $m'_i$ being positive constants, the constructed $f_i$ and $g_i$ obey the requirement: for $\ub>0$
$$f_i(\o, \ub)\ge m_i\at \mbox{ on } B_{p_i}(\e_i)\subset D_i, \quad g_i(\o, \ub)\ge m_i'\at \mbox{ on } B_{p_i'}(\e_i')\subset D_i',$$
and
$$f_i(\o, \ub)=0, \quad  g_i(\o, \ub)=0 \quad \mbox{ for } \quad \ub\leq 0.$$
Hence, functions $f_i(\o, \ub)$ and $g_i(\o, \ub)$ are not continuous in $\ub$, when $\ub\rightarrow 0^+$. In Section \ref{C1 instability}, by exploiting the extra $\Omega$-weight gained in the hyperbolic estimates, we extend this instability theorem by allowing to choose continuous functions $f_i(\o, \ub)$ and $g_i(\o, \ub)$ for the $\ub$ variable with $\ub\geq 0$. In Theorem \ref{main thm 1.2}, we impose the below requirements for $f_i$ and $g_i$: 
$$\mbox{With } \,\, 1\leq i\leq k, \,\, C_1>2, \,\, \Lambda>0 \,\, \mbox{ and } \, \, g(\ub)=({C_1\Lambda+C_1})^{\f12}\cdot [\ln(\ln\f{1}{\ub})]^{-\f12},$$ 
we request the prescribed $f_i$ and $g_i$ to obey
$$0\leq f_i\leq g(\ub)^{\f12}\at \mbox{ on } \mathbb S^2\times (0,\delta], \quad f_i(\cdot, \ub)\ge m_i g(\ub)^{\f12}\at \mbox{ on } B_{p_i}(\e_i)\subset D_i ,$$
$$0\leq g_i\leq g(\ub)^{\f12}\at \mbox{ on } \mathbb S^2\times (0,\delta], \quad  g_i(\cdot, \ub)\ge m_i' g(\ub)^{\f12}\at \mbox{ on } B_{p_i'}(\e_i')\subset D_i'$$  
with points  $p_i\in D_i\subset \mathbb S^2$, $p_i'\in D_i'\subset \mathbb S^2$ and constants  $m_i, m_i'\in (0,1)$, $\e_i, \e_i'\in (0,\pi/2)$. Note that we can set $a=1$. And as $\ub\rightarrow 0^+$, it holds $$\lim_{\ub\rightarrow 0^+}g(\ub)^{\f12}=0.$$ 
Therefore, we can choose continuous-in-$\ub$ and smooth-in-$\o$ functions $f_i, g_i, \chih, \partial_{\ub}\phi$ along $u=-1$ from $\ub=0$ to prove the corresponding high co-dimensional instability results stated in Theorem \ref{main thm 1.2}.

\end{itemize}

\subsection{Related Results}

\subsubsection{Trapped Surfaces, MOTS and Short-Pulse Method}
After Christodoulou's monumental work \cite{Chr:book}, a systematical approach by introducing signature for the short pulse was developed by Klainerman-Rodnianski in \cite{KR:Trapped} to simplify and to extend \cite{Chr:book} in a finite region. Yu in \cite{Yu2} extended Klainerman-Rodnianski's approach to the Einstein-Maxwell system. The author in \cite{An12} introduced a new signature, called the signature for decay rates, and extended Klainerman-Rodnianski's results to past null infinity. Based on the hyperbolic estimates in \cite{KR:Trapped}, Klainerman-Luk-Rodnianski in \cite{KLR} found a fully anisotropic mechanism for trapped surface formation. 

In \cite{AL} the author and Luk designed a new set of short-pulse hierarchies and proved the first scale-critical trapped surface formation criterion for the Einstein vacuum equations. With only the signature for decay rates, in \cite{An19} the author gave a short proof of the scale-critical trapped surface formation in the far-field regime and related the short-pulse hierarchies to the peeling properties of the gravitational waves. In \cite{AA}, the author and Athanasiou further extended the new approach in \cite{An19} to the Einstein-Maxwell system. With spherically symmetric singular initial data prescribed along the incoming null hypersurface, Li and Liu in \cite{LL2} studied the Einstein-scalar field system and proved an almost scale-critical trapped surface formation criterion. Their result showed that the spherically symmetric singularities considered in \cite{LL2} are unstable (as a result of trapped surface formation) subject to close-to-isotropic gravitational perturbations and the instability in \cite{LL2} is in the sense of the first category, slightly weaker than co-dimension $1$. Very recently, the methods in \cite{An19} by the author and in \cite{AA} by the author and Athanasiou are also extended to the Einstein-Yang-Mills system and the Einstein-scalar field system, by Athanasiou-Mondal-Yau \cite{AMY} and Zhao-Hilditch-Kroon \cite{ZHK}, respectively. 

Based on the hyperbolic estimates in \cite{AL}, in the short-pulse regime, the MOTS along each incoming null hypersurface and the corresponding apparent horizon are constructed in \cite{An17} by the author and in \cite{AH} by the author and Han for the general anisotropic scenarios. For solving the MOTS and the apparent horizon based on spacelike foliations, interested readers are referred to \cite{AEM, AMS, AMS08, AM} by Andersson, Eichmair, Mars, Metzger, Simon and references therein.

\subsubsection{Weak Cosmic Censorship within Spherical Symmetry and Naked Singularity Formation}
In a series of works \cite{Chr.1}-\cite{Chr.3}, Christodoulou achieved the proof of weak cosmic censorship for the spherically symmetric Einstein-(real) scalar field system. After constructing the first example of the naked-singularity solution under spherical symmetry, he ingenuously obtained $4$-integrated theorems:
\begin{enumerate}
\item BV extension principle in triangle and in rectangle regions; 
\item a sharp trapped surface formation criterion; 
\item the first instability theorem;
\item the second instability theorem.
\end{enumerate}
Christodoulou's arguments in \cite{Chr.1}-\cite{Chr.3} rely on the special structures of the spherically symmetric Einstein-(real) scalar field system and the associated monotonicity properties. His $4$-step argument is unified. Each step is exactly matched with the other steps. In recent years, each of Christodoulou's steps has been separately extended. In \cite{An-Lim} the author and Lim extended \cite{Chr.1} to the spherically symmetric Einstein-Maxwell-charged (complex) scalar field system. See also the extension in \cite{Costa} by Costa to the spherically symmetric Einstein-(real) scalar field system with positive cosmological constant. For the instability argument of the Einstein-(real) scalar field system, Li-Liu in \cite{LL1} gave a proof with apriori estimates. Their instability theorem is slightly different from Christodoulou's. 

Extending the proof of the weak cosmic censorship to the other spherically symmetric Einstein-matter field is still very challenging.  The Einstein-Maxwell-charged (complex) scalar field system is the next to be considered. However, the corresponding weak cosmic censorship has remained open since the late 1990s. Several difficulties are easily seen:

\begin{itemize}
\item Due to the presence of non-constant charge $Q$, almost all previously employed monotonicity formulas fail to hold. These monotonicity formulas are critically used in Christodoulou's BV extension/trapped surface formation/instability arguments when $Q=0$.  A new strategy to incorporate non-constant charge $Q$ remains to be developed. 
\item For the trapped surface formation criterion, with $\eta$ being the initial mass input and $\d$ being the  small deformation parameter, Christodoulou needs $\eta\gtrsim \delta\log\f{1}{\delta}$. Following Christodoulou's method, for the charged case, in \cite{An-Lim} the author and Lim obtained a result requiring $\eta\gtrsim \d^{1-\f{\o}{2}}(>\delta\log{\f{1}{\delta}})$ with $0<\o\ll 1$. But Christodoulou's ``blue-shift" $\gamma$ in the instability argument is around $\log{\f{1}{\delta}}$ rather than $\d^{-\f{\o}{2}}$. How to reconcile?  
\item Christodoulou's way of proving his instability theorems is very delicate. Can one obtain the same sharp arguments as Christodoulou did?
\end{itemize}
In a new preprint \cite{An-Tan}, the author and Tan address these questions and we show that \textit{the weak cosmic censorship still holds (in the sense of Christodoulou) for the gravitational collapse of the spherically symmetric Einstein-Maxwell-charged (complex) scalar field system.}\\

\noindent Recently, there is also the exciting development of constructing the naked-singularity solutions for Einstein's field equations. Under spherical symmetry and self-similar ansatz, in \cite{An-ZhangXF} the author and Zhang showed naked-singularity formation for Einstein vacuum equations in higher dimensions, and in \cite {GHJ} Guo-Hadzic-Jang extended \cite{Chr.2} to the Einstein-Euler system. Outside spherical symmetry, important progress has been made by Rodnianski-Shlapentokh-Rothman in \cite{R-S} and by Shlapentokh-Rothman in \cite{YS}. There they prove the first naked-singularity formation result for the Einstein vacuum equations in $3+1$ dimensions.

\subsection{Organizations of the Paper} 
The paper is organized as follows. In Section \ref{setting} we present the settings. In Section \ref{secbasic}, we derive the preliminary estimates. In Section \ref{secRicci}, we provide the $L^2(\S)$ estimates for Ricci coefficients. In Section \ref{scalar field}, we derive the $L^2(\S)$ estimates for 
the scalar field. In Section \ref{energy scalar field}, we prove energy estimates for the scalar field. In Section \ref{elliptic estimates}, we establish the elliptic estimates. In Section \ref{energy estimate curvature}, we conduct the energy estimates for curvature components. In Section \ref{sec-Barriers}, we construct the sub and super barriers for the MOTS. In Section \ref{sec-Schauder-estimates}, we carry out Moser's iteration and the Schauder estimates. In Section \ref{linearization}, we demonstrate the crucial linearization and continuity arguments. In Section \ref{sec-existence-solutions}, we obtain the existence and properties for the apparent horizon.  In Section \ref{section instability theorems}, we prove the instability theorems.  

\subsection{Acknowledgement}
XA would like to thank Yu Deng and Willie Wong for valuable correspondences. XA is supported by MOE Tier 1 grants A-0004287-00-00, A-0008492-00-00 and MOE Tier 2 grant A-8000977-00-00.

\section{Setting}\label{setting}
\subsection{Double Null Coordinates}\label{secdnf} 
In this paper, we adopt the double null coordinates $(u,\ub, \theta^1, \theta^2)$ for $(\M M, g)$. For $A,B=1,2$, the spacetime metric $g$ takes the form
\begin{equation}\label{metric ansatz}
g=-2\Omega^2(du\otimes d\ub+d\ub\otimes du)+\sg_{AB}(d\th^A-d^Adu)\otimes (d\th^B-d^Bdu). 
\end{equation}
Here $u$ and $\ub$ are optical functions satisfying 
$$g^{\mu\nu}\partial_\mu u\partial_\nu u=0,\quad g^{\mu\nu}\partial_\mu\ub\partial_\nu \ub=0.$$
We require both $u$ and $\ub$ are increasing towards the future. The level sets of $u$ and $\ub$ are called $H_u$ and $\Hb_{\ub}$, respectively. They are outgoing and incoming null cones. The intersections of $H_u$ and $\Hb_{\ub}$ are topologically $2$-shpere and are denoted as $\S$. {\color{black}In later sections, along $\Hb_{\ub}$ we consider the $2$-sphere $u=-R(\ub, \theta_1, \theta_2)$ with $R(\ub, \theta_1, \theta_2)$ being a given function, and for simplicity we will write this $2$-sphere as $S_{-R, \ub}$.}

\subsection{Equations}\label{seceqn} 
The equations studied here are the Einstein-scalar field system: 
\begin{equation*}
\begin{split}
&R_{\mu\nu}-\f12 Rg_{\mu\nu}=T_{\mu\nu},\\
&T_{\mu\nu}=D_{\mu}\phi D_{\nu}\phi-\f12g_{\mu\nu}D^{\lambda}\phi D_{\lambda}\phi.
\end{split}
\end{equation*} 

\noindent Utilizing the null frame $e_3$, $e_4$ and an frame ${e_1,e_2}$ tangent to the 2-spheres $S_{u,\ub}$, with the indices $A,B\in\{1,2\}$, we define the Ricci coefficients: 
 \begin{equation}
\begin{split}
&\chi_{AB}=g(D_A e_4,e_B),\, \,\, \quad \chib_{AB}=g(D_A e_3,e_B),\\
&\eta_A=-\frac 12 g(D_3 e_A,e_4),\quad \etab_A=-\frac 12 g(D_4 e_A,e_3),\\
&\h\omega=-\frac 14 g(D_4 e_3,e_4),\quad\,\,\, \omegab=-\frac 14 g(D_3 e_4,e_3),\\
&\zeta_A=\frac 1 2 g(D_A e_4,e_3),
\end{split}
\end{equation}
where $D_A=D_{e_{A}}$. {\color{black}Note that here we use $\h\omega$ to represent the Ricci coefficient and this is because we save $\omega$ to stand for the angular variable on $\mathbb{S}^2$.} We also denote $\nab$ to be the 
induced covariant derivative operator on $S_{u,\ub}$ and let $\nab_3$, $\nab_4$ be
the projections to $S_{u,\ub}$ of the covariant derivatives $D_3$, $D_4$. By above definition, it also holds
\begin{equation}
\begin{split}
&\h\omega=-\frac 12 \nab_4 (\log\Omega),\qquad \omegab=-\frac 12 \nab_3 (\log\Omega),\\
&\eta_A=\zeta_A +\nab_A (\log\Omega),\quad \etab_A=-\zeta_A+\nab_A (\log\Omega).
\end{split}
\end{equation}

In $3+1$ dimensions, for $i,k,l,m=1,2,3,4$ we further decompose the below Weyl tensor
\begin{equation*}
W_{iklm}:=R_{iklm}+\f{1}{2}(R_{im}g_{kl}-R_{il}g_{km}+R_{kl}g_{im}-R_{km}g_{il})+\f{1}{6}R(g_{il}g_{km}-g_{im}g_{kl})
\end{equation*}
and define the following null curvature components
 \begin{equation}
\begin{split}
\a_{AB}&=W(e_A, e_4, e_B, e_4),\quad \, \,\,   \ab_{AB}=W(e_A, e_3, e_B, e_3),\\
\b_A&= \frac 1 2 W(e_A,  e_4, e_3, e_4) ,\quad \bb_A =\frac 1 2 W(e_A,  e_3,  e_3, e_4),\\
\rho&=\frac 1 4 W(e_4,e_3, e_4,  e_3),\quad \sigma=\frac 1 4  \,^*W(e_4,e_3, e_4,  e_3),
\end{split}
\end{equation}
where $\, ^*W$ is the Hodge dual of $W$.

We further set $\phi^{(1)}\cdot\phi^{(2)}$ to stand for a contraction of the tensor product of $\phi^{(1)}$ and $\phi^{(2)}$ with respect to the metric $\sg$ on $\S$. We also denote
$$(\phi^{(1)}\hot\phi^{(2)})_{AB}:=\phi^{(1)}_A\phi^{(2)}_B+\phi^{(1)}_B\phi^{(2)}_A-\delta_{AB}(\phi^{(1)}\cdot\phi^{(2)}) \quad\mbox{for one forms $\phi^{(1)}_A$, $\phi^{(2)}_A$,}$$
$$\phi^{(1)}\wedge\phi^{(2)}:=\eps^{AB}(\sg^{-1})^{CD}\phi^{(1)}_{AC}\phi^{(2)}_{BD}\quad\mbox{for symmetric two tensors $\phi^{(1)}_{AB}$, $\phi^{(2)}_{AB}$},$$
where $\eps$ is the volume form associated to the metric $\sg$.
For totally symmetric tensors, we also employ the below $\div$ and $\curl$ operators
$$(\div\phi)_{A_1...A_r}:=\nabla^B\phi_{BA_1...A_r},$$
$$(\curl\phi)_{A_1...A_r}:=\eps^{BC}\nabla_B\phi_{CA_1...A_r}.$$
In addition, we define the trace to be
$$(\mbox{tr}\phi)_{A_1...A_{r-1}}:=(\sg^{-1})^{BC}\phi_{BCA_1...A_{r-1}}.$$

Let $\chih$ and $\chibh$ be the traceless parts of $\chi$ and $\chib$ respectively. Employing above notations, rewriting $R_{\mu\nu}-\f12 R g_{\mu\nu}=T_{\mu\nu}$ according to the null frame $\{e_3, e_4, e_A, e_B\}$ with $A,B=1,2$, we have that $\chi$ and $\chib$ obey the following null structure equations: 
\begin{equation}
\label{null.str1}
\begin{split}
\nab_4 \trch+\frac 12 (\trch)^2&=-|\chih|^2-2\h\omega \trch-D_4\phi D_4\phi,\\
\nab_4\chih_{AB}+\trch \chih_{AB}&=-2 \h\omega \chih_{AB}-\alpha_{AB},\\
\nab_3 \trchb+\frac 12 (\trchb)^2&=-2\omegab \trchb-|\chibh|^2-D_3\phi D_3\phi,\\
\nab_3\chibh_{AB} + \trchb\,  \chibh_{AB}&= -2\omegab \chibh_{AB} -\alphab_{AB},\\
\nab_4 \trchb+\frac1 2 \trch \trchb &=2\h\omega \trchb +2\rho- \chih\cdot\chibh +2\div \etab +2|\etab|^2-\f13D_3\phi D_4\phi+\f13D_A\phi D^A\phi,\\
\nab_4\chibh_{AB} +\frac 1 2 \trch \chibh_{AB}&=(\nab\widehat{\otimes} \etab)_{AB}+2\h\omega \chibh_{AB}-\frac 12 \trchb \chih_{AB} +(\etab\widehat{\otimes} \etab)_{AB}+D_A\phi D_B\phi-\f12 g_{AB}D_C\phi D^C\phi,\\
\nab_3 \trch+\frac1 2 \trchb \trch &=2\omegab \trch+2\rho- \chih\cdot\chibh+2\div \eta+2|\eta|^2-\f13D_3\phi D_4\phi+\f13 D_A\phi D^A\phi,\\
\nab_3\chih_{AB}+\frac 1 2 \trchb \chih_{AB}&=(\nab\widehat{\otimes} \eta)_{AB}+2\omegab \chih_{AB}-\frac 12 \trch \chibh_{AB} +(\eta\widehat{\otimes} \eta)_{AB}+D_A\phi D_B\phi-\f12g_{AB}D_C\phi D^C\phi.
\end{split}
\end{equation}
And the Ricci coefficients $\eta_A, \etb_A, \omb, \h\o$ obey the remaining null structure equations:
\begin{equation}\label{null.str2}
\begin{split}
\nabla_4\eta_A&=-\chi_{AB}\cdot(\eta-\etab)_B-\b_A-\f12D_A\phi D_4\phi,\\
\nabla_3\etab_A &=-\chib_{AB}\cdot (\etab-\eta)_B+\bb_A-\f12D_A\phi D_3\phi,\\
\nabla_4\omegab&=2\h\omega\omegab\Red{- \eta_A\cdot\etab^A+\f 12|\eta|^2}+\frac 12 \rho+\f16 D_3\phi D_4\phi+\f{1}{12}D_A\phi D^A\phi,\\
\nabla_3\h\omega&=2\h\omega\omegab-\Red{\eta_A\cdot\etab^A+\f 12|\etab|^2}+\frac 12 \rho+\f16 D_3\phi D_4\phi+\f{1}{12}D_A\phi D^A\phi.
\end{split}
\end{equation}
On $\S$, with the null frame, the Gauss-Codazzi equations boil down to the constraint equations: 
\begin{equation}
\label{null.str3}
\begin{split}
\div\chih_A=&\frac 12 \nabla_A \trch - \frac 12 (\eta_B-\etab_B)\cdot (\chih_{BA} -\frac 1 2 \trch g_{BA}) -\beta_A+\f12D_A\phi D_4\phi,\\
\div\chibh_A=&\frac 12 \nabla_A \trchb + \frac 12 (\eta_B-\etab_B)\cdot (\chibh_{BA}-\frac 1 2   \trchb g_{BA}) +\betab_A+\f12 D_A\phi D_3\phi,\\
\curl\eta =&-\curl\etab=\sigma +\frac 1 2\chibh \wedge\chih,\\
K=&-(\rho+\f16 R)-\f14\tr\chi \tr\chib+\f12\chih\cdot \chibh+\f12 g^{AB}D_A\phi D_B\phi\\
=&-\rho-\f14\tr\chi \tr\chib+\f12 \chih\cdot\chibh+\f16 D_3\phi D_4\phi+\f13 D_A\phi D^A\phi.
\end{split}
\end{equation}
Here $K$ is the Gauss curvature of the spheres $S_{u,\ub}$. \\  

For the system \eqref{1.1}, the scalar field $\phi$ obeys $\Box_g\phi=0$. And we rewrite the wave equation as a transport equation
\begin{equation*}
\begin{split}
0=\Box_g\phi=&g^{\lambda\mu}D_{\lambda}D_{\mu}\phi=g^{\lambda\mu}D_{\lambda}(D_{\mu}\phi)-g^{\lambda\mu}D_{D_{\lambda}e_\mu}\phi\\
=&g^{34}D_3(D_4\phi)-g^{34}D_{D_3 e_4}\phi+g^{43}D_4(D_3\phi)-g^{43}D_{D_4 e_3}\phi+g^{AB}D_A(D_B\phi)-g^{AB}D_{D_A e_B}\phi\\
=&-\f12 D_3 D_4\phi-\f12 D_4 D_3\phi+g^{AB}D_A(D_B\phi)-g^{AB}D_{\nab_A e_B+\f12\chi_{AB}e_3+\f12\chib_{AB}e_4}\phi\\
=&-D_3 D_4\phi+\Delta_g\phi-\f12\tr\chi e_3\phi-\f12\tr\chib e_4\phi.
\end{split}
\end{equation*}
It hence holds
\begin{equation}\label{e3e4phi}
D_3 D_4\phi+\f12\tr\chi D_3\phi+\f12\tr\chib D_4\phi=g^{AB}\nab_A\nab_B\phi.
\end{equation}
For convenience, we also employ the following two equivalent forms
(see Section \ref{energy scalar field}) 
\begin{equation}\label{e3e4phi2}
(\O e_3)(\O e_4\phi)+\f12{\O\tr\chib}\O e_4\phi=\O^2\Delta_g \phi-\f12{\O\tr\chi}\O e_3\phi{\color{black}+2\O^2\eta^A e_A\phi},
\end{equation}

\begin{equation}
(\O e_4)(\O e_3\phi)+\f12{\O\tr\chib}\O e_4\phi=\O^2\Delta_g \phi-\f12{\O\tr\chi}\O e_3\phi{\color{black}+2\O^2\etb^A e_A\phi}.
\end{equation}

Using the property of covariant derivatives, the scalar field $\phi$ also satisfies
\begin{equation}
\begin{split}
\nab_{\O e_3}\nab_A \phi=&e_A(\O e_3\phi)-\O {\chib_{A}}^B e_B\phi,\\
\nab_{\O e_4}\nab_A \phi=&e_A(\O e_4\phi)-\O {\chi_{A}}^B e_B\phi.
\end{split}
\end{equation} 

Recall that in $3+1$ dimensions for $i,k,l,m=1,2,3,4$, the Weyl curvatures obey  
\begin{equation}\label{nabla Weyl}
\begin{split}
\nab^iW_{iklm}=&-\f12\nab_m R_{lk}+\f12\nab_l R_{mk}-\f{1}{12}(\nab_l R g_{mk}-\nab_m R g_{lk}).
\end{split}
\end{equation}
Employing the null frame $\{e_3, e_4, e_A, e_B\}$ with $A,B=1,2$, we rewrite the above \eqref{nabla Weyl} as the below null Bianchi equations: 
\begin{equation}
\label{eq:null.Bianchi} 
\begin{split}
&\nab_3\alpha+\frac 12 \trchb \alpha=\nabla\hot \beta+ 4\omegab\alpha-3(\chih\rho+^*\chih\sigma)+
(\zeta+4\eta)\hot\beta+\f14(D_3 R_{44}-D_4 R_{43})g_{AB}+\f16\nab_4 R g_{AB},\\
&\nab_4\beta+2\trch\beta = \div\alpha - 2\h\omega\beta +  \eta\cdot \alpha-\f12(D_A R_{44}-D_4 R_{4A}),\\
&\nab_3\beta+\trchb\beta=\nabla\rho + 2\omegab \beta +^*\nabla\sigma +2\chih\cdot\betab+3(\eta\rho+^*\eta\sigma)+{\color{black}\f12\nab_A(g^{CD}R_{CD})-\nab^B R_{BA}+\f16\nab_A R+\f12\nab_3 R_{4A}},\\
&\nab_4\sigma+\frac 32\trch\sigma=-\div^*\beta+\frac 12\chibh\cdot ^*\alpha-\zeta\cdot^*\beta-2\etab\cdot
^*\beta-\f14(D_{\mu}R_{4\nu}-D_{\nu}R_{4\mu}){\epsilon^{\mu\nu}}_{34},\\
&\nab_3\sigma+\frac 32\trchb\sigma=-\div ^*\betab+\frac 12\chih\cdot ^*\alphab-\zeta\cdot ^*\betab-2\eta\cdot 
^*\betab+\f14(D_{\mu}R_{3\nu}-D_{\nu}R_{3\mu}){\epsilon^{\mu\nu}}_{34},\\
&\nab_4\rho+\frac 32\trch\rho=\div\beta-\frac 12\chibh\cdot\alpha+\zeta\cdot\beta+2\etab\cdot\beta-\f14(D_3 R_{44}-D_4 R_{43})+\f{1}{12}\nab_4 R,\\
&\nab_3\rho+\frac 32\trchb\rho=-\div\betab- \frac 12\chih\cdot\alphab+\zeta\cdot\betab-2\eta\cdot\betab+\f14(D_3 R_{34}-D_4 R_{33})+\f{1}{12}\nab_3 R,\\
&\nab_4\betab+\trch\betab=-\nabla\rho +^*\nabla\sigma+ 2\h\omega\betab +2\chibh\cdot\beta-3(\etab\rho-^*\etab\sigma){\color{black}-\f12\nab_A(g^{CD}R_{CD})+\nab^B R_{BA}-\f16\nab_A R-\f12\nab_4 R_{3A} },\\\
&\nab_3\betab+2\trchb\,\betab=-\div\alphab-2\omegab\betab+\etab \cdot\alphab+\f12(D_A R_{33}-D_3 R_{3A}),\\
&\nab_4\alphab+\frac 12 \trch\alphab=-\nabla\hot \betab+ 4\h\omega\alphab-3(\chibh\rho-^*\chibh\sigma)+
(\zeta-4\etab)\hot \betab+\f14(D_4 R_{33}-D_3 R_{34})g_{AB}+\f16\nab_3 R g_{AB}.
\end{split}
\end{equation}
Here, $^*$ denotes the Hodge dual on $S_{u,\ub}$.

In our proof, we utilize the null Bianchi equations to deduce energy estimates. We will employ the renormalized quantities\footnote{\label{renormalized K}For more general incoming initial data, the renormalized quantity for $K$ is $K-K(u,0)-\f14\nab^A\phi\nab_A\phi$.}
\begin{equation*}
\begin{split}
&\b_A-\f12\nab_4\phi\nab_A\phi, \quad K-\f{1}{|u|^2}-\f14\nab^A\phi\nab_A\phi,\\
&\sigmac=\sigma+\frac 12 \chih\wedge\chibh, \quad \beb_A+\f12\nab_3\phi\nab_A\phi.
\end{split}
\end{equation*}
These renormalized quantities satisfy the below renormalized null Bianchi equations:

\begin{equation}
\label{eq:null.Bianchi2}
\begin{split}
\nab_3(\beta_{\cdot}-\f12\nab_4\phi\nab_{\cdot}\phi)_A+\trchb\beta_A=&-\nabla K  +^*\nabla\sigmac + 2\omegab \beta+2\chih\cdot\betab\Red{-}3(\eta K\Red{-}^*\eta\sigmac)\\
&+\frac 1 2(\nabla(\chih\cdot\chibh)-^*\nabla(\chih\wedge\chibh))-\frac 34 \eta\trch\trchb\\
&+{\color{black}\f32}(\eta\chih\cdot\chibh-^*\eta\chih\wedge\chibh)-\frac 14 (\nab\trch \trchb+\trch\nab\trchb)\\
&+\nab_A(g^{CD} R_{CD})-\nab^B R_{BA}\\
&{\color{black}+\f12\eta_A\nab_3\phi\nab_4\phi+\eta_A\nab_C\phi\nab^C\phi.}
\end{split}
\end{equation}
\begin{equation}
\begin{split}
\nab_4\sigmac+\frac 32\trch\sigmac=&-\div^*(\beta-\f12\nab_4\phi\nab\phi)-\zeta\cdot^*\beta-2\etab\cdot
^*\beta+\frac 12 \chih\cdot^*(\nab\widehat{\otimes}\etab)+\frac 12 \chih\cdot^*(\etab\widehat{\otimes}\etab)\\
&-\nab_1 R_{42}+\nab_2 R_{41}+\f12(\nab_A\phi \nab_B\phi-\f12 g_{AB}\nab_C\phi \nab^C\phi)\wedge \chih^{AB},\\
\end{split}
\end{equation}
\begin{equation}
\begin{split}
&\nab_4 (K-\f14\nab^A\phi \nab_A\phi)+\trch (K-\f14\nab^A\phi \nab_A\phi)\\
=&-\div\beta-\zeta\cdot\beta-2\etab\cdot\beta+\frac 12 \chih\cdot\nab\widehat{\otimes}\etab+\frac 12 \chih\cdot(\etab\widehat{\otimes}\etab)\\
&-\frac 12 \trch\div\etab-\frac 12\trch |\etb|^2+\f12\etb^A\nab_A\phi\nab_4\phi+\f12\Delta_g\phi\nab_4\phi\\
&+\f12(\nab_A\phi\nab_B\phi-\f12g_{AB}\nab_C\phi\nab^C\phi)\cdot\chih^{AB},\\
\end{split}
\end{equation}
\begin{equation}
\begin{split}
\nab_3\sigmac+\frac 32\trchb\sigmac=&-\div ^*(\betab+\f12\nab_3\phi\nab\phi)-\zeta\cdot ^*\betab-2\eta\cdot^*\betab+\frac 12 \chibh\cdot^*(\nab\widehat{\otimes}\eta)+\frac 12 \chibh\cdot^*(\eta\widehat{\otimes}\eta)\\
&+\nab_1 R_{32}-\nab_2 R_{31}+\f12\chibh^{AB}\wedge(\nab_A\phi \nab_B\phi-\f12 g_{AB}\nab_C\phi \nab^C\phi),\\
\end{split}
\end{equation}

\begin{equation}
\begin{split}
&\nab_3 (K-\f{1}{|u|^2}-\f14\nab^A\phi\nab_A\phi)+\f32\trchb (K-\f{1}{|u|^2}-\f14\nab^A\phi\nab_A\phi)\\
=& \div(\betab+\f12\nab_3\phi\nab\phi)-\f12\nab^A\nab_3\phi\nab_A\phi-\zeta\cdot\betab+2\eta\cdot\betab\\
&+\frac 12 \chibh\cdot\nab\widehat{\otimes}\eta+\frac 12 \chibh\cdot(\eta\widehat{\otimes}\eta)-\frac 12 \trchb |\eta|^2+\f12\eta^A\nab_A\phi\nab_3\phi\\
&+\f12(\nab_A\phi\nab_B\phi-\f12 g_{AB}\nab_C\phi\nab^C\phi)\cdot\chibh^{AB}-\f18\tr\chib\nab^A\phi\nab_A\phi\\
&+\f12\tr\chib(-\div \eta+\Kt)-\f{\Omega^{-1}}{|u|^2}(\O\tr\chib+\f{2}{|u|}), 
\end{split}
\end{equation}

\begin{equation} 
\begin{split}
\nab_4(\betab_{\cdot}+\f12\nab_3\phi\nab_{\cdot}\phi)_A+\trch\betab_A=&\nabla (K-\f{1}{|u|^2}-\f14\nab^A\phi \nab_A\phi)+^*\nabla\sigmac+ 2\omega\betab +2\chibh\cdot\beta+3(\Red{-}\etab K+^*\etab\sigmac)\\
&+\nabla(-\f14\nab^A\phi \nab_A\phi+\f14\tr\chi \tr\chib-\f12\chih\cdot\chibh)-\f12{^*\nabla}(\chih\wedge\chibh)\\
&+\nab_A(g^{CD}R_{CD})+\nab^B R_{BA}\\
&{\color{black}-\f12\etb_A\nab_3\phi\nab_4\phi-\etb_A\nab_C\phi\nab^C\phi.}
\end{split}
\end{equation}

\begin{remark} 
In below, if no further clarification, we employ the capital Latin letters $A\in \{1,2\}$ as frame indices on the spheres $S_{u,\ub}$, and use Greek letters $\mu\in\{1,2,3,4\}$ as frame indices in the whole 3+1 dimensional spacetime.
\end{remark}

\subsection{Naked Singularity Initial Data} 
In \cite{Chr.2}, Christodoulou constructed a spherically-symmetric naked-singularity spacetime. It is portrayed in below picture.

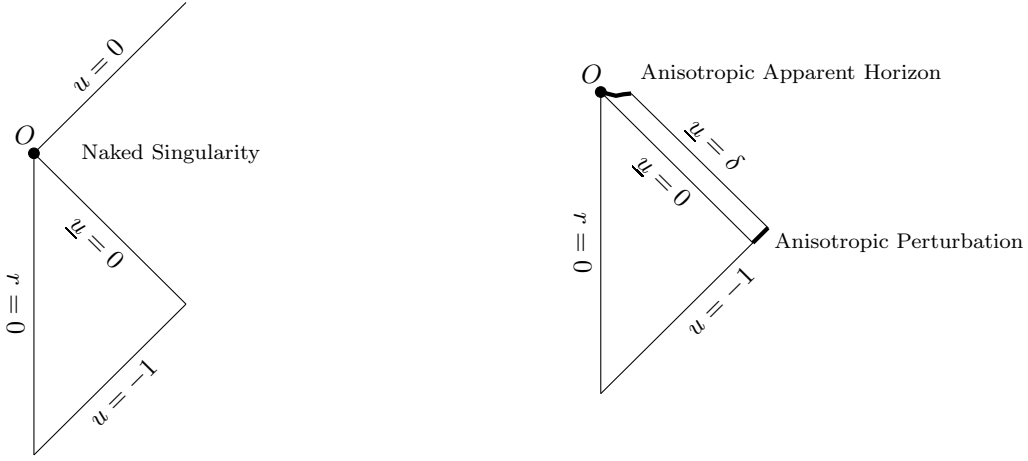
\begin{minipage}[!t]{0.48\textwidth}
\begin{tikzpicture}[scale=1]
   \draw [white](1.3, -0.25)-- node[midway, sloped, above,black]{{\footnotesize Naked Singularity}}(2.3, -0.25);
  \draw (0,0) -- node[pos=0.5, below, sloped]{$r=0$}
         (0,-4) node[below]{$$};
   \draw (0,-4) -- node[pos=0.5, below, sloped]{$u=-1$}
         (2,-2) node[below]{$$};
      \draw node[above]{$O\,\,\,\,$} (0,0) -- node[pos=0.5, below, sloped]{$\ub=0$}
         (2,-2) node[below]{$$};      
   \draw (0, 0) -- node[midway, above, sloped]{$u=0$}
         (2,2) node[right=4pt]{$$};
   \filldraw[black] (0,0) circle (2pt);      
  \end{tikzpicture}
  \end{minipage}
\hspace{0.02\textwidth}
\begin{minipage}[!t]{0.48\textwidth}
\begin{tikzpicture}[scale=1]
   \draw [white](2, 0)-- node[midway, sloped, above,black]{{\footnotesize Anisotropic Apparent Horizon}}(3, 0);
   \draw (0,0) -- node[pos=0.5, below, sloped]{$r=0$}
         (0,-4) node[below]{$$};
   \draw (0,-4) -- node[pos=0.7, below, sloped]{$u=-1$}
         (2,-2) node[below]{$$};
      \draw node[above]{$O\,\,\,\,$} (0,0) -- node[pos=0.5, below, sloped]{$\ub=0$}
         (2,-2) node[below]{$$};      
   \draw[ultra thick] (0, 0) -- node[midway, above, sloped]{$$}
         (0.2, -0.05)--(0.3, -0.03)--(0.4, -0.02) node[right=4pt]{$$};
     \draw (0.4, -0.02) -- node[midway, above, sloped]{$\ub=\delta$}
         (2.2, -1.8)--(2, -2) node[right=4pt]{$$};       
     \draw[ultra thick]    (2.2, -1.8)--(2, -2) node[right=4pt]{{\footnotesize Anisotropic Perturbation}};
   \filldraw[black] (0,0) circle (2pt);      
  \end{tikzpicture}
\end{minipage}

Christodoulou's solution solves the Einstein-scalar field system \ref{1.1}. In \cite{Chr.2} a spacetime singularity forms in gravitational collapse and the whole spacetime is free of trapped surfaces and the apparent horizon. 

{\color{black}
In this article, we consider a characteristic initial value problem with data prescribed along $\ub=0$ with $-1 \leq u \leq 0$ and $u=-1$ with $0\leq \ub \leq \delta$. For $\ub\leq 0$, we adopt the spacetime $(\M M_{int}, g_{int})$ constructed in \cite{Chr.2} with $\M M_{int}$ standing for the interior spacetime region $\ub\leq 0$. 
$$\mbox{Setting } v=4\ub, \mbox{ we let } g_{int}=-\f12\O^2(u,v)(du\otimes dv+dv\otimes du)+r^2(u,v)[(d\theta^1)^2+\sin^2\theta^1(d\theta^2)^2].$$ 
Here $r(u,v)$ is the radius function. 

We will add general (anisotropic) perturbation along $u=-1$ starting from $\ub=0$ and for $\ub\geq 0$ the spacetime metric $g$ is of the form \eqref{metric ansatz}. In particular, along $\ub=0$, we prescribe Christodoulou's spherically-symmetric naked-singularity initial data: via Christodoulou's construction in \cite{Chr.2}, for Hawking mass
$m=\f{r}{2}(1+4\O^{-2}\partial_u r\partial_{v}r)$, it holds that 

$$\f{2m}{r}(u,0)=\ms \mbox{ with } \ms \mbox{ being a positive constant less than }1 .$$ 
We also set
$$r(-1,0)=1, \quad \partial_{v}r(-1,0)=\f{1-\ms}{4} \quad \mbox{ and } \quad \partial_u r(u,0)=-1.$$
Since along $\ub=0$ it holds $\ms=\f{2m}{r}=1+4\O^{-2}\partial_u r \partial_{v} r$, we get
$$\O^{-2}\partial_{v} r(u,0)=\f{1-\ms}{4} \quad \mbox{along} \quad \ub=0.$$
Employing the equation\footnote{It is equivalent to $r\partial_u \partial_{v} r=-\partial_u r \partial_{v}r-\f14\O^2$.}
$$\partial_u[\ln(\partial_{v}r)]=\f{\partial_u r}{r}\cdot\f{\f{2m}{r}}{1-\f{2m}{r}},$$ 
with $r(-1,0)=1$ and $\partial_{v}r(-1,0)=(1-\ms)/4$, we obtain
$$\partial_{v}r(u,0)=\f{1-\ms}{4}r^{\f{\ms}{1-\ms}}(u,0) \quad \mbox{ and } \quad \O^2(u,0)=r^{\f{\ms}{1-\ms}}(u,0).$$
Using $v=4\ub$,  these further imply
$$\partial_{\ub} r(u,0)=(1-\ms)\cdot r^{\f{\ms}{1-\ms}}(u,0) \quad \mbox{ and } \quad \O^{-2}\partial_{\ub} r(u,0)=1-\ms.$$
Recall  $\tr\chi(u,0)=g^{AB}g(D_A e_4, e_B)(u, 0)$, $\tr\chib(u,0)=g^{AB}g(D_A e_3, e_B)(u, 0)$ and along $\ub=0$ it holds $e_3=\O^{-1}\partial_u, e_4=\O^{-1}\partial_{\ub}$. We hence derive 
$$\O^{-1}\tr\chi(u,0)=\O^{-1}g^{AB}\cdot\Gamma^{D}_{A4}\cdot g_{BD}=\f{2}{r}\O^{-2}\partial_{\ub}r(u,0)=\f{2(1-\ms)}{r(u,0)},$$
$$\O\tr\chib(u,0)=\O g^{AB}\cdot\Gamma^{D}_{A3}\cdot g_{BD}=\f{2}{r}\partial_{u}r(u,0)=-\f{2}{r(u,0)}.$$

}

\noindent Utilizing $\omb=-\f12\nab_3\log \O$, we also have
\begin{equation}\label{omb initial data}
\O\omb(u,0)=-\f12\partial_u \log\O(u,0)=-\f12\partial_u\log [r^{\f{\ms}{2(1-\ms)}}(u,0)]=\f{\ms}{4(1-\ms)}\cdot\f{1}{r(u,0)}.
\end{equation}

{\color{black}
\noindent We proceed to calculate the initial value of $K$ along $\ub=0$. Applying the below null structure equation and the Gauss equation
\begin{equation*}
\begin{split}
\nab_3 \trch+\frac1 2 \trchb \trch=&2\omegab \trch+2\rho- \chih\cdot\chibh+2\div \eta+2|\eta|^2-\f13D_3\phi D_4\phi+\f13 D_A\phi D^A\phi,\\
K=&-\rho-\f14\tr\chi \tr\chib+\f12 \chih\cdot\chibh+\f16 D_3\phi D_4\phi+\f13 D_A\phi D^A\phi,
\end{split}
\end{equation*}
further employing $\omb=-\f12\O^{-1}\nab_3\O$ and restricting the equations to $\ub=0$, we obtain 
$$2\O^2 K(u,0)+\partial_u(\O^2\cdot\O^{-1}\tr\chi)(u,0)+\O^2\cdot\O^{-1}\tr\chi\cdot\O\tr\chib(u,0)=0.$$
Together with the fact
$$\O^{-1}\tr\chi(u,0)=\f{2(1-\ms)}{r(u,0)}, \quad \O\tr\chib(u,0)=-\f{2}{r(u,0)}, \quad \O^2(u,0)=r^{\f{\ms}{1-\ms}}(u,0), \quad \partial_u r(u,0)=-1,$$
we deduce
$$2K(u,0)+r^{-\f{\ms}{1-\ms}}\cdot\partial_u(r^{\f{\ms}{1-\ms}-1})\cdot 2(1-\ms)(u,0)-\f{4(1-\ms)}{r(u,0)^2}=0,$$
which infers
$$K(u,0)=-(2-4\ms-4+4\ms)\cdot\f{1}{2r(u,0)^2}=\f{1}{r(u,0)^2}.$$

}

Quoting the results in \cite{Chr.2}, for the scalar field along $\ub=0$ it holds
\begin{equation}\label{scalar field initial data}
\begin{split}
\O\partial_3 \phi(u,0)=&\partial_u\phi(u,0)=\sqrt{2}\sqrt{\f{\ms}{1-\ms}}\f{1}{r(u,0)},\\ 
\O\partial_{4}\phi(u,0)=&\partial_{\ub}\phi(u,0)=\sqrt{2}\sqrt{\f{1-\ms}{\ms}}\cdot\f{\partial_{\ub}r(u,0)}{r(u,0)}.
\end{split}
\end{equation}

\begin{remark}
Along $\ub=0$, employing $\partial_u r(u,0)=-1$ and
$\partial_u(\O^{-2}\partial_u r)=-\f12 r\O^{-2}(\partial_u \phi)^2,$
via (\ref{scalar field initial data}) we also obtain 
\begin{equation*}
\O\omb(u,0)=-\f12\partial_u\log\O(u,0)=\f{|r\partial_u \phi|^2}{8r}(u,0)=\f{\ms}{4(1-\ms)}\cdot\f{1}{r(u,0)},
\end{equation*}
which is consistent with \eqref{omb initial data}. 
\end{remark}

\begin{remark}
The $\sqrt{2}$ on the right of (\ref{scalar field initial data}) is owing to in \cite{Chr.2} Christodoulou's use of
\begin{equation*}
\begin{split}
&\mbox{Ric}_{\mu\nu}-\f12 R g_{\mu\nu}=2T_{\mu\nu},\\
&T_{\mu\nu}=\partial_{\mu}\phi_c\partial_{\nu}\phi_c-\f12g_{\mu\nu}\partial^{\sigma}\phi_c \partial_{\sigma}\phi_c.
\end{split}
\end{equation*}
Here we denote $\phi_c$ to be the scalar field function employed in \cite{Chr.2}. Denoting $k:=\sqrt{\f{\ms}{1-\ms}}$, on page 644 of \cite{Chr.2}, we have that along $s=s_*$ (i.e. $\ub=0$) it holds
\begin{equation}\label{Christodoulou scalar field initial data}
r\f{\partial_u\phi_c}{\partial_u r}(u,0)=-k \quad \mbox{ and } \quad r\f{\partial_{\ub} \phi_c}{\partial_{\ub} r}(u,0)=\f{1}{k}.
\end{equation}
Setting the scalar field function $\phi(u, \ub, \theta^1, \theta^2)$ used in this paper to be 
$$\phi(u, \ub, \theta_1, \theta_2)=\sqrt{2}\phi_c(u, \ub, \theta_1, \theta_2),$$
via $\partial_u r(u,0)=-1$ we then translate (\ref{Christodoulou scalar field initial data}) to (\ref{scalar field initial data}).  
\end{remark}

{\color{black}
\begin{remark}
At the end of this subsection, we point out that, along $\ub=0$ by employing equations
$r\partial_u\partial_v\phi=-\partial_u r\partial_v\phi-\partial_v r\partial_u\phi$ and  $r\partial_u\partial_v r=-\partial_u r\partial_v r-\f14\O^2$, it holds
\begin{equation}\label{equation partial v phi over partial v r}
\f{\partial}{\partial u}\f{\partial_v \phi}{\partial_v r}-\f{\O^2}{4r\partial_v r}\cdot\f{\partial_v\phi}{\partial_v r}+\f{\partial_u \phi}{r}=0.
\end{equation}
Treat $\partial_v\phi(u,0)$ as it is determined by other geometric quantities. Applying the fact 
$$-\f{\O^2}{4r\partial_v r}(u,0)=\f{-1}{1-\ms}\cdot\f{1}{r(u,0)}\quad \mbox{ and } \quad \f{\partial_u\phi}{r}(u,0)=\sqrt{2}\sqrt{\f{\ms}{1-\ms}}\cdot\f{1}{r(u,0)^2},$$
we can solve an ODE and we see that the below equality as stated in \eqref{scalar field initial data}
$$\f{\partial_{v}\phi(u,0)}{\partial_v r(u,0)}=\f{\partial_{\ub}\phi(u,0)}{\partial_{\ub}r(u,0)}=\sqrt{2}\sqrt{\f{1-\ms}{\ms}}\cdot\f{1}{r(u,0)}\approx\f{1}{r(u,0)}$$ 
is debt to a special initial choice of $\partial_v\phi/\partial_v r$ at $S_{-1,0}$.  If we do not use this choice (by adding a general perturbation), back to \eqref{equation partial v phi over partial v r}, we would have that $\partial_v\phi(u,0)/\partial_v r(u,0)$ behaviors like $r(u,0)^{\f{-1}{1-\ms}}=r(u,0)^{-1}\cdot r(u,0)^{\f{-\ms}{1-\ms}}$. With $\partial_v r(u,0)\approx r^{\f{\ms}{1-\ms}}(u,0)$, it yields $\partial_v \phi(u,0)\approx r(u,0)^{-1}$. In below sections, we demonstrate an approach, which can deal with these general scenarios. 
\end{remark} 
}

\subsection{Norms}\label{subsection norms} For geometric quantities defined on $\S$,  we now set norms that we will work with. We use $\p$ and $\q$ to denote the below geometric quantities
\begin{equation}
\p\in\{\O \tr\chi, \O\chih, \O\h\o, \O\partial_{4}\phi\},
\end{equation}
\begin{equation}
\q\in\{\eta, \etb,\, \O\omb-\O \omb(u,0), \, \O\chibh, \, \O\tr\chib-\O\tr\chib(u,0),  \nab_{e_1}\phi, \, \nab_{e_2}\phi, \, \O\partial_{3}\phi-\O\partial_3\phi(u,0)\}.
\end{equation}

\noindent For conventional simplicity, we also write
\begin{equation}
\begin{split}
\O\omb(u,0)=&(\O\omb)_0, \quad \O\partial_3\phi(u,0)=(\O\partial_3\phi)_0,\\
\O\tr\chib(u,0)=(\O\tr\chib)_0,& \quad \O^{-1}\tr\chi(u,0)=(\O^{-1}\tr\chi)_0, \quad \O(u,0)=\O_0.
\end{split}
\end{equation}

\noindent With constants $a, b$ chosen as in Section \ref{secbasic}, we then define
\begin{equation}\label{O i infty}
\begin{split}
\mathcal{O}_{i,\infty}(u,\underline{u}):= &\f{1}{\at}\|u^{i+1}\nab^i\p\|_{L^{\infty}(\S)}+\f{1}{\ub\at}\|u^{i+2}\nab^i\q\|_{L^{\infty}(\S)} \mbox{ with } 0\leq i \leq 2,
\end{split}
\end{equation}

\begin{equation}\label{R i infty}
\begin{split}
\mathcal{R}_{i,\infty}(u,\underline{u}):=& \f{1}{\at}\|u^{i+2}\nab^i(\O\b)\|_{L^{\infty}(\S)}+\f{1}{\ub a}\|u^{i+3}\nab^i \rho\|_{L^{\infty}(\S)}\\
&+\f{1}{\ub\at}\|u^{i+3}\nab^i \sigmac\|_{L^{\infty}(\S)}+\f{1}{\ub\at}\|u^{i+3}\nab^i (K-\f{1}{|u|^2})\|_{L^{\infty}(\S)}\\
&+\f{1}{\ub\at}\|u^{i+3}\nab^i (\O\beb)\|_{L^{\infty}(\S)} \mbox{ with } 0\leq i\leq 1.
\end{split}
\end{equation}

\begin{equation}\label{O i 2}
\begin{split}
\mathcal{O}_{i,2}(u,\underline{u}):= &\f{1}{\at}\|u^{i}\nab^i\p\|_{L^{2}(\S)}+\f{1}{\ub\at}\|u^{i+1}\nab^i\q\|_{L^{2}(\S)} \mbox{ with } 0\leq i \leq 4.
\end{split}
\end{equation}
\noindent We also set
\begin{equation}\label{tilde O 5 2}
\begin{split}
\tilde{\mathcal{O}}_{5,2}(u,\underline{u}):=&\f{1}{\ub^{\f12}\at}\|u^5\nab^5 (\O\tr\chi, \O\h\o, \O\chih)\|_{L^{2}_{\ub}L^2(\S)}\\
&+\bigg(\f{\ub^{\f12}a^{\f14}}{|u|^{\f12}\O}\bigg)^{-1} \f{1}{\ub^{\f12}a^{\f14}}\|u^5\nab^5 \etb\|_{L^2_{\ub}L^2(\S)}+\f{1}{\ub^{\f32}a^{\f34}}\|u^6\nab^5 \eta\|_{L^2_{\ub}L^2(\S)}\\
&+\f{\O}{\ub^{\f12}\at}\|u^5\nab^5 \eta\|_{L^2_{u}L^2(\S)}+\f{1}{\ub^{\f12}\at}\|\O u^5\nab^5 \eta\|_{L^2_{u}L^2(\S)}\\
&+\bigg(1+\f{\ub^{\f12}a^{\f14}}{|u|^{\f12}\O}\bigg)^{-1} \f{1}{\ub^{\f12}a^{\f14}}\|u^5\nab^5 (\O\tr\chib, \O\chibh, \tr\chib, \chibh)\|_{L^2_{u}L^2(\S)},
\end{split}
\end{equation}

\begin{equation}\label{tilde' O 5 2}
\tilde{\mathcal{O}}'_{5,2}(u,\underline{u}):=\f{1}{\ub^{\f32}a}\|u^6\nab^5(\O\tr\chi)\|_{L^2_{\ub}L^2(\S)}. 
\end{equation}
For $1\leq i \leq 4$, we further let
\begin{equation}\label{R i 1}
\begin{split}
\mathcal{R}_{i,1}(u,\underline{u}):=&\f{1}{\ub^{\f12}\at}\|u^{i+1}\nab^i(\O\b_A-\f12\O\nab_4\phi\nab_A\phi)\|_{L^2_{\ub}L^2(\S)}\\
&+\f{1}{\ub^{\f12}\at}\|\O u^{i+1}\nab^i(K-\f{1}{|u|^2}-\f14\nab^A\phi\nab_A\phi, \sigmac)\|_{L^2_{u}L^2(\S)},
\end{split}
\end{equation}
\begin{equation}\label{R i 2}
\begin{split}
\mathcal{R}_{i,2}(u,\underline{u}):=&\bigg(1+\f{\ub^{\f12}a^{\f14}}{|u|^{\f12}\O}\bigg)^{-1}\f{1}{\ub^{\f32}a^{\f34}}\|u^{i+2}\nab^i(K-\f{1}{|u|^2}-\f14\nab^A\phi\nab_A\phi, \sigmac)\|_{L^2_{\ub}L^2(\S)} \\
&+\bigg(1+\f{\ub^{\f12}a^{\f14}}{|u|^{\f12}\O}\bigg)^{-1}\f{1}{\ub^{\f32}a^{\f34}}\|u^{i+2}\nab^i(\beb_A+\f12\nab_3\phi\nab_A\phi)\|_{L^2_{u}L^2(\S)}.
\end{split}
\end{equation}
For $1\leq i \leq 5$, we then set 
\begin{equation}\label{S i 1}
\begin{split}
\mathcal{S}_{i,1}(u,\underline{u}):=&\f{1}{\ub^{\f12}\at}\|u^{i}\nab^i(\O e_4\phi)\|_{L^2_{\ub}L^2(\S)}+\f{1}{\ub^{\f12}\at}\|\O u^{i}\nab^i(e_A\phi)\|_{L^2_{u}L^2(\S)},
\end{split}
\end{equation}
\begin{equation}\label{S i 2}
\begin{split}
\mathcal{S}_{i,2}(u,\underline{u}):=&\bigg(1+\f{\ub^{\f12}a^{\f14}}{|u|^{\f12}\O}\bigg)^{-1}\f{1}{\ub^{\f12}a^{\f14}}\|u^i\nab^i(e_A\phi)\|_{L^2_{\ub}L^2(\S)}\\
&+\bigg(1+\f{\ub^{\f12}a^{\f14}}{|u|^{\f12}\O}\bigg)^{-1}\f{1}{\ub^{\f12}a^{\f14}}\|\O^{-1}u^i\nab^i(\O e_3\phi)\|_{L^2_u L^2(\S)}.
\end{split}
\end{equation}
And we further define $\mathcal{O}(u,\ub),\, \mathcal{R}(u,\ub),\, \mathcal{S}(u,\ub)$ as
$$\mathcal{O}(u,\ub):=\sum_{i\leq 2}\mathcal{O}_{i,\infty}(u,\ub)+\sum_{i\leq 4}\M O_{i,2}(u,\ub),\quad \mathcal{R}(u,\ub):=\sum_{1\leq i\leq 4}\mathcal{R}_{i,1}(u,\ub)+\mathcal{R}_{i,2}(u,\ub),$$
$$ \mathcal{S}(u,\ub):=\sum_{1\leq i\leq 4}\mathcal{S}_{i,1}(u,\ub)+\mathcal{S}_{i,2}(u,\ub).$$

\noindent Next we denote $\mathcal{O}_{i,\i},\, \mathcal{O}_{i,2},\, \tilde{\mathcal{O}}_{5,2}, \, \tilde{\mathcal{O}}'_{5,2}, \,\mathcal{R}_{i,\infty},\, \mathcal{R}_{i, 1},\, \mathcal{R}_{i,2}, \, \mathcal{S}_{i, 1},\, \mathcal{S}_{i,2}$ to be the supremum over $u,\ub$ in the considered spacetime region of $\mathcal{O}_{i,\i}(u,\ub),\, \mathcal{O}_{i,2}(u,\ub),\, \tilde{\mathcal{O}}_{5,2}(u,\ub),$ $ \tilde{\mathcal{O}}'_{5,2}(u,\ub), \, \mathcal{R}_{i,\infty}(u,\ub),\mathcal{R}_{i, 1}(u,\ub),\, \mathcal{R}_{i,2}(u,\ub), \, \mathcal{S}_{i, 1}(u,\ub),\, \mathcal{S}_{i,2}(u,\ub)$ respectively. Finally, we set  $\mathcal{O},\, \mathcal{R},\, \mathcal{S}$ to be:
$$\mathcal{O}:=\sum_{i\leq 2}\mathcal{O}_{i,\infty}+\sum_{i\leq 4}\M O_{i,2},\quad \mathcal{R}:=\sum_{1\leq i\leq 4}\mathcal{R}_{i,1}+\mathcal{R}_{i,2}, \quad \mathcal{S}:=\sum_{1\leq i\leq 4}\mathcal{S}_{i,1}+\mathcal{S}_{i,2}.$$
We also denote $\mathcal{O}^{(0)},\,\t{\mathcal{O}}^{(0)}, \, \t{\mathcal{O}}'^{(0)}, \, \mathcal{R}^{(0)},\, \mathcal{S}^{(0)}$ to be the corresponding norms along the initial hypersurfaces $u=-1$ and $\ub=0$. Lastly, we define an initial data quantity
$$\mathcal{I}^{(0)}:=\sup_{0\leq \ub \leq \delta}\mathcal{I}^{(0)}(\ub)$$
with
\begin{equation*}
\begin{split}
\mathcal{I}^{(0)}(\ub):=1+\sum_{0\leq m \leq 6,}\sum_{0\leq k \leq 5}\f{1}{\at}\|(\nab)^{m}(\nab_4)^k(\chih, \nab_4\phi)\|_{L^2(S_{-1,\ub})}.
\end{split}
\end{equation*}

\begin{remark} 
{\color{black}  
Following the arguments carried out in detail in Chapter of \cite{Chr:book}, we have
$$\sum_{i\leq 2}\mathcal{O}_{i,\infty}(-1, \ub)+\sum_{i\leq 4}\M O_{i,2}(-1, \ub)+\sum_{i\leq 2}\mathcal{R}_{i,\infty}(-1, \ub)\lesssim \mathcal{I}^{(0)}.$$
In Section \ref{elliptic estimates}, we will employ the initial values of $\O\kappa, \, \underline{\mu}, \, \O^{-1}\tr\chib$ along $H_{-1}$, with
$$\kappa:=\nab\h\omega+{^*\nab\omega^{\dagger}}-\f12(\beta-\f12\nab_4\phi\nab\phi), \quad \underline{\mu}:=-\div\etb+K-\f{1}{|u|^2}-\f14\nab^A\phi\nab_A\phi.$$
Note that we set $\omega^{\dagger}$ to be $0$ and $\Omega$ to be $1$ along $H_{-1}$. And it holds
$$\|u^4\nab^4(\O\kappa)\|_{L^2(S_{-1,\ub})}\lesssim \at, \quad \|u^5\nab^4\underline{\mu}\|_{L^2(S_{-1,\ub})}\lesssim \ub\at, \quad  \|u^5\nab^5(\O^{-1}\tr\chib)\|_{L^2(S_{-1,\ub})}\lesssim \ub\at.$$

}
\end{remark}

\subsection{Exterior and Interior Regions}
To establish the hyperbolic existence result, we first choose a $u_1<0$ satisfying $0<|u_1|\ll 1$ and then divide the spacetime $\M M$ into two regions: 
\begin{itemize}
\item[$\cdot$] the exterior: \quad {\color{black} $-1 \leq u \leq u_1$ \quad and \quad $0\leq\ub\at\leq |u_1|\O^{2-\t\d}(u_1,0)$,}
\item[$\cdot$] the interior: \quad \, {\color{black}$u_1\leq u\leq 0$ \quad \, and \quad $0\leq\ub\at\leq |u|\O^{2-\t\delta}(u,0)$.}
\end{itemize} 
{\color{black}
\begin{remark}
For notational simplicity, in the rest of this paper, we set constant $\t\d$ obeying $0<\t\d\ll1$. In reality, if needed we can set $\t\d$ to be any number satisfying $0<\t\d<1/2$ and all the below proof stays the same. 
\end{remark} 
}
In this article, we will use the below inequality frequently 
\begin{equation}\label{smallness}
\f{\ub\at}{|u|}\leq \f{\ub\at}{|u|\O}\leq \f{\ub\at}{|u|\O^{\f32}}\ll 1. 
\end{equation}
This inequality is correct because (with $a=1$)
\begin{itemize}
\item[$\cdot$] in the exterior we have {\color{black} $\f{\ub\at}{|u|\O^{\f32}}\leq \f{|u_1|\O^{2-\t\d}(u_1,0)}{|u|\O^{\f32}}=\f{|u_1|\O^{\f32}(u_1,0)}{|u|\O^{\f32}}\cdot\O^{\f12-\t\d}(u_1,0)\ll 1$,}
\item[$\cdot$] in the interior we have $\f{\ub\at}{|u|\O^{\f32}}\leq \O^{\f12-\t\delta}(u, 0)\lesssim \O^{\f12-\t\delta}(u_1, 0)\ll 1$.
\end{itemize}
With above discussion, we can further choose a constant $B$, such that in both regions 
\begin{equation}\label{largeness B}
\f{\ub\at}{|u|}\leq\f{\ub\at}{|u|\O}\leq \f{\ub\at}{|u|\O^{\f32}}\leq\f{1}{B}\ll 1.
\end{equation} 

\subsection{An Approach of Bootstrap}\label{order of bootstrap}
In this article, we carry out a bootstrap argument to prove uniform upper bounds for $\M O, \M R$ and $\M S$. For initial data along $\ub=0$ and $u=-1$, we have
$$\M O^{(0)}+\M R^{(0)}+\M S^{(0)}\leq C_0 \M I^{(0)},$$
where $C_0$ depends only on initial data. For notational simplicity, we rewrite the above line as
$$\M O^{(0)}+\M R^{(0)}+\M S^{(0)}\lesssim \M I^{(0)}.$$
For the rest of this article, we keep the same notation: if $P\leq C_0 Q$ with constant $C_0$ depending only on the initial data, we could write it as 
$$P\lesssim Q.$$
Under this notation, we have $ \M I^{(0)}\lesssim 1$.\\

Our goal in below sections is to show that in the considered spacetime region, it holds
$$\M O+\M R+\M S\leq \M I^{(0)}+(\M I^{(0)})^2+1.$$
To achieve this, we make bootstrap assumptions
\begin{equation}\label{BA.1}
\M O(u,\ub)\leq O,
\end{equation}
\begin{equation}\label{BA.2}
\M R(u,\ub)\leq R,
\end{equation}
\begin{equation}\label{BA.3}
\M S(u,\ub)\leq S,
\end{equation}
\begin{equation}\label{BA.4}
\t{\M O}_{5,2}(u,\ub)\leq \t O, \quad \t{\M O}'_{5,2}(u,\ub)\leq \t O'.
\end{equation}

Recall \eqref{largeness B}. For constants in bootstrap assumptions \eqref{BA.1}-\eqref{BA.4}, for the rest of this paper, we assume that they satisfy
\begin{equation}\label{BA.5}
1\leq O, R, S, \t O \ll (\t O')^{\f18}\ll B^{\f{1}{32}}\leq (\f{\ub \at}{|u|\O^{\f32}})^{-\f{1}{32}}\leq (\f{\ub\at}{|u|\O})^{-\f{1}{32}}. 
\end{equation}

{\color{black}
In below, we demonstrate the steps for improving the bootstrap assumptions:\footnote{A more detailed explanation of the bootstrap method can be found in Section 3.1 of \cite{An19}.}  

\noindent In Section \ref{secRicci} we prove 
\begin{itemize}
\item[] $\sum_{i\leq 4}\M O_{i,2}(\chih, \tr\chi)\lesssim 1$ with improved estimates for $\tr\chi$,
\item[] $\sum_{i\leq 4}\M O_{i,2}(\h\o)\lesssim 1+\sum_{i\leq 4}\M O_{i,2}(\nab_4\phi)$,
\item[] $\sum_{i\leq 4}\M O_{i,2}(\q)\lesssim 1+\sum_{i\leq 4}\M O_{i,2}(\h\o)+\t {\M O}_{5,2}(\h\o)+\M R(\b)$ with improved bound for $\O\tr\chib+\f{2}{|u|}$.
\end{itemize}

\noindent In Section \ref{scalar field} we obtain 
\begin{itemize}
\item[] $\sum_{i\leq 4}\M O_{i,2}(\nab_4\phi)\lesssim 1, \quad \sum_{i\leq 4}\M O_{i,2}(\nab_A\phi)\lesssim S, \quad \sum_{i\leq 4}\M O_{i,2}(\nab_3\phi)\lesssim S$.
\end{itemize}

\noindent In Section \ref{energy scalar field} we derive 
\begin{itemize}
\item[] $\M S_{5,1}\lesssim 1$, \quad $\M S_{5,2}\lesssim 1$.
\end{itemize}
With the bounds obtained in Section \ref{scalar field} and Section \ref{energy scalar field}, we revisit the estimates in Section \ref{secRicci}. Hence we now get 
\begin{itemize}
\item[] $\sum_{i\leq 4}\M O_{i,2}(\p)\lesssim 1$, \quad $\sum_{i\leq 4}\M O_{i,2}(\q)\lesssim 1+\t {\M O}_{5,2}(\h\o)+\M R(\b)$.
\end{itemize}

\noindent In Section \ref{elliptic estimates} we show
\begin{itemize}
\item[] $\t{\M O}_{5,2}(\h\o)\lesssim 1+\M R(\b)$, \quad $\t{\M O}_{5,2}(\tr\chi, \chih)\lesssim 1+\M R(\b)$, 
\item[] $\t{\M O}'_{5,2}(\tr\chi)\lesssim 1$,
\item[] $\t{\M O}_{5,2}(\etb)\lesssim 1$, \quad $\t{\M O}_{5,2}(\eta)\lesssim 1+\M R(K, \sigmac)+\M R(\beta)+\underline{\M R}(K, \sigmac)$,
\item[] $\t{\M O}_{5,2}(\O\tr\chib, \tr\chib)\lesssim 1$, \quad $\t{\M O}_{5,2}(\O\chibh)\lesssim 1$, \quad $\t{\M O}_{5,2}(\chibh)\lesssim \M R(\b)$.
\end{itemize}

\noindent In Section \ref{energy estimate curvature} we derive
\begin{itemize}
\item[] $\M R(\b)+\underline{\M R}(K, \sigmac)\lesssim 1$, \quad $\M R(K, \sigmac)+\underline{\M R}(\beb)\lesssim 1$.
\end{itemize} 

}

{\color{black}
\begin{remark}\label{small scale critical norm} 
The proof strategy and the main theorems of this paper also hold for perturbed Christodoulou's initial data prescribed along $\ub=0$. Denote 
$$\mathring{\O}(u, 0)=\O(u,0) \mbox{ and } \mathring{\Gamma}(u,0)\in\{\f{\partial_u r}{r}(u,0), \, \f{\partial_{\ub} r}{r}(u,0), \, \partial_u\phi(u,0), \, \partial_{\ub}\phi(u,0) \}.$$
We further set $\mathring{\O}_{c}(u,0)$ and $\mathring{\Gamma}_{c}(u,0)$ to be the corresponding values of Christodoulou's naked-singularity solution along $\ub=0$. Our below proof strategy also extends to spherically-symmetric initial data along $\ub=0$ satisfying 
$$|\f{\mathring{\O}(u,0)}{\mathring{\O}_{c}(u,0)}|\leq 1+\f{\mathring{\O}_{c}(u,0)^{\t\delta}}{a^{\f13}}, \quad |\f{\mathring{\Gamma}(u,0)}{\mathring{\Gamma}_{c}(u,0)}|\leq 1+\f{\mathring{\O}(u,0)^{\t\delta}}{a^{\f13}}$$
with constants $0<\t\delta\ll 1$ and $a\geq 1$.\footnote{The reader is also referred to a spherically-symmetric extension of Christodoulou's naked-singularity solution by Singh in \cite{Singh}.} For simplicity of the proof and for more precise statements of the main theorems, in the main body of this paper we stick to Christodoulou's naked-singularity initial data with Footnotes \ref{renormalized K}, \ref{trchib evolution lemma}, \ref{initial omb perturbed data}, \ref{alpha1 perturbed data}, \ref{initial data cancellation 1}, \ref{initial data cancellation 2}, \ref{instability theorems perturbed initial data} and Remark \ref{phi def with perturbed data} added to explain how to deal with the perturbed scenario.\footnote{{\color{black}The methods developed in this paper also allow non-spherically-symmetric initial data prescribed along $\ub=0$, the author will provide the details in a separate paper.}}  

\end{remark}
}

\section{Preliminary Estimates}\label{secbasic}

\subsection{Estimates for metric components}\label{metric} We now start the hyperbolic estimates. We first obtain a bound for $\Omega$ under the bootstrap assumptions:
\begin{proposition}\label{Omega}
Under the assumptions of Theorem \ref{main thm} and the bootstrap assumptions \eqref{BA.1}, \eqref{BA.2}, \eqref{BA.3}, \eqref{BA.4}, we have
$$\|\f{\O}{\O(u,0)}-1\|_{L^\infty(S_{u,\ub})}\leq B^{-\f12}.$$
\end{proposition}
\begin{proof}
We employ the equation
\begin{equation}\label{Omegatransport}
 \O\h\omega=-\frac{\O}{2}\nabla_4\log\Omega={\color{black}-\f12(\f{\partial}{\partial\ub}+d^A\f{\partial}{\partial \theta^A})\log\O.}
 \end{equation}
 Integrating equation (\ref{Omegatransport}) along $e_4$ direction, we obtain
$$||\log\f{\O}{\O(u,0)}||_{L^\infty(S_{u,\ub})}\ls \int_0^{\ub}||\O\h\omega||_{L^\infty(S_{u,\ub'})}d\ub'\ls \frac{\ub\at O}{|u|}.$$
Here we use the bootstrap assumption \eqref{BA.1}.
This implies that
$$\|\f{\O}{\O(u,0)}-1\|_{L^\infty(S_{u,\ub})}\ls \frac{\ub\at O}{|u|}\leq B^{-\f12}.$$

\end{proof}

We then control $\sg$ under the bootstrap assumptions. Following the same argument as for Proposition 5.2 in \cite{AL}, we have

\begin{proposition}\label{gamma} 
Fix a point $(u, \o)$ on the initial hypersurface $\Hb_0$. Along the outgoing characteristic emanating from $(u, \o)$, we define $\Lambda(\ub)$ and $\lambda(\ub)$ to be the larger and smaller eigenvalue of $\sg^{-1}(u, \ub=0, \o)\sg(u, \ub, \o)$. Under the assumptions of Theorem \ref{main thm} and the bootstrap assumptions \eqref{BA.1}, \eqref{BA.2}, \eqref{BA.3} and \eqref{BA.4}, then it holds
$$|\Lambda(\ub)-1|+|\lambda(\ub)-1|\leq B^{-\f12}$$
for every $\ub\in[0,\de]$. Consequently, for every $\ub\in[0,\delta]$, this also implies 
$$|\xi(\ub)-1|\leq B^{-\f12} \quad \mbox{ with } \quad \xi(\ub)=\f{dvol_{\sg_{\ub}}}{dvol_{\sg_0}}.$$
\end{proposition} 

As corollaries, we also obtain the below area estimate and the Sobolev embedding.
\begin{proposition}\label{area}
Under the assumptions of Theorem \ref{main thm} and the bootstrap assumptions \eqref{BA.1}, \eqref{BA.2}, \eqref{BA.3}, \eqref{BA.4}, it holds that
$$\sup_{u,\ub}|\mbox{Area}(S_{u,\ub})-\mbox{Area}(S_{u,0})|\leq |u|^2 B^{-\f12}.$$
\end{proposition}

\begin{proposition}\label{Sobolev}
With the assumptions of Theorem \ref{main thm} and the bootstrap assumptions \eqref{BA.1}, \eqref{BA.2}, \eqref{BA.3} and \eqref{BA.4}, it holds that
\begin{equation}
\begin{split}
\|\phi\|_{L^\infty(S_{u,\ub})} \ls & \sum_{i\leq 2}\|u^{i-1}\nab^i\phi\|_{L^2(\S)}.
\end{split}
\end{equation}
\end{proposition}
\noindent Note that the detailed proof of the Sobolev embedding is provided in Propositions 5.6-5.9 in \cite{AL}. 

\subsection{Estimates for Transport Equations}\label{transportsec}
In this article, we will employ the below two propositions frequently: 

\begin{proposition}\label{transport}
With the assumptions of Theorem \ref{main thm} and the bootstrap assumptions \eqref{BA.1}, \eqref{BA.2}, \eqref{BA.3} ,\eqref{BA.4}, for an $\S$-tangent tensor field $\phi$ of arbitrary rank it holds
\[
 ||\phi||_{L^2(S_{u,\ub})}\ls ||\phi||_{L^2(S_{u,\ub'})}+\int_{\ub'}^{\ub} ||\O\nabla_4\phi||_{L^2(S_{u,\ub''})}d{\ub''}.
\]
\end{proposition}

\begin{proposition}\label{evolution lemma}
Set $\phi$ and $F$ to be $\S$-tangent tensor fields of rank $k$ obeying the below transport equation:
\begin{equation*}
\O\nab_3 \phi_{A_1...A_k}+\lambda_0\O{\tr\underline{\chi}}\phi_{A_1...A_k}=\O F_{A_1...A_k}
\end{equation*}
and denote $\lambda_1=2(\lambda_0-\frac{1}{2})$. Under the assumptions of Theorem \ref{main thm} and the bootstrap assumptions \eqref{BA.1}, \eqref{BA.2}, \eqref{BA.3}, \eqref{BA.4}, then it holds that
\begin{equation*}
|u|^{\lambda_1}\|\phi\|_{L^2(\S)}\ls
\|\phi\|_{L^2(S_{-1,\underline{u}})}+\int_{-1}^u |u'|^{\lambda_1}\|\O F\|_{L^2(S_{u',\underline{u}})}du'.
\end{equation*}
\end{proposition}

\noindent The proof of the above two propositions are the same as in Proposition 5.4 and Proposition 5.5 of \cite{AL} and hence are omitted.\footnote{\label{trchib evolution lemma}With perturbed Christodoulou initial data, the quantity $\|\O\tr\chib(u,\ub, \o)+\f{2}{|u|}\|_{L^1_u}$ is still uniformly bounded and hence the proof for Proposition \ref{evolution lemma} is the same.}

\subsection{Commutation Formula}\label{commutation}
We quote the following formula from \cite{KNI:book}:
\begin{proposition}
The commutator $[\nabla_4,\nabla]$ acting on an $(0,r)$ S-tensor satisfies
\begin{equation*}
 \begin{split}
&[\nabla_4,\nabla_B]\phi_{A_1...A_r}\\
=&[D_4,D_B]\phi_{A_1...A_r}+(\nabla_B\log\Omega)\nabla_4\phi_{A_1...A_r}-(\sg^{-1})^{CD}\chi_{BD}\nabla_C\phi_{A_1...A_r} \\
&-\sum_{i=1}^r (\sg^{-1})^{CD}\chi_{BD}\etab_{A_i}\phi_{A_1...\hat{A_i}C...A_r}+\sum_{i=1}^r (\sg^{-1})^{CD}\chi_{A_iB}\etab_{D}\phi_{A_1...\hat{A_i}C...A_r}.
 \end{split}
\end{equation*}
\end{proposition}

\begin{proposition}
The commutator $[\nabla_3,\nabla]$ acting on an $(0,r)$ S-tensor satisfies
\begin{equation*}
 \begin{split}
&[\nabla_3,\nabla_B]\phi_{A_1...A_r}\\
=&[D_3,D_B]\phi_{A_1...A_r}+(\nabla_B\log\Omega)\nabla_3\phi_{A_1...A_r}-(\sg^{-1})^{CD}\chib_{BD}\nabla_C\phi_{A_1...A_r} \\
&-\sum_{i=1}^r (\sg^{-1})^{CD}\chib_{BD}\eta_{A_i}\phi_{A_1...\hat{A_i}C...A_r}+\sum_{i=1}^r (\sg^{-1})^{CD}\chib_{A_iB}\eta_{D}\phi_{A_1...\hat{A_i}C...A_r}.
 \end{split}
\end{equation*}
\end{proposition}

\noindent By induction, using commutation formula repeatedly we get the schematic form:
\begin{proposition}\label{commuteeqn}
Assume $\O\nabla_4\phi=\O F_0$. Let $\O \nabla_4\nabla^i\phi=\O F_i$.
Then it holds
\begin{equation*}
\begin{split}
\O F_i= &\sum_{i_1+i_2+i_3=i}\nabla^{i_1}(\eta+\underline{\eta})^{i_2}\nabla^{i_3} (\O F_0)+\sum_{i_1+i_2+i_3+i_4=i}\nabla^{i_1}(\eta+\underline{\eta})^{i_2}\nabla^{i_3}(\O\chi)\nabla^{i_4} \phi\\
&+\sum_{i_1+i_2+i_3+i_4=i-1} \nabla^{i_1}(\eta+\underline{\eta})^{i_2}\nabla^{i_3}(\O\beta)\nabla^{i_4} \phi.
\end{split}
\end{equation*}
Here by $\nabla^{i_1}(\eta+\underline{\eta})^{i_2}$ we mean the sum of all terms, which itself is a product of $i_2$ factors, each factor being $\nabla^j (\eta+\underline{\eta})$ for some $j\in \mathbb{Z}^{\geq 0}$ and that the sum of all $j$ being $i_1$, i.e., $\nabla^{i_1}(\eta+\underline{\eta})^{i_2}=\displaystyle\sum_{j_1+...+j_{i_2}=i_1}\nabla^{j_1}(\eta+\underline{\eta})...\nabla^{j_{i_2}}(\eta+\underline{\eta})$. Similarly, assume $\O\nabla_3\phi=\O G_{0}$. Suppose $\O\nabla_3\nabla^i\phi=\O G_{i}$.
Then we have
\begin{equation*}
\begin{split}
\O G_{i}=-\frac{i}{2}\O\tr\chib \nab^i \phi&+\sum_{i_1+i_2+i_3=i}\nabla^{i_1}(\eta+\underline{\eta})^{i_2}\nabla^{i_3} (\O G_{0})\\
&+\sum_{i_1+i_2+i_3+i_4=i}\nabla^{i_1}(\eta+\underline{\eta})^{i_2}\nabla^{i_3}(\O\chibh,\O\tr\chib-(\O\tr\chib)_0)\nabla^{i_4} \phi\\
&+\sum_{i_1+i_2+i_3+i_4=i-1} \nabla^{i_1}(\eta+\underline{\eta})^{i_2}\nabla^{i_3}(\O\underline{\beta})\nabla^{i_4} \phi.
\end{split}
\end{equation*}
\end{proposition} 

By replacing $\O\beta$ and $\O\betab$ via the schematic Codazzi equations:
$$\O\beta=\nabla\p+\p\q, \quad \O\beb=\nabla\q+\f{1}{|u|}\q+\q\,\q,$$
rewriting $\O\trch$, $\O\chih$ as $\p$ and substituting $\eta$, $\etab$, $\O\chibh$, $\O\trchb-(\O\tr\chib)_0$ with $\q$, we obtain a simplified version 
\begin{proposition}
Suppose $\nabla_4\phi=F_0$. Let $\nabla_4\nabla^i\phi=F_i$.
Then we have
\begin{equation*}
\begin{split}
\O F_i= &\sum_{i_1+i_2+i_3=i}\nabla^{i_1}\q^{i_2}\nabla^{i_3} (\O F_0)+\sum_{i_1+i_2+i_3+i_4=i}\nabla^{i_1}\q^{i_2}\nabla^{i_3}\p\nabla^{i_4} \phi.\\
\end{split}
\end{equation*}
Similarly, assume $\nabla_3\phi=G_{0}$. Let $\nabla_3\nabla^i\phi=G_{i}$.
Then it holds
\begin{equation*}
\begin{split}
\O G_{i}+\f {i}{2} \O\tr\chib\nab^i\phi =&\sum_{i_1+i_2+i_3=i}\nabla^{i_1}\q^{i_2}\nabla^{i_3} (\O G_{0})\sum_{i_1+i_2+i_3=i}\nabla^{i_1}\q^{i_2} \nab^{i_3}\phi\\
&+\sum_{\substack{i_1+i_2+i_3=i,\\ i_3\leq i-1}} (\O\tr\chib, \O\nab_3\phi)\nabla^{i_1}\q^{i_2}\nabla^{i_3} \phi. 
\end{split}
\end{equation*}
\end{proposition}

{\color{black}
\begin{remark}
For the rest of this paper, when there is no danger of confusion, the constants on the left of the equations are kept precise. The uniform constants in front of coefficients on the right are usually omitted.
\end{remark} 
}

\subsection{Elliptic Estimates for Hodge Systems}\label{elliptic}

We first recall the definition of the divergence and curl of a symmetric covariant tensor with an arbitrary rank:
$$(\div\phi)_{A_1...A_r}:=\nabla^B\phi_{BA_1...A_r},$$
$$(\curl\phi)_{A_1...A_r}:=\eps^{BC}\nabla_B\phi_{CA_1...A_r}.$$
Here $\eps$ is the volume form associated to the metric $\sg$ on $\S$.
Recall that the trace is defined via
$$(\mbox{tr }\phi)_{A_1...A_{r-1}}=(\sg^{-1})^{BC}\phi_{BCA_1...A_{r-1}}.$$
In this article, we will use the following $L^2(\S)$ elliptic estimates proved in {\color{black}Chapter 7} of \cite{Chr:book}. 

\begin{proposition}\label{ellipticthm}
Under the assumptions of Theorem \ref{main thm} and the bootstrap assumptions \eqref{BA.1}, \eqref{BA.2}, \eqref{BA.3}, \eqref{BA.4}, denoting $\phi$ to be a totally symmetric $r+1$ covariant tensorfield on a 2-sphere $(\S,\sg)$ satisfying
$$\div\phi=\t f,\quad \curl\phi=\t g,\quad \mbox{tr}\phi=\t h,$$
then, for $1\leq i\leq 5$, it holds that
\begin{equation*}
 \begin{split}
  \|u^i\nabla^{i}\phi||_{L^2(\S)}\ls \sum_{j=0}^{i-1}(||u^{j+1}\nabla^{j}(\t f, \t g)||_{L^2(\S)}+\|u^j\nab^j \t h\|_{L^2(\S)}+||u^j\nab^j\phi||_{L^2(\S)}).
 \end{split}
\end{equation*}
\end{proposition}

As a special case with $\phi$ being a symmetric traceless 2-tensor, it holds that $\curl \phi=^* \div \phi$. Using the above proposition, we then get

\begin{proposition}\label{elliptictraceless}
Assume that $\phi$ is a symmetric traceless 2-tensor obeying
$$\div\phi=\t f.$$
Then, under the assumptions of Theorem \ref{main thm} and the bootstrap assumptions \eqref{BA.1}, \eqref{BA.2}, \eqref{BA.3}, \eqref{BA.4}, for $1\leq i\leq 5$, it holds
\begin{equation*}
 \begin{split}
||u^i\nabla^{i}\phi||_{L^2(\S)}
\ls& \sum_{j=0}^{i-1}(||u^{j+1}\nabla^{j}\t f||_{L^2(\S)}+||u^j\nab^j\phi||_{L^2(\S)}).
 \end{split}
\end{equation*}
\end{proposition}

\section{$L^2(\S)$ estimates for Ricci coefficients}\label{secRicci}

We then proceed to derive the $L^2(\S)$ estimates for the Ricci coefficients and their first four angular derivatives. 

We first state a useful lemma, which can be verified directly.\footnote{A detailed proof is provided in Proposition 6.1 of \cite{AL}.} 
\begin{lemma}\label{product}
Under the assumptions of Theorem \ref{main thm} and the bootstrap assumptions \eqref{BA.1}, \eqref{BA.2}, \eqref{BA.3}, \eqref{BA.4}, it holds that
\begin{equation*}
\begin{split}
\sum_{i_1+i_2\leq 4}\|u^{i_1+i_2+1}\nabla^{i_1}\q^{i_2+1}\|_{L^{2}(\S)}\ls &\sum_{i_1\leq 4}\|u^{i_1+1}\nab^{i_1} \q\|_{L^2(\S)}
\end{split}
\end{equation*}
and
\begin{equation*}
\begin{split}
\sum_{i_1+i_2\leq 2}\|u^{i_1+i_2+2}\nabla^{i_1}\q^{i_2+1}\|_{L^{\infty}(\S)}
\ls &\sum_{i_1\leq 2}\|u^{i_1+2}\nab^{i_1} \q\|_{L^\infty(\S)}.
\end{split}
\end{equation*}
In particular, by \eqref{BA.1}, we have
$$\sum_{i_1+i_2\leq 4}\|u^{i_1+i_2+1}\nabla^{i_1}\q^{i_2+1}\|_{L^{2}(\S)}+\sum_{i_1+i_2\leq 2}\|u^{i_1+i_2+2}\nabla^{i_1}\q^{i_2+1}\|_{L^{\infty}(\S)}\ls \ub \at O.$$
\end{lemma}

We now start from the $L^2(\S)$ estimates for $\O \chih$.

\begin{proposition}\label{chih.bd}
With the assumptions of Theorem \ref{main thm} and the bootstrap assumptions \eqref{BA.1}, \eqref{BA.2}, \eqref{BA.3}, \eqref{BA.4}, it holds
\[
 \sum_{i\leq 4}\|u^i\nab^i(\O\chih)\|_{L^2(S_{u,\ub})} \ls a^{\frac 12}.
\]

\end{proposition}

\begin{proof}
We employ the null structure equation 
$${\color{black}\O}\nab_3({\color{black}\O}\chih)+\f12 {\color{black}\O}\tr\chib({\color{black}\O}\chih)={\color{black}\O^2}\nab\eta+\p\q+\q\q.$$ 
After commuting this equation with $i$ angular derivatives, we obtain

\begin{equation*}
\begin{split}
&\O\nab_3 \nab^i(\O\chih) +\frac{i+1}{2}\O\tr\chib\nab^i(\O\chih)\\
= &{\color{black}\O^2}\nab^{i+1}\eta+\sum_{i_1+i_2=i} \nab^{i_1}\q^{i_2+2}+\sum_{i_1+i_2+i_3=i}\nabla^{i_1}\q^{i_2+1}\nab^{i_3}\p+\sum_{i_1+i_2+i_3=i-1}\frac {1}{|u|} \nabla^{i_1}\q^{i_2+1}\nab^{i_3}\p.
\end{split}
\end{equation*}
{\color{black}Note that the fourth term only presents when $i\geq 1$.}
To bound $\|u^i\nab^{i}(\O\chih)\|_{L_u^{\i}L_{\ub}^{\i}L^{2}(\S)}$, we then apply Proposition \ref{evolution lemma} with $\lambda_0=\frac{i+1}{2}$ and estimate each term on the right. 

For the $\nab^{i+1}\eta$ term, when $i=4$ we have
\begin{equation*}
\begin{split}
\|u^4{\color{black}\O^2}\nab^5\eta\|_{L_{u}^{1}L^{2}(\S)}\leq \|\frac{{\color{black}\O}}{|u|}\|_{L^2_u}\|{\color{black}\O}u^5\nab^5\eta\|_{L_{u}^{2}L^{2}(\S)}\ls \frac{\ub^{\f12}a^{\f12}}{|u|^{\f12}} \tilde{\M {O}}_{5,2}{\color{black}\leq \at B^{-\f14}}.
\end{split}
\end{equation*}

The rest terms are lower order. For example, we control the last term as

\begin{equation}\label{chih.4}
\begin{split}
&\sum_{i\leq 4}\|\sum_{i_1+i_2+i_3=i-1}u^{i-1}\nabla^{i_1}\q^{i_2+1}\nab^{i_3}\p\|_{L_{u}^{1}L^{2}(\S)}
\ls \frac{\ub a O^2}{|u|}\leq \at B^{-\f14}.
\end{split}
\end{equation}
All the remaining terms on the right obey the same upper bound $\at B^{-\f14}$. 

For initial data, we have $\sum_{i\leq 4}\|\nab^i (\O\chih)\|_{L^\infty_{\ub}L^2(S_{-1,\ub})}\leq \at$. Gathering all above estimates, we get 
$$\sum_{i\leq 4}\|u^i\nab^i(\O\chih)\|_{L^{\i}_{u}L^2(\S)}\ls a^{\frac 12}.$$

\end{proof} 

We then derive the estimates for $\O\trch$ and its derivatives. We will also point out that $\nab^i(\O\trch)$ obeys a better bound than the standard estimate for $\nab^i\p$.

\begin{proposition}\label{trch.bd}
Under the assumptions of Theorem \ref{main thm} and the bootstrap assumptions \eqref{BA.1}, \eqref{BA.2}, \eqref{BA.3}, \eqref{BA.4}, it holds
\begin{align*}
\sum_{i\leq 4} \|u^{i+1}\nab^i(\O\trch-&\O^2(\O^{-1}\tr\chi)_0)\|_{L^2(\S)}\ls \ub a O^2\ll |u|\at,\\
\sum_{i\leq 4} \|u^{i}\nab^i(\O&\trch)\|_{L^2(\S)}\ls \f{\ub a O^2}{|u|}\ll \at.
\end{align*}
\end{proposition}

\begin{proof}
We use the following null structure equation 
$${\color{black}\O}\nab_4 ({\color{black}\O^{-1}}\tr\chi)=-|\chih|^2-(\nab_4\phi)^2-\f12(\tr\chi)^2.$$
Note that $\O^{-1}\tr\chi(u,0)=2(1-\ms)/|u|$. Utilizing $\nab_4(\O^{-1}\tr\chi)_0=0$, we have
 $${\color{black}\O}\nab_4 ({\color{black}\O^{-1}}\tr\chi-(\O^{-1}\tr\chi)_0)=-|\chih|^2-(\nab_4\phi)^2-\f12(\tr\chi)^2.$$

\noindent Commuting this equation with $i$ angular derivatives, we get the below form

\begin{equation*}
\begin{split}
&\O\nab_4 \nab^i(\O^{-1}\trch-(\O^{-1}\tr\chi)_0)\\
=& \sum_{i_1+i_2+i_3+i_4=i}\nab^{i_1}\q^{i_2}\nab^{i_3}(\chih, \partial_4\phi)\nab^{i_4}(\chih, \partial_4\phi)+\sum_{i_1+i_2+i_3+i_4=i}\O^{-1}\nab^{i_1}\q^{i_2}\nab^{i_3}\p\nab^{i_4}\trch\\
&+\sum_{i_1+i_2+i_3=i}\frac{{\color{black}1}}{|u|}\nab^{i_1}\q^{i_2}\nab^{i_3}\p.
\end{split}
\end{equation*}
We then employ Proposition \ref{transport} and control the right hand side in the $\|u^{i+1}\cdot\|_{L^1_{\ub}L^2(\S)}$ norm to bound $\|u^{i+1}\nab^i(\O^{-1}\trch-(\O^{-1}\tr\chi)_0)\|_{L^\i_{\ub}L^2(\S)}$.
\begin{remark}
{\color{black}In below estimates and for the rest of this article, when there is no danger of confusion, we use $\O^{-1}$ to stand for $\|\O^{-1}\|_{L^{\infty}(\S)}$.}
\end{remark}
Employing the bootstrap assumption \eqref{BA.1} together with Sobolev embedding in Proposition \ref{Sobolev}, for the case $i_2=0$, we first control 

\begin{equation*}
\begin{split}
&\sum_{i\leq 4}\|\sum_{i_3+i_4=i}u^{i+1}\nab^{i_3}(\chih,\partial_4\phi)\nab^{i_4}(\chih, \partial_4\phi)\|_{L_{\ub}^{1}L^{2}(\S)}\\
\ls& \ub\sum_{i_3\leq 2}\|u^{i_3+1}\nab^{i_3}(\chih, \partial_4\phi)\|_{L_{u}^{\i} L^{\i}_{\ub} L^{\i}(\S)}\sum_{i_4\leq 4}\|u^{i_4}\nab^{i_4}(\chih, \partial_4\phi)\|_{L_{u}^\infty L^{\i}_{u} L^2(\S)}
\ls \O^{-2}\ub a O^2.
\end{split}
\end{equation*} 

\noindent Similarly, contribution from $\sum_{i_1+i_2+i_3+i_4=i}\nab^{i_1}\q^{i_2}\nab^{i_3}\p\nab^{i_4}\trch$ with $i_2=0$ obeys
\begin{equation*}
\begin{split}
&\sum_{i\leq 4}\|\sum_{i_3+i_4=i}\O^{-1}u^{i+1}\nab^{i_3}\p\nab^{i_4}\trch\|_{L_{\ub}^{1}L^{2}(\S)}\\
\ls & \frac{\ub \at O}{|u|}\sum_{i\leq 4}\|u^{i+1}\nab^{i}(\O^{-1}\trch-(\O^{-1}\tr\chi)_0)\|_{L_{u}^\infty L^{\i}_{u} L^2(\S)}+{\color{black}\O^{-2}\ub a O^2} .
\end{split}
\end{equation*} 
\noindent When $i_2=0$, we can also control the third term $\sum_{i_1+i_2+i_3=i}\frac{1}{|u|}\nab^{i_1}\q^{i_2}\nab^{i_3}\p$ by
\begin{equation*}
\begin{split}
\sum_{i\leq 4}\|u^{i}\nab^{i}\p\|_{L_{\ub}^{1}L^{2}(\S)}
\ls \ub\sum_{i\leq 4}\|u^{i}\nab^{i}\p\|_{L_{u}^\infty L^{\i}_{u} L^2(\S)}
\ls & \ub \at O=\O^{-2}\ub a\cdot \f{\O^2 O}{\at} \leq \O^{-2}\ub a B^{-\f12}.
\end{split}
\end{equation*} 

\noindent For the case $i_2\geq 1$, using the bootstrap assumption \eqref{BA.1} together with Proposition \ref{product}, for the first two terms we also have 
\begin{equation*}
\begin{split}
\sum_{i\leq 4}\|\sum_{i_1+i_2+i_3+i_4=i-1}u^{i+1}\nab^{i_1}\q^{i_2+1}\nab^{i_3}\p\nab^{i_4}\p\|_{L_{\ub}^{1}L^{2}(\S)}\ls  \frac{\ub^2 a^{\f32}O^3}{|u|}\leq \O^{-2}\ub a B^{-\f14}.
\end{split}
\end{equation*}
The third term can be bounded similarly via
\begin{equation*}
\begin{split}
\sum_{i\leq 4}\|\sum_{i_1+i_2+i_3=i-1}u^{i}\nab^{i_1}\q^{i_2+1}\nab^{i_3}\p\|_{L_{\ub}^{1}L^{2}(\S)}\ls  \frac{\ub^2 a O^2}{|u|}\leq \O^{-2}\ub a B^{-\f14}. 
\end{split}
\end{equation*} 

Gathering above estimates, we obtain
\begin{equation*}
\begin{split}
&\sum_{i\leq 4} \|u^{i+1}\nab^i(\O^{-1}\trch-(\O^{-1}\tr\chi)_0)\|_{L^{\i}_u L^\infty_{\ub}L^2(\S)}\\
\ls &\O^{-2}\ub a O^2+\f{\ub\at O}{|u|}\sum_{i\leq 4} \|u^{i+1}\nab^i(\O^{-1}\trch-(\O^{-1}\tr\chi)_0)\|_{L^{\i}_ uL^\infty_{\ub}L^2(\S)}.
\end{split}
\end{equation*}
Since $\ub\at O/|u|\leq B^{-\f12}\ll 1$, we hence prove
{\color{black}
$$\sum_{i\leq 4} \|u^{i+1}\nab^i(\O^{-1}\trch-(\O^{-1}\tr\chi)_0)\|_{L^{\i}_u L^\infty_{\ub}L^2(\S)}\ls \O^{-2}\ub a O^2.$$
Using $\nab_A\O=\O(\eta_A+\etb_A)/2$, this further implies
}
\begin{align*}
\sum_{i\leq 4} \|u^{i+1}\nab^i(\O\trch-&\O^2(\O^{-1}\tr\chi)_0)\|_{L^{\i}_u L^\infty_{\ub}L^2(\S)}\ls \ub a O^2,\\
\sum_{i\leq 4} \|u^{i}\nab^i(\O&\trch)\|_{L^{\i}_u L^\infty_{\ub}L^2(\S)}\ls \f{\ub a O^2}{|u|} \ll \at.
\end{align*}
\end{proof} 

We proceed to prove the estimates for $\h\om$ and its derivatives. And $\h\o$ also obeys a slightly better bound than a general component $\p$. 

\begin{proposition}\label{om.bd}
Under the assumptions of Theorem \ref{main thm} and the bootstrap assumptions \eqref{BA.1}, \eqref{BA.2}, \eqref{BA.3}, \eqref{BA.4}, it holds
$$\sum_{i\leq 4} \|u^i\nab^i(\O\h\om)\|_{L^2(\S)}\ls \at+\sum_{i\leq 4} \|u^i\nab^i(\O \nab_4\phi)\|_{L^2(\S)}.$$ 
\end{proposition}

\begin{proof}
We employ the following schematic null structure equation for $\o$:
$$\O\nab_3(\O\h\o)={\color{black}\O^2}K+\p\q+\q\q+ \O\tr\chi\O\tr\chib+\O\nab_4\phi\cdot\O\nab_3\phi.$$ 
By commuting it with angular derivative for $i$ times, we obtain
\begin{equation*}
\begin{split}
&\O\nab_3 \nab^i(\O\h\o) +\frac i2 \O\trchb\nab^i(\O\h\o)\\
= &\sum_{i_1+i_2+i_3\leq i}{\color{black}\O^2}\nab^{i_1}\q^{i_2}\nab^{i_3}K+\sum_{i_1+i_2=i} \nab^{i_1}\q^{i_2+2}+\sum_{i_1+i_2+i_3=i}\nabla^{i_1}\q^{i_2+1}\nab^{i_3}\p\\
&+\sum_{i_1+i_2+i_3=i-1}\frac {1}{|u|} \nabla^{i_1}\q^{i_2+1}\nab^{i_3}\p+\frac {1}{|u|} \nab^i(\O\trch, \O\nab_4\phi).
\end{split}
\end{equation*}
{\color{black}Note that we only encounter the fourth term when $i\geq 1$.} 

We then apply Proposition \ref{evolution lemma} with $\lambda_0=\frac i2$. Since $\h\o(-1, \ub)$=0, we control $\|u^{i-1}\nab^{i}(\O\h\o)\|_{L_u^{\i}L_{\ub}^{\i}L^{2}(\S)}$ by the $\|u^{i-1}\cdot\|_{L_{\ub}^{\i}L_{u}^{1}L^{2}(\S)}$ norm of the right hand side.  

\begin{equation*}
\begin{split}
\mbox{Denote} \quad F_i:=&\sum_{i_1+i_2=i} \nab^{i_1}\q^{i_2+2}+\sum_{i_1+i_2+i_3=i}\nabla^{i_1}\q^{i_2+1}\nab^{i_3}\p\\
&+\sum_{i_1+i_2+i_3=i-1}\frac {1}{|u|} \nabla^{i_1}\q^{i_2+1}\nab^{i_3}\p.
\end{split}
\end{equation*}
We first have
\begin{equation*}
\begin{split}
\sum_{i\leq 4} \|u^{i-1}F_i\|_{L^{\i}_{\ub}L^1_uL^2(\S)}\ls\sum_{i\leq 4} \frac{1}{|u|}\|u^{i}F_i\|_{L^{\i}_{\ub}L^1_uL^2(\S)}\ls \frac{\ub^{\f12}\at O^2}{|u|^{\f32}}\leq{\color{black}\f{\at B^{-\f14}}{|u|}}.
\end{split}
\end{equation*}

\noindent For the terms containing $K$, if needed, we also employ $K\ls (\K)+\frac 1{|u|^2}$. Applying bootstrap assumptions \eqref{BA.3},\eqref{BA.1} and Proposition \ref{product}, their contributions obey
$$\lesssim \f{\ub^{\f12}\at(R+O)}{|u|^{\f32}}+\f{\ub \at O'}{|u|^{\f52}}\cdot \ub^{\f12}\at (R+O)\leq \f{\at B^{-\f14}}{|u|}.$$

We proceed to control the contribution from $\frac{1}{|u|^2}$. Note that for this term, the case $i_2=0$ happens only when $i=0$. For this scenario, the term obeys
\begin{equation}\label{om.0}
\begin{split}
\|u^{-1}\frac{1}{|u|^2}\|_{L_{u}^{1}L^{2}(\S)}
\ls &\frac 1{|u|}.
\end{split}
\end{equation}
For $i_2\geq 1$, it holds 
\begin{equation*}
\begin{split}
\sum_{i\leq 4}\|\sum_{i_1+i_2=i-1}u^{i-1}\nab^{i_1}\q^{i_2+1}\frac{1}{|u|^2}\|_{L_{u}^{1}L^{2}(\S)}\ls \frac{\ub\at O}{|u|^2}{\color{black}\leq\f{\at B^{-\f14}}{|u|}}. 
\end{split}
\end{equation*}
Here in the second inequality, we employ the bootstrap assumption \eqref{BA.1}. 
Collecting all the above estimates, we prove
\begin{equation}\label{om.1}
\sum_{i\leq 4}\|\sum_{i_1+i_2+i_3=i-1}{\color{black}\O^2}u^{i-1}\nab^{i_1}\q^{i_2+1}\nab^{i_3}(\K)\|_{L_{u}^{1}L^{2}(\S)}\ls \frac{1}{|u|}+{\color{black}\f{\at B^{-\f14}}{|u|}}.
\end{equation}
For the remaining term, using the bound for $\nab^i(\O\trch)$ in Proposition \ref{trch.bd}, we have 
\begin{equation}\label{om.2}
\begin{split}
\sum_{i\leq 4}\|u^{i-2}\nab^{i}(\O\trch, \O\nab_4\phi)\|_{L_{u}^{1}L^{2}(\S)}\ls& \frac{\ub a}{|u|^2}+\frac 1{|u|}+\f{1}{|u|}\sum_{i\leq 4}\|u^i \nab^i (\O\nab_4\phi)\|_{L^2(\S)}\\
 \ls& \f{\ub^{\f12} a^{\f34}}{|u|^{\f32}}+\f{1}{|u|}\sum_{i\leq 4}\|u^i \nab^i (\O\nab_4\phi)\|_{L^2(\S)}.
\end{split}
\end{equation} 

Combining all above estimates, we obtain 
$$\sum_{i\leq 4}\|u^{i-1}\nab^i (\O\h\om)\|_{L^2(\S)}\ls  \frac 1{|u|}+\f{\at B^{-\f14}}{|u|}+\f{1}{|u|}\sum_{i\leq 4}\|u^i \nab^i (\O\nab_4\phi)\|_{L^2(\S)} .$$
Multiplying by $|u|$ on both sides, we then derive
\begin{equation}\label{omega L2 nabla4phi}
\sum_{i\leq 4}\|u^{i}\nab^i (\O\h\om)\|_{L^2(\S)}\ls 1+\at B^{-\f14}+\sum_{i\leq 4}\|u^i \nab^i (\O\nab_4\phi)\|_{L^2(\S)}.
\end{equation}
\end{proof}
\begin{remark}
Note that in Proposition \ref{L2 scalar field}, we will prove
$$\sum_{i\leq 4}\|u^i \nab^i (\O\nab_4\phi)\|_{L^2(\S)}\ls \at.$$
Once this is obtained, back to \eqref{omega L2 nabla4phi} we will get
$$\sum_{i\leq 4}\|u^i \nab^i (\O\h\om)\|_{L^2(\S)}\ls \at.$$
\end{remark}

We now estimate the $L^2(\S)$ norms of the remaining Ricci coefficients.
\begin{proposition} \label{q.bd}
Under the assumptions of Theorem \ref{main thm} and the bootstrap assumptions \eqref{BA.1}, \eqref{BA.2}, \eqref{BA.3}, \eqref{BA.4}, we get
\[
 \sum_{i\leq 4}\frac{1}{\ub\at}\|u^{i+1}\nab^i\q\|_{L^2(\S)} \ls 1{\color{black}+\M O(\h\o)+\tilde{\M {O}}_{5,2}(\h\o)}+\mathcal R(\beta).
\]
{\color{black}In particular, it} holds 
\[
 \sum_{i\leq 4}\frac{1}{\ub\at}\|u^{i+1}\nab^i\q\|_{L^2(\S)} \ls 1{\color{black}+\M O(\h\o)+\f{1}{\ub^{\f12}a^{\f12}}\|u^5\nab^5(\O\h\om)\|_{L^\i_uL^2_{\ub}L^2(\S)}}+\mathcal R(\b).
\]

\end{proposition} 

\begin{proof}

For $\q\in\{\eta,\nab\log\Omega,\O\trchb-(\O\tr\chib)_0,\O\chibh {\color{black},\O\omb-(\O\omb)_0}\}$, the below schematic transport equation holds 
\begin{equation}\label{Omega nabla 4 psi}
\O\nab_4\q=\O\b{\color{black}+\nab(\O\h\om)}+{\color{black}\O^2}K+{\color{black}\O^2}\nab\etab+\p\q+\q\q+\f{1}{|u|}\p. \footnote{
{\color{black}
On the right of the equation \eqref{Omega nabla 4 psi}, we have that $\O^2 K$ obeys at least the same and even better estimates compared with $K$.  When there is no danger of confusion, we sometimes also write $K$ instead of $\O^2 K$ for the notational convenience and then the $\O^2$ could be viewed as a coefficient like $1$ and is usually omitted in the schematic writing.  
}
}
\end{equation}

\noindent By commuting the above equation with $i$ angular derivatives, we have 
\begin{equation*}
\begin{split}
\O\nab_4 \nab^i\q= &\nab^i(\O\b)+\nab^{i+1}({\color{black}\O\h\om}, {\color{black}\O^2}\etab)+\sum_{i_1+i_2+i_3=i} {\color{black}\O^2}\nabla^{i_1}\q^{i_2}\nab^{i_3}K\\
&+\sum_{i_1+i_2+i_3=i}\nabla^{i_1}\q^{i_2+1}\nab^{i_3}(\p,\q)+\sum_{i_1+i_2+i_3=i}\frac {1}{|u|}\nabla^{i_1}\q^{i_2}\nab^{i_3}\p.
\end{split}
\end{equation*}
\noindent Notice $\nab(\O\omb)_0=0$. Due to the {\color{black}naked-singularity} initial data along $\ub=0$, all the quantities $\q$ are initially $0$. By Proposition \ref{transport}, hence to estimate the quantity $\|u^i\nab^i\q\|_{L_{\ub}^{\i}L^{2}(\S)}$, we only need to bound the $\|u^i\cdot\|_{L_{\ub}^{1}L^{2}(\S)}$ norm of the right hand side. We proceed to control each of the terms
in the equation. We bound the $\O\beta$ term first 

\begin{equation}\label{q.1}
\sum_{i\leq 4}\|u^i\nab^i(\O\b)\|_{L_{\ub}^{1}L^{2}(\S)}\leq \frac{\ub^{\f12}}{|u|}\sum_{i\leq 4}\|u^{i+1}\nab^i(\O\b)\|_{L_{\ub}^{2}L^{2}(\S)}\ls \f{\ub\at}{|u|}\M R(\b).
\end{equation}
{\color{black}
We then estimate the term $\nab^{i+1}(\O\h\om)$ and get
\begin{equation}\label{q.2}
\begin{split}
&\sum_{i\leq 4}\|u^i\nab^{i+1}(\O\h\om)\|_{L_{\ub}^{1}L^{2}(\S)}\\
\ls &\frac{\ub}{|u|}\sum_{i\leq 3}\|u^{i+1}\nab^{i+1}{\color{black}(\O\h\o)}\|_{L_{\ub}^{\i}L^{2}(\S)}+\frac{\ub^{\f12}}{|u|}\|u^5\nab^5(\O\h\om)\|_{L_{\ub}^{2}L^{2}(\S)}\\
\ls &\frac{\ub a^{\frac 12}}{|u|}(\M O_{i,2}(\h\o)+\tilde{\M {O}}_{5,2}{\color{black}(\h\o)}).
\end{split}
\end{equation}
}

For the $\nab^{i+1}\etab$ term, using the bootstrap assumptions \eqref{BA.1} and \eqref{BA.4}, it holds
\begin{equation}\label{q.3}
\begin{split}
&\sum_{i\leq 4}\|{\color{black}\O^2}u^i\nab^{i+1}\etab\|_{L_{\ub}^{1}L^{2}(\S)}\\
\ls &\frac{\ub}{|u|^2}\sum_{i\leq 3}\|u^{i+2}\nab^{i+1}\q\|_{L_{\ub}^{\i}L^{2}(\S)}+\frac{\ub^{\f12}}{{\color{black}|u|}}\|{\color{black}\O^2 u^5}\nab^5\etab\|_{L_{\ub}^{2}L^{2}(\S)}\\
\ls &\frac{\ub^2 a^{\frac 12} O}{|u|^2}+{\color{black}\f{\ub^{\f12}}{|u|}\cdot\O^2\cdot\ub^{\f12}\af\cdot\f{\ub^{\f12}a^{\f14}}{|u|^{\f12}\O} }\cdot\tilde{O}\leq \f{\ub^2 a^{\f34} O}{|u|^2}+\f{\ub^{\f32}a^{\f34}}{|u|^{\f32}}\leq \f{\ub\at}{|u|}.
\end{split}
\end{equation} 

Employing the bootstrap assumptions \eqref{BA.1}, \eqref{BA.3}, we then bound the term containing the Gauss curvature $K$ {\color{black}and we write  
$$K=(K-\f{1}{|u|^2}-\f14\nab^A\phi\nab_A\phi)+\f{1}{|u|^2}+\f14\nab^A\phi\nab_A\phi$$
to deduce the estimate.} It holds
\begin{equation}\label{q.4}
\begin{split}
&\sum_{i\leq 4}\|\sum_{i_1+i_2+i_3=i} {\color{black}\O^2}u^i\nab^{i_1}\q^{i_2}\nab^{i_3}K\|_{L_{\ub}^{1}L^{2}(\S)} \leq \f{\ub\at}{|u|}. 
\end{split}
\end{equation}

Applying \eqref{BA.1} to control $\p$ and using Proposition \ref{product} to estimate the product of $\q$, as well as using Sobolev embedding in Proposition \ref{Sobolev}, for $i\leq 4$ we bound the fourth term and the final term as follows: 
\begin{equation}\label{q.6}
\begin{split}
\|\sum_{i_1+i_2+i_3=i}u^i\nabla^{i_1}\q^{i_2+1}\nab^{i_3}(\p,\q)\|_{L_{\ub}^{1}L^{2}(\S)}\ls \frac{\ub^2 a O^2}{|u|^2}\leq \f{\ub\at B^{-\f14}}{|u|},
\end{split}
\end{equation}

\begin{equation}\label{q.7} 
\begin{split}
\|\sum_{i_1+i_2+i_3=i} u^{i-1}\nabla^{i_1}\q^{i_2}\nab^{i_3}\p\|_{L_{\ub}^{1}L^{2}(\S)}\ls \f{\ub \at}{|u|}(1+\M O_{i,2}(\h\o))+\frac{\ub^2 a O^2}{|u|^2}.
\end{split}
\end{equation}
{\color{black}

\noindent Notice that here $\|u^i\nab^i\p\|_{L^{\i}_{\ub}L^2(\S)}\leq \at (1+\M O_{i,2}(\h\o))$ was derived before.}

Collecting all above estimates, we hence obtain 
$$\sum_{i\leq 4}\|u^i\nab^i\q\|_{L^{\i}_{\ub}L^2(\S)}\ls \f{\ub\at}{|u|}[1{\color{black}+\sum_{i\leq 4}\M O_{i,2}(\h\o)+\tilde{\M {O}}_{5,2}(\h\o)}+\mathcal R(\b)].$$
\end{proof}

For future use, we also prove an improved bound for $\nab^i({\color{black}\O}\trchb-(\O\tr\chib)_0)$ for $i\geq 1$:

\begin{proposition} \label{trchb.bd}
Under the assumptions of Theorem \ref{main thm} and the bootstrap assumptions \eqref{BA.1}, \eqref{BA.2}, \eqref{BA.3}, \eqref{BA.4}, it holds that
\[
 \sum_{1\leq i\leq 4}\|u^i\nab^i(\O\trchb-(\O\tr\chib)_0)\|_{L^2(\S)} \ls \frac{\ub^{\f32} a^{\f34}}{|u|^{\f32}}{\color{black}\leq \f{\ub\at B^{-\f12}}{|u|}}.
\]

\end{proposition}  

\begin{proof}
Note that $\O\trchb-(\O\tr\chib)_0$ obeys the equation
$$\O\nab_4(\O\trchb-(\O\tr\chib)_0)={\color{black}\O^2}K+{\color{black}\O^2}\nab\etab+\p\q+\q\q+\f1u \O\trch.$$
By commuting the above equation with $i$ angular derivatives, we have 
\begin{equation*}
\begin{split}
&\O\nab_4 \nab^i(\O\trchb-(\O\tr\chib)_0) \\
= &\O^2\nab^{i+1}\etab+\sum_{i_1+i_2+i_3=i}\O^2\nabla^{i_1}\q^{i_2}\nab^{i_3}K+\sum_{i_1+i_2+i_3=i}\nabla^{i_1}\q^{i_2+1}\nab^{i_3}(\p,\q)\\
&+\sum_{i_1+i_2+i_3=i}\frac {1}{|u|}\nabla^{i_1}\q^{i_2}\nab^{i_3}(\O\trch).
\end{split}
\end{equation*}

We revisit the proof of Proposition \ref{q.bd}. Except for the last term, all other terms obey the $\ub^{\f32}a^{\f34}/|u|^{\f32}$ upper bound. And for the last term, we have 
\begin{equation*}
\begin{split}
\sum_{1\leq i\leq 4}\|\sum_{i_1+i_2+i_3=i} u^{i-1}\nabla^{i_1}\q^{i_2}\nab^{i_3}(\O\trch)\|_{L_{\ub}^{1}L^{2}(\S)}\ls \f{\ub^{\f32} a^{\f34}}{|u|^{\f32}}{\color{black}\leq \f{\ub\at B^{-\f12}}{|u|}}.
\end{split}
\end{equation*}
Here besides using Proposition \ref{product}, we also employ the improved bound for $\nab^i(\O\trch)$ derived in Propositions \ref{trch.bd}. Collecting this with \eqref{q.3} and \eqref{q.6}, we then obtain the desired conclusion.
\end{proof}

We also notice that the $\nab_4$ equation for $K-\f{1}{|u|^2}{\color{black}-\f14 \nab^A\phi\nab_A\phi}$ contains exactly the same type of terms as the $\nab_4$ equation for $\nab\q$. Hence, for $\nab^i(\K{\color{black}-\f14 \nab^A\phi\nab_A\phi})$ we have the same estimates as for $\nab^{i+1}\q$ for $i\leq 3$. 
\begin{proposition} \label{K.bd}
With the assumptions of Theorem \ref{main thm} and the bootstrap assumptions \eqref{BA.1}, \eqref{BA.2}, \eqref{BA.3} and \eqref{BA.4}, it holds
\[
 \sum_{i\leq 3}\|u^{i+2}\nab^{i}(\K{\color{black}-\f14 \nab^A\phi\nab_A\phi})\|_{L^\i_uL^\i_{\ub}L^2(\S)} \ls \ub a^{\f12} [1+\sum_{i\leq 4}\M O_{i,2}(\h\o)+\tilde{\M O}_{5,2}(\h\o)+\M R].
\]
In particular,
\[
 \sum_{i\leq 3}\|u^{i+1}\nab^{i}(\K{\color{black}-\f14 \nab^A\phi\nab_A\phi})\|_{L^\i_uL^\i_{\ub}L^2(\S)} \ls \f{\ub\at}{|u|}[1+\sum_{i\leq 4}\M O_{i,2}(\h\o)+\tilde{\M O}_{5,2}(\h\o)+\M R]\ll 1.
\]
\end{proposition} 

\section{$L^2(\S)$ Estimates for Scalar Field}\label{scalar field} 
We now move to derive the $L^2(\S)$ estimates for the scalar field. We first prove 
  
\begin{proposition}\label{L2 scalar field}
Under the assumptions of Theorem \ref{main thm} and the bootstrap assumptions \eqref{BA.1}, \eqref{BA.2}, \eqref{BA.3}, \eqref{BA.4}, we have
$$\sum_{i=0}^4\|u^i\nab^i(\O e_4\phi)\|_{L^2(\S)}\ls \at, \quad \sum_{i=0}^4\|u^i\nab^i(e_A\phi)\|_{L^2(\S)}\ls \f{\ub\at S}{|u|},$$
$$\sum_{i=0}^4\|u^i\nab^i(\O e_3\phi-(\O e_3\phi)_0)\|_{L^2(\S)}\ls \f{\ub\at S}{|u|}.$$
\end{proposition}
\begin{proof} 

To estimate $\O e_4\phi$, we first rewrite equation (\ref{e3e4phi}) as
\begin{equation}\label{e3e4phi v2} 
(\O e_3)(\O e_4\phi)+\f12{\O\tr\chib}\O e_4\phi=\O^2\Delta_g \phi-\f12{\O\tr\chi}\O e_3\phi{\color{black}+2\O^2\eta^A e_A\phi}.
\end{equation}
Commuting the above equation with $i$ angular derivatives, we then obtain 
\begin{equation*}
\begin{split}
&(\O\nab_3)\nab^i(\O e_4\phi)+\f{i+1}{2}\O\tr\chib\nab^i(\O e_4\phi)\\
=&\sum_{i_1+i_2+i_3=i}\nab^{i_1}(\eta+\etb)^{i_2}\nab^{i_3}\bigg(\O^2\Delta_g \phi-\f12{\O\tr\chi}\O e_3\phi{\color{black}+2\O^2\eta^A e_A\phi} \bigg) \\
&+\sum_{i_1+i_2+i_3+i_4=i}\nab^{i_1}(\eta+\etb)^{i_2}\nab^{i_3}(\O\chibh, \O\tr\chib+\f{2}{|u|})\nab^{i_4}(\O e_4\phi)\\
&+\sum_{i_1+i_2+i_3+i_4=i-1}\nab^{i_1}(\eta+\etb)^{i_2}\nab^{i_3}(\O\beb)\nab^{i_4}(\O e_4\phi).
\end{split}
\end{equation*} 
{\color{black}Note that the last term only appears when $i\geq 1$.}Using Proposition \ref{evolution lemma} with $\lambda_0=\f{i+1}{2}$, we have that $\|u^i\nab^i(\O e_4\phi)\|_{L^2(\S)}$ is controlled by the  $\|u^i\cdot\|{L^1_{u}L^2(\S)}$ norm of the right. We first estimate 
\begin{equation*}
\begin{split}
&\|u^i \sum_{i_1+i_2+i_3=i}\nab^{i_1}(\eta+\etb)^{i_2}\nab^{i_3}(\O^2\Delta_g \phi)\|_{L^1_u L^2(\S)}\\
\leq& \|\O^2 u^i \nab^{i+1} (e_A\phi)\|_{L^1_{u}L^2(\S)}{\color{black}+\|\O^2 u^i \sum_{i_1+i_2+i_3=i-1}\nab^{i_1}(\eta+\etb)^{i_2+1}\nab^{i_3}\Delta_g\phi\|_{L^1_{u}L^2(\S)}}\\
\leq& \| \O u^{i+1} \nab^{i+1} (e_A\phi)\|_{L^2_{u}L^2(\S)}\|\f{\O}{|u|}\|_{L^2_{u}L^{\infty}(\S)}{\color{black}+\f{\ub\at O}{|u|}\cdot\f{\ub\at S}{|u|^3}|u|^2}\\
\leq& \f{\ub^{\f12}\at\M S}{|u|^{\f12}}{\color{black}+\f{\ub^2 a OS}{|u|^2}}\leq \at B^{-\f12}.
\end{split}
\end{equation*}
Here for the last line we use bootstrap assumption 
$$\|u^5\nab^5(\O e_4\phi)\|_{L^2_{\ub}L^2(\S)}+\|\O u^5\nab^5(e_A\phi)\|_{L^2_{u}L^2(\S)}\ls  \ub^{\f12}\at S.$$
The terms including $-\f12\O\tr\chi\cdot\O e_3\phi$ are lower order and they obey 
\begin{equation*}
\begin{split}
\|u^i\sum_{i_1+i_2+i_3=i}\nab^{i_1}(\eta+\etb)^{i_2}\nab^{i_3}(-\f12\O\tr\chi\cdot\O e_3\phi)\|_{L^1_u L^2(\S)}\leq& \f{\ub a}{|u|}\leq \at B^{-1},
\end{split}
\end{equation*}
where to control $\|u^i\nab^i(\O\tr\chi)\cdot\O e_3\phi\|_{L^1_u L^2(\S)}$ we employ improved estimates in Proposition \eqref{trch.bd}.  

Using the schematic Codazzi equation, we can rewrite $\O\beb=\nab\q+\q(\O\tr\chib+\q)$. All the rest terms are already estimated in Proposition \ref{chih.bd} and they are bounded by $\at$. Hence we prove
\begin{equation}\label{L2 e4 phi}
\|u^i\nab^i(\O e_4\phi)\|_{L^{\infty}_{\ub}L^2(\S)}\leq \at. 
\end{equation}

To derive estimates for $e_A\phi$, we employ the equation
{\color{black}
$$\nab_{\O e_4}\nab_A\phi=e_A(\O e_4\phi)-\O{\chi_A}^B e_B\phi.$$
}
Commuting this equation with $i$ angular derivatives, it holds
\begin{equation*}
\begin{split}
\O D_4(\nab^i\nab_A\phi)=&\sum_{i_1+i_2+i_3=i}\nab^{i_1}(\eta+\etb)^{i_2}\nab^{i_3}\nab_A(\O e_4\phi)\\
&+\sum_{i_1+i_2+i_3+i_4=i}\nab^{i_1}(\eta+\etb)^{i_2}\nab^{i_3}(\O \chi)\nab^{i_4}\nab_A\phi\\
&+\sum_{i_1+i_2+i_3+i_4=i-1}\nab^{i_1}(\eta+\etb)^{i_2}\nab^{i_3}(\O\b)\nab^{i_4}\nab_A\phi.
\end{split}
\end{equation*}

\noindent Via Proposition \ref{transport}, we then bound $\|u^{i}\nab^{i}(e_A\phi)\|_{L^2(\S)}$ by the $\|u^{i}\cdot\|_{L^1_{\ub}L^2(\S)}$ norm of the right hand side. With the aid of bootstrap assumption \ref{BA.3}, the term with the top-order derivative can be controlled by
\begin{equation*}
\begin{split}
 \|u^i \nab^{i+1} (\O e_4\phi)\|_{L^1_{\ub}L^2(\S)}\leq& \|u^{i+1}\nab^{i+1}(\O e_4\phi)\|_{L^2_{\ub}L^2(\S)}\|\f{1}{|u|}\|_{L^2_{\ub}L^{\infty}(\S)}\\
 \leq&{\ub^{\f12}\at \M S}\cdot{\f{\ub^{\f12}}{|u|}}= \f{\ub \at}{|u|}\M S.
\end{split}
\end{equation*}
All the rest terms can be estimated as in Proposition \ref{q.bd}. There is neither $\nab^i(\O\b)$ nor $\nab^{i+1}(\O\h\o)$. And they obey the upper bound $\ub\at/|u|$. Thus, we arrive at
$$\sum_{i=0}^4\|u^i\nab^i(e_A\phi)\|_{L^{\infty}_{\ub}L^2(\S)}\ls \f{\ub\at S}{|u|}.$$

To estimate $(\O e_3\phi)-(\O e_3\phi)_0$, using (\ref{e3e4phi2}) and commutation formula, we get 
\begin{equation*}
\begin{split}
&(\O \nab_4)\nab^i(\O e_3\phi-(\O e_3\phi)_0)\\
=&\sum_{i_1+i_2+i_3=i}\nab^{i_1}(\eta+\etb)^{i_2}\nab^{i_3}(\O^2\Delta_g\phi-\f12\O\tr\chi\O e_3\phi-\f12\O\tr\chib\O e_4\phi+2\O^2\eta^A e_A\phi)\\
&+\sum_{i_1+i_2+i_3+i_4=i}\nab^{i_1}(\eta+\etb)^{i_2}\nab^{i_3}(\O\chi)\nab^{i_4}(\O e_3\phi)+\sum_{i_1+i_2+i_3+i_4=i-1}\nab^{i_1}(\eta+\etb)^{i_2}\nab^{i_3}(\O\b)\nab^{i_4}(\O e_3\phi).
\end{split}
\end{equation*}  

With Proposition \ref{transport}, we control $\|u^i\nab^i[(\O e_3\phi)-(\O e_3\phi)_0]\|_{L^{\infty}_{\ub}L^2(\S)}$ by the $\|u^i\cdot\|_{L^1_{\ub}L^2(\S)}$ norm of the right. The term with the top-order derivative is bounded by   
\begin{equation*}
\begin{split}
 \|\O^2 u^i \nab^{i+1} (e_A\phi)\|_{L^1_{\ub}L^2(\S)}\leq&\|u^{i+1}\nab^{i+1}(e_A\phi)\|_{L^2_{\ub}L^2(\S)}\|\f{\O^2}{|u|}\|_{L^2_{\ub}L^{\infty}(\S)}\\
 \leq&{\ub^{\f12}\af}{\color{black}\F}\M S\cdot{\f{\ub^{\f12}}{|u|}}{\color{black}\O^2} \leq \f{\ub\at \M S}{|u|}. 
\end{split}
\end{equation*}
Here for the second inequality, we use bootstrap assumption \ref{BA.3}
$$\|u^5\nab^5(e_A\phi)\|_{L^2_{\ub}L^2(\S)}+\|\O^{-1}u^5\nab^5(\O e_3\phi)\|_{L^2_{u}L^2(\S)}\ls  {\color{black}\big(1+\f{\ub^{\f12}\af}{|u|^{\f12}\O}\big)\ub^{\f12}a^{\f14}}\M S.$$ 

\noindent Two borderline terms on the right are $(\O e_3\phi)\nab^i(\O\chi)$ and $(\O\tr\chib)\nab^i(\O e_4\phi)$ with $i=4$. And they obey
\begin{equation*}
\begin{split}
\|u^i(\O e_3\phi)\nab^{i}(\O\chi)\|_{L^1_{\ub}L^2(\S)}\leq&\mathcal{O}_{i,2}(\O\chi)\f{\ub\at}{|u|},
\end{split}
\end{equation*}
\begin{equation*} 
\begin{split}
\|u^i(\O\tr\chib)\nab^{i}(\O e_4\phi)\|_{L^1_{\ub}L^2(\S)}\leq&\mathcal{O}_{i,2}(\O e_4\phi)\f{\ub\at}{|u|}.
\end{split}
\end{equation*}
Note that by Proposition \ref{chih.bd}, Proposition \ref{trch.bd} and inequality \eqref{L2 e4 phi}, we have that
$\mathcal{O}_{i,2}(\O\chi)$, $\mathcal{O}_{i,2}(\O e_4\phi)$ depend only on initial data. 

All the rest terms can be estimated as in Proposition \ref{q.bd}. And they are
$$\leq \f{\ub\at}{|u|}\cdot\f{\ub\at}{|u|}\cdot (O+S)\leq \f{\ub\at B^{-\f12}}{|u|}.$$

\noindent We hence arrive at
$$\sum_{i=0}^4\|u^i\nab^i(\O e_3\phi-(\O e_3\phi)_0)\|_{L^2(\S)}\ls \f{\ub\at \M S}{|u|}.$$
\end{proof}

\section{Energy Estimates for Scalar Field}\label{energy scalar field} 
We now move to derive the energy bounds for the top-order derivatives of the scalar field. We first recall the below covariant decompositions\footnote{ The decompositions are the same as in \cite{Chr:book}.} 

		\begin{gather}
		D_A e_{B} = \nabla_A e_{B} +{\color{black}\frac12} \chi_{AB} e_3 +{\color{black}\frac12} \chib_{AB} e_4, \\ 
		D_3 e_A = \nabla_3 e_A + {\color{black}\eta_A} e_3, \hspace{2mm} D_4 e_A = \nabla_4 e_A + {\color{black}\ \etb_A} e_4, \\  
		D_A e_3= {\chib_A}^{\sharp B}e_B + {\color{black} \zeta_A} e_3, \hspace{2mm} D_A e_4 = {\chi_A}^{\sharp B} e_B - {\color{black}\zeta_A} e_4, \\   
		D_3 e_4 = {\color{black}2}\eta^{\sharp A} e_A +{\color{black}2}\hspace{.5mm} \omb \hspace{.5mm}e_4, \hspace{2mm} D_4 e_3 = {\color{black}2}\etb^{\sharp A} e_A +{\color{black}2}\hspace{.5mm}\h\omega \hspace{.5mm} e_3, \\ 
		D_3 e_3 = {\color{black}-2}\omb e_3, \hspace{2mm} D_4 e_4 = {\color{black}-2} \h\omega e_4.
		\end{gather}
Next we start to use equation (\ref{e3e4phi}) 
\begin{equation*}
D_3 D_4\phi=D_4 D_3\phi=\Delta_g\phi-\f12\tr\chi e_3\phi-\f12\tr\chib e_4\phi.
\end{equation*}
With the fact $-2\h\o=\nab_4\log\O=\O^{-1}\nab_4\O$ and the properties of covariant derivative
$$D_4D_3\phi=D_4(D_3\phi)-D_{D_4 e_3}\phi=D_4(D_3\phi)-2\etb^A e_A(\phi)+\O^{-1} D_4\O \,e_3\phi,$$
we can write (\ref{e3e4phi}) as
$$D_4(\O D_3\phi)-2\O \etb^A e_A\phi=\O D_4 D_3\phi=\O D_3D_4\phi=\O\Delta_g\phi-\f12\O\tr\chi e_3\phi-\f12\O\tr\chib e_4\phi.$$
This implies
\begin{equation}\label{covariant transport 0}
\O D_4(\O D_3\phi)=\O^2 \nab^A\nab_A\phi-\f12\O\tr\chi \cdot\O e_3\phi-\f12\O\tr\chib \cdot \O e_4\phi+2\O^2 \etb^A e_A\phi.
\end{equation}
Similarly, via employing the fact $-2\omb=\nab_3\log\O=\O^{-1}\nab_3\O$ and using
\begin{equation*}
\begin{split}
D_3D_4\phi=&D_3(D_4\phi)-D_{D_3 e_4}\phi=D_3(D_4\phi)-D_{2\eta^A e_A+2\omb e_4}\phi\\
=&D_3(D_4\phi)-2\eta^A e_A\phi-2\omb e_4\phi=D_3(D_4\phi)-2\eta^A e_A\phi+\O^{-1}D_3\O e_4\phi,
\end{split}
\end{equation*}
we also obtain
\begin{equation}\label{Omega wave equation}
\O D_3(\O D_4\phi)=\O^2 \nab^A\nab_A\phi-\f12\O\tr\chi \cdot\O e_3\phi-\f12\O\tr\chib \cdot \O e_4\phi+2\O^2 \eta^A e_A\phi.
\end{equation}

We then proceed to derive the desired form for the transport equation of $\nab_A\phi$. For $f$ being a scalar function, {\color{black}we define $\slashed{d}f$ to be the $1$-form restricted to each surface $\S$ of $df$ (the differential of $f$). Denote }$L, X$ being vector fields, we first have
$$(\mathcal{L}_{L}\slashed{d}f)(X)=L((\slashed{df})(X))-(\slashed{d}f)[L,X]=L(Xf)-[L,X]f=X(Lf).$$
Letting $f=\phi, \,\, L=\f{\partial}{\partial u}$ and $X=e_A$, we hence obtain 
\begin{equation}\label{Lie transport 1}
(\mathcal{L}_{\f{\partial}{\partial u}}\slashed{d}\phi)_A=\slashed{d}_A(\f{\partial}{\partial u}\phi).
\end{equation}
For $(\mathcal{L}_{\f{\partial}{\partial u}}\slashed{d}\phi)_A$, we also have
\begin{equation}\label{Lie transport 2}
(\mathcal{L}_{\f{\partial}{\partial u}}\slashed{d}\phi)_A=\f{\partial}{\partial u}(e_A\phi)-[\f{\partial}{\partial u}, e_A]\phi=\f{\partial}{\partial u}(e_A\phi)-(D_{\f{\partial}{\partial u}}e_A) \, \phi+(D_{e_A}\f{\partial}{\partial u}) \,\phi=D_{\f{\partial}{\partial u}}D_{e_A}\,\phi+D_{e_A}(\O e_3)\phi.
\end{equation}
For the last term, it holds
\begin{equation}\label{covariant transport 1}
\begin{split}
D_{e_A}(\O e_3)\phi=&\O^{-1}e_A(\O)(\O e_3\phi)+\O D_{e_A}e_3 \phi=e_A(\log \O)(\O e_3\phi)+\O({{\chib_A}^B} e_B\phi+\zeta_A e_3\phi)\\
=&\eta_A(\O e_3\phi)+\O{\chib_A}^B e_B\phi.
\end{split}
\end{equation}
Back to (\ref{Lie transport 2}), employing (\ref{Lie transport 1}) and (\ref{covariant transport 1}), we get
\begin{equation}\label{covariant transport 2}
\begin{split}
D_{\f{\partial}{\partial u}}D_{e_A}\,\phi=&e_A(\f{\partial}{\partial u}\phi)-D_{e_A}(\O e_3) \,\phi\\
=&e_A(\f{\partial}{\partial u}\phi)-\eta_A(\O e_3\phi)-\O{\chib_A}^B e_B\phi.
\end{split}
\end{equation}
At the same time, employing properties of the covariant derivative and the fact $D_3 e_A=\nab_3 e_A+\eta_A e_3$, we also have
\begin{equation}\label{covariant transport 3}
\begin{split}
D_{\f{\partial}{\partial u}}D_{e_A}\,\phi=&\f{\partial}{\partial u}(e_A\phi)-D_{D_{\f{\partial}{\partial u}e_A}}\phi=\f{\partial}{\partial u}(e_A\phi)-D_{\nab_{\O e_3}e_A+\O\eta_A e_3}\phi\\
=&\nab_{\O e_3}e_A\phi-\O \eta_A e_3\phi.
\end{split}
\end{equation}
Comparing \eqref{covariant transport 2} with \eqref{covariant transport 3} and noting $\partial/\partial u=\O e_3$, we hence obtain
\begin{equation}\label{nabla transport 1}
\nab_{\O e_3}\nab_A\phi=e_A(\O e_3\phi)-\O{\chib_A}^B e_B\phi. 
\end{equation}
In the same fashion, similar to \eqref{covariant transport 1}, we get
$$(D_{e_A}\O e_4)\phi=e_A(\log\O)\O e_4\phi+\O({\chi_A}^B e_B-\zeta_A e_4)\phi=\etb_A(\O e_4\phi)+\O{\chi_A}^B e_B\phi.$$
This implies
\begin{equation}\label{covariant transport 4}
\begin{split}
D_{e_A}D_{\O e_4}\,\phi=e_A(\O e_4\phi)-(D_{e_A}\O e_4)\,\phi=e_A(\O e_4\phi)-\etb_A(\O e_4\phi)-\O {\chi_A}^B e_B\phi.
\end{split}
\end{equation}
On the other side, we also have
\begin{equation}\label{covariant transport 5}
\begin{split}
D_{\O e_4}D_{e_A}\,\phi=&\O e_4(e_A\phi)-D_{D_{\O e_4}e_A}\phi=\O e_4(e_A \phi)-D_{\nab_{\O e_4}e_A+\O\etb_A e_4}\phi\\
=&\nab_{\O e_4} e_A\,\phi-\O \etb_A e_4\phi.
\end{split}
\end{equation}
Employing $D_{e_A}D_{\O e_4}\,\phi=D_{\O e_4}D_{e_A}\,\phi$ and collection both \eqref{covariant transport 4} and \eqref{covariant transport 5}, we then arrive at
\begin{equation}\label{nabla transport 2}
\nab_{\O e_4}\nab_A\phi=e_A(\O e_4\phi)-\O {\chi_A}^B e_B\phi. 
\end{equation}

In the rest of this section, to derive energy estimates for the scalar field, we rearrange \eqref{covariant transport 0}, \eqref{Omega wave equation}, \eqref{nabla transport 1}, \eqref {nabla transport 2} as the below paired equations 
\begin{equation}\label{phi transport 1}
\begin{split}
&{\color{black}\nab_{\O e_4}(\O \nab_A)\phi=\O e_A(\O e_4\phi)-\O^2{\chih_A\,}^B\nab_B\phi-\f12\O^2\tr\chi\nab_A\phi+\O e_4(\O)\nab_A\phi,}
\end{split}
\end{equation}
\begin{equation}\label{phi transport 2} 
\begin{split}
&\O e_3(\O e_4\phi)=\O^2 \nab^A\nab_A\phi-\f12\O\tr\chib \cdot \O e_4\phi-\f12\O\tr\chi \cdot\O e_3\phi+2\O^2 {\color{black}\eta^A}\cdot e_A\phi,
\end{split}
\end{equation}
and
\begin{equation}\label{phi transport 3}
\begin{split}
&\nab_{\O e_3}\nab_A\phi+\f12\O\tr\chib\nab_A\phi=e_A(\O e_3\phi)-\O {\underline{\chih}_A }^B \nab_B\phi,
\end{split}
\end{equation}
\begin{equation}\label{phi transport 4}
\begin{split}
&\O e_4(\O e_3 \phi)=\O^2\nab^A\nab_A\phi-\f12\O\tr\chib\cdot\O e_4\phi-\f12\O\tr\chi\cdot\O e_3\phi+2\O^2{\color{black}\etb^A}\cdot e_A\phi.
\end{split}
\end{equation}

\noindent For the first pair, on both sides multiplying ${\color{black}\O}\nab^A\phi$ for the first equation and multiplying $\O e_4\phi$ for the second equation, we get
\begin{equation}\label{phi first pair}
\begin{split}
({\color{black}\O}\nab^A\phi) (\O e_4)({\color{black}\O}\nab_A\phi)=&{\color{black}\O^2}(\nab^A\phi) e_A(\O e_4\phi)+\O\nab^A\phi\nab_A\phi(\O e_4\O)\\
&+(\O \nab^A\phi)\cdot\big(-{\color{black}\O^2} {\chih_A\,}^B \nab_B\phi-\f12{\color{black}\O^2}\tr\chi \nab_A\phi\big),\\
(\O e_4\phi)\O e_3 (\O e_4\phi)=(\O e_4\phi)\O^2& \nab^A \nab_A\phi+(\O e_4\phi)\cdot\bigg(-\f12\O\tr\chib \cdot \O e_4\phi-\f12\O\tr\chi \cdot\O e_3\phi+2\O^2 {\color{black}\eta^A}\cdot e_A\phi\bigg).
\end{split}
\end{equation}
For the second pair, multiplying $\nab^A\phi$ for the first equation and multiplying ${\color{black}\O^{-2}} (\O e_3\phi)$ for the second equation, it holds
\begin{equation}\label{phi second pair}
\begin{split}
\nab^A\phi \nab_{\O e_3}\nab_A\phi+\f12\O\tr\chib&\nab^A\phi\nab_A\phi=\nab^A\phi e_A(\O e_3\phi)-\O {\underline{\chih}_A}^B \nab_B\phi\nab^A\phi,\\
{\color{black}\O^{-2}}(\O e_3\phi)\O e_4 (\O e_3\phi)=&{\color{black}\O^{-2}}(\O e_3\phi)\O^2\nab^A \nab_A\phi\\
&+{\color{black}\O^{-2}}(\O e_3\phi)\cdot\bigg(-\f12\O\tr\chib\cdot\O e_4\phi-\f12\O\tr\chi\cdot\O e_3\phi+2\O^2{\color{black}\etb^A}\cdot e_A\phi\bigg).
\end{split}
\end{equation}

If we add the two equations in \eqref{phi first pair} and integrate on $\S$, we will encounter the potential top-order-derivative terms
\begin{equation}\label{first pair sum}
\int_{\S}\O^2(\nab^A\phi) e_A(\O e_4\phi)+(\O e_4\phi)\O^2\nab^A\nab_A\phi.
\end{equation}
Repeating the same procedure, for \eqref{phi second pair} we will have
\begin{equation}\label{second pair sum}
\int_{\S}\nab^A\phi \, e_A(\O e_3\phi)+\O^{-2}(\O e_3\phi)\, \O^2 \nab^A\nab_A\phi.
\end{equation}
Conducting integration by parts with respect to $\nab_A$ for both \eqref{first pair sum} and \eqref{second pair sum}, the potential top-order-derivative terms will be canceled. \\ 

We then recall two integration by parts formulas in \cite{AL}. Define
\begin{equation}\label{integration definition}
\begin{split}
\int_{H_{u}}\phi:=&2\int_0^{\ub}(\int_{S_{u, \ub'}}\O\phi)d\ub', \quad \int_{\Hb_{\ub}}\phi:=2\int_{-1}^u(\int_{S_{u',\ub}}\O\phi)du'\\
&\int_{D_{u,\ub}}\phi:=2\int_0^{\ub}\int_{-1}^u(\int_{S_{u', \ub'}}\O^2\phi)du'd\ub'.
\end{split}
\end{equation}

\noindent Applying $\phi_1=\O^{-1}\,{^{(1)}\phi}$ and $\phi_2={^{(2)}\phi}$ in Proposition 8.1 of \cite{AL}, we get 

\begin{proposition}\label{intbyparts34}
Suppose ${^{(1)}\phi}$ and ${^{(2)}\phi}$ are $r$ tensorfields, then it holds
{\color{black}
\begin{equation*}
\begin{split}
&\int_{D_{u,\ub}} \O^{-1}\,{^{(1)}\phi} \nabla_4{^{(2)}\phi}+\int_{D_{u,\ub}}\O^{-1}\,{^{(2)}\phi}\nabla_4{^{(1)}\phi}\\
=& \int_{\Hb_{\ub}(-1, u)} \O^{-1}\,{^{(1)}\phi}{^{(2)}\phi}-\int_{\Hb_0(-1, u)} \O^{-1}\,{^{(1)}\phi}{^{(2)}\phi}-\int_{D_{u,\ub}}\O^{-1}\trch{^{(1)}\phi}{^{(2)}\phi}.
\end{split}
\end{equation*}
}
\end{proposition}
\noindent Set ${^{(1)}\phi}={^{(2)}\phi}=\phi$ in the above proposition. Together with (\ref{integration definition}) we then obtain
{\color{black}
\begin{proposition}\label{intbyparts34} 
Suppose $\phi$ is a $r$ tensorfield, then it holds
\begin{equation*}
\begin{split}
\int_{-1}^u\int_{S_{u', \ub}}|\phi|^2=\int_{-1}^u\int_{S_{u', 0}}|\phi|^2+\int_0^{\ub}\int_{-1}^u\int_{S_{u', \ub'}}2\phi(\O\nab_4\phi)+\int_0^{\ub}\int_{-1}^u\int_{S_{u', \ub'}}\O\tr\chi\phi^2.
\end{split}
\end{equation*}

\end{proposition}
}

By Proposition 8.2 in \cite{AL}, we also have

\begin{proposition}\label{intbypartssph}
For an $r$ tensorfield $^{(1)}\phi$ and an $r-1$ tensorfield $^{(2)}\phi$, they satisfy
\begin{equation*}
\begin{split}
&\int_{D_{u,\ub}}{ }^{(1)}\phi^{A_1A_2...A_r}\nabla_{A_r}{ }^{(2)}\phi_{A_1...A_{r-1}}+\int_{D_{u,\ub}}\nabla^{A_r}{ }^{(1)}\phi_{A_1A_2...A_r}{ }^{(2)}\phi^{A_1...A_{r-1}}\\
=& -\int_{D_{u,\ub}}(\eta+\etab){ }^{(1)}\phi{ }^{(2)}\phi.
\end{split}
\end{equation*}
\end{proposition}

{\color{black}
At the same time, it also holds 

\begin{proposition}\label{intbyparts3}
Let $\phi$ be an $r$ tensorfield and $\lambda_1=2(\lambda_0-\f12)$. Then we get
\begin{equation*}
\begin{split}
\int_0^{\ub}\int_{S_{u, \ub'}}|u|^{2\lambda_1}|\phi|^2=&\int_0^{\ub}\int_{S_{-1,\ub'}}|\phi|^2+\int_0^{\ub}\int_{-1}^u\int_{S_{u', \ub'}}2|u'|^{2\lambda_1}<\phi, \O\nab_3\phi+\lambda_0\O\trchb\phi>\\
&+\int_{0}^{\ub}\int_{-1}^u\int_{S_{u', \ub'}}|u'|^{2\lambda_1}(1-2\lambda_0)(\O\tr\chib+\f{2}{|u'|})|\phi|^2.
\end{split}
\end{equation*}
\end{proposition}
\begin{proof}
Using the first variation of area formula, it holds
\begin{equation*}
\begin{split}
\frac{d}{du}(\int_{\S}|u|^{2\lambda_1}|\phi|^2)=&\int_{\S} 2|u|^{2\lambda_1}<\phi, \O\nab_3\phi+\lambda_0\O\trchb\phi>\\
&+\int_{\S}  |u|^{2\lambda_1}\l\f{2\lambda_1 \cdot (-\O e_3u)}{|u|}+(1-2\lambda_0)\O\trchb\r|\phi|^2.
\end{split}
\end{equation*}
We then choose $\lambda_1=2\lambda_0-1$, which satisfies

$$\f{2\lambda_1 (-\O e_3u)}{|u|}+(1-2\lambda_0)\O\trchb=(1-2\lambda_0)(\O\tr\chib+\f{2}{|u|}).$$
By integrating with respect to $du\,d\ub$ and applying the fundamental theorem of calculus in $u$, we then prove this proposition. 
\end{proof}
}

We now proceed to prove
\begin{proposition}\label{Omegae4phi energy} 
Under the assumptions of Theorem \ref{main thm} and the bootstrap assumptions \eqref{BA.1}, \eqref{BA.2}, \eqref{BA.3}, \eqref{BA.4}, it holds
\begin{equation}
\|u^5\nab^5(\O e_4\phi)\|_{L^2_{\ub}L^2(\S)}+\|\O u^5\nab^5(e_A\phi)\|_{L^2_{u}L^2(\S)}\ls  \ub^{\f12}\at.
\end{equation}
\end{proposition}
\begin{proof}
Taking equation \eqref{phi transport 2}
$$\O e_3(\O e_4 \phi)=\O^2\nab^A\nab_A\phi-\f12\O\tr\chib\cdot\O e_4\phi-\f12\O\tr\chi\cdot\O e_3\phi+2\O^2\eta^A\cdot e_A\phi,$$
and commuting it with $i$ angular derivatives, we obtain
\begin{equation}\label{e3e4phi energy}
\begin{split}
&{\color{black}\O}\nab_3\nab^i(\O e_4\phi)+\f{i+1}{2}\O\tr\chib\nab^i(\O e_4\phi)\\
=&\sum_{i_1+i_2+i_3=i}\nab^{i_1}(\eta+\etb)^{i_2}\nab^{i_3}\big(\O^2 \Delta_g\phi-\f12\O\tr\chi\cdot {\color{black}\O}e_3\phi+2\O^2\eta^A e_A\phi\big)\\
&+\sum_{\substack{ i_1+i_2+i_3+i_4=i,\\ i_4<i}}\nab^{i_1}(\eta+\etb)^{i_2}\nab^{i_3}(\O\tr\chib)\nab^{i_4}(\O e_4\phi)\\
&+\sum_{i_1+i_2+i_3+i_4=i}\nab^{i_1}(\eta+\etb)^{i_2}\nab^{i_3}({\color{black}\O}\chibh, {\color{black}\O}\tr\chib+\f{2}{|u|})\nab^{i_4}(\O e_4\phi)\\
&+\sum_{i_1+i_2+i_3+i_4=i-1}\nab^{i_1}(\eta+\etb)^{i_2}\nab^{i_3}({\color{black}\O}\beb)\nab^{i_4}(\O e_4\phi).
\end{split}
\end{equation}
{\color{black}
Note that the term with the top-order derivative on the right is $ \O^2\nab^i\Delta_g\phi$.
}

Similarly, commuting \eqref{phi transport 1} with $i$ angular derivatives, we obtain 
\begin{equation}\label{e4eAphi energy}
\begin{split}
\O\nab_4\nab^i(\O\nab_A)\phi=&\sum_{i_1+i_2+i_3=i}\nab^{i_1}(\eta+\etb)^{i_2}\nab^{i_3}\bigg(\O e_A(\O e_4\phi)-\O^2 {\chi_A}^B e_B\phi+\O e_4(\O)\nab_A\phi \bigg)\\
&+\sum_{i_1+i_2+i_3+i_4=i}\nab^{i_1}(\eta+\etb)^{i_2}\nab^{i_3}(\O\chi)\nab^{i_4}(\O\nab_A)\phi\\
&+\sum_{i_1+i_2+i_3+i_4=i-1}\nab^{i_1}(\eta+\etb)^{i_2}\nab^{i_3}(\O\b)\nab^{i_4}(\O\nab_A)\phi.
\end{split}
\end{equation}
{\color{black}
Notice that the term with the top-order derivative on the right is $\nab^i[\O e_A(\O e_4\phi) ]$. 

We now consider
\begin{equation}\label{e4eAphi IBP}
\begin{split}
&\int_0^{\ub}\int_{u_0}^u\int_{S_{u,\ub}}u^{2i}\nab_{A_1\cdot\cdot\cdot A_i}(\O e_4\phi)[\O \nab_3\nab^{A_1\cdot\cdot\cdot A_i}(\O e_4\phi)+\f{i+1}{2}\O\tr\chib \nab^{A_1\cdot\cdot\cdot A_i}(\O e_4\phi)]\\
&+\int_0^{\ub}\int_{u_0}^u\int_{S_{u,\ub}}u^{2i}\nab_{A_1\cdot\cdot\cdot A_i}(\O\nab^A)\phi[\O\nab_4\nab^{A_1\cdot\cdot\cdot A_i}(\O\nab_A)\phi].
\end{split}
\end{equation}
Employing (\ref{e3e4phi energy}) and (\ref{e4eAphi energy}), on the right we encounter 
$$\int_{\S}\O^2 u^{2i}\nab_{A_1\cdot\cdot\cdot A_i}(\O e_4\phi)\nab^{A_1\cdot\cdot\cdot A_i}\Delta_g\phi+\O^2 u^{2i}\nab_{A_1\cdot\cdot\cdot A_i}\nab^A\phi \nab^{A_1\cdot\cdot\cdot A_i}\nab_A(\O e_4\phi).$$
Employing integration by parts with respect to $\nab_A$, we then cancel the top-order derivatives $\nab^{i+1}(\nab\phi)$, $\nab^{i+1}(\O e_4\phi)$. Applying Proposition \ref{intbyparts34} for \eqref{e4eAphi IBP}, together with \eqref{e3e4phi energy} and \eqref{e4eAphi energy} we arrive at 
}

{\color{black}
\begin{equation}\label{e4 phi ea phi energy}
\begin{split}
&\|u^i\nab^i(\O e_4\phi)\|^2_{L^2_{\ub}L^2(\S)}+\|\O u^i\nab^i(e_A\phi)\|^2_{L^2_{u}L^2(\S)}\\
\leq&\|u^i\nab^i(\O e_4\phi)\|^2_{L^2_{\ub}L^2(S_{-1,\ub})}+\|u^{2i}\O^2\sum_{i_1+i_2+i_3=i-1}\nab^{i_1}(\eta+\etb)^{i_2+1}\nab^{i_3}\Delta_g\phi \nab^i(\O e_4\phi) \|_{L^1_uL^1_{\ub}L^1(\S)}\\
&+\|u^{2i}\sum_{i_1+i_2+i_3=i}\nab^{i_1}(\eta+\etb)^{i_2}\nab^{i_3}\big(-\f12\O\tr\chi\cdot {\color{black}\O}e_3\phi+2\O^2\eta^A e_A\phi\big) \nab^i(\O e_4\phi) \|_{L^1_uL^1_{\ub}L^1(\S)}\\
&+\|u^{2i}\sum_{\substack{i_1+i_2+i_3+i_4=i,\\ i_4<i}}\nab^{i_1}(\eta+\etb)^{i_2}\nab^{i_3}(\O\tr\chib)\nab^{i_4}(\O e_4\phi) \nab^i(\O e_4\phi) \|_{L^1_uL^1_{\ub}L^1(\S)}\\
&+\|u^{2i}\sum_{i_1+i_2+i_3+i_4=i}\nab^{i_1}(\eta+\etb)^{i_2}\nab^{i_3}({\color{black}\O}\chibh, {\color{black}\O}\tr\chib-(\O\tr\chib)_0)\nab^{i_4}(\O e_4\phi) \nab^i(\O e_4\phi) \|_{L^1_uL^1_{\ub}L^1(\S)}\\
&+\|u^{2i}\sum_{i_1+i_2+i_3+i_4=i-1}\nab^{i_1}(\eta+\etb)^{i_2}\nab^{i_3}(\O\beb)\nab^{i_4}(\O e_4\phi) \nab^i(\O e_4\phi) \|_{L^1_uL^1_{\ub}L^1(\S)}\\
&+\|\O^2 u^{2i} \sum_{i_1+i_2+i_3=i-1}\nab^{i_1}(\eta+\etb)^{i_2+1}\nab^{i_3}\big(e_A(\O e_4\phi)\big) \nab^i\nab_A\phi\|_{L^1_uL^1_{\ub}L^1(\S)}\\
&+\|\O^2 u^{2i} \sum_{i_1+i_2+i_3+i_4=i}\nab^{i_1}(\eta+\etb)^{i_2}\nab^{i_3}(\O\chi)\nab^{i_4}\nab_A\phi \nab^i\nab_A\phi\|_{L^1_uL^1_{\ub}L^1(\S)}\\
&+\|\O^2 u^{2i} \sum_{i_1+i_2+i_3+i_4=i-1}\nab^{i_1}(\eta+\etb)^{i_2}\nab^{i_3}(\O\b)\nab^{i_4}\nab_A\phi \nab^i\nab_A\phi\|_{L^1_uL^1_{\ub}L^1(\S)}.
\end{split}
\end{equation}  

We then bound each of the terms on the right. For the initial data term, we have
$$\|u^i\nab^i(\O e_4\phi)\|^2_{L^2_{\ub}L^2(S_{-1,\ub})}\leq (\f{\at}{|u|}\ub^{\f12}|u|)^2=\ub a.$$

The terms with top-order derivatives (when $i=5$) can be controlled as follows: we first bound the terms involving $\O e_3\phi$ and get
\begin{equation*}
\begin{split}
&\|u^{10} (\O\tr\chi)\nab^5(\O e_3\phi)\nab^5(\O e_4\phi) \|_{L^1_uL^1_{\ub}L^1(\S)}\\
\leq&\|\O^{-1}u^{5}\nab^5 (\O e_3\phi)\|_{L^{\infty}_{\ub} L^2_{u}L^2(\S)}\|u^{5}\nab^5 (\O e_4\phi) \|_{L^{\infty}_{u} L^2_{\ub}L^2(\S)}\| \O(\O\tr\chi) \|_{L^{2}_{u} L^2_{\ub}L^{\infty}(\S)}\\
\leq&\ub^{\f12}\af{\color{black}\F}S\cdot\ub^{\f12}\at S\cdot\f{\at}{|u|}|u|^{\f12}\ub^{\f12}O=\f{\ub^{\f12}\af}{|u|^{\f12}}{\color{black}\F}S^2 O\cdot\ub a\leq \ub a B^{-\f14},
\end{split}
\end{equation*}
\begin{equation*} 
\begin{split}
&\|u^{10} (\O e_3\phi)\nab^5(\O\tr\chi)\nab^5(\O e_4\phi) \|_{L^1_uL^1_{\ub}L^1(\S)}\\
\leq&\| u^{6}\nab^5 (\O\tr\chi) \|_{L^{\infty}_{u} L^2_{\ub}L^2(\S)}\|u^{5}\nab^5 (\O e_4\phi) \|_{L^{\infty}_{u} L^2_{\ub}L^2(\S)}\| u^{-1}(\O e_3\phi)\|_{L^{\infty}_{\ub} L^1_{u}L^{\infty}(\S)}\\
\leq&\ub^{\f32}a\t O\cdot\ub^{\f12}\at S\cdot\f{1}{|u|}= \f{\ub\at \t O S}{|u|}\ub a\leq \ub a B^{-\f12}.
\end{split}
\end{equation*}

For the top-order terms with $\eta$, we have
\begin{equation*}
\begin{split}
&\|\O^2 u^{10}(\eta, \etb)\nab^5(e_A\phi)\nab^5(\O e_4\phi) \|_{L^1_uL^1_{\ub}L^1(\S)}\\
\leq&\| u^{5}\nab^5(e_A\phi) \|_{L^{\infty}_{u} L^2_{\ub}L^2(\S)}\|u^{5}\nab^5 (\O e_4\phi) \|_{L^{\infty}_{u} L^2_{\ub}L^2(\S)}\|\O^2(\eta, \etb)\|_{L^{\infty}_{\ub} L^1_{u}L^{\infty}(\S)}\\
\leq&\ub^{\f12}\af{\color{black}\F}S\cdot\ub^{\f12}\at S\cdot\f{\ub\at}{|u|^2}|u| O=\f{\ub a^{\f14}}{|u|}\ub a{\color{black}\F}S^2 O\leq \ub a B^{-1},
\end{split}
\end{equation*}
\begin{equation*}
\begin{split}
&\|\O^2 u^{10}e_A\phi\nab^5\eta \nab^5(\O e_4\phi)\|_{L^1_uL^1_{\ub}L^1(\S)}\\
\leq&\|u^{5}\nab^5\eta  \|_{L^{\infty}_{u} L^2_{\ub}L^2(\S)}\|u^{5}\nab^5 (\O e_4\phi) \|_{L^{\infty}_{u} L^2_{\ub}L^2(\S)}\| \O^2 e_A\phi \|_{L^{\infty}_{\ub} L^1_{u}L^{\infty}(\S)}\\
\leq&\O^{-1}\ub^{\f12}\at\t O\cdot \ub^{\f12}\at S \cdot \f{\ub\at}{|u|}O=\f{\ub\at}{|u|\O}\t O S O\cdot\ub a \leq \ub a B^{-\f12}.
\end{split}
\end{equation*}

For the top-order terms containing $\O\chibh, \O\tr\chib-(\O\tr\chib)_0$, they can be controlled via
\begin{equation*}
\begin{split}
&\|u^{10}(\O\chibh, \O\tr\chib-(\O\tr\chib)_0)\nab^5(\O e_4\phi) \nab^5(\O e_4\phi)\|_{L^1_uL^1_{\ub}L^1(\S)}\\
\leq&\| u^{5}\nab^5 (\O e_4\phi) \|_{L^{\infty}_{u} L^2_{\ub}L^2(\S)}\|u^{5}\nab^5 (\O e_4\phi) \|_{L^{\infty}_{u} L^2_{\ub}L^2(\S)}\| (\O\chibh, \O\tr\chib-(\O\tr\chib)_0) \|_{L^{\infty}_{\ub} L^1_{u}L^{\infty}(\S)}\\
\leq&\ub^{\f12} a^{\f12}S\cdot\ub^{\f12}\at S\cdot\f{\ub\at}{|u|}O\leq \ub a B^{-\f12},
\end{split}
\end{equation*}
\begin{equation*} 
\begin{split}
&\|u^{10} (\O e_4\phi)\nab^5 (\O\chibh, \O\tr\chib-(\O\tr\chib)_0) \nab^5(\O e_4\phi)\|_{L^1_uL^1_{\ub}L^1(\S)}\\
\leq&\|u^{5}\nab^5 (\O\chibh, \O\tr\chib-(\O\tr\chib)_0) \|_{L^{\infty}_{\ub} L^2_{u}L^2(\S)}\|u^{5}\nab^5 (\O e_4\phi) \|_{L^{\infty}_{u} L^2_{\ub}L^2(\S)}\| \O e_4\phi \|_{L^{2}_{u} L^2_{\ub}L^{\infty}(\S)}\\
\leq&\ub^{\f12}\af{\color{black}\F}\t O\cdot\ub^{\f12}\at S\cdot\f{\at}{|u|}|u|^{\f12}\ub^{\f12}O=\f{\ub^{\f12}\af \t O S O}{|u|^{\f12}}\ub a{\color{black}\F}\leq \ub a B^{-\f14}.
\end{split}
\end{equation*}

The top-order terms involving $\O\chi$ obey 
\begin{equation*}
\begin{split}
&\|\O^2 u^{10} (\O\chi)\nab^5 (e_A\phi) \nab^5(e_A\phi)\|_{L^1_uL^1_{\ub}L^1(\S)}\\
\leq&\|\O u^{5}\nab^5 (e_A\phi) \|_{L^{\infty}_{\ub} L^2_{u}L^2(\S)}\|\O u^{5}\nab^5 (e_A\phi) \|_{L^{\infty}_{\ub} L^2_{u}L^2(\S)}\| \O\chi\|_{L^{\infty}_{u} L^1_{\ub}L^{\infty}(\S)}\\
\leq&\ub^{\f12} \at S\cdot \ub^{\f12}\at S\cdot\f{\at}{|u|}\ub O\leq \ub a B^{-\f12},
\end{split}
\end{equation*}
\begin{equation*}
\begin{split}
&\|u^{10} \O^2 e_A\phi\nab^5(\O\chi) \nab^5(e_A\phi)\|_{L^1_uL^1_{\ub}L^1(\S)}\\
\leq&\|u^{10} \O^2 e_A\phi\nab^5(\O\chih) \nab^5(e_A\phi)\|_{L^1_uL^1_{\ub}L^1(\S)}+\|u^{10} \O^2 e_A\phi\nab^5(\O\tr\chi) \nab^5(e_A\phi)\|_{L^1_uL^1_{\ub}L^1(\S)}\\
\leq&\|u^{5}\nab^5 (\O\chih) \|_{L^{\infty}_{u} L^2_{\ub}L^2(\S)}\|\O u^{5}\nab^5 (e_A\phi) \|_{L^{\infty}_{\ub} L^2_{u}L^2(\S)}\| \O e_A\phi \|_{L^{2}_{u} L^2_{\ub}L^{\infty}(\S)}\\
&+\|u^{6}\nab^5 (\O\tr\chi) \|_{L^{\infty}_{u} L^2_{\ub}L^2(\S)}\|\O u^{5}\nab^5 (e_A\phi) \|_{L^{\infty}_{\ub} L^2_{u}L^2(\S)}\| \O u^{-1} e_A\phi \|_{L^{2}_{u} L^2_{\ub}L^{\infty}(\S)}\\
\leq&\ub^{\f12}\at \t O\cdot\ub^{\f12}\at S\cdot\f{\ub\at}{|u|^2}|u|^{\f12}\ub^{\f12}O+\ub^{\f32}a\t O'\cdot\ub^{\f12}\at S\cdot\f{\ub\at}{|u|^3}|u|^{\f12}\ub^{\f12}O\\
=&\f{\ub^{\f32}\at}{|u|^{\f32}}\t O S O\cdot \ub a+\f{\ub^{\f52}a\t O' S O}{|u|^{\f52}}\ub a\leq \ub a B^{-1}.
\end{split}
\end{equation*}

Finally, the top-order terms with $\beta$ and $\beb$ can be bounded via
\begin{equation*}
\begin{split}
&\|u^{10}(\O e_4\phi)\nab^4(\O\beb) \nab^5(\O e_4\phi)\|_{L^1_uL^1_{\ub}L^1(\S)}\\
\leq&\|\O^{-1} u^{6}\nab^4(\O\beb) \|_{L^{\infty}_{\ub} L^2_{u}L^2(\S)}\|u^{5}\nab^5 (\O e_4\phi) \|_{L^{\infty}_{u} L^2_{\ub}L^2(\S)}\| \O u^{-1}(\O e_4\phi) \|_{L^{2}_{u} L^2_{\ub}L^{\infty}(\S)}\\
\leq&\ub^{\f32}a^{\f34}{\color{black}\F}\mathcal{R}\cdot\ub^{\f12}\at S\cdot\f{\at}{|u|^2}|u|^{\f12}\ub^{\f12}O=\f{\ub^{\f32}a^{\f34}\mathcal{R}SO}{|u|^{\f32}}{\color{black}\F}\ub a\leq \ub a B^{-1},
\end{split}
\end{equation*}
\begin{equation*}
\begin{split}
&\|\O^2 u^{10} e_A\phi\nab^4(\O\b) \nab^5(e_A\phi)\|_{L^1_uL^1_{\ub}L^1(\S)}\\
\leq&\|u^{5}\nab^4 (\O\b) \|_{L^{\infty}_{u} L^2_{\ub}L^2(\S)}\|\O u^{5}\nab^5 (e_A\phi) \|_{L^{\infty}_{\ub} L^2_{u}L^{2}(\S)}\| \O e_A\phi \|_{L^{2}_{u} L^2_{\ub}L^{\infty}(\S)}\\
\leq&\ub^{\f12}\at \M R\cdot\ub^{\f12}\at S\cdot\f{\ub\at}{|u|^2}|u|^{\f12}\ub^{\f12}O=\f{\ub^{\f32}\at \M R S O}{|u|^{\f32}}\ub a\leq \ub a B^{-1}.
\end{split}
\end{equation*} 

}

\noindent It is a straight-forward check that the non-top-order-derivative terms in (\ref{e4 phi ea phi energy}) obey the same (and even smaller) upper bounds, compared with the estimates above. There is no borderline term and their upper bounds are much smaller than $\ub a$. Putting all the estimates together, we then finish the proof of this current proposition. 
 \end{proof}

Next, we move to establish the energy estimate for the other pair and we have 

\begin{proposition}\label{prop6.2}
Under the assumptions of Theorem \ref{main thm} and the bootstrap assumptions \eqref{BA.1}, \eqref{BA.2}, \eqref{BA.3}, \eqref{BA.4}, it holds
$$\|u^5\nab^5(e_A\phi)\|_{L^2_{\ub}L^2(\S)}+\|\O^{-1}u^5\nab^5(\O e_3\phi)\|_{L^2_{u}L^2(\S)}\ls \F \ub^{\f12}\af.$$
In addition, we also have the improved estimate
 \begin{equation*}
\begin{split}
\|u^5\nab^5(e_A\phi)\|_{L^2_{\ub}L^2(\S)}+\|\O^{-1}u^5\nab^5(\O e_3\phi)\|_{L^2_{u}L^2(\S)}\ls \f{\ub^{\f12}\af}{|u|^{\f12}\O}\cdot \ub^{\f12}\af+(\f{\ub\at}{|u|\O^{\f32}})^{\f14}\cdot\ub^{\f12}\af.  
\end{split}
\end{equation*}

\end{proposition} 
\begin{proof}
Employing the paired equations \eqref{phi transport 3}, \eqref{phi transport 4}
\begin{equation*}
\begin{split}
\nab_{\O e_3}\nab_A\phi+\f12\O\tr\chib\nab_A\phi=&e_A(\O e_3\phi)-\O {\chibh_A\,}^B \nab_B\phi,\\
\nab_{\O e_4}[(\O e_3\phi)-(\O e_3\phi)_0]=&\O^2\Delta\phi-\f12\O\tr\chib\cdot\O e_4\phi-\f12\O\tr\chi\cdot\O e_3\phi+2\O^2\etb^A\cdot e_A\phi.
\end{split}
\end{equation*}
and commutating them with $i$ angular derivatives, for $i
\geq 1$ we arrive at 
\begin{equation*}
\begin{split}
&(\O\nab_3)\nab^i(e_A\phi)+\f{i+1}{2}\O\tr\chib\nab^i(e_A\phi)\\
=&\sum_{i_1+i_2+i_3=i}\nab^{i_1}(\eta+\etb)^{i_2}\nab^{i_3}\bigg(\nab_A(\O\nab_3\phi)-\O {\chibh_A\,}^B \nab_B\phi\bigg)\\
&+\sum_{\substack{i_1+i_2+i_3+i_4=i, \\ i_4<i}}\nab^{i_1}(\eta+\etb)^{i_2}\nab^{i_3}(\O\tr\chib)\nab^{i_4}(e_A\phi)\\
&+\sum_{i_1+i_2+i_3+i_4=i}\nab^{i_1}(\eta+\etb)^{i_2}\nab^{i_3}(\O\chibh, \O\tr\chib-(\O\tr\chib)_0)\nab^{i_4}(e_A\phi)\\
&+\sum_{i_1+i_2+i_3+i_4=i-1}\nab^{i_1}(\eta+\etb)^{i_2}\nab^{i_3}(\O\beb)\nab^{i_4}(e_A\phi),
\end{split}
\end{equation*}
\begin{equation*}
\begin{split}
&(\O\nab_4)\nab^i(\O e_3\phi)\\
=&(\O\nab_4)\nab^i[(\O e_3\phi)-(\O e_3\phi)_0]\\
=&\sum_{i_1+i_2+i_3=i}\nab^{i_1}(\eta+\etb)^{i_2}\nab^{i_3}\bigg(\O^2\Delta_g\phi-\f12\O\tr\chib\cdot\O e_4\phi-\f12\O\tr\chi\cdot\O e_3\phi+2\O^2\etb^A e_A\phi \bigg)\\
&+\sum_{i_1+i_2+i_3+i_4=i}\nab^{i_1}(\eta+\etb)^{i_2}\nab^{i_3}(\O\chi)\nab^{i_4}[(\O e_3\phi)-(\O e_3\phi)_0]\\
&+\sum_{i_1+i_2+i_3+i_4=i-1}\nab^{i_1}(\eta+\etb)^{i_2}\nab^{i_3}(\O\b)\nab^{i_4}[(\O e_3\phi)-(\O e_3\phi)_0].
\end{split}
\end{equation*}
With $i\geq 1$, we then consider 
\begin{equation*}
\begin{split}
&\int_0^{\ub}\int_{-1}^u\int_{S_{u,\ub}}u^{2i}\nab_{A_1\cdot\cdot\cdot A_i}(e^A\phi)[\O \nab_3\nab^{A_1\cdot\cdot\cdot A_i}(e_A\phi)+\f{i+1}{2}\O\tr\chib \nab^{A_1\cdot\cdot\cdot A_i}(e_A\phi)]\\
&+\int_0^{\ub}\int_{-1}^u\int_{S_{u,\ub}}u^{2i}\nab_{A_1\cdot\cdot\cdot A_i}[(\O e_3\phi)-(\O e_3\phi)_0]\{\O^{-2}\cdot \O\nab_4\nab^{A_1\cdot\cdot\cdot A_i}[(\O e_3\phi)-(\O e_3\phi)_0] \}.
\end{split}
\end{equation*}
Using (\ref{phi transport 3}) and (\ref{phi transport 4}) into the above expression, we see that the terms involving the top-order derivatives are of the form
$$\int_{\S} u^{2i}\nab_{A_1\cdot\cdot\cdot A_i}(e^A\phi)\nab^{A_1\cdot\cdot\cdot A_i}[e_A(\O e_3\phi)]+ u^{2i}\O^{-2}\O^2\nab_{A_1\cdot\cdot\cdot A_i}(\O e_3\phi) \nab^{A_1\cdot\cdot\cdot A_i}(\Delta_g\phi).$$
Employing integration by parts with respect to $\nab_A$, we notice that the top-order-derivative terms $\nab^{i+1}(\O e_3\phi)$ and $\nab^{i+1}(e_A\phi)$ are \underline{cancelled}. Carrying out the details of this approach, we get 

\begin{equation*}
\begin{split}
&\|u^i\nab^i(e_A\phi)\|^2_{L^2_{\ub}L^2(\S)}+\|\O^{-1} u^i\nab^i(\O e_3\phi)\|^2_{L^2_{u}L^2(\S)}\\
\leq&\|u^i\nab^i(e_A\phi)\|^2_{L^2_{\ub}L^2(S_{-1,\ub})}+\|u^{2i}\sum_{i_1+i_2+i_3=i-1}\nab^{i_1}(\eta+\etb)^{i_2+1}\nab^{i_3}\nab_A(\O e_3\phi)\nab^i(e_A\phi) \|_{L^1_uL^1_{\ub}L^1(\S)}\\
&+\|u^{2i}\sum_{i_1+i_2+i_3=i}\nab^{i_1}(\eta+\etb)^{i_2}\nab^{i_3}(\O {\chibh_A\,}^B \nab_B\phi) \nab^i(e_A\phi) \|_{L^1_uL^1_{\ub}L^1(\S)}\\
&+\|u^{2i}\sum_{i_1+i_2+i_3+i_4=i, \, i_4<i}\nab^{i_1}(\eta+\etb)^{i_2}\nab^{i_3}(\O\tr\chib)\nab^{i_4}(e_A\phi) \nab^i(e_A\phi) \|_{L^1_uL^1_{\ub}L^1(\S)}\\
&+\|u^{2i}\sum_{i_1+i_2+i_3+i_4=i}\nab^{i_1}(\eta+\etb)^{i_2}\nab^{i_3}(\O\chibh, \O\tr\chib-(\O\tr\chib)_0)\nab^{i_4}(e_A\phi) \nab^i(e_A\phi) \|_{L^1_uL^1_{\ub}L^1(\S)}\\
&+\|u^{2i}\sum_{i_1+i_2+i_3+i_4=i-1}\nab^{i_1}(\eta+\etb)^{i_2}\nab^{i_3}(\O\beb)\nab^{i_4}(e_A\phi) \nab^i(e_A\phi) \|_{L^1_uL^1_{\ub}L^1(\S)}\\
&+\|\O^{-2}\O^2 u^{2i}\sum_{i_1+i_2+i_3=i-1}\nab^{i_1}(\eta+\etb)^{i_2+1}\nab^{i_3}\Delta_g\phi\nab^i(\O e_3\phi) \|_{L^1_uL^1_{\ub}L^1(\S)}\\
&+\|\O^{-2} u^{2i}\sum_{i_1+i_2+i_3=i}\nab^{i_1}(\eta+\etb)^{i_2}\\
&\quad\quad\quad\quad\quad\quad\times\nab^{i_3}\bigg(-\f12\O\tr\chib\cdot\O e_4\phi-\f12\O\tr\chi \cdot\O e_3\phi+2\O^2\etb^A e_A\phi \bigg) \nab^i(\O e_3\phi) \|_{L^1_uL^1_{\ub}L^1(\S)}\\
\end{split}
\end{equation*}
\begin{equation}\label{energy estimate details e3phi}
\begin{split}
&+\|\O^{-2} u^{2i}\sum_{i_1+i_2+i_3+i_4=i}\nab^{i_1}(\eta+\etb)^{i_2}\nab^{i_3}(\O\chi)\nab^{i_4}[(\O e_3\phi)-(\O e_3\phi)_0] \nab^i(\O e_3\phi) \|_{L^1_uL^1_{\ub}L^1(\S)}\\
&+\|\O^{-2} u^{2i}\sum_{i_1+i_2+i_3+i_4=i-1}\nab^{i_1}(\eta+\etb)^{i_2}\nab^{i_3}(\O\b)\nab^{i_4}[(\O e_3\phi)-(\O e_3\phi)_0] \nab^i(\O e_3\phi) \|_{L^1_uL^1_{\ub}L^1(\S)}.
\end{split}
\end{equation}

We now derive bound for terms on the right. The initial-data term obeys
$$\|u^i\nab^i(e_A\phi)\|^2_{L^2_{\ub}L^2(S_{-1,\ub})}\leq (\f{\ub\at}{|u|^2}\ub^{\f12}|u|)^2=(\f{\ub^{\f32}\at}{|u|})^2= (\f{\ub\af}{|u|}\ub^{\f12}\af)^2\leq \ub\at\cdot\f{\ub^2\at}{|u|^2}.$$

{\color{black}
The terms with top-order derivatives (when $i=5$) are controlled as below: we first bound the 4 terms involving $\nab^5(e_A\phi)$ and they satisfy
\begin{equation*}
\begin{split}
&\|u^{10} (\eta,\etb)\nab^5(\O e_3\phi)\nab^5(e_A\phi)) \|_{L^1_uL^1_{\ub}L^1(\S)}\\
\leq&\|\O^{-1} u^{5}\nab^5 (\O e_3\phi) \|_{L^{\infty}_{\ub} L^2_{u}L^2(\S)}\|u^{5}\nab^5(e_A\phi)\|_{L^{\infty}_{u} L^2_{\ub}L^2(\S)}\| \O(\eta, \etb) \|_{L^{2}_{\ub} L^{2}_{u}L^{\infty}(\S)}\\
\leq&\ub^{\f12}\af\cdot\F S\cdot\ub^{\f12}\af\cdot\F S\cdot\f{\ub\at}{|u|^2}\ub^{\f12}|u|^{\f12}O\leq \ub\at\cdot\f{\ub\at}{|u|\O},
\end{split}
\end{equation*} 

\begin{equation*} 
\begin{split}
&\|u^{10} (\O\chibh, \O\tr\chib-(\O\tr\chib)_0)\nab^5(e_A\phi)\nab^5(e_A\phi) \|_{L^1_uL^1_{\ub}L^1(\S)}\\
\leq&\|u^{5}\nab^5(e_A\phi)\|_{L^{\infty}_{u} L^2_{\ub}L^2(\S)}\|u^{5}\nab^5 (e_A\phi) \|_{L^{\infty}_{u} L^2_{\ub}L^2(\S)}\| \O\chibh, \O\tr\chib-(\O\tr\chib)_0 \|_{L^{1}_{u} L^{\infty}_{\ub}L^{\infty}(\S)}\\
\leq&\ub^{\f12}\af\cdot\F S\cdot\ub^{\f12}\af\cdot\F S\cdot\f{\ub\at}{|u|}O\leq\ub \at\cdot(\f{\ub\at}{|u|\O^{\f32}})^{\f12},
\end{split}
\end{equation*}

\begin{equation*}
\begin{split}
&\|u^{10} (e_A\phi)\nab^5(\O \tr\chib, \O\chibh)\nab^5(e_A\phi) \|_{L^1_uL^1_{\ub}L^1(\S)}\\
\leq&\|u^{5}\nab^5 (\O\tr\chib, \O\chibh) \|_{L^{\infty}_{\ub} L^2_{u}L^2(\S)}\|u^{5}\nab^5 (e_A\phi) \|_{L^{\infty}_{u} L^2_{\ub}L^2(\S)}\| e_A\phi \|_{L^{2}_{u} L^2_{\ub}L^{\infty}(\S)}\\
\leq&\ub^{\f12}\af\cdot\F \t O\cdot \ub^{\f12}\af\cdot\F S\cdot\f{\ub\at}{|u|^2}|u|^{\f12}\ub^{\f12}O\leq \ub\at\cdot\f{\ub\at}{|u|\O}, 
\end{split}
\end{equation*}

\begin{equation*}
\begin{split}
&\|u^{10} (e_A\phi)\nab^4(\O\beb)\nab^5(e_A\phi) \|_{L^1_uL^1_{\ub}L^1(\S)}\\
\leq&\|u^{6}\O^{-1}\nab^4(\O\beb)  \|_{L^{\infty}_{\ub} L^2_{u}L^2(\S)}\|u^{5}\nab^5 (e_A\phi) \|_{L^{\infty}_{u} L^2_{\ub}L^2(\S)}\| u^{-1}\O e_A\phi \|_{L^{2}_{u} L^2_{\ub}L^{\infty}(\S)}\\
\leq&\ub^{\f32}a^{\f34}\F\mathcal{R}\cdot\ub^{\f12}\af\cdot\F S\cdot \f{\ub\at}{|u|^3}|u|^{\f12}\ub^{\f12}O\leq \ub\at\cdot\f{\ub\at}{|u|\O}.
\end{split}
\end{equation*}

For top-order-derivative terms containing $\nab^5(\O e_3\phi)$ and $\O\chi$, we have 
\begin{equation*}
\begin{split}
&\|\O^{-2} u^{10} (\O\chi)\nab^5(\O e_3\phi)\nab^5(\O e_3\phi) \|_{L^1_uL^1_{\ub}L^1(\S)}\\
\leq&\|\O^{-1} u^{5}\nab^5 (\O e_3\phi) \|_{L^{\infty}_{\ub} L^2_{u}L^2(\S)}\|\O^{-1} u^{5}\nab^5 (\O e_3\phi) \|_{L^{\infty}_{\ub} L^2_{u}L^2(\S)}\| \O\chi \|_{L^{1}_{\ub} L^{\infty}_{u}L^{\infty}(\S)}\\
\leq&\ub^{\f12}\af\F S\cdot\ub^{\f12}\af\F S\cdot\f{\at}{|u|}\ub O\leq \ub\at\cdot(\f{\ub\at}{|u|\O})^{\f12},
\end{split}
\end{equation*}

\begin{equation*}
\begin{split}
&\|\O^{-2} u^{10} (\O e_3\phi)\nab^5(\O \tr\chi)\nab^5(\O e_3\phi) \|_{L^1_uL^1_{\ub}L^1(\S)}\\
\leq&\|u^{6}\nab^5 (\O\tr\chi) \|_{L^{\infty}_{u} L^2_{\ub}L^2(\S)}\|\O^{-1} u^{5}\nab^5 (\O e_3\phi) \|_{L^{\infty}_{\ub} L^2_{u}L^2(\S)}\| \O^{-1} u^{-1}(\O e_3\phi) \|_{L^{2}_{u} L^2_{\ub}L^{\infty}(\S)}\\
\leq&\ub^{\f32}a\t O'\cdot\ub^{\f12}\af S \cdot\F\cdot\f{1}{|u|^2\O}|u|^{\f12}\ub^{\f12}\leq \ub\at\cdot\f{\ub\at}{|u|\O}.\end{split}
\end{equation*}
Noting that in above inequality, we employ the crucial estimates via $\t O'$. 

To deal with the top-order-derivative terms involving both $\nab^5(\O e_3\phi)$ and $\O\tr\chib$, we encounter the borderline estimate. Using the derived Proposition \ref{Omegae4phi energy}, we have 

\begin{equation*}
\begin{split}
&\|\O^{-2} u^{10} (\O\tr\chib)\nab^5(\O e_4\phi)\nab^5(\O e_3\phi) \|_{L^1_uL^1_{\ub}L^1(\S)}\\
\leq&\|u^{5}\nab^5 (\O e_4\phi) \|_{L^{\infty}_{u} L^2_{\ub}L^2(\S)}\|\O^{-1} u^{5}\nab^5 (\O e_3\phi) \|_{L^{\infty}_{\ub} L^2_{u}L^2(\S)}\| \O^{-1}(\O\tr\chib) \|_{L^{2}_{u} L^2_{\ub}L^{\infty}(\S)}\\
\leq&\ub^{\f12}\at\cdot\f{1}{|u|\O}|u|^{\f12}\ub^{\f12}\cdot \|\O^{-1}u^{5}\nab^5 (\O e_3\phi) \|_{L^{\infty}_{\ub} L^2_{u}L^2(\S)}\\
\leq& c_0 \|\O^{-1}u^{5}\nab^5 (\O e_3\phi) \|^2_{L^{\infty}_{\ub} L^2_{u}L^2(\S)}+\f{1}{c_0}(\f{\ub^{\f12}\af}{|u|^{\f12}\O})^2\cdot\ub\at. 
\end{split}
\end{equation*}
By choosing $c_0$ to be suitably small, the term $c_0 \|\O^{-1}u^{5}\nab^5 (\O e_3\phi) \|^2_{L^{\infty}_{\ub} L^2_{u}L^2(\S)}$ can be absorbed to the left of \eqref{energy estimate details e3phi}. 

The other top-order-derivative terms with $\O\tr\chib$ satisfies a better upper bound
\begin{equation*}
\begin{split}
&\|\O^{-2} u^{10} (\O e_4\phi)\nab^5(\O\tr\chib)\nab^5(\O e_3\phi) \|_{L^1_uL^1_{\ub}L^1(\S)}\\
\leq&\|u^{5}\nab^5 (\O\tr\chib) \|_{L^{\infty}_{\ub} L^2_{u}L^2(\S)}\|\O^{-1} u^{5}\nab^5 (\O e_3\phi) \|_{L^{\infty}_{\ub} L^2_{u}L^2(\S)}\| \O^{-1}(\O e_4\phi) \|_{L^{1}_{\ub} L^{\infty}_{u}L^{\infty}(\S)}\\
\leq&\ub^{\f12}\af\F\t O\cdot\ub^{\f12}\af\F S\cdot\f{\at}{|u|\O}\ub O\leq \ub\at\cdot(\f{\ub\at}{|u|\O^{\f32}})^{{\color{black}\f12}}.
\end{split}
\end{equation*} 

For the top-order-derivative terms containing $\nab^5(\O e_3\phi)$ and $\etb$, we have
\begin{equation*}
\begin{split}
&\|u^{10} \etb^A\nab^5(e_A\phi)\nab^5(\O e_3\phi) \|_{L^1_uL^1_{\ub}L^1(\S)}\\
\leq&\|u^{5}\nab^5 (e_A\phi) \|_{L^{\infty}_{u} L^2_{\ub}L^2(\S)}\|\O^{-1}u^{5}\nab^5 (\O e_3\phi) \|_{L^{\infty}_{\ub} L^2_{u}L^2(\S)}\| \O \etb^A \|_{L^{2}_{u} L^2_{\ub}L^{\infty}(\S)}\\
\leq&\ub^{\f12}\af\F S\cdot\ub^{\f12}\af\F S\cdot\f{\ub\at}{|u|^2}|u|^{\f12}\ub^{\f12}O\leq \ub\at\cdot\f{\ub\at}{|u|\O},
\end{split}
\end{equation*}
\begin{equation*}
\begin{split}
&\|u^{10} (e_A\phi)\nab^5\etb\nab^5(\O e_3\phi) \|_{L^1_uL^1_{\ub}L^1(\S)}\\
\leq&\|u^{5}\nab^5 \etb \|_{L^{\infty}_{u} L^2_{\ub}L^2(\S)}\|\O^{-1}u^{5}\nab^5 (\O e_3\phi) \|_{L^{\infty}_{\ub} L^2_{u}L^2(\S)}\| \O e_A\phi \|_{L^{2}_{u} L^2_{\ub}L^{\infty}(\S)}\\
\leq&\ub^{\f12}\af\F\t O\cdot \ub^{\f12}\af\F S\cdot\f{\ub\at}{|u|^2}|u|^{\f12}\ub^{\f12}O\leq \ub\at\cdot\f{\ub\at}{|u|\O}.
\end{split}
\end{equation*}

The top-order-derivative terms involving $\nab^5(\O e_3\phi)$ and $\O e_3\phi-(\O e_3\phi)_0$ obey 

\begin{equation*}
\begin{split}
&\|\O^{-2} u^{10}[(\O e_3\phi)-(\O e_3\phi)_0]\nab^5(\O \tr\chi)\nab^5(\O e_3\phi) \|_{L^1_uL^1_{\ub}L^1(\S)}\\
\leq&\|u^{5}\nab^5 (\O\tr\chi) \|_{L^{\infty}_{u} L^2_{\ub}L^2(\S)}\|\O^{-1} u^{5}\nab^5 (\O e_3\phi) \|_{L^{\infty}_{\ub} L^2_{u}L^2(\S)}\\
&\quad\quad\quad\quad\quad\quad\quad\quad\quad\quad\quad\quad \times\|\O^{-1}[(\O e_3\phi)-(\O e_3\phi)_0] \|_{L^{2}_{u} L^2_{\ub}L^{\infty}(\S)}\\
\leq&\ub^{\f12}\at\t O\cdot \ub^{\f12}\af \F S\cdot\f{\ub\at}{|u|^2\O}|u|^{\f12}\ub^{\f12}O\leq \ub\at\cdot\f{\ub\at}{|u|\O},
\end{split}
\end{equation*}

\begin{equation*}
\begin{split}
&\|\O^{-2} u^{10}[(\O e_3\phi)-(\O e_3\phi)_0]\nab^5(\O \chih)\nab^5(\O e_3\phi) \|_{L^1_uL^1_{\ub}L^1(\S)}\\
\leq&\|u^{5}\nab^5 (\O\chih) \|_{L^{\infty}_{u} L^2_{\ub}L^2(\S)}\|\O^{-1} u^{5}\nab^5 (\O e_3\phi) \|_{L^{\infty}_{\ub} L^2_{u}L^2(\S)}\|\O^{-1}[(\O e_3\phi)-(\O e_3\phi)_0] \|_{L^{2}_{u} L^2_{\ub}L^{\infty}(\S)}\\
\leq&\ub^{\f12}\at \t O\cdot \ub^{\f12}\af\F S\cdot\f{\ub\at}{|u|^2\O}|u|^{\f12}\ub^{\f12} O\leq\ub\at\cdot\f{\ub\at}{|u|\O}.
\end{split}
\end{equation*}

And the top-order-derivative term with $\nab^5(\O e_3\phi)$ and $\nab^4(\O\b)$ obey
\begin{equation*}
\begin{split}
&\|\O^{-2} u^{10}[(\O e_3\phi)-(\O e_3\phi)_0]\nab^4(\O \beta)\nab^5(\O e_3\phi) \|_{L^1_uL^1_{\ub}L^1(\S)}\\
\leq&\|u^{5}\nab^4 (\O\beta) \|_{L^{\infty}_{u} L^2_{\ub}L^2(\S)}\|\O^{-1} u^{5}\nab^5 (\O e_3\phi) \|_{L^{\infty}_{\ub} L^2_{u}L^2(\S)}\\
&\quad\quad\quad\quad\quad\quad\quad\quad\quad\quad\quad\quad  \times\|\O^{-1}[(\O e_3\phi)-(\O e_3\phi)_0] \|_{L^{2}_{u} L^2_{\ub}L^{\infty}(\S)}\\
\leq&\ub^{\f12}\at \M R \cdot \ub^{\f12}\af\F S\cdot\f{\ub\at}{|u|^2\O}|u|^{\f12}\ub^{\f12}O\leq\ub\at\cdot\f{\ub\at}{|u|\O}.
\end{split}
\end{equation*}

\noindent We hence bound all the top-order-derivative terms. It is a straight-forward check that compared with the upper bounds of the obtained estimates above, the rest terms obey the same (and even smaller) upper bounds. There is no other borderline term and their upper bounds are much smaller than $\ub a^{\f12}$. Gathering all these estimates, we finish the proof of this current proposition. 

}
\end{proof}

\section{Elliptic estimates for the fifth derivatives of the Ricci coefficients}\label{elliptic estimates} 

{\color{black}
We now move to bound the fifth angular derivatives for the Ricci coefficients. We start from the estimates for $\h\om$. 
\begin{proposition} \label{om.5.bd}
Under the assumptions of Theorem \ref{main thm} and the bootstrap assumptions \eqref{BA.1}, \eqref{BA.2}, \eqref{BA.3}, \eqref{BA.4}, it holds
$$\|u^5\nab^5({\color{black}\O}\h\o)\|_{L^{\i}_uL^2_{\ub}L^2(\S)}\ls \ub^{\f12}\at[1+\M R(\b)].$$

\end{proposition}

\begin{proof}
We first construct an auxiliary function $\od$, which satisfies 
$$\nab_3 \od=\f 12\sigmac$$
and has zero initial data on $H_{-1}$. We further let
$$\kappa:=\nab\h\o+^*\nab\od-\f12(\b-\f12\nab_4\phi\nab\phi).$$

Using Proposition \ref{evolution lemma} and bootstrap assumption \eqref{BA.2}, we have
\begin{align*}
\sum_{i\leq 4}\|u^{i}\nab^i\od\|_{L^\i_uL^\i_{\ub}L^2(\S)}\leq& \|\O u^{i+1}\nab^i\sigmac\|_{L^\i_{\ub}L^2_{u}L^2(\S)}\|u^{-1}\|_{L^{\i}_{\ub}L^2_u L^{\i}(\S)}+\f{\ub^{\f12} a^{\f12}}{|u|^{\f12}}\\
\leq&\f{\ub^{\f12}\at \M R}{|u|^{\f12}}+\f{\ub^{\f12}\at}{|u|^{\f12}}\leq \at B^{-\f14}.
\end{align*}
This also implies
$$\sum_{i\leq 4}\|u^{i}\nab^i(\O \od)\|_{L^\i_uL^\i_{\ub}L^2(\S)}\leq\at B^{-\f14}\leq \at.$$

\noindent Note that in above inequality $\O\od$ obeys the same bounds as for $\p$. Hence we regard $\O\od$ as one of the {\color{black}$\p$ terms and we employ the notation $\p\in\{\O\trch,\O\chih,\O\h\om,\O\od\}$.} Using this notation, we write the transport equation for $\O\kappa$ as
\begin{equation*}
\begin{split}
&{\color{black}\O}\nab_3({\color{black}\O}\kappa)+\f12 {\color{black}\O}\tr\chib({\color{black}\O}\kappa)\\
=&\sum_{i_1+i_2+i_3=1}\q^{i_1}\nab^{i_2}(\O\tr\chib+\f{2}{|u|}, \O\chibh)\cdot\nab^{i_3}\p+\sum_{i_1+i_2=1}\O^2\q^{i_1+1}\nab^{i_2}(\eta,\etab, e_A\phi)\\
&+\frac{1}{|u|}\q\p+\frac{{\color{black}\O}}{|u|}\beta+\q K+\nab(\O e_3\phi)\p+\f{1}{|u|}\nab(\O e_4\phi).
\end{split}
\end{equation*}
Commuting with angular derivatives for $4$ times and employing the schematic Codazzi equation ${\color{black}\O}\beta=\sum_{i_1+i_2=1}\q^{i_1}\nab^{i_2}\p$, we obtain 
\begin{equation*}
\begin{split}
&{\color{black}\O}\nab_3 \nab^4 ({\color{black}\O}\kappa)+\f{5}{2}{\color{black}\O}\tr\chib \nab^4 ({\color{black}\O}\kappa)\\
=&\sum_{i_1+i_2+i_3+i_4=5} \nab^{i_1}\q^{i_2}\nab^{i_3}({\color{black}\O}\trchb+\f2{|u|},{\color{black}\O}\chibh)\cdot\nab^{i_4}\p+\sum_{\substack{i_1+i_2+i_3=5\\ i_1\leq 4}}\O^2\nab^{i_1}\q^{i_2+1}\nab^{i_3}(\eta,\etab, e_A\phi)\\
&+\frac{1}{|u|}\nab^4({\color{black}\O}\beta)+\frac{1}{|u|}\sum_{i_1+i_2+i_3=4} \nab^{i_1}\q^{i_2+1}\nab^{i_3}\p {\color{black}+\sum_{\substack{i_1+i_2+i_3=4\\ i_1\leq 3}}\nab^{i_1+1}\q^{i_2+1}\nab^{i_3}\p}\\
&+\sum_{i_1+i_2+i_3+i_4=4} \nab^{i_1}\q^{i_2}\nab^{i_3}\q\nab^{i_4}K{\color{black}+\p\nab^5(\O e_3\phi)+\f{1}{|u|}\nab^5(\O e_4\phi)}.
\end{split}
\end{equation*}
Employing Proposition \ref{evolution lemma} with $\lambda_0=\f52$, we then bound the $\|u^4\nab^4({\color{black}\O}\kappa)\|_{L^2_{\ub}L_u^{\i}L^{2}(\S)}$ norm by the $\|u^4\cdot\|_{L_{\ub}^{2}L_{u}^{1}L^{2}(\S)}$ norm of the right. 

In the first term, we encounter $u^4\p\nab^5(\O\tr\chib, \O\chibh)$ and it satisfies
\begin{equation*}
\begin{split}
\|u^4\p \nab^5(\O\tr\chib, \O\chibh)\|_{L^2_{\ub}L^1_u L^2(\S)}\leq&\|u^5\nab^5(\O\tr\chib, \O\chibh)\|_{L^{\i}_{\ub}L^2_u L^2(\S)}\cdot\|u^{-1}\p\|_{L^2_{\ub}L^2_u L^{\infty}(\S)}\\
\leq& \ub^{\f12}\af\F \M R\cdot \f{\ub^{\f12}\at}{|u|^{\f32}}O\leq\f{\ub^{\f12}\at B^{-\f14}}{|u|}.
\end{split}
\end{equation*}
The rest of the first term obey
\begin{equation*}
\begin{split}
\ls&\|u^4({\color{black}\O}\trchb+\f2{|u|},{\color{black}\O}\chibh)\nab^5\p\|_{L^2_{\ub}L_{u}^{1}L^{2}(\S)}+\sum_{\substack{i_1+i_2+i_3=5\\i_1,i_3\leq 4}}\|u^4\nab^{i_1}\q^{i_2+1}\nab^{i_3}\p\|_{L^2_{\ub}L_{u}^{1}L^{2}(\S)}\\
\ls& \|u^2\q\|_{L^\infty_{\ub}L^\infty_u L^\infty(\S)}\|u^{-3}\|_{L^1_u}\\
&\quad\times\bigg(\|u^5\nab^5({\color{black}\O}\trch, {\color{black}\O}\chih, {\color{black}\O}\h\om)\|_{L^\infty_uL_{\ub}^{2}L^{2}(\S)}+\|u^5\nab^5({\color{black}\O}\od)\|_{L^\infty_uL_{\ub}^{2}L^{2}(\S)}\bigg)+\f{\ub^{\f32} a O^2}{|u|^2}\\
\ls & \frac{\ub a^{\f12}O}{|u|^2}\cdot\bigg(\ub^{\f12} a^{\f12} \t O+\|u^5\nab^5({\color{black}\O}\od)\|_{L^\infty_uL_{\ub}^{2}L^{2}(\S)}\bigg)+\f{\ub^{\f32} a O^2}{|u|^2}\\
\leq&\f{\ub^{\f12}\at B^{-\f12}}{|u|}+\f{B^{-\f12}}{|u|}\|u^5\nab^5({\color{black}\O}\od)\|_{L^\infty_uL_{\ub}^{2}L^{2}(\S)}.
\end{split}
\end{equation*}

In the second term, there is $\O^2 u^4\q \nab^5(\eta, \etb, e_A\phi)$ and it obeys
\begin{equation*}
\begin{split}
&\|\O^2 u^4 \q \nab^5(\eta,\etb, e_A\phi)\|_{L^2_{\ub}L^1_u L^2(\S)}\\
\leq&\|u^5\nab^5\etb\|_{L^{\i}_u L^2_{\ub} L^2(\S)}\|\O^2 u^{-1}\q\|_{L^1_u L^{\i}_{\ub} L^{\i}(\S)}+\|\O u^5\nab^5\eta\|_{L^{\i}_{\ub} L^2_{u} L^2(\S)}\|\O u^{-1}\q\|_{L^2_{\ub} L^{2}_{u} L^{\i}(\S)}\\
&+\|\O u^5\nab^5(e_A\phi)\|_{L^{\i}_{\ub} L^2_{u} L^2(\S)}\|\O u^{-1}\q\|_{L^2_{\ub} L^{2}_{u} L^{\i}(\S)}\\
\leq&\ub^{\f12}\af\F \cdot\t O\cdot\f{\ub\at}{|u|^2}O+\ub^{\f12}\at(\t O+S)\cdot\f{\ub\at O}{|u|^3}\ub^{\f12}|u|^{\f12}\leq \f{\ub^{\f12}\at B^{-1}}{|u|}.
\end{split}
\end{equation*}

The rest of the second term can be bounded by
\begin{equation*}
\begin{split}
\sum_{\substack{i_1+i_2+i_3=5\\i_1,i_3\leq 4}}\|\O^2 u^4\nab^{i_1}\q^{i_2+1}\nab^{i_3}\q\|_{L^2_{\ub}L_{u}^{1}L^{2}(\S)}
\ls &\f{\ub^{\f32} a O^2}{|u|^2}\leq \f{\ub^{\f12}\at B^{-\f12}}{|u|}.
\end{split}
\end{equation*}

The third term containing $\beta$ satisfies
\begin{equation*}
\begin{split}
\|u^3\nab^4({\color{black}\O}\beta)\|_{L^2_{\ub}L_{u}^{1}L^{2}(\S)}
\ls& \|u^{-2}\|_{L^1_u}\|u^5\nab^4({\color{black}\O}\beta)\|_{L^\infty_uL^2_{\ub}L^2(\S)}
\ls \f{\ub^{\f12} \at}{|u|}\M R(\b).
\end{split}
\end{equation*}

For the next two terms, we have 
\begin{equation*}
\begin{split}
&\sum_{i_1+i_2+i_3=4}\|u^3\nab^{i_1}\q^{i_2+1}\nab^{i_3}\p\|_{L^2_{\ub}L_{u}^{1}L^{2}(\S)}\ls \f{\ub^{\f32} a O^2}{|u|^2}\leq \f{\ub^{\f12}\at B^{-\f12}}{|u|},
\end{split}
\end{equation*}
{\color{black}
$$\sum_{\substack{i_1+i_2+i_3=4\\ i_1\leq 3}}\|u^4 \nab^{i_1+1}\q^{i_2+1}\nab^{i_3}\p\|_{L^2_{\ub}L_{u}^{1}L^{2}(\S)}\ls \f{\ub^{\f32} a O^2}{|u|^2}\leq \f{\ub^{\f12}\at B^{-\f12}}{|u|}.$$
}

The sixth term containing $K$ obeys
\begin{equation*}
\begin{split}
&\sum_{i_1+i_2+i_3=4}\|u^4\nab^{i_1}\q^{i_2+1}\nab^{i_3}K\|_{L^2_{\ub}L_{u}^{1}L^{2}(\S)}\\
\ls &\ub a^{\f12}O\bigg(\frac{\ub^{\f32}a^{\f34}}{|u|^3}\F \M R+\frac{\ub^{\f12}}{|u|^2}\bigg)\ls \f{\ub^{\f12}\at B^{-\f12}}{|u|},
\end{split}
\end{equation*}
where we use Proposition \ref{product} and the bootstrap assumption \eqref{BA.3}.

{\color{black}
For the last two terms, we have 
\begin{equation*}
\begin{split}
\|u^4\p\nab^5(\O e_3\phi)\|_{L^2_{\ub}L_{u}^{1}L^{2}(\S)}
\leq& \|\O^{-1} u^5\nab^5(\O e_3\phi)\|_{L^{\infty}_{\ub} L^2_u L^2(\S)}\|\O u^{-1}\p\|_{L^2_{\ub}L^2_{u}L^{\i}(\S)}\\
 \leq& \ub^{\f12}\af\F S\cdot\f{\at}{|u|^2}|u|^{\f12}\ub^{\f12}O\leq \f{\ub^{\f12}\at B^{-\f12}}{|u|}.
\end{split}
\end{equation*}
And by Proposition \eqref{Omegae4phi energy}, it holds
\begin{equation*}
\begin{split}
\|u^3\nab^5(\O e_4\phi)\|_{L^2_{\ub}L_{u}^{1}L^{2}(\S)}
\ls \|u^5\nab^5(\O e_4\phi)\|_{L^{\infty}_{u} L^2_{\ub} L^2(\S)}\|u^{-2}\|_{L^1_u}
 \ls \f{\ub^{\f12}\at}{|u|}.
 \end{split}
\end{equation*} 
}

\noindent Collecting all the above estimates, we arrive at
$$\|u^4\nab^4({\color{black}\O}\kappa)\|_{L^2_{\ub}L^{\i}_uL^2(\S)}\ls \frac{\ub^{\f12}\at}{|u|}[1+\M R(\b)]+\f{B^{-\f12}}{|u|}\|u^5\nab^5({\color{black}\O}\od)\|_{L^\infty_uL_{\ub}^{2}L^{2}(\S)}.$$
This also implies
\begin{equation}\label{nab4kappa}
\|u^5\nab^4({\color{black}\O}\kappa)\|_{L^2_{\ub}L^{\i}_uL^2(\S)}\ls \ub^{\f12}\at[1+\M R(\b)+\f{B^{-\f12}}{\ub^{\f12}a^{\f12}}\|u^5\nab^5({\color{black}\O}\od)\|_{L^\infty_uL_{\ub}^{2}L^{2}(\S)}].
\end{equation}

We then appeal to the following $\div$-$\curl$ system: 
$$\div\nab\h\o=\div\kappa+\f12\div(\b-\f12\nab_4\phi\nab\phi),\quad \curl\nab\h\o=0,$$
$$\curl\nab\od=\curl\kappa+\f12\curl(\b-\f12\nab_4\phi\nab\phi), \quad \div\nab\od=0.$$
Applying the elliptic estimates from Proposition \ref{ellipticthm}, we get
\begin{equation*}
\begin{split}
&\|u^5\nab^5({\color{black}\O}\h\o, {\color{black}\O}\od)\|_{L^2(\S)}\\
\ls &\sum_{j\leq 4}\bigg(\|u^{j+1}\nab^j({\color{black}\O}\kappa)\|_{L^2(\S)}+\|u^{j+1}\nab^{j}({\color{black}\O}\b-\f12\O\nab_4\phi\nab\phi)\|_{L^2(\S)}+\|u^j\nab^j({\color{black}\O}\h\o,{\color{black}\O}\od)\|_{L^2(\S)}\bigg)\\
\ls &\|u^5\nab^4({\color{black}\O}\kappa)\|_{L^2(\S)}+\sum_{j\leq 4}\bigg(\|u^{j+1}\nab^{j}({\color{black}\O}\b-\f12\O\nab_4\phi\nab\phi)\|_{L^2(\S)}+\|u^j\nab^j({\color{black}\O}\h\o,{\color{black}\O}\od)\|_{L^2(\S)}\bigg).
\end{split}
\end{equation*}
After taking $L^2$ in $\ub$, together with \eqref{nab4kappa}, the above inequality gives 
\begin{equation*}
\begin{split}
&\|u^5\nab^5({\color{black}\O}\h\o,{\color{black}\O}\od)\|_{L^2_{\ub}L^2(\S)}\\
\ls& \ub^{\f12}\at\bigg(1+\M R(\b)+\f{B^{-\f12}}{\ub^{\f12}\at}\|u^5\nab^5({\color{black}\O}\od)\|_{L^\infty_uL_{\ub}^{2}L^{2}(\S)}\bigg)+\ub^{\f12}\at[1+\M R(\b)].
\end{split}
\end{equation*}
With $B$ sufficiently large, we absorb the $u^5\nab^5(\O\od)$ term to the left and get
\begin{equation*}
\begin{split}
\|u^5\nab^5({\color{black}\O}\h\o,{\color{black}\O}\od)\|_{L^2_{\ub}L^2(\S)}\ls& \ub^{\f12}\at[1+\M R(\b)].
\end{split}
\end{equation*}
\end{proof} 

}

We then derive the estimates for $\trch$ and $\chih$.

\begin{proposition}\label{chi.5.bd}
With the assumptions of Theorem \ref{main thm} and the bootstrap assumptions \eqref{BA.1}, \eqref{BA.2}, \eqref{BA.3},\eqref{BA.4}, we have
$$\|u^6\nab^5({\color{black}\O}\trch)\|_{L^2_{\ub}L^2(\S)}\ls  \ub^{\f32}a [1+\M R(\b)],\quad \|u^5\nab^5({\color{black}\O}\trch)\|_{L^2_{\ub}L^2(\S)}\ls \ub^{\f12}\at,$$
$$\|u^5\nab^5({\color{black}\O}\chih)\|_{L^2_{\ub}L^2(\S)}\ls  \ub^{\f12}\at[1+\mathcal R(\b)].$$

\end{proposition}

\begin{proof}
We use the following equation:
$$(\O\nab_4)(\O\tr\chi)=-\f12(\O\tr\chi)^2-|\O\chih|^2-4\O\h\o \O\tr\chi-(\O\nab_4\phi)^2.$$
Commuting it with $i$ angular derivatives, we get
\begin{equation*}
\begin{split}
(\O\nab_4) \nab^{5} (\O\tr\chi)
=&\sum_{i_1+i_2+i_3+i_4=5} \nab^{i_1}\q^{i_2}\nab^{i_3}\p\nab^{i_4}\p.
\end{split}
\end{equation*}
Via applying Proposition \ref{transport}, we can bound $\|u^6\nab^5(\O\tr\chi)\|_{L_{\ub}^{\i}L_{u}^{\i}L^{2}(\S)}$ by the $\|u^6\cdot\|_{L_{u}^{\i}L_{\ub}^{1}L^{2}(\S)}$ norm of the right. Together with $\nab^5(\O\tr\chi)(u,0)=0$, we get 
\begin{equation}\label{nab5 Omega trchi mid step}
\begin{split}
\|u^6\nab^5(\O\tr\chi)\|_{L^2(\S)}\leq&\|u^6\sum_{i_1+i_2+i_3+i_4=5} \nab^{i_1}\q^{i_2}\nab^{i_3}\p\nab^{i_4}\p\|_{L^{\i}_u L^1_{\ub}L^2(\S)}\\
\ls & \ub^{\f12} \|u\p\|_{L^\infty_{\ub}L^\infty(\S)}\|u^5\nab^5\p\|_{L^{\i}_uL^2_{\ub}L^2(\S)}\\
+&\|u^6\sum_{\substack{i_1+i_2+i_3+i_4=5\\  i_3, i_4\leq 4} } \nab^{i_1}\q^{i_2}\nab^{i_3}\p\nab^{i_4}\p\|_{L^{\i}_u L^1_{\ub}L^2(\S)}\\
\ls& \ub a (O\tilde{O}+O^2). 
\end{split}
\end{equation}
This gives
\begin{equation}\label{trch.5.bd v1}
\|u^6\nab^5(\O\trch)\|_{L^2_{\ub}L^2(\S)}\ls \ub^{\f32} a (O \t O+O^2)\ll \ub^{\f32}a \t O'.
\end{equation}
{\color{black}It also implies
\begin{equation}\label{trch.5.bd 2nd}
\|u^5\nab^5(\O\trch)\|_{L^2_{\ub}L^2(\S)}\ls \ub^{\f12} a^{\f12}.
\end{equation} 
}

We then use the Codazzi equation with the schematic form
$$\div({\color{black}\O}\chih)-\frac 12\nab({\color{black}\O}\trch)+{\color{black}\O}\beta=\p\q.$$
Applying the elliptic estimates in Proposition \ref{elliptictraceless}, we obtain
\begin{equation*}
\begin{split}
\|u^5\nab^5({\color{black}\O}\chih)\|_{L^2(\S)}\ls & \sum_{i\leq 5}\|u^i\nab^i(\O\tr\chi)\|_{L^2(\S)}+\sum_{i\leq 4}\|u^{i+1}\nab^i(\O\b)\|_{L^2(\S)}\\
&+\sum_{i\leq 4}\sum_{i_1+i_2=i}\|u^{i+1}\nab^{i_1}\q\nab^{i_2}\p\|_{L^2(\S)}+\sum_{i\leq 4}\|u^i\nab^i(\O\chih)\|_{L^2(\S)}.
\end{split}
\end{equation*}
We further take the $L^2$ norm in $\ub$. With the bound \eqref{trch.5.bd v1} for $\nab^5\trch$, we deduce 
\begin{equation*}
\begin{split}
\|u^5\nab^5&({\color{black}\O}\chih)\|_{L^2_{\ub}L^2(\S)}
\ls \sum_{i\leq 5}\|u^i\nab^i({\color{black}\O}\tr\chi)\|_{L^2_{\ub}L^2(\S)}\\
&+\sum_{i\leq 4}\|u^{i+1}\nab^i(\O\b{\color{black}-\f{\O}{2}\nab_4\phi\nab_A\phi})\|_{L^2_{\ub}L^2(\S)}{\color{black}+\sum_{i\leq 4}\|u^{i+1}\nab^i(\O\nab_4\phi\nab_A\phi)\|_{L^2_{\ub}L^2(\S)}}\\
&+\sum_{i\leq 4}\sum_{i_1+i_2=i}\|u^{i+1}\nab^{i_1}\q\nab^{i_2}\p\|_{L^2_{\ub}L^2(\S)}+\sum_{i\leq 4}\|u^i\nab^i(\O\chih)\|_{L^2_{\ub}L^2(\S)}\\
\ls & \ub^{\frac 12}\at+\ub^{\f12}\at\mathcal R(\b)+\f{\ub^{\f32}a O^2}{|u|}\ls \ub^{\f12}\at[1+\mathcal R(\beta)].
\end{split}
\end{equation*} 

With the above estimate for $u^5\nab^5(\O\chih)$, together with Proposition \ref{trch.bd}, {\color{black}Proposition \ref{L2 scalar field}, Proposition \ref{Omegae4phi energy}}, we then revisit \eqref{nab5 Omega trchi mid step} and obtain 
\begin{equation*}
\begin{split}
\|u^6\nab^5&(\O\tr\chi)\|_{L^2(\S)}\leq\ub^{\f12}\|(\O\tr\chi, \O\h \o)\|_{L^{\infty}_{\ub}L^{\infty}(\S)}\|u^6 \nab^5(\O\tr\chi)\|_{L^2_{\ub}L^2(\S)}\\
&\quad\quad\quad\quad\quad\quad\quad+\ub^{\f12}\|u\O\tr\chi\|_{L^{\infty}_{\ub}L^{\infty}(\S)}\|u^5 \nab^5(\O\h \o)\|_{L^2_{\ub}L^2(\S)}\\
&\quad\quad\quad\quad\quad\quad\quad+\ub^{\f12}\|u\O\chih\|_{L^{\infty}_{\ub}L^{\infty}(\S)}\|u^5 \nab^5(\O\chih)\|_{L^2_{\ub}L^2(\S)}\\
&\quad\quad\quad\quad\quad\quad\quad{\color{black}+\ub^{\f12}\|u\O\nab_4\phi\|_{L^{\infty}_{\ub}L^{\infty}(\S)}\|u^5 \nab^5(\O\nab_4\phi)\|_{L^2_{\ub}L^2(\S)}}\\
&\quad\quad\quad\quad\quad\quad\quad+\|u^6\sum_{\substack{i_1+i_2+i_3+i_4=5\\  i_3, i_4\leq 4} } \nab^{i_1}\q^{i_2}\nab^{i_3}\p\nab^{i_4}\p\|_{L^{\i}_u L^1_{\ub}L^2(\S)}\\
\leq&\f{\ub^{\f12}\at}{|u|}\cdot \ub^{\f32}a \t O+\ub^{\f12}\O^2\cdot\ub^{\f12}\af[1+\M R(\b)]+\ub^{\f12}\at\cdot\ub^{\f12}\at[1+\M R(\b)]{\color{black}+\ub^{\f12}\at\cdot\ub^{\f12}\at} \\
\leq&\ub a[1+\M R(\b)].
\end{split}
\end{equation*}
This further implies

\begin{equation}\label{trch.5.bd}
\|u^6\nab^5(\O\trch)\|_{L^2_{\ub}L^2(\S)}\ls \ub^{\f32} a [1+\M R(\b)].
\end{equation}

\end{proof}

We then control the highest derivative of $\etb$.

\begin{proposition} \label{etab.5.bd}
Under the assumptions of Theorem \ref{main thm} and the bootstrap assumptions \eqref{BA.1}, \eqref{BA.2}, \eqref{BA.3}, \eqref{BA.4}, it holds

{\color{black}
\[
 \|u^5\nab^5\etb\|_{L^\i_uL^2_{\ub}L^2(\S)} \ls \ub^{\f12}\af\cdot \f{\ub^{\f12}\af}{|u|^{\f12}\O}.
\]
}

\end{proposition}

\begin{proof}
We define $\mub$ via 
$$\mub=-\div\etb+\Kt.$$ 
Thus $\etb$ and $\mub$ satisfy the following Hodge system:
$$\div\etb=-\mub+\Kt,\quad\curl\etb=-\sigmac.$$
Recall $\nab_3\etb_A=-\chib_{AB}\cdot(\etb-\eta)_B+\beb_A-\f12\nab_A\phi \cdot e_3\phi$ and
\begin{equation*}
\begin{split}
\nab_3 (K-\f{1}{|u|^2}-&\f14\nab^A\phi\nab_A\phi)+\f32\trchb (K-\f{1}{|u|^2}-\f14\nab^A\phi\nab_A\phi)\\
=& \div(\betab+\f12\nab_3\phi\nab\phi)-\f12\nab^A\nab_3\phi\nab_A\phi-\zeta\cdot\betab+2\eta\cdot\betab\\
&+\frac 12 \chibh\cdot\nab\widehat{\otimes}\eta+\frac 12 \chibh\cdot(\eta\widehat{\otimes}\eta)-\frac 12 \trchb |\eta|^2+\f12\eta^A\nab_A\phi\nab_3\phi\\
&+\f12(\nab_A\phi\nab_B\phi-\f12 g_{AB}\nab_C\phi\nab^C\phi)\cdot\chibh^{AB}-\f18\tr\chib\nab^A\phi\nab_A\phi\\
&+\f12\tr\chib(-\div \eta+\Kt)-\f{\Omega^{-1}}{|u|^2}(\O\tr\chib+\f{2}{|u|}). 
\end{split}
\end{equation*}

\noindent Hence, in the schematic form, we have that $\mub$ obeys 
\begin{equation*}
\begin{split}
\O\nab_3\mub+\O\tr\chib\,\mub=&\q\nab(\eta,\etb)+\q\,\q\,\q+\q\nab(\O\chibh, \O\trchb)+\O\tr\chib\div\eta+\O\tr\chib K+\f{\O}{|u|^3}\\
&{\color{black}+\O\tr\chib\,\q\,\q+\Delta_g\phi (\O e_3\phi)+e_A\phi e^A(\O e_3\phi)+\eta e_A\phi (\O e_3\phi)}.
\end{split}
\end{equation*}
Commuting with $4$ angular derivatives, we get
\begin{equation*}
\begin{split}
&\O\nab_3 \nab^4 \mub+3\O\tr\chib\nab^4\mub\\
=&\frac{\O}{|u|}\nab^5 \eta+\O e_3\phi\cdot \nab^5(e_A\phi)+e_A\phi\cdot\nab^5(\O e_3\phi)+\q\nab^5(\eta,\etab,\O\trchb,\O\chibh)+\frac{1}{|u|}\sum_{i_1+i_2+i_3=4}\nab^{i_1}\q^{i_2+1}\nab^{i_3} \q\\
&+\frac{1}{|u|}\sum_{i_1+i_2+i_3=4}\nab^{i_1}\q^{i_2}\nab^{i_3} K+\sum_{i_1+i_2+i_3=4} \nab^{i_1}\q^{i_2+2}\nab^{i_3}\q+\sum_{i_1+i_2+i_3=4} \nab^{i_1}\q^{i_2+1}\nab^{i_3}K\\
&{\color{black}+\sum_{\substack{i_1+i_2+i_3=4\\ i_3\leq 3}}\nab^{i_1}\q^{i_2}\nab^{i_3+1}e_A\phi\nab^{i_4}(\O e_3\phi)+\sum_{\substack{i_1+i_2+i_3=4\\ i_4\leq 3}}\nab^{i_1}\q^{i_2}\nab^{i_3}e_A\phi\nab^{i_4+1}(\O e_3\phi).}
\end{split}
\end{equation*}

We then employ Proposition \ref{evolution lemma} with $\lambda_0=3$ and $\|u^{5}\nab^4\mub\|_{L_{u}^{\i}L^{2}(\S)}$ can be controlled by the $\|u^5\cdot\|_{L_{u}^{1}L^{2}(\S)}$ norm of the right hand side. We bound each of these terms. We begin with terms containing the $5$ angular derivatives on the Ricci coefficients. First note that by Proposition \ref{eta.5.bd}, we have
\begin{equation*}
\begin{split}
&\|\O u^4\nab^5\eta\|_{L^2_{\ub}L^1_uL^2(\S)}\leq \|\O u^{-2}\|_{L^1_u}\|u^6\nab^5\eta\|_{L^\infty_uL^2_{\ub}L^2(\S)}\ls \frac{\ub^{\f32}a^{\f34}}{|u|}\F \mathcal R \leq \ub^{\f12}\af\cdot\f{\ub^{\f12}\af}{|u|^{\f12}\O}.
\end{split}
\end{equation*}

{\color{black}
The next term is a borderline term. Via Proposition \ref{Omegae4phi energy}, we control it as 
\begin{equation*}
\begin{split}
&\|u^5\nab^5(e_A\phi)(\O e_3\phi)\|_{L^2_{\ub}L^1_uL^2(\S)}\\
\leq& \|\O u^5\nab^5(e_A\phi)\|_{L^{\infty}_{\ub}L^2_uL^2(\S)}       \|\O^{-1}(\O e_3\phi)\|_{L^2_{u}L^2_{\ub}L^{\infty}(\S)}\\
\leq&\ub^{\f12}\at\cdot\f{\O^{-1}}{|u|}\cdot|u|^{\f12}\ub^{\f12}=\ub^{\f12}\af\cdot \f{\ub^{\f12}\af}{|u|^{\f12}\O} .
\end{split}
\end{equation*} 

And the term containing $\nab^5(\O e_3\phi)$ obeys
\begin{equation*}
\begin{split}
&\|u^5\nab^5(\O e_3\phi)(e_{{\color{black}A}}\phi)\|_{L^2_{\ub}L^1_uL^2(\S)}\\
\leq& \|\O^{-1} u^5\nab^5(\O e_3\phi)\|_{L^{\infty}_{\ub}L^2_uL^2(\S)}       \|\O e_{{\color{black}A}}\phi\|_{L^2_{u}L^2_{\ub}L^{\infty}(\S)}\\
\leq&\ub^{\f12}\af\F\M S\cdot\f{\ub\at}{|u|^2}|u|^{\f12}\ub^{\f12} O=\ub^{\f12}\af\cdot\f{\ub^{\f32}\at}{|u|^{\f32}}\F\M S O\leq \ub^{\f12}\af\cdot\f{\ub^{\f12}\af}{|u|^{\f12}\O}.
\end{split}
\end{equation*} 

}

{\color{black}
Then the top-order-derivative terms with $\nab^5\eta$ or $\nab^5\etb$ obey
\begin{equation*}
\begin{split}
&\|u^5\q\nab^5(\eta,\etab)\|_{L^2_{\ub}L^1_uL^2(\S)}\\
\leq &\|u^{-2}\|_{L^1_u}\|u^2\q\|_{L^\i_{\ub}L^\infty_uL^\infty(\S)}\|u^5\nab^5\etab\|_{L^\infty_uL^2_{\ub}L^2(\S)}+\|u^{-1}\q\|_{L^{\infty}_{\ub}L^1_u L^{\infty}(\S)}\|u^6 \nab^5\eta\|_{L^{\infty}_u L^2_{\ub} L^2(\S)}\\
\leq&\f{\ub\at O}{|u|}\cdot \ub^{\f12}\af\F \t O+\f{\ub\at}{|u|^2}\ub^{\f32}a^{\f34} O \t O\leq \ub^{\f12}\af\cdot\f{\ub^{\f12}\af}{|u|^{\f12}\O}.
\end{split}
\end{equation*}
}

And the top-order-derivative terms $\nab^5(\O\trchb,\O\chibh)$ satisfy 
{\color{black}
\begin{equation*}
\begin{split}
&\|u^5\q\nab^5(\O\trchb, \O\chibh)\|_{L^2_{\ub}L^1_uL^2(\S)}\\
\leq&\|\q\|_{L^2_{\ub}L^2_u L^{\infty}(\S)}\|u^5 \nab^5(\O\tr\chib, \O\chibh)\|_{L^{\infty}_{\ub} L^2_u L^2(\S)}\\
\leq&\f{\ub\at}{|u|^2}\ub^{\f12}|u|^{\f12}O\cdot \ub^{\f12}\af\F \t O\leq \ub^{\f12}\af\cdot\f{\ub^{\f12}\af}{|u|^{\f12}\O}.
\end{split}
\end{equation*}
}

The next term is lower-order and we have
\begin{equation*}
\begin{split}
&\|u^5\frac{1}{|u|}\sum_{i_1+i_2+i_3=4}\nab^{i_1}\q^{i_2+1}\nab^{i_3} \q\|_{L^2_{\ub}L_{u}^{1}L^{2}(\S)}\leq  \frac{\ub^{\f52} a O^2}{|u|^2}\leq \ub^{\f12}\af\cdot\f{\ub^{\f12}\af}{|u|^{\f12}\O}.
\end{split}
\end{equation*}

For the term containing $K$, via using Proposition \ref{Sobolev}, we deduce
\begin{equation*} 
\begin{split}
&\|u^5\frac{1}{|u|}\sum_{i_1+i_2+i_3=4}\nab^{i_1}\q^{i_2}\nab^{i_3} K\|_{L^2_{\ub}L_{u}^{1}L^{2}(\S)}\leq \ub^{\f12}\af\cdot\f{\ub^{\f12}\af}{|u|^{\f12}\O}.
\end{split}
\end{equation*}

The remaining lower-order terms are all bounded by $\ub^{\f12}\af\cdot\f{\ub^{\f12}\af}{|u|^{\f12}\O}$. Hence, combining the above estimates, we arrive at
$$\|u^5\nab^4\mub\|_{L^2_{\ub}L^\i_u L^2(\S)}\leq \ub^{\f12}\af\cdot\f{\ub^{\f12}\af}{|u|^{\f12}\O}.$$

By using the div-curl system
$$\div\etb=-\mub+\Kt,\quad\curl \etb=-\sigmac$$
and employing elliptic estimates from Proposition \ref{ellipticthm}, we further deduce
{\color{black}
\begin{equation*}
\begin{split}
\|u^5\nab^5\etb\|_{L^2_{\ub}L^2(\S)}\ls& \|u^5\nab^4\mub\|_{L^2_{\ub}L^2(\S)}+\|u^5\nab^4(\Kt, \sigmac)\|_{L^2_{\ub}L^2(\S)}\\
&+\sum_{i\leq 3}\|u^{i+1}\nab^i(-\underline{\mu}+\Kt, \sigmac)\|_{L^2_{\ub}L^2(\S)}\\
\leq& {\ub^{\f12}\af}\cdot\f{\ub^{\f12}\af}{|u|^{\f12}\O}+\f{\ub^{\f32}a^{\f34}}{|u|}\F\M R\leq \ub^{\f12}\af\cdot\f{\ub^{\f12}\af}{|u|^{\f12}\O}.
\end{split}
\end{equation*}  

}

\end{proof}

We proceed to obtain the highest order derivative estimates for $\eta$. In the proof we will derive an improved bound for $\nab^i\mu$ with $1\leq i\leq 4$.

\begin{proposition}\label{eta.5.bd}
With the assumptions of Theorem \ref{main thm} and the bootstrap assumptions \eqref{BA.1}, \eqref{BA.2}, \eqref{BA.3}, \eqref{BA.4} and the proved {\color{black}$\M O_{5,2}(\etb)\ls 1$}, we have 
{\color{black}
\begin{equation*}
\begin{split}
\|u^6\nab^5\eta\|_{L^{\i}_uL^2_{\ub}L^2(\S)}
\ls \ub^{\f32} a^{\f34}\F[1+\M R(K,\sigmac)]+\ub^{\f32}a^{\f34}\cdot a^{-\f14}\M R(\b), 
\end{split}
\end{equation*}
$$\|u^5\nab^5\eta\|_{L^{\i}_{\ub}L^2_uL^2(\S)}\ls {\color{black}\O^{-1}}\ub^{\f12}\at [1+\underline{\M R}(K,\sigmac)],$$
$$\|u^5\O\nab^5\eta\|_{L^{\i}_{\ub}L^2_uL^2(\S)}\ls \ub^{\f12}\at [1+\underline{\M R}(K,\sigmac)].$$
}
And $\nab^i\mu$ obeys the below improved bound:
$$\sum_{1\leq i\leq 4}\|u^{i+2}\nab^i\mu\|_{L^\i_u L^\i_{\ub}L^2(\S)}\ls\ub a^{\f34}\cdot\f{\ub^{\f12}\af}{|u|^{\f12}\O}.$$
\end{proposition}

\begin{proof} 

We first define the mass aspect function $\mu$  
$$\mu:=-\div\eta+\K-\f14\nab^A\phi\nab_A\phi.$$
Recall that $\eta$ obeys the following elliptic system
$$\div\eta=-\mu+\K-\f14\nab^A\phi\nab_A\phi, \quad \quad \curl\eta=\sigmac.$$

{\color{black}
And via employing the below renormalized null Bianchi equation
\begin{equation}
\begin{split}
&\nab_4 [\O(K-\f{1}{|u|^2}-\f14\nab^A\phi \nab_A\phi)]+\trch [\O(K-\f{1}{|u|^2}-\f14\nab^A\phi \nab_A\phi)]\\
=&-\O\div(\beta-\f12\nab_4\nab\phi)-\f12\O\nab^A\nab_4\phi\cdot \nab_A\phi-2\O\h\o(K-\f{1}{|u|^2}-\f14\nab^A\phi\nab_A\phi)-\O\tr\chi\cdot\f{1}{|u|^2}\\
&+\O[-\zeta\cdot\beta-2\etab\cdot\beta+\frac 12 \chih\cdot\nab\widehat{\otimes}\etab+\frac 12 \chih\cdot(\etab\widehat{\otimes}\etab)-\frac 12 \trch\div\etab-\frac 12\trch |\etb|^2]\\
&+\O[\f12\etb^A\nab_A\phi\nab_4\phi+\f12(\nab_A\phi\nab_B\phi-\f12g_{AB}\nab_C\phi\nab^C\phi)\cdot\chih^{AB}],
\end{split}
\end{equation}
we deduce that $\mu$ satisfies the schematic equation 
}
\begin{equation*}
\begin{split}
({\color{black}\O}\nab_4) \mu=&\chi\nab(\eta,\etb)+\p\q\hspace{1pt}\q+\q\nab({\color{black}\O}\trch,{\color{black}\O}\chih)+{\color{black}\p(\Kt)+\f{\O\tr\chi}{|u|^2}}\\
& {\color{black}+\Delta_g\phi ({\color{black}\O}e_4\phi)+\nab_A\phi\nab^A ({\color{black}\O}e_4\phi}).
\end{split}
\end{equation*}
Commuting with $i$ angular derivatives, we have
\begin{equation*}
\begin{split}
&({\color{black}\O}\nab_4) \nab^i \mu\\
=&\O e_4\phi\cdot\nab^{i+1}(e_A\phi)+e_A\phi \nab^{i+1}(\O e_4\phi)+\chi\nab^{i+1}(\eta,\etb)+\q\nab^{i+1}({\color{black}\O}\trch,{\color{black}\O}\chih){\color{black}+\f{1}{|u|^2}\nab^4(\O\tr\chi)}\\
&+\sum_{\substack{i_1+i_2+i_3=i+1\\i_1,i_3\leq i}} \nab^{i_1}\q^{i_2+1}\nab^{i_3}\p+\sum_{i_1+i_2+i_3+i_4=i}\nab^{i_1}\q^{i_2}\nab^{i_3}{\color{black}\p \nab^{i_4} (\Kt)}\\
&{\color{black}+\sum_{\substack{i_1+i_2+i_3+i_4=i\\ i_3\leq i-1}}\nab^{i_1}\q^{i_2}\nab^{i_3+1}(\nab_A\phi)\nab^{i_4}({\color{black}\O}e_4\phi) }{\color{black}+\sum_{\substack{i_1+i_2+i_3+i_4=i\\ i_3\leq i-1}}\nab^{i_1}\q^{i_2}\nab^{i_3+1}({\color{black}\O}e_4\phi)\nab^{i_4}(e_A\phi) }.\\
\end{split}
\end{equation*}
Since $\nab^i\mu$ vanishes on $\ub=0$, by Proposition \ref{transport}, to bound $\|u^{i+2}\nab^i\mu\|_{L_{\ub}^{\i}L_{u}^{\i}L^{2}(\S)}$, we only need to control the $\|u^{i+2}\cdot\|_{L_{u}^{\i}L_{\ub}^{1}L^{2}(\S)}$ norm of the right. We now estimate them.

{\color{black}
Via employing the obtained estimates in Proposition \eqref{Omegae4phi energy}, Proposition \eqref{prop6.2} and Proposition \eqref{L2 scalar field}, the terms involving the highest angular derivatives of the scalar field satisfy 
\begin{equation*}
\begin{split}
&\|u^{i+2}\nab^{i+1}(e_A\phi)(\O e_4\phi)\|_{L^1_{\ub}L^2(\S)}\\
\leq&\|u^{i+1}\nab^{i+1}(e_A\phi)\|_{L^2_{\ub}L^2(\S)}   \|u(\O e_4\phi)\|_{L^2_{\ub}L^{\infty}(\S)}\\
\leq&\bigg(\f{\ub^{\f12}\af}{|u|^{\f12}\O}\cdot\ub^{\f12}\af+(\f{\ub \at}{|u|\O^{\f32}})^{\f14}\cdot \ub^{\f12}\af \bigg)\cdot \ub^{\f12}\at\leq \ub a^{\f34}\cdot\f{\ub^{\f12}\af}{|u|^{\f12}\O},
\end{split}
\end{equation*}

\begin{equation*}
\begin{split}
&\|u^{i+2}\nab^{i+1}(\O e_4\phi)(e_A\phi)\|_{L^1_{\ub}L^2(\S)}\\
\leq&\|u^{i+1}\nab^{i+1}(\O e_4\phi)\|_{L^2_{\ub}L^2(\S)}   \|u e_A\phi\|_{L^2_{\ub}L^{\infty}(\S)}\\
\leq&\ub^{\f12}\at \M S\cdot \f{\ub\at}{|u|}\ub^{\f12}O=\ub\at\f{\ub\at}{|u|}\M S O\leq \ub a^{\f34}\cdot\f{\ub^{\f12}\af}{|u|^{\f12}\O}.
\end{split}
\end{equation*}
}

The terms containing the highest angular derivatives of $\eta$ and $\etb$ obey
\begin{equation*}
\begin{split}
&\sum_{i\leq 4}\|u^{i+2}\chi\nab^{i+1}(\eta,\etb)\|_{L^1_{\ub}L^2(\S)}\\
\leq&\ub^{\f12}\|\chi\|_{L^{\infty}_{\ub}L^{\infty}(\S)}\|u^6\nab^5\eta\|_{L^2_{\ub}L^2(\S)}+\ub^{\f12}\|u\chi\|_{L^{\infty}_{\ub}L^{\infty}(\S)}\|u^5\nab^5\etb\|_{L^2_{\ub}L^2(\S)}\\
\leq&\f{\ub^{\f12}\at O}{|u|}\ub^{\f32}a^{\f34}\M R+\ub^{\f12}\at O(\chi)\cdot\ub^{\f12}\af\cdot\f{\ub^{\f12}\af}{|u|^{\f12}\O}\t O(\etb)\ls \ub a^{\f34}\cdot\f{\ub^{\f12}\af}{|u|^{\f12}\O}.
\end{split}
\end{equation*}
For the last line we utilize Proposition \eqref{chih.bd}, Proposition \eqref{trch.bd} and Proposition \eqref{etab.5.bd}.

The terms with the highest derivatives of ${\color{black}\O}\trch$ and ${\color{black}\O}\chih$ satisfy 
\begin{equation*}
\begin{split}
&\sum_{i\leq 4}\|u^{i+2}\q\nab^{i+1}({\color{black}\O}\trch,{\color{black}\O}\chih)\|_{L^1_{\ub}L^2(\S)}\\
\ls & \frac{\ub^{\f12}}{|u|}\|u^2\q\|_{L^{\i}_{\ub}L^{\i}(\S)}\cdot\bigg(\sum_{i\leq 3}\|u^{i+1}\nab^{i+1}({\color{black}\O}\trch,{\color{black}\O}\chih)\|_{L^2_{\ub}L^2(\S)}+\|u^5\nab^5({\color{black}\O}\trch,{\color{black}\O}\chih)\|_{L^2_{\ub}L^2(\S)}\bigg)\\
\ls & \frac{\ub^{\frac 12}}{|u|}\ub\at O\cdot\ub^{\f12}\at \t O
\ls \frac{\ub^2 a O \t O}{|u|}\ls \ub a^{\f34}\cdot\f{\ub^{\f12}\af}{|u|^{\f12}\O}.
\end{split}
\end{equation*}

{\color{black}
By Proposition \eqref{trch.bd}, we also have
\begin{equation*}
\begin{split}
\|u^4\nab^4(\O\tr\chi)\|_{L^1_{\ub}L^2(\S)}\leq&\|\f{1}{|u|}\|_{L^1_{\ub}L^{\i}(\S)}\|u^5\nab^4(\O\tr\chi-\O^2(\O^{-1}\tr\chi)_0)\|_{L^{\i}_{\ub}L^2(\S)}+\|u^3\nab^4 \O^2\|_{L^1_{\ub}L^2(\S)}\\
\leq&\f{\ub}{|u|}\cdot\ub a+\f{\ub\at}{|u|^2}\cdot\ub |u|O\leq \ub a^{\f34}\cdot\f{\ub\at}{|u|}+\ub\at\cdot\f{\ub O}{|u|}\leq \ub a^{\f34}\cdot\f{\ub^{\f12}\af}{|u|^{\f12}\O}.
\end{split}
\end{equation*}
}

\noindent Employing Proposition \ref{product} together with the bootstrap assumption \eqref{BA.2}, we then bound the other lower order terms and they obey $\ls \ub a^{\f34}\cdot \f{\ub^{\f12}\af}{|u|^{\f12}\O}$.

Combining all the estimates above, we obtain
{\color{black}
\begin{equation}\label{mu.est}
\begin{split}
\sum_{1\leq i\leq 4}\|u^{i+2}\nab^i\mu\|_{L^{\i}_uL^2(\S)}
\ls& \ub a^{\f34}\cdot\f{\ub^{\f12}\af}{|u|^{\f12}\O}.
\end{split}
\end{equation}
}

We then recall the div-curl system
 $$\div\eta=-\mu+\Kt,\quad\curl \eta=\sigmac.$$
Via employing elliptic estimates from Proposition \ref{ellipticthm} and the fact $\O\leq 1$, we obtain
\begin{equation*}
\begin{split}
&\|u^6\nab^5\eta\|_{L^2(\S)}\\
\ls& \sum_{i\leq 4}\bigg(\|u^{i+2}\nab^i\mu\|_{L^2(\S)}+\|u^{i+2}\nab^i (\Kt,\sigmac)\|_{L^2(\S)}+\|u^{i+1}\nab^i\eta\|_{L^2(\S)}\bigg)\\
\ls& \ub a^{\f34}\F\cdot \M R(K)+\sum_{i\leq 4}\|u^{i+2}\nab^i (\Kt,\sigmac)\|_{L^2(\S)}+\ub\at \M R(\b),
\end{split}
\end{equation*}
\begin{equation*}
\begin{split}
&\|u^5\O\nab^5\eta\|_{L^2(\S)}\\
\ls&\sum_{i\leq 4} \bigg(\|\O u^{i+1}\nab^i\mu\|_{L^2(\S)}+\|\O u^{i+1}\nab^i (\Kt,\sigmac)\|_{L^2(\S)}+\|\O u^{i}\nab^i\eta\|_{L^2(\S)}\bigg)\\
\ls& \sum_{i\leq 4}\|\O u^{i+1}\nab^i (\Kt,\sigmac)\|_{L^2(\S)}+\f{\ub a^{\f34}}{|u|}\cdot\O\cdot\F \cdot \M R(K)+\f{\ub\at}{|u|}\cdot \O O.
\end{split}
\end{equation*}
We then take the $L^2$ norm in $\ub$ and $u$ respectively.  Via using \eqref{mu.est}, we further obtain
\begin{equation*}
\begin{split}
\|u^6\nab^5\eta\|_{L^{\i}_uL^2_{\ub}L^2(\S)}\ls& {\ub^{\f32} a^{\f34}}\F\cdot\M R(K, \sigmac)+\ub^{\f32}a^{\f34}\cdot a^{-\f14}\M R(\b),
\end{split}
\end{equation*}
 
\begin{equation}\label{eta elliptic Omega}
\begin{split}
\|\O u^5\nab^5\eta\|_{L^{\i}_{\ub}L^2_uL^2(\S)}\ls \ub^{\f12}\at \underline{\M R}(K, \sigmac)+\ub^{\f12}\at \ls \ub^{\f12}\at[1+\underline{\M R}(K, \sigmac)].
\end{split}
\end{equation}
This also infers
$$\|u^5\nab^5\eta\|_{L^{\i}_{\ub}L^2_u L^2(\S)}\leq \O^{-1}\ub^{\f12}\at\cdot[1+\underline{\M R}(K, \sigmac)].$$

We hence finish the proof of this proposition.
\end{proof}

Finally, we move to bound the highest order estimates for $\trchb$ and $\chibh$.

\begin{proposition}\label{nabla5chib}
Under the assumptions of Theorem \ref{main thm} and the bootstrap assumptions \eqref{BA.1}, \eqref{BA.2}, \eqref{BA.3}, \eqref{BA.4}, it holds

\begin{equation*}
\begin{split}
\|u^{5}\nab^5(\O\trchb, \O\chibh)\|_{L^\i_{\ub}L^2_uL^2(\S)}\ls& {\ub^{\f12}\af}\bigg(1+\f{\ub^{\f12}\af}{|u|^{\f12}\O} \bigg),\\
\|u^{5}\nab^5\tr\chib\|_{L^\i_{\ub}L^2_u L^2(\S)}\ls& {\ub^{\f12}\af}{\color{black}\bigg(1+\f{\ub^{\f12}\af}{|u|^{\f12}\O}\bigg)},\\
\|u^{5}\nab^5\chibh\|_{L^\i_{\ub}L^2_u L^2(\S)}\ls& {\ub^{\f12}\af}{\color{black}\bigg(1+\f{\ub^{\f12}\af}{|u|^{\f12}\O}\bigg)}\cdot\M R(\b).
\end{split}
\end{equation*}

\end{proposition}

\begin{proof}
We utilize the following equation for $\trchb$:
{\color{black}
$$\O\nab_3(\O^{-1}\tr\chib)+\f12\O\tr\chib(\O^{-1}\tr\chib)=-|\chibh|^2-\nab_3\phi\nab_3\phi.$$
}
Via commuting the above equation with $5$ angular derivatives, we obtain
\begin{equation}\label{nabla 5 trchib}
\begin{split}
&\O\nab_3 \nab^5 (\O^{-1}\tr\chib)+3\O\tr\chib\nab^5(\O^{-1}\tr\chib)\\
=&\q\nab^5(\O^{-1}\trchb,\O^{-1}\chibh)+{\color{black}\O^{-1}}e_3\phi\nab^5({\color{black}\O}e_3\phi)+{\color{black}\O^{-2}}\sum_{\substack{i_1+i_2+i_3=5\\i_1,i_3\leq 4}} \nab^{i_1}\q^{i_2+1}\nab^{i_3}\q\\
&+\sum_{i_1+i_2+i_3=4}\f{{\color{black}\O^{-2}}}{|u|}\nab^{i_1}\q^{i_2+1}\nab^{i_3}\q+\sum_{i_1+i_2=4}\f{{\color{black}\O^{-2}}}{|u|^2}\nab^{i_1}\q^{i_2+1}.
\end{split}
\end{equation} 

{\color{black}
In the derivation of above equation, we employ 
\begin{equation*}
\begin{split}
e_3\phi\nab^5(e_3\phi)=&e_3\phi\nab^5[e_3\phi-(e_3\phi)_0]=e_3\phi\nab^5[\O^{-1}\O(e_3\phi-(e_3\phi)_0)]\\
=&\O^{-1}e_3\phi\sum_{\substack{i_1+i_2=5\\ i_2\leq 4}}\q^{i_1}\nab^{i_2}[\O(e_3\phi-(e_3\phi)_0)]+\O^{-1}e_3\phi\nab^5[\O(e_3\phi-(e_3\phi)_0)]\\
=&\O^{-1}e_3\phi\sum_{\substack{i_1+i_2=5\\ i_2\leq 4}}\q^{i_1}\nab^{i_2}[\O(e_3\phi-(e_3\phi)_0]+\O^{-1}e_3\phi\nab^5[(\O e_3\phi)-(\O e_3\phi)_0]\\
&-\O^{-1}e_3\phi\nab^5[(\O-\O_0)(e_3\phi)_0].
\end{split}
\end{equation*}
Note that $\O^{-1}e_3\phi\nab^5[(\O e_3\phi)-(\O e_3\phi)_0]$ is the same as the last term in (\ref{nabla 5 trchib}) and the rest terms are absorbed by the other same-type terms in (\ref{nabla 5 trchib}).
} We also notice
\begin{equation*}
\begin{split}
\|u^5\nab^5(\O\tr\chib)\|_{L^2(\S)}\ls& \|\O^2\|_{L^{\i}(\S)}\|u^5\nab^5(\O^{-1}\tr\chib)\|_{L^2(\S)}+\f{\ub\at O}{|u|},\\
\|u^5\nab^5\tr\chib\|_{L^2(\S)}\ls& \|\O\|_{L^{\i}(\S)}\|u^5\nab^5(\O^{-1}\tr\chib)\|_{L^2(\S)}+\f{\ub\at O}{|u|\O}.
\end{split}
\end{equation*}
To bound $\|u^5\nab^5(\O^{-1}\tr\chib)\|_{L^2(\S)}$, we then apply Proposition \ref{evolution lemma} with $\lambda_0=3$ and it can be controlled by the $\|u^5\cdot\|_{L_{\ub}^{\i}L_{u}^{1}L^{2}(\S)}$ norm of the right hand side of \eqref{nabla 5 trchib}. 

Together, to estimate $\|u^5\nab^5(\O\tr\chib)\|_{L^2(\S)}$, we first bound
{\color{black}
\begin{equation*}
\begin{split}
&\|\O^2\|_{L^{\infty}(\S)}\cdot \|u^5\q\nab^5(\O^{-1}\trchb,\O^{-1}\chibh)\|_{L_{\ub}^{\i}L_{u}^{1}L^{2}(\S)}\\
\leq&\|u^5\nab^5(\O\tr\chib, \O\chibh)\|_{L^{\infty}_{\ub}L^2_u L^2(\S)}\cdot\|\q\|_{L^{\infty}_{\ub}L^2_u L^{\i}(\S)}\\
\leq&\ub^{\f12}\af\F\t O\cdot\f{\ub\at}{|u|^2}|u|^{\f12}O=\f{\ub^{\f12}\af}{|u|^{\f12}}\cdot\f{\ub\at}{|u|}\F \t O O\leq \f{\ub^{\f12}\af}{|u|^{\f12}}\cdot\f{\ub^{\f12}\af}{|u|^{\f12}\O^{\f12}}.
\end{split}
\end{equation*}
}
To bound $\|u^5\nab^5\tr\chib\|_{L^2(\S)}$, similarly, we get
\begin{equation*}
\begin{split}
&\|\O\|_{L^{\infty}(\S)}\cdot \|u^5\q\nab^5(\O^{-1}\trchb,\O^{-1}\chibh)\|_{L_{\ub}^{\i}L_{u}^{1}L^{2}(\S)}\\
\leq&\|u^5\nab^5(\tr\chib, \chibh)\|_{L^{\infty}_{\ub}L^2_u L^2(\S)}\cdot\|\q\|_{L^{\infty}_{\ub}L^2_u L^{\i}(\S)}\\
\leq&\ub^{\f12}\af\F\t O\cdot\f{\ub\at}{|u|^2}|u|^{\f12}O=\f{\ub^{\f12}\af}{|u|^{\f12}}\cdot\f{\ub\at}{|u|}\F \t O O\leq \f{\ub^{\f12}\af}{|u|^{\f12}}\cdot\f{\ub^{\f12}\af}{|u|^{\f12}\O^{\f12}}.
\end{split}
\end{equation*} 

We then move to control the remaining terms. We first bound
\begin{equation*}
\begin{split}
&\|(\O^2, \O)\|_{L^{\infty}(\S)}\cdot\|\O^{-1} u^5  e_3\phi\nab^5(\O e_3 \phi)\|_{L^{\i}_{\ub}L^1_uL^2(\S)}\\
\ls&\|(\O, 1)\|_{L^{\infty}(\S)}\|\O e_3\phi\|_{L^{\i}_{\ub}L^2_uL^{\i}(\S)}    \|\O^{-1} u^5\nab^5(\O e_3 \phi)\|_{L^{\i}_{\ub}L^2_uL^2(\S)}\\
\ls &\|(\O, 1)\|_{L^{\infty}(\S)}\cdot\|u^{-1}\|_{L^2_u}\cdot \bigg(\f{\ub^{\f12}\af}{|u|^{\f12}\O}\cdot\ub^{\f12}\af+(\f{\ub\at}{|u|\O^{\f32}})^{\f14}\cdot\ub^{\f12}\af\bigg)\\
\ls& \f{1}{|u|^{\f12}}\cdot \bigg(\f{\ub^{\f12}\af}{|u|^{\f12}\O}\cdot\ub^{\f12}\af+(\f{\ub\at}{|u|\O^{\f32}})^{\f14}\cdot\ub^{\f12}\af\bigg).
\end{split}
\end{equation*} 
Here for the third line, we employ Proposition {\color{black}\ref{prop6.2}}. For the rest terms, we have 

\begin{equation*}
\begin{split}
&{\color{black}\|(\O^2, \O)\|_{L^{\infty}(\S)}}\cdot\|\O^{-2}u^5\sum_{\substack{i_1+i_2+i_3=5\\i_1,i_3\leq 4}} \nab^{i_1}\q^{i_2+1}\nab^{i_3}\q\|_{L_{\ub}^{\i}L_{u}^{1}L^{2}(\S)}
\leq \f{\ub^{\f12}\af}{|u|^{\f12}}\cdot\f{\ub^{\f12}\af}{|u|^{\f12}\O^{\f12}},
\end{split}
\end{equation*}
\begin{equation*}
\begin{split}
&{\color{black}\|(\O^2, \O)\|_{L^{\infty}(\S)}}\cdot \|u^5\sum_{i_1+i_2+i_3=4} \frac{\O^{-2}}{|u|}\nab^{i_1}\q^{i_2+1}\nab^{i_3}\q\|_{L_{\ub}^{\i}L_{u}^{1}L^{2}(\S)}\leq \f{\ub^{\f12}\af}{|u|^{\f12}}\cdot\f{\ub^{\f12}\af}{|u|^{\f12}\O^{\f12}},
\end{split}
\end{equation*}

\begin{equation*}
\begin{split}
&{\color{black}\|(\O^2,\O)\|_{L^{\infty}(\S)}}\cdot \|u^5\sum_{i_1+i_2=4}\f{\O^{-2}}{|u|^2}\nab^{i_1}\q^{i_2+1}\|_{L_{\ub}^{\i}L_{u}^{1}L^{2}(\S)}\leq\f{\ub^{\f12}\af}{|u|^{\f12}}\cdot(\f{\ub^{\f12}\af}{|u|^{\f12}\O^{\f12}})^{\f12}.
\end{split}
\end{equation*}

Collecting all the above estimates, we derive
\begin{equation}
\begin{split}
&\|u^5\nab^5(\O\tr\chib)\|_{L^2(\S)}\leq \f{\ub^{\f12}\af}{|u|^{\f12}}\cdot\O\cdot{\color{black}(1+\f{\ub^{\f12}\af}{|u|^{\f12}\O})}+\f{\ub^{\f12}\af}{|u|^{\f12}}(\f{\ub^{\f12}\af}{|u|^{\f12}\O^{\f12}})^{\f12},\\
&\|u^5\nab^5(\O\tr\chib)\|_{L^{\i}_{\ub}L^2_u L^2(\S)}\leq  {\ub^{\f12}\af}{\color{black}(1+\f{\ub^{\f12}\af}{|u|^{\f12}\O})}, \\
&\|u^5\nab^5\tr\chib\|_{L^2(\S)}\leq\f{\ub^{\f12}\af}{|u|^{\f12}\O}\cdot\f{\ub^{\f12}\af}{|u|^{\f12}}+(\f{\ub\at}{|u|\O^{\f32}})^{\f14}\cdot\f{\ub^{\f12}\af}{|u|^{\f12}}+\f{\ub^{\f12}\af}{|u|^{\f12}}(\f{\ub^{\f12}\af}{|u|^{\f12}\O^{\f12}})^{\f12},\\
&\|u^5\nab^5\tr\chib\|_{L^{\infty}_{\ub}L^2_uL^2(\S)}\leq {\ub^{\f12}\af}{\color{black}(1+\f{\ub^{\f12}\af}{|u|^{\f12}\O})}.
\end{split}
\end{equation}

To bound the fifth angular derivatives of $\chibh$, we now employ the Codazzi equation
$$\div\chibh=\beb+\f12\nab\tr\chib-\f12(\eta-\etb)\cdot(\chibh-\f12\tr\chib).$$
Applying Proposition \ref{elliptictraceless}, we get
\begin{equation*}
\begin{split}
\|u^{5}\nab^5(\O\chibh)\|_{L^2(\S)}\ls& \sum_{i\leq 4}\bigg(\|u^{i+1}\nab^{i+1}(\O\tr\chib)\|_{L^2(\S)}+\|u^{i+1}\nab^{i}(\O\beb)\|_{L^2(\S)}\bigg)+\f{\ub\at O}{|u|}\\
\ls &\f{\ub^{\f12}\af \O}{|u|^{\f12}}(1+\f{\ub^{\f12}\af}{|u|^{\f12}\O})+\f{\ub^{\f12}\af}{|u|^{\f12}}\cdot (\f{\ub^{\f12}\af}{|u|^{\f12}\O^{\f12}})^{\f12}+\|u^{i+1}\nab^{i}(\O\beb)\|_{L^2(\S)}+\f{\ub\at O}{|u|},
\end{split}
\end{equation*}
\begin{equation*}
\begin{split}
\|u^{5}\nab^5\chibh\|_{L^2(\S)}\ls& \sum_{i\leq 4}\bigg(\|u^{i+1}\nab^{i+1}\tr\chib\|_{L^2(\S)}+\|u^{i+1}\nab^{i}\beb\|_{L^2(\S)}\bigg)+\f{\ub\at O(\eta, \etb)}{|u|\O}.
\end{split}
\end{equation*}
We then take the $L^2$ norm in $u$ and obtain
{\color{black}
\begin{equation*}
\begin{split}
&\|u^{5}\nab^5(\O\chibh)\|_{L^\i_{\ub}L^2_u L^2(\S)}\\
\ls& \ub^{\f12}\af\F+\sum_{i\leq 4}\|u^{i+2}\nab^{i}\beb\|_{L^\i_{\ub}L^2_u L^2(\S)}\|u^{-1}\O\|_{L^{\i}_{\ub}L^{\i}_u L^{\i}(\S)}+\f{\ub\at O}{|u|^{\f12}}\\
\ls &{\ub^{\f12}\af}{\color{black}(1+\f{\ub^{\f12}\af}{|u|^{\f12}\O})}+\ub^{\f32}a^{\f34}\F \M R\cdot\f{1}{|u|}+\f{\ub\at O}{|u|^{\f12}}\leq {\ub^{\f12}\af}{\color{black}(1+\f{\ub^{\f12}\af}{|u|^{\f12}\O})},
\end{split}
\end{equation*} 
\begin{equation*}
\begin{split}
&\|u^{5}\nab^5\chibh\|_{L^\i_{\ub}L^2_u L^2(\S)}\\
\ls& \sum_{i\leq 4}\bigg(\|u^{i+1}\nab^{i+1}\tr\chib\|_{L^\i_{\ub}L^2_u L^2(\S)}+\|u^{i+1}\nab^{i}\beb\|_{L^\i_{\ub}L^2_u L^2(\S)}\bigg)+\f{\ub\at O(\eta, \etb)}{|u|^{\f12}\O}\\
\ls& \sum_{i\leq 4}\|u^{i+1}\nab^{i+1}\tr\chib\|_{L^\i_{\ub}L^2_u L^2(\S)}+\sum_{i\leq 4}\|u^{i+2}\nab^{i}\beb\|_{L^\i_{\ub}L^2_u L^2(\S)}\cdot\|u^{-1}\|_{L^{\i}_{\ub}L^{\i}_u L^{\i}(\S)}+\f{\ub\at \M R(\b)}{|u|^{\f12}\O}\\
\ls &{\ub^{\f12}\af}{\color{black}(1+\f{\ub^{\f12}\af}{|u|^{\f12}\O})}+\ub^{\f12}\af\cdot\f{\ub^{\f12}\af}{|u|^{\f12}\O}\cdot\M R(\b),
\end{split}
\end{equation*} 
where we use Proposition {\color{black}\ref{q.bd}} and Proposition  {\color{black}\ref{om.5.bd}}.

}

Collecting all above estimates, we hence prove this proposition. 

\end{proof}

We conclude this section by summarizing all the proved estimates in this section with the aid of the $\tilde{\mathcal O}_{5,2}$ norm:
\begin{proposition}\label{O52.bd}
Under the assumptions of Theorem \ref{main thm} and the bootstrap assumptions \eqref{BA.1}, \eqref{BA.2}, \eqref{BA.3}, \eqref{BA.4}, it holds
\[
 \tilde{\mathcal O}_{5,2}\ls 1+\mathcal R.
\]

\end{proposition}

\section{Energy Estimates for Curvature Components}\label{energy estimate curvature} 

{\color{black} To carry out the energy estimates, we employ the below paired equations with suitable $\O$ weights. They are from the renormalized null Bianchi equations. The first pair are 
\begin{equation}
\label{eq:null.Bianchi2}
\begin{split}
&\nab_3[\O(\beta_{\cdot}-\f12\nab_4\phi\nab_{\cdot}\phi)]_A+\trchb\O(\beta_A-\f12\nab_4\phi\nab_A\phi)\\
=&-\O\nabla (K-\f{1}{|u|^2}-\f14\nab^C\phi\nab_C\phi)  +\O{^*\nabla}\sigmac-\f12\O\tr\chib\nab_4\phi\nab_A\phi+\O\omb\nab_4\phi\nab_A\phi\\
&+\O[2\chih\cdot\betab\Red{-}3(\eta K\Red{-}^*\eta\sigmac)+\frac 1 2(\nabla(\chih\cdot\chibh)-^*\nabla(\chih\wedge\chibh))-\frac 34 \eta\trch\trchb]\\
&+\O[{\color{black}\f32}(\eta\chih\cdot\chibh-^*\eta\chih\wedge\chibh)-\frac 14 (\nab\trch \trchb+\trch\nab\trchb)]\\
&+\f34\O\nab_A(\nab^C\phi\nab_C\phi)-\O\nab^B R_{BA}+\f12\eta_A\nab_3\phi\nab_4\phi+\eta_A\nab_C\phi\nab^C\phi,\\
\end{split}
\end{equation}

\begin{equation}
\begin{split}
\nab_4(\O\sigmac)+\frac 32\trch\O\sigmac=&-\O\div^*(\beta-\f12\nab_4\phi\nab\phi)+\O(-\nab_1 R_{42}+\nab_2 R_{41})-2\O\h\o\sigmac\\
&+\O[-\zeta\cdot^*\beta-2\etab\cdot ^*\beta+\frac 12 \chih\cdot^*(\nab\widehat{\otimes}\etab)+\frac 12 \chih\cdot^*(\etab\widehat{\otimes}\etab)]\\
&+\f12\O(\nab_A\phi \nab_B\phi-\f12 g_{AB}\nab_C\phi \nab^C\phi)\wedge \chih^{AB},\\
\end{split}
\end{equation}

\begin{equation}
\begin{split}
&\nab_4 [\O(K-\f{1}{|u|^2}-\f14\nab^A\phi \nab_A\phi)]+\trch \O(K-\f{1}{|u|^2}-\f14\nab^A\phi \nab_A\phi)\\
=&-\O\div(\beta-\f12\nab_4\nab\phi)-\f12\O\nab^A\nab_4\phi \nab_A\phi-2\O\h\o(K-\f{1}{|u|^2}-\f14\nab^A\phi\nab_A\phi)-\O\tr\chi\cdot\f{1}{|u|^2}\\
&+\O[-\zeta\cdot\beta-2\etab\cdot\beta+\frac 12 \chih\cdot\nab\widehat{\otimes}\etab+\frac 12 \chih\cdot(\etab\widehat{\otimes}\etab)-\frac 12 \trch\div\etab-\frac 12\trch |\etb|^2]\\
&+\O[\f12\etb^A\nab_A\phi\nab_4\phi+\f12(\nab_A\phi\nab_B\phi-\f12g_{AB}\nab_C\phi\nab^C\phi)\cdot\chih^{AB}]. 
\end{split}
\end{equation} 

The second pair are
\begin{equation}
\begin{split}
\nab_3\sigmac+\frac 32\trchb\sigmac=&-\div ^*(\betab+\f12\nab_3\phi\nab\phi)-\zeta\cdot ^*\betab-2\eta\cdot^*\betab+\frac 12 \chibh\cdot^*(\nab\widehat{\otimes}\eta)+\frac 12 \chibh\cdot^*(\eta\widehat{\otimes}\eta)\\
&+\nab_1 R_{32}-\nab_2 R_{31}+\f12\chibh^{AB}\wedge(\nab_A\phi \nab_B\phi-\f12 g_{AB}\nab_C\phi \nab^C\phi),\\
\end{split}
\end{equation}

\begin{equation}
\begin{split}
\nab_3 (K-\f{1}{|u|^2}-&\f14\nab^A\phi\nab_A\phi)+\f32\trchb (K-\f{1}{|u|^2}-\f14\nab^A\phi\nab_A\phi)\\
=& \div(\betab+\f12\nab_3\phi\nab\phi)-\f12\nab^A\nab_3\phi\nab_A\phi-\zeta\cdot\betab+2\eta\cdot\betab\\
&+\frac 12 \chibh\cdot\nab\widehat{\otimes}\eta+\frac 12 \chibh\cdot(\eta\widehat{\otimes}\eta)-\frac 12 \trchb |\eta|^2+\f12\eta^A\nab_A\phi\nab_3\phi\\
&+\f12(\nab_A\phi\nab_B\phi-\f12 g_{AB}\nab_C\phi\nab^C\phi)\cdot\chibh^{AB}-\f18\tr\chib\nab^A\phi\nab_A\phi\\
&+\f12\tr\chib(-\div\eta+\Kt)-\f{\Omega^{-1}}{|u|^2}(\O\tr\chib+\f{2}{|u|}),
\end{split}
\end{equation}

\begin{equation}
\begin{split}
&\nab_4(\betab_{\cdot}+\f12\nab_3\phi\nab_{\cdot}\phi)_A+\trch(\betab_A{\color{black}+\f12\nab_3\phi\nab_A\phi})\\
=&\nabla (K-\f{1}{|u|^2}-\f14\nab^A\phi \nab_A\phi)+^*\nabla\sigmac+ 2\h\omega\betab +2\chibh\cdot\beta+3(\Red{-}\etab K+^*\etab\sigmac)\\
&+\nabla(-\f14\nab^A\phi \nab_A\phi+\f14\tr\chi \tr\chib-\f12\chih\cdot\chibh)-\f12{^*\nabla}(\chih\wedge\chibh){\color{black}+\f12\tr\chi\nab_3\phi\nab_A\phi}\\
&+\nab_A(g^{CD}R_{CD})+\nab^B R_{BA}-\f12\etb_A\nab_3\phi\nab_4\phi-\etb_A\nab_C\phi\nab^C\phi.\\
\end{split}
\end{equation}
}

Via applying Propositions \ref{intbypartssph}-\ref{intbyparts3}, for the first pair we have 
\begin{proposition} \label{EE.1}
Under the assumptions of Theorem \ref{main thm} and the bootstrap assumptions \eqref{BA.1}, \eqref{BA.2}, \eqref{BA.3}, \eqref{BA.4}, it holds
\begin{equation*}
\begin{split}
&\sum_{i\leq 4}(\|u^{i+1}\nab^i(\O\b-\f12\O\nab_4\phi\nab\phi)\|_{L^{\i}_uL^{2}_{\ub}L^2(\S)}+\|\O u^{i+1}\nab^i(\Kt,\sigmac)\|_{L^{\i}_{\ub}L^{2}_{u}L^2(\S)})\ls \ub^{\f12}\at.
\end{split}
\end{equation*}
By the obtained $L^2(\S)$ estimates in Section \ref{L2 scalar field} for $\phi$, it also holds
\begin{equation*}
\begin{split}
&\sum_{i\leq 4}(\|u^{i+1}\nab^i(\O\b)\|_{L^{\i}_uL^{2}_{\ub}L^2(\S)}+\|\O u^{i+1}\nab^i(\K,\sigmac)\|_{L^{\i}_{\ub}L^{2}_{u}L^2(\S)})\ls \ub^{\f12}\at.
\end{split}
\end{equation*}

\end{proposition} 

\begin{proof}
Commuting equations for 
$$\nab_3[\O(\beta-\f12\nab_4\phi\nab\phi)], \quad \nab_{ {\color{black}4}}(\O\sigmac), \quad \nab_4[\O(K-\f{1}{|u|^2}-\f14\nab^A\phi\nab_A\phi)]$$
with $i$ angular derivatives, we get 
\begin{equation}\label{eqn.G1}
\nab_3\nab^i[\O(\b_{\cdot}-\f12\nab_4\phi\nab_{\cdot}\phi]_A+\O\nab\nab^i(\K-\f14\nab^B\phi\nab_B\phi)-\O{^*\nab}\nab^i\sigmac+\f{2+i}{2}\tr\chib\nab^i[\O(\b_{\cdot}-\f12\nab_4\phi\nab_{\cdot}\phi]_A=G_{1,i}
\end{equation}
with $G_{1,i}$ being
\begin{equation*}
\begin{split}
G_{1,i}=&{\color{black}\O^{-1}}\q\nab^{i+1}\p+\f{1}{|u|}\nab^{i+1}\trch+\p\nab^{i+1}(\chibh,\tr\chib)+\sum_{i_1+i_2+i_3=i}\nab^{i_1}\q^{i_2+1}\nab^{i_3}(\K,\sigmac)\\
+&\sum_{\substack{i_1+i_2+i_3=i+1\\i_1, i_3\leq i}}\nab^{i_1}\q^{i_2+1}\nab^{i_3}\p+\f{{\color{black}\O^{-1}}}{|u|}\sum_{i_1+i_2+i_3=i}\nab^{i_1}\q^{i_2+1}\nab^{i_3}\p+\sum_{\substack{i_1+i_2=i+1\\ i_1\leq i}}\nab^{i_1}\q^{i_2+2}\\
+&\f{1}{|u|^2}\sum_{i_1+i_2=i}\nab^{i_1}\q^{i_2+1}{\color{black}+\O\nab^{i+2}\phi\nab\phi,}
\end{split}
\end{equation*}
\begin{equation}\label{eqn.G2}
\nab_4\nab^i(\O\sigmac)+\O\div^*\nab^i(\b-\f12\nab_4\phi\nab\phi)=G_{2,i},
\end{equation}
where $G_{2,i}$ is 
\begin{equation*}
\begin{split}
G_{2,i}=&\q\nab^{i+1}\p+\p\nab^{i+1}\etb+\sum_{i_1+i_2+i_3+i_4=i}\nab^{i_1}\q^{i_2}\nab^{i_3}\p\nab^{i_4}\sigmac\\
&+\sum_{\substack{i_1+i_2+i_3=i+1\\i_1, i_3\leq i}}\nab^{i_1}\q^{i_2+1}\nab^{i_3}\p{\color{black}+\nab^{i+1}(\O e_4\phi)\nab\phi},
\end{split}
\end{equation*}
\begin{equation}\label{eqn.G3}
\nab_4\nab^i[\O(\K-\f14\nab^A\phi\nab_A\phi)]+\O\div\nab^i(\b-\f12\nab_4\phi\nab\phi)=G_{3,i}
\end{equation}
with $G_{3,i}$ satisfying
\begin{equation*}
\begin{split}
G_{3,i}=&\q\nab^{i+1}\p+\p\nab^{i+1}\etb+\sum_{i_1+i_2+i_3+i_4=i}\nab^{i_1}\q^{i_2}\nab^{i_3}\p\nab^{i_4}(\K,\sigmac)\\
&+\sum_{\substack{i_1+i_2+i_3=i+1\\i_1, i_3\leq i}}\nab^{i_1}\q^{i_2+1}\nab^{i_3}\p+\f{1}{|u|^2}\sum_{i_1+i_2+i_3=i}\nab^{i_1}\q^{i_2}\nab^{i_3}\p+\nab^{i+1}(\O e_4\phi)\nab\phi.
\end{split}
\end{equation*}
Note that by the structure of the renormalized null Bianchi equation, we will have $-\nab_1 R_{42}+\nab_2 R_{41}$, which enables us to avoid the troublesome term $\nab^{i+2}\phi e_4\phi$ in above two equations.

We now employ Proposition \ref{intbyparts34} for $\phi=u^{i+1}\nab^i[\O(\K-\f14\nab^A\phi\nab_A\phi)]$ and $\phi=u^{i+1}\nab^i[\O \sigmac]$, and apply Proposition \ref{intbyparts3} for $\phi=\nab^i[\O(\b-\f12\nab_4\phi\nab\phi)]$ with $\lambda_1=i+1$. We then add these identities together and obtain
\begin{equation}\label{main.ee.2}
\begin{split}
&\sum_{i\leq 4}\l\|u^{i+1}\nab^{i}({\color{black}\O}\b)\|_{L^2_{\ub}L^2(\S)}+\|{\color{black}\O}u^{i+1}\nab^{i}(\K)\|_{L^2(\Hb)}+\|{\color{black}\O}u^{i+1}\nab^{i}\sigmac\|_{L^2_uL^2(\S)}\r\\
\ls&\sum_{i\leq 4}\|u^{i+1}\nab^{i}({\color{black}\O}\b)\|_{L^2_{\ub}L^2(S_{0,\ub})}+\sum_{i\leq4}\l\|{\color{black}\O}u^{i+1}G_{1,i}\|_{L^1_{u}L^2_{\ub}L^2(S_{u,\ub})}+\|{\color{black}\O}u^{i+1}G_{2,i}\|_{L^1_{\ub}L^2_{u}L^2(S_{u,\ub})}\r\\
&+\sum_{i\leq 4}\|{\color{black}\O}u^{i+1}G_{3,i}\|_{L^1_{\ub}L^2_{u}L^2(S_{u,\ub})}.
\end{split}
\end{equation}

Note that the highest derivatives of the scalar field in in $G_{1,i}$ satisfies 
\begin{equation*}
\begin{split}
\|\O u^{i+1}\nab^{i+1}\nab\phi \nab\phi&\|_{L^1_{u}L^2_{\ub}L^2(S_{u,\ub})}\leq\|u^{i+1}\nab^{i+1}\nab\phi\|_{L^{\infty}_{u}L^2_{\ub}L^2(S_{u,\ub})} \|\O\nab\phi\|_{L^{1}_{u}L^{\infty}_{\ub}L^{\infty}(S_{u,\ub})}\\
\leq&\ub^{\f12}\at\cdot \F\M S\cdot \f{\ub\at O}{|u|}\leq \ub^{\f12}\at.
\end{split}
\end{equation*}

The corresponding terms in $G_{2,i}$ and $G_{3,i}$ obey
\begin{equation*} 
\begin{split}
\|\O u^{i+1}\nab^{i+1}(\O e_4\phi) \nab\phi&\|_{L^1_{\ub}L^2_{u}L^2(S_{u,\ub})}\leq\|u^{i+1}\nab^{i+1}(\O e_4\phi)\|_{L^{\infty}_{u}L^2_{\ub}L^2(S_{u,\ub})} \|\O\nab\phi\|_{L^{2}_{u}L^2_{\ub}L^{\infty}(S_{u,\ub})}\\
\leq&\ub^{\f12}\at \M S \cdot \f{\ub\at}{|u|^2}|u|^{\f12}\ub^{\f12}O=\ub^{\f12}\at \f{\ub^{\f32}\at}{|u|^{\f32}}\M S O\leq \ub^{\f12}\at. 
\end{split}
\end{equation*} 

We then proceed to bound other terms. Among the terms in ${\color{black}\O u^{i+1}}G_{1,i}$, we first control the contributions from the fifth derivatives the Ricci coefficients. These terms are
\begin{equation}\label{terms.5.ee.2}
\q\nab^5(\O\tr\chi, \O\chih),\quad \f{1}{|u|}\nab^5(\O\trch),\quad \p\nab^5(\O\chibh,\O\trchb)
\end{equation}
and in $\|u^5\cdot\|_{L^1_u L^2_{\ub}L^2(\S)}$ norm they satisfy
\begin{equation*}
\begin{split}
\|u^5\q\nab^5(\O\tr\chi, \O\chih)\|_{L^1_{u}L^2_{\ub}L^2(S_{u,\ub})}\ls& \|u^5\nab^5(\O\tr\chi, \O\chih)\|_{L^{\i}_{u}L^2_{\ub}L^2(S_{u,\ub})}\|\q\|_{L^1_{u}L^{\i}_{\ub}L^{\i}(\S)}\\
\ls& \ub^{\f12}a^{\f12}\t O\cdot\f{\ub a^{\f12}O}{|u|}\leq\ub^{\f12}a^{\f12},
\end{split}
\end{equation*}
\begin{equation*}
\begin{split}
\|u^4\nab^5(\O\trch)\|_{L^1_{u}L^2_{\ub}L^2(S_{u,\ub})}\ls& \|u^6\nab^5(\O\trch)\|_{L^{\i}_{u}L^2_{\ub}L^2(S_{u,\ub})}\|\f{1}{|u|^2}\|_{L^1_{u}}\ls \f{\ub^{\f32}a\t O}{|u|}\leq \ub^{\f12}\at.
\end{split}
\end{equation*}
Here we use the improved bound for $u^6\nab^5(\O\trch)$ in Proposition \ref{chi.5.bd}, which is better than other $\nab^5(\O\p)$ components. For the terms containing $(\O\chibh, \O\tr\chib)$, employing the bootstrap assumption \eqref{BA.3}, we get 
\begin{equation*} 
\begin{split}
\|u^5 \p\nab^5(\O\chibh,\O\tr\chib)&\|_{L^1_{u}L^2_{\ub}L^2(S_{u,\ub})}\ls\|u^5\nab^5(\O\chibh,\O\tr\chib)\|_{L^{\i}_{\ub}L^2_{u}L^2(S_{u,\ub})}\|\p\|_{L^2_{u}L^2_{\ub}
L^{\i}(S_{u,\ub})}\\
\ls&\ub^{\f12}\af\F \t O\cdot\f{\ub^{\f12}\at O}{|u|^{\f12}}
\leq\ub^{\f12}\at \f{\ub^{\f12}\af}{|u|^{\f12}}\F \t O O\leq \ub^{\f12}\at.
\end{split}
\end{equation*}

For $\|\O u^{i+1}G_{1,i}\|_{L^1_u L^2_{\ub} L^2(\S)}$, the other terms obey
$$\|\O u^{i+1}\sum_{i_1+i_2+i_3=i}\nab^{i_1}\q^{i_2+1}\nab^{i_3}(K-\f{1}{|u|^2}, \sigmac)\|_{L^1_u L^2_{\ub} L^2(\S)}\leq \f{\ub\at O}{|u|^2}|u|^{\f12}\ub^{\f12}\cdot \ub^{\f12}\at\leq \ub^{\f12}\at,$$
$$\|\O u^{i+1}\sum_{\substack{i_1+i_2+i_3=i+1\\i_1, i_3\leq i}}\nab^{i_1}\q^{i_2+1}\nab^{i_3}\p\|_{L^1_u L^2_{\ub} L^2(\S)}\leq \f{\ub\at O}{|u|^2}\cdot\f{\at O}{|u|}\cdot|u|^2\ub^{\f12}\leq \ub^{\f12}\at,$$
$$\|u^{i}\sum_{i_1+i_2+i_3=i}\nab^{i_1}\q^{i_2+1}\nab^{i_3}\p\|_{L^1_u L^2_{\ub} L^2(\S)}\leq \f{\ub\at O}{|u|^2}\cdot\f{\at O}{|u|}\cdot|u|^2\ub^{\f12}\leq \ub^{\f12}\at,$$
$$\|\O u^{i+1}\sum_{\substack{i_1+i_2=i+1\\i_1\leq i}}\nab^{i_1}\q^{i_2+2}\|_{L^1_u L^2_{\ub} L^2(\S)}\leq \f{\ub\at\cdot\ub\at O^2}{|u|^4}\cdot|u|^2\ub^{\f12}\leq \ub^{\f12}\at,$$
$$\|\O u^{i-1}\sum_{i_1+i_2=i}\nab^{i_1}\q^{i_2+1}\|_{L^1_u L^2_{\ub} L^2(\S)}\leq\f{\ub\at O}{|u|^3}\cdot|u|^2\ub^{\f12}\leq \ub^{\f12}\at.$$

We then proceed to derive the estimates for $G_{2,i}$ and $G_{3,i}$. Since all the terms of $G_{2,i}$ are contained in the expression for $G_{3,i}$, hence we only need to bound the terms in $G_{3,i}$.  The terms containing the $5$ derivatives of the Ricci coefficients are
$$\O u^5\q\nab^5\p,\quad \O u^5\p\nab^5\etab,$$
and they obey
\begin{equation*}
\begin{split}
\|\O u^5\q\nab^5\p\|_{L^1_{\ub}L^2_{u}L^2(S_{u,\ub})}\ls& \|u^5\nab^5\p\|_{L^{\i}_{u}L^2_{\ub}L^2(S_{u,\ub})}\|\O\q\|_{L^{2}_{\ub}L^2_{u}L^{\i}(S_{u,\ub})}\ls \ub^{\f12}a^{\f12}\t O\f{\ub^{\f32}a^{\f12}O}{|u|^{\f32}}\leq \ub^{\f12}\at,
\end{split}
\end{equation*}

\begin{equation*}
\begin{split}
\|\O u^5\p\nab^5\etb\|_{L^1_{\ub}L^2_{u}L^2(S_{u,\ub})}\ls&\|u^{{\color{black}5}}\nab^5\etb\|_{L^{\i}_{u}L^2_{\ub}L^2(S_{u,\ub})}\|\O{\color{black}\p}\|_{L^2_{u}L^2_{\ub}L^{\i}(S_{u,\ub})}\\
\ls&\ub^{\f12}\af\F\t O\cdot\f{\ub^{\f12}\at}{|u|^{\f12}}\leq \ub^{\f12}\at\f{\ub^{\f12}\af}{|u|^{\f12}}\F\t O\leq \ub^{\f12}\at.
\end{split}
\end{equation*}
\noindent The rest terms in $\|\O u^{i+1}(G_{2,i}, G_{3,i})\|_{L^1_{\ub}L^2_u L^2(\S)}$ are lower order and they satisfy
$$\|\O u^{i+1}\sum_{i_1+i_2+i_3+i_4=i}\nab^{i_1}\q^{i_2}\nab^{i_3}\p \nab^{i_4}(K-\f{1}{|u|^2}, \sigmac)\|_{L^1_{\ub} L^2_{u} L^2(\S)}\leq \f{\ub\at O}{|u|}\cdot\ub^{\f12}\at\leq \ub^{\f12}\at,$$
$$\|\O u^{i+1}\sum_{\substack{i_1+i_2+i_3=i+1\\i_1, i_3\leq i}}\nab^{i_1}\q^{i_2+1}\nab^{i_3}\p\|_{L^1_{\ub} L^2_{u} L^2(\S)}\leq \f{\ub\at O}{|u|^2}\cdot\f{\at O}{|u|}\cdot|u|^{\f32}\ub\leq \ub^{\f12}\at,$$
$$\|\O u^{i-1}\sum_{i_1+i_2+i_3=i}\nab^{i_1}\q^{i_2}\nab^{i_3}\p\|_{L^1_{\ub} L^2_{u} L^2(\S)}\leq \f{\at}{|u|^2}|u|^{\f32}\ub O\leq \ub^{\f12}\at.$$ 
Back to \eqref{main.ee.2}, we hence prove this proposition.\\
\end{proof}

{\color{black} 
We then utilize the other paired equations and proceed to prove
\begin{proposition} \label{EE.2}
Under the assumptions of Theorem \ref{main thm} and the bootstrap assumptions \eqref{BA.1}, \eqref{BA.2}, \eqref{BA.3}, \eqref{BA.4}, it holds
\begin{equation*}
\begin{split}
&\sum_{1\leq i\leq 4}\bigg(\|u^{i+2}\nab^i(\K-\f14\nab^A\phi\nab_A\phi,\sigmac)\|_{L^{\i}_uL^{2}_{\ub}L^2(\S)}+\|u^{i+2}\nab^i(\beb+\f12\nab_3\phi\nab\phi)\|_{L^{\i}_{\ub}L^{2}_{u}L^2(\S)}\bigg)\\
\ls& \ub^{\f32} a^{\f34}\bigg(1+\f{\ub^{\f12}\af}{|u|^{\f12}\O}\bigg).
\end{split}
\end{equation*}
By further employing the $L^2(\S)$ estimates proved in Section \ref{L2 scalar field} for $\phi$, it also holds
\begin{equation*}
\begin{split}
&\sum_{1\leq i\leq 4}\bigg(\|u^{i+2}\nab^i(\K, \sigmac)\|_{L^{\i}_uL^{2}_{\ub}L^2(\S)}+\|u^{i+2}\nab^i\beb\|_{L^{\i}_{\ub}L^{2}_{u}L^2(\S)}\bigg)\ls\ub^{\f32} a^{\f34}\bigg(1+\f{\ub^{\f12}\af}{|u|^{\f12}\O}\bigg). 
\end{split}
\end{equation*}

\end{proposition} 

}

\begin{proof}
We first commute the renormalized null Bianchi equations with $i$ angular derivatives and get
\begin{equation}\label{eqn.sigma}
\nab_3\nab^i\sigmac+\f{3+i}{2}\tr\chib\nab^i\sigmac+\div^*\nab^i(\beb+\f12\nab_3\phi\nab\phi)=F_{1,i},
\end{equation}
where $F_{1,i}$ is 
\begin{equation*}
\begin{split}
F_{1,i}=&{\color{black}\nab^{i+1}(e_3\phi)\nab\phi}+{\color{black}\O^{-1}}\sum_{i_1+i_2+i_3=i+1}\nab^{i_1}\q^{i_2+1}\nab^{i_3}\q+\f{{\color{black}\O^{-1}}}{|u|}\sum_{i_1+i_2+i_3=i-1}\nab^{i_1}\q^{i_2+1}\nab^{i_3}\sigmac,
\end{split}
\end{equation*}

\begin{equation}\label{eqn.K}
\nab_3\nab^i(\K{\color{black}-\f14\nab^A\phi \nab_A\phi})-\div \nab^i(\beb+\f12\nab_3\phi\nab\phi)+\f{3+i}{2}\tr\chib\nab^i(\K{\color{black}-\f14\nab^A\phi \nab_A\phi})=F_{2,i}
\end{equation}
with $F_{2,i}$ being
\begin{equation*}
\begin{split}
F_{2,i}=&{\color{black}\nab^{i+1}(e_3\phi)\nab\phi}+\frac{{\color{black}\O^{-1}}}{|u|}\nab^i\mu+{\color{black}\O^{-1}}\sum_{i_1+i_2+i_3=i+1}\nab^{i_1}\q^{i_2+1}\nab^{i_3}\q\\
&+{\color{black}\O^{-1}}\sum_{i_1+i_2+i_3=i}\nab^{i_1}\q^{i_2+1}\nab^{i_3}(\K{\color{black}-\f14\nab^A\phi\nab_A\phi})\\
&+\frac{{\color{black}\O^{-1}}}{|u|}\sum_{i_1+i_2+i_3=i}\nab^{i_1}\q^{i_2+1}\nab^{i_3}\q+\frac{{\color{black}\O^{-1}}}{|u|^2}\sum_{i_1+i_2+i_3=i}\nab^{i_1}\q^{i_2}\nab^{i_3}({\color{black}\O}\trchb-(\O\tr\chib)_0)\\
&+{\color{black}\O^{-1}}\sum_{i_1+i_2+i_3=i-1}\frac 1{|u|}\nab^{i_1}\q^{i_2+1}\nab^{i_3}(\K{\color{black}-\f14\nab^A\phi\nab_A\phi}),
\end{split}
\end{equation*} 

\begin{equation}\label{eqn.betab} 
\nab_4\nab^i(\beb+\f12 \nab_3\phi\nab\phi)-\div \nab^{i}(\K-\f14\nab^A\phi\nab_A\phi)-^*\nab\nab^i\sigmac=F_{3,i},
\end{equation}
where $F_{3,i}$ satisfies
\begin{equation}\label{F3,i}
\begin{split}
F_{3,i}=&\O^{-2}(\O\tr\chib)\nab^{i+1}(\O\tr\chi)+\O^{-2}\cdot\O\chibh\nab^{i+1}(\O\chih,\O\tr\chi)+\O^{-2}\cdot\O\chih\nab^{i+1}(\O\chibh)\\
&+\O^{-2}(\O\tr\chi)\nab^{i+1}(\O\tr\chib)+\O^{-2}\cdot\O\h\o\nab^{i+1}(\O\chibh, 
\O\tr\chib-(\O\tr\chib)_0)+{\color{black}\nab^{i+1}(\nab\phi)\nab\phi}\\
&+{\color{black}\O^{-2}}\sum_{\substack{i_1+i_2+i_3+i_4=i+1\\i_3, i_4\leq i}}\nab^{i_1}\q^{i_2}\nab^{i_3}{\color{black}(\O\chibh, \O\tr\chib-(\O\tr\chib)_0)}\nab^{i_4}(\O\tr\chi, \O\chih, \O\h\o)\\
&+\sum_{i_1+i_2+i_3=i}\nab^{i_1}\q^{i_2+1}\nab^{i_3}(K,\sigmac)
+\frac{{\color{black}\O^{-2}}}{|u|}\sum_{i_1+i_2+i_3=i}\nab^{i_1}\q^{i_2+1}\nab^{i_3}\p\\
&{\color{black}+\sum_{i_1+i_2=i-1}\nab^{i_1+1}\q\nab^{i_2+1}\q+\O^{-2}\sum_{i_1+i_2+i_3=i}\nab^{i_1}\q \nab^{i_2}\q \nab^{i_3}\p}.
\end{split}
\end{equation}

We then applying Proposition \ref{intbyparts34} for $\phi=u^{i+2}\nab^i(\beb+\f12\nab_3\phi\nab\phi)$ and Proposition \ref{intbyparts3} for $\phi=\nab^i(K-\f{1}{|u|^2}-\f14\nab^A\phi\nab_A\phi)$ and $\phi=\nab^i\sigmac$ with $\lambda_1=i+2$ and add these identities together. We then obtain
\begin{equation}\label{main.ee.0}
\begin{split}
&\|u^{i+2}\nab^{i}(\K{\color{black}-\f14\nab^A\phi\nab_A\phi},\sigmac)\|_{L^2_{\ub}L^2(\S)}^2+\|u^{i+2}\nab^{i}(\beb{\color{black}+\f12\nab_3\phi\nab\phi})\|_{L^2_u L^2(\S)}^2\\
\ls &\|\nab^{i}(\K{\color{black}-\f14\nab^A\phi\nab_A\phi},\sigmac)\|_{L^2_{\ub}L^2(S_{0,\ub})}^2+\|u^{2i+4} \nab^i\sigmac {\color{black}\cdot\O} F_{1,i}\|_{L^1_uL^1_{\ub}L^1(\S)}^2\\
&+\|u^{2i+4}\nab^i (\K{\color{black}-\f14\nab^A\phi\nab_A\phi}) {\color{black}\cdot\O}F_{2,i}\|_{L^1_uL^1_{\ub}L^1(\S)}+\|u^{2i+4}\nab^i(\beb{\color{black}+\f12\nab_3\phi\nab\phi}) {\color{black}\cdot\O}F_{3,i}\|_{L^1_{\ub}L^1_uL^1(\S)}.
\end{split}
\end{equation}

We proceed to control terms on the right. For the term involving $F_{3,i}$ we have
\begin{equation*} 
\begin{split}
&\|u^{2i+4}\nab^i(\beb{\color{black}+\f12\nab_3\phi\nab\phi}) {\color{black}\cdot\O}F_{3,i}\|_{L^1_{\ub}L^1_uL^1(\S)}\\
\leq&\|u^{i+2}\nab^i(\beb{\color{black}+\f12\nab_3\phi\nab\phi})\|_{L^{\infty}_{\ub}L^2_uL^2(\S)} \|{\color{black}\O} u^{i+2}F_{3,i}\|_{L^1_{\ub}L^2_uL^2(\S)}.
\end{split}
\end{equation*} 

\noindent In $\|{\color{black}\O} u^{i+2}F_{3,i}\|_{L^1_{\ub}L^2_uL^2(\S)}$, we note that the borderline terms are from the first two lines of \eqref{F3,i}. These borderline terms obey
\begin{equation*}
\begin{split}
\|u^6\nab^5(\O\tr\chi)\tr\chib\|_{L^1_{\ub}L^2_u L^2(\S)}\ls& \|\O^{-1}u^5\nab^5(\O\tr\chi)\|_{L^1_{\ub}L^2_uL^2(\S)}\\
\leq&\|u^6\nab^5(\O\tr\chi)\|_{L^{\i}_u L^2_{\ub}L^2(\S)}\cdot\|\f{1}{|u|\O}\|_{L^2_u L^2_{\ub} L^2(\S)}\\
\leq&\ub^{\f32}a\t O(\tr\chi)\cdot\f{\ub^{\f12}}{|u|^{\f12}\O}\leq \ub^{\f32}a^{\f34}\cdot\f{\ub^{\f12}\af}{|u|^{\f12}\O}\M R(\b)\ls\ub^{\f32}a^{\f34}\cdot\f{\ub^{\f12}\af}{|u|^{\f12}\O},
\end{split}
\end{equation*}

\begin{equation*}
\begin{split}
&\|u^6 \nab^5(\tr\chib, \chibh) (\O\tr\chi, \O\chih, \O\h\o)\|_{L^1_{\ub}L^2_uL^2(\S)}+\|u^6 \nab^5(\O\chih, \O\tr\chi)\chibh\|_{L^1_{\ub}L^2_uL^2(\S)}\\
\leq&\|u^5\nab^5(\tr\chib, \chibh)\|_{L^{\infty}_{\ub}L^2_{u}L^2(\S)}\cdot\|u\p\|_{L^1_{\ub}L^{\infty}_u L^{\i}(\S)}\\
&+\|u^5\nab^5(\O\chih, \O\tr\chi)\|_{L^{\infty}_{u}L^2_{\ub}L^{2}(\S)}\cdot\|u\O^{-1}(\O\chibh)\|_{L^2_{\ub}L^{2}_u L^{\infty}(\S)}\\
\leq&\ub^{\f12}\af(1+\f{\ub^{\f12}\af}{|u|^{\f12}\O})\M R(\b)\cdot\ub\at+\ub^{\f12}\at\M R(\b)\cdot\f{\ub\at}{|u|\O}\ub^{\f12}|u|^{\f12}\\
\leq&\ub^{\f32}a^{\f34}\bigg(1+\f{\ub^{\f12}\af}{|u|^{\f12}\O} \bigg)\M R(\b)\ls \ub^{\f32}a^{\f34}\bigg(1+\f{\ub^{\f12}\af}{|u|^{\f12}\O} \bigg).
\end{split}
\end{equation*}

Here we employ Proposition \eqref{chi.5.bd}, Proposition \eqref{nabla5chib} and Proposition \eqref{EE.1}. And to control $\|\O\chibh\|_{L^{\i}(\S)}$, we utilize
\begin{equation*}
\begin{split}
(\O\nab_4)(\O\chibh)_{AB}+\f12\O\tr\chi(\O\chibh)_{AB}=&\O^2(\nab\widehat{\otimes}\etb)_{AB}-\f12\O\tr\chib(\O\chih)_{AB}+\O^2(\etb\widehat{\otimes}\etb)_{AB}\\
&+\O^2\nab_A\phi\nab_B\phi-\f12\O^2g_{AB}\nab^C\phi\nab_C\phi.
\end{split}
\end{equation*}
Applying Proposition \eqref{chih.bd}, it holds
$$\|\O\chibh\|_{L^{\infty}(\S)}\leq \f{\ub}{|u|}\|\O\chih\|_{L^{\i}(\S)}+\f{\ub}{|u|}\cdot\f{\ub\at(O+O^2)}{|u|^2}\leq\f{\ub\at}{|u|}.$$

Similarly, we bound the rest terms in $\|\O u^{i+2}F_{3,i}\|_{L^1_{\ub}L^2_u L^2(\S)}$ and they obey 
\begin{equation*}
\begin{split}
\|\O u^{i+2}&\nab^{i+1}\nab\phi \nab\phi\|_{L^1_{\ub}L^2_uL^2(\S)}\leq\|u^{i+1}\nab^{i+1}\nab\phi\|_{L^{\infty}_{u}L^2_{\ub}L^2(\S)} \|u\O\nab\phi\|_{L^{2}_{u}L^2_{\ub}L^{\infty}(\S)}\\
\leq&\ub^{\f12}\af\F\M S\cdot\f{\ub\at}{|u|}|u|^{\f12}\ub^{\f12}O=\ub^{\f32}a^{\f34}\F\f{\ub^{\f12}}{|u|^{\f12}} \M S O \leq \ub^{\f32}a^{\f34}, 
\end{split}
\end{equation*}
\begin{equation*}
\begin{split}
&\|\O^{-2}u^{i+2}\sum_{\substack{i_1+i_2+i_3+i_4=i+1\\i_3, i_4\leq i}}\nab^{i_1}\q^{i_2}\nab^{i_3}(\O\chibh, \O\tr\chib-(\O\tr\chib)_0)\nab^{i_4}(\O\tr\chi, \O\chih, \O\h\o)\|_{L^1_{\ub}L^2_u L^2(\S)}\\
\leq& \ub^{\f32}a^{\f34}\F \M R(\b)\leq \ub^{\f32}a^{\f34}\F,
\end{split}
\end{equation*}

\begin{equation*}
\begin{split}
&\|\O^{-2}u^{i+1}\sum_{i_1+i_2+i_3=i}\nab^{i_1}\q^{i_2+1}\nab^{i_3}\p\|_{L^1_{\ub}L^2_u L^2(\S)}\\
\leq&\O^{-2}|u|\cdot\f{\ub\at}{|u|^2}\f{\at}{|u|}\cdot\f{\ub\at}{|u|}O^2\cdot\ub|u|^{\f32}\leq \ub^{\f32}a^{\f34}\cdot \f{\ub^{\f32}a^{\f34}O^2}{|u|^{\f32}\O^2}\leq \ub^{\f32}a^{\f34},
\end{split}
\end{equation*}
\begin{equation*}
\begin{split}
&\|u^{i+2}\sum_{i_1+i_2+i_3=i}\nab^{i_1}\q^{i_2+1}\nab^{i_3}(K, \sigmac) \|_{L^1_{\ub}L^2_u L^2(\S)}\\
\leq& \sum_{i_3\leq 4}\|\O u^{i_3+1}\nab^{i_3}(K-\f{1}{|u|^2}, \sigmac)\|_{L^{\i}_{\ub}L^2_u L^2(\S)}\cdot \sum_{i_2\leq 2}\|\O^{-1}u^{i_1+i_2+1}\nab^{i_1}\q^{i_2+1}\|_{L^1_{\ub}L^{\i}_u L^{\i}(\S)}\\
&+\sum_{i_3\leq 2}\|u^{i_3+2}\nab^{i_3}(K, \sigmac)\|_{L^{\i}_{\ub}L^{\i}_u L^{\i}(\S)}\cdot\sum_{i_1\leq 4}\|u^{i_1+i_2}\nab^{i_1}\q^{i_2+1}\|_{L^1_{\ub}L^2_u L^2(\S)}\\
\leq&\ub^{\f12}\at \M R\cdot\f{\ub\at O}{|u|\O}\cdot\ub+\f{\ub\at O}{|u|^2}\cdot\ub|u|^{\f32}\leq \ub^{\f32}a^{\f34}\cdot\f{\ub\af \M R O}{|u|\O}+\ub^{\f32}\at\cdot\f{\ub^{\f12}O}{|u|^{\f12}}\leq \ub^{\f32}a^{\f34},
\end{split}
\end{equation*}
\begin{equation*}
\begin{split}
\|\O u^{i+2} \sum_{i_1+i_2=i-1}\nab^{i_1+1}\q\nab^{i_2+1}\q  \|_{L^{1}_{\ub}L^2_{u}L^2(\S)}\leq\f{\ub^3 a O^2}{|u|^{\f32}}\leq \ub^{\f32}a^{\f34},
\end{split}
\end{equation*}
\begin{equation*}
\begin{split}
\|\O^{-1}u^{i+2}\cdot \sum_{i_1+i_2+i_3=i}\nab^{i_1}\q\nab^{i_2}\q\nab^{i_3}\p  \|_{L^{1}_{\ub}L^2_{u}L^2(\S)}\leq\ub^{\f32}a^{\f34}\cdot\f{\ub^{\f32}a^{\f34}O^3}{|u|^{\f32}\O}\leq \ub^{\f32}a^{\f34}.
\end{split}
\end{equation*}

In the similar manner, we also control terms in $F_{2,i}$ and $F_{1,i}$. In $\|\O u^{i+2}F_{2,i}\|_{L^1_u L^2_{\ub}L^2(\S)}$ and $\|\O u^{i+2}F_{1,i}\|_{L^1_u L^2_{\ub}L^2(\S)}$, the top-order derivatives appear in the next four terms. The first term satisfies
\begin{equation*}
\begin{split}
&\|u^{i+2}\nab^{i+1}(\O e_3\phi)\nab\phi\|_{L^1_uL^2_{\ub}L^2(\S)}\\
\leq&\|\O^{-1}u^{i+1}\nab^{i+1}(\O e_3\phi)\|_{L^{\infty}_{\ub}L^{2}_{u}L^{2}(\S)}\|u\O\nab\phi\|_{L^{2}_{u}L^2_{\ub}L^{\infty}(\S)}\\
\leq&\ub^{\f12}\af\F\M S\cdot\f{\ub\at}{|u|}|u|^{\f12}\ub^{\f12}\leq\ub^{\f32}a^{\f34}\F\f{\ub^{\f12}}{|u|^{\f12}}\M S\leq \ub^{\f32}a^{\f34}.
\end{split}
\end{equation*}

\noindent The second term obeys
\begin{equation*}
\begin{split}
&\|u^{i+1}\nab^i\mu\|_{L^1_u L^2_{\ub}L^2(\S)}\\
\leq&\|u^{i+\f52}\nab^i\mu\|_{L^{\i}_u L^{\i}_{\ub} L^2(\S)}\|u^{-\f32}\|_{L^1_u L^2_{\ub}L^{\i}(\S)}\leq\ub a^{\f34}\cdot\f{\ub^{\f12}\af}{\O}\cdot\f{\ub^{\f12}}{|u|^{\f12}}=\ub^{\f32}a^{\f34}\cdot\f{\ub^{\f12}\af}{|u|^{\f12}\O}
\end{split}
\end{equation*}
where Proposition \eqref{eta.5.bd} is used. For the third term, we have
\begin{equation}\label{F1.1} 
\begin{split}
&\sum_{i\leq 4}\|u^{i+2}\sum_{i_1+i_2+i_3=i+1}\nab^{i_1}\q^{i_2+1}\nab^{i_3}\q\|_{L^1_uL^2_{\ub}L^2(\S)}\\
\ls &\sum_{\substack{i_1+i_2\leq 5\\ i_1\leq 2}}\|u^{i_1+i_2+2}\nab^{i_1}\q^{i_2+1} \|_{L^\infty_u L^\infty_{\ub} L^\infty(\S)}\cdot\bigg(\ub^{\f12}\sum_{i_3\leq 4}\|u^{i_3+1}\nab^{i_3}\q \|_{L^\infty_uL^\infty_{\ub}L^2(\S)}\|u^{-2} \|_{L^1_u}\\
&+\ub^{\f12}\|u^{i+1}\nab^{i+1}(\O\trchb,\O\chibh)\|_{L^\infty_{\ub}L^2_uL^2(\S)}\|u^{-1}\|_{L^2_u}+\|u^{i+2}\nab^{i+1}{\color{black}\eta}\|_{L^\infty_u L^2_{\ub}L^2(\S)}\|u^{-2}\|_{L^1_u}\bigg)\\
\ls &\ub\at O\cdot\bigg(\frac{\ub^{\f32}\at O}{|u|}+\f{\ub^{\f12}}{|u|^{\f12}}\ub^{\f12}\af\F\bigg)+\ub\at O\cdot\f{\ub^{\f32}a^{\f34}\t O}{|u|}\leq \ub^{\f32}a^{\f34},
\end{split}
\end{equation}
where for the first inequality of \eqref{F1.1} we use the top-order terms are of form $\q\nab^{i+1}(\O\chibh, \O\tr\chib)$ and $\q\nab^{i+1}\eta$. And for the fourth term, it holds
\begin{equation*}
\begin{split}
&\|\O^{-1}u^{i+2}\sum_{i_1+i_2+i_3=i}\nab^{i_1}\q^{i_2+1}\nab^{i_3}(\Kt)\|_{L^1_u L^2_{\ub}L^2(\S)}\\
\leq&\f{\ub\at O}{|u|\O}\cdot\ub^{\f32}a^{\f34}\cdot\F \M R\leq \ub^{\f32}a^{\f34}.
\end{split}
\end{equation*}

\bigskip

For the rest terms in $\|\O u^{i+2}F_{2,i}\|_{L^1_u L^2_{\ub}L^2(\S)}$ and $\|\O u^{i+2}F_{1,i}\|_{L^1_u L^2_{\ub}L^2(\S)}$, we have
\begin{equation*}
\begin{split}
\|u^{i+1}\sum_{i_1+i_2+i_3=i}\nab^{i_1}\q^{i_2+1}\nab^{i_3}\q\|_{L^1_u L^2_{\ub}L^2(\S)}\leq |u|\cdot\f{\ub\at}{|u|^2}\cdot\f{\ub\at O^2}{|u|^2}|u|^2\ub^{\f12}=\ub^{\f32}a^{\f34}\cdot\f{\ub\af O^2}{|u|}\leq \ub^{\f32}a^{\f34},
\end{split}
\end{equation*}
\begin{equation*}
\begin{split}
&\|\O^{-1}u^{i+1}\sum_{i_1+i_2+i_3=i-1}\nab^{i_1}\q^{i_2}\nab^{i_3}(\Kt, \sigmac)\|_{L^1_u L^2_{\ub}L^2(\S)}\\
\leq&\O^{-1}|u|^2\cdot\f{\ub\at}{|u|^2}\cdot\f{\ub\at O^2}{|u|^3}|u|^2\ub^{\f12}\leq \ub^{\f32}a^{\f34}\cdot\f{\ub\af}{|u|\O}\leq \ub^{\f32}a^{\f34},
\end{split}
\end{equation*}
\begin{equation*}
\begin{split}
&\|u^{i}\sum_{i_1+i_2+i_3=i}\nab^{i_1}\q^{i_2}\nab^{i_3}(\O\tr\chib-(\O\tr\chib)_0)\|_{L^1_u L^2_{\ub}L^2(\S)}\\
\leq&\|u^{i+\f32}\sum_{i_1+i_2+i_3=i}\nab^{i_1}\q^{i_2}\nab^{i_3}(\O\tr\chib-(\O\tr\chib)_0)\|_{L^{\i}_u L^{\i}_{\ub} L^2(\S)}\|u^{-\f32}\|_{L^1_u L^2_{\ub}L^{\i}(\S)}\leq\ub^{\f32} a^{\f34}\cdot\f{\ub^{\f12}}{|u|^{\f12}}\leq \ub^{\f32}a^{\f34},
\end{split}
\end{equation*}
where Proposition \eqref{trchb.bd} is employed. 

Gathering the above estimates, together with Proposition \eqref{EE.1} we then obtain
\begin{equation*}
\begin{split}
&\sum_{i\leq 4}\bigg(\|u^{i+2}\nab^{i}(\K,\sigmac)\|_{L^2_{\ub}L^2(\S)}+\|u^{i+2}\nab^{i}\beb\|_{L^2_u L^2(\S)}\bigg)\\
\ls& \ub^{\f32}a^{\f34}{\color{black}\F(1+\mathcal{R}[\b])}\ls  \ub^{\f32}a^{\f34}\F.
\end{split}
\end{equation*}

\end{proof}

Combining Propositions \ref{EE.2}, \ref{EE.1}, \ref{EE.lower}, we thus conclude
\begin{equation}\label{R.final}
\mathcal R\ls 1.
\end{equation}
Substituting this bound \eqref{R.final} into Proposition \ref{O52.bd}, we furthermore acquire
\begin{equation}\label{O52.final}
\tilde{\mathcal O}_{5,2}\ls 1.
\end{equation} 
This proves the hyperbolic part of Theorem \ref{main thm}. \\

We further derive an additional estimate for $\K-\f14 \nab_A\phi \nab^A\phi$. 
\begin{proposition}\label{EE.lower}
Under the assumptions of Theorem \ref{main thm} and the bootstrap assumptions \eqref{BA.1}, \eqref{BA.2}, \eqref{BA.3},\eqref{BA.4}, it holds
\begin{equation*}
\begin{split}
\|u^{2}(\K-\f14 \nab_A\phi \nab^A\phi)\|_{L^{\i}_uL^{2}_{\ub}L^2(S)}
\ls & \ub^{\f32} a^{\f34}.
\end{split}
\end{equation*}
\end{proposition}
\begin{proof}
Applying Proposition \ref{K.bd} and the bootstrap assumption \eqref{BA.3}, we get
$$\|u^{2}(\K-\f14 \nab_A\phi \nab^A\phi)\|_{L^{\i}_uL^{2}_{\ub}L^2(S)}\ls \ub^{\f12}\cdot\ub\at=\ub^{\f32}a^{\f34}.$$
\end{proof}

The above sections finish the hyperbolic estimates. Together with a standard local existence and extension argument, we have obtained a desired spacetime existence region. We now move to the elliptic part and construct the apparent horizon within this obtained existence regime.  

\section{Subsolutions and Supersolutions for the MOTS}\label{sec-Barriers}

Along each fixed $\Hb_{\ub}$, we consider a $2$-sphere $(u, \ub, \o)$ on it with $\o\in\mathbb{S}^2$ and $u$ satisfying $u=-R(\o,\ub)$. Considering $R=R(\o,\ub)$ as a function on $\mathbb S^2$, we define 
\begin{equation}\label{eq-expression-cal-G-ub}
\begin{split}
\mathcal G(R,\ub)&:=\D_{g} R+\f12 \O \tr_g\chib |\nab_g R|^2-\f12\O^{-1} \tr_g\chi\\
&\qquad+4\O \omb |\nab_g R|^2+2\Omega\chibh(\nab_g R, \nab_g R)+2g(\eta, \nab_g R),
\end{split}
\end{equation}
where $g$ is the induced metric on the 2-sphere $M=\{u=-R\}$ along $\Hb_{\ub}$. 

When there is no confusion, we suppress $\ub$ and write $R=R(\o)$. Following \cite{AH}, along $\Hb_{\ub}$ it holds that   
$$u=-R(\o) \text{ is a} 
\begin{cases}
\text{MOTS} & \text{if $\mathcal G(R,\ub)=0$ on }\mathbb{S}^2 ,\\
\text{trapped surface} & \text{if $\mathcal G(R,\ub)>0$ on }\mathbb{S}^2,\\
\text{untrapped surface} & \text{if $\mathcal G(R,\ub)<0$ on }\mathbb{S}^2.
\end{cases}$$

\begin{figure}[h]
\begin{center}
\begin{tikzpicture}[scale=0.8]
\draw [white](-1, 0)-- node[midway, sloped, below,black]{$\Hb_{\ub}$}(-2.3, 0);
\draw [white](-0.1, -2.1)-- node[midway, sloped, below,black]{$M_{\ub}(\o)$}(0.1, -2.1);
\draw [white](-0.1, -3.8)-- node[midway, sloped, below,black]{$\o\in\mathbb{S}^2$}(0.1, -3.8);
\draw [white](1, 0.5)-- node[midway, sloped, below,black]{$$}(3, 0.5);
\draw [white](2.4, -1.5)-- node[midway, sloped, below,black]{$u=-R(\ub, \o)$}(3.9, -1.5);
\draw [white](2.2, -3.7)-- node[midway, sloped, below,black]{$u=-1$}(3.7, -3.7);
\draw (0,0) ellipse (1cm and 0.3cm);
\draw (0,-4) ellipse (2cm and 0.6cm);
\draw[thin] (-1,0)--(-2,-4);
\draw[thin] (1,0)--(2,-4);
\draw [thick] (-1.375,-1.5) to [out=10,in=150] (1.5,-2);
\draw [thick] (1.5,-2) to [out=190,in=-35] (-1.375,-1.5);
\draw[fill] (1.5,-2) circle [radius=0.05];
\end{tikzpicture}
\end{center}
\caption{}
	\label{Figure MOTS}
\end{figure} 

A remaining main goal of this article is to solve $R=R(\o, \ub)$, which satisfies $\mathcal G(R,\ub)=0$ for each $\ub$. In later sections, we will also study the regularity of $R=R(\o, \ub)$ in terms of $\ub$. 

To solve for $R=R(\o, \ub)$, we first construct the anisotropic sub- and super- solutions to $\M G(R, \ub)=0$. With $u=-R(\o, \ub)$, on $\S$ we define $\gamma$ and $\phi$ via
$$g=R^2\gamma \mbox{ and } R=(\ub a)^{1-\mu^*}e^{-\phi}.$$  

{\color{black}
\begin{remark}\label{phi def with perturbed data}
For perturbed Christodoulou's initial data prescribed along $\ub=0$, we can define the corresponding $\phi$ via letting
$$R\cdot\O(-R, 0)^2=\ub a e^{-\f{\phi}{1-\ms}}.$$ 
\end{remark}
\begin{remark}
Compared with \cite{AH} by An-Han, there $\ms=0$ and $a\gg 1$. In this paper, we have $0<\ms<1$ and we can set $a=1$. {\color{black}Two major observations in below elliptic arguments are that i) owing to Christodoulou's initial data, a potentially dangersome term is with a favored sign and hence the approach to construct the multi- and single-valley anisotropic trapped surface as in \cite{AH} and \cite{KLR} still works. ii) even though the borderline quadratic nonlinear term with gradient is no longer negligible as $1/a$ fails to be the small parameter, the method of deriving apriori estimates via Moser's iteration developed in \cite{AH} by An-Han is robust enough to be extended to this new setting. Besides these observations, a novel ingredient in the step of linearization is a new ansatz for the to-be-solved solution and the use of it plays a vital role.}
\end{remark} 
}

{\color{black}In below, we provide the details.} Note that on $\S$ it holds
$${\color{black}\Delta_g\phi=R^{-2}\Delta_{\gamma}\phi}.$$ 
And equation $\mathcal G(R,\ub)=0$ can be rewritten as
\begin{equation}\label{G R ub=0 v2}
\begin{split}
0=\Delta_{\gamma}\phi&-|\nab_{\gamma}\phi|^2-4\O\omb R|\nab_{\gamma}\phi|^2+2\gamma^{ij}\eta_i \partial_j\phi+{\color{black}2R^{-1}\O\chibh_{kl}\gamma^{ik}\gamma^{jl}\nab_i\phi\nab_j\phi}\\
&-\f{R}{2}\O\tr_g\chib|\nab_{\gamma}\phi|^2+\f{R}{2}\O^{-1}\tr_g\chi.
\end{split}
\end{equation}
We point out that the main contribution in $\O^{-1}\tr_g\chi$ is
$$\f{2}{R}\cdot(1-\ms)-\f{\ub a f(\ub,\o)}{R^2\O^2_0(R)}\cdot (1-\ms)$$
and $\O_0^2(R):=\O^2(-R, 0)=R^{\f{\ms}{1-\ms}}$. We then write $\f12 R\O^{-1}\tr_g\chi$ as 
$$\f12 R\O^{-1}\tr_g\chi=1-\ms-\f{\ub a f(\ub,\o)\cdot(1-\ms)}{2R\O^2_0(R)}+[\f12 R\O^{-1}\tr_g\chi-(1-\ms)+\f{\ub a f(\ub,\o)\cdot(1-\ms)}{2R\O^2_0(R)}].$$  

Together with the fact
\begin{equation}\label{ROmega2}
R\O_0^2(R)=R\cdot R^{\f{\ms}{1-\ms}}=R^{\f{1}{1-\ms}}=[(\ub a)^{1-\ms}e^{-\phi}]^{\f{1}{1-\ms}}=\ub a e^{-\f{\phi}{1-\ms}},
\end{equation}
we can rewrite $\M G(R,\ub)=0$ as
\begin{equation*}
\begin{split}
0=\Delta_{\gamma}\phi&+(1-\ms)-\f12 f(\ub,\o)(1-\mu^*)e^{\frac{\phi}{1-\mu^*}}-4\O\omb R|\nab_{\gamma}\phi|^2\\
&-(1+\f{R}{2}\O\tr_g\chib)|\nab_\gamma \phi|^2+2\gamma^{ij}\eta_i\partial_j \phi+{\color{black}2R^{-1}\O\chibh_{kl}\gamma^{ik}\gamma^{jl}\nab_i\phi\nab_j\phi}\\
&+[\f12 R\O^{-1}\tr_g\chi-(1-\ms)+\f12f(\ub,\o)(1-\ms)e^{\f{\phi}{1-\ms}}]. 
\end{split}
\end{equation*}
Multiplying $1/(1-\ms)$ on both sides of the equation and setting
$$\tp:=\f{\phi}{1-\mu^*},$$
we then transfer $\M G(R,\ub)=0$ into 

\begin{equation}\label{S=0}
\begin{split}
0=S(\tp,\ub)=&\Delta_{\gamma}\tp+1-\f12 f(\ub,\o)e^{\tp}-{4(1-\ms)\O\omb R}|\nab_\gamma \tp|^2\\
&-(1+\f{R}{2}\O\tr_g\chib)(1-\ms)|\nab_{\gamma}\tp|^2+2\gamma^{ij}\eta_i\partial_j\tp\\
&+{\color{black}2(1-\ms)R^{-1}\O\chibh_{kl}\gamma^{ik}\gamma^{jl}\nab_i\tp\nab_j\tp}{\color{black}+[\f{R}{2(1-\ms)} \O^{-1}\tr_g\chi-1+\f12 f(\ub,\o)e^{\tp}]}.
\end{split}
\end{equation}
And we denote
\begin{equation}\label{Ftp}
\begin{split}
F(\tp):=&-(1+\f{R}{2}\O\tr_g\chib)(1-\ms)|\nab_{\gamma}\tp|^2+2\gamma^{ij}\eta_i\partial_j\tp\\
&+{\color{black}2(1-\ms)R^{-1}\O\chibh_{kl}\gamma^{ik}\gamma^{jl}\nab_i\tp\nab_j\tp}{\color{black}+[\f{R}{2(1-\ms)} \O^{-1}\tr_g\chi-1+\f12 f(\ub,\o)e^{\tp}]}.
\end{split}
\end{equation}

Applying $g=R^2\gamma$ and the estimates in Section \ref{metric}, we have that $\gamma$ satisfies 
\begin{align}
\label{Lambda-gamma-1}
\f12\leq \sqrt{|\gamma|(u, \o)}\leq \f32, \quad &\Lambda^{-1}I\leq \gamma_{ij}(u, \o)\leq \Lambda I,\\
\label{Lambda-gamma-2}|D_{\o}\gamma_{ij}(u, \o)|+&R|\partial_{u}\gamma_{ij}(u, \o)|\leq \Lambda,\\
\label{Lambda-gamma-3} |D^2_{\o}\gamma_{ij}(u, \o)|+R|D_{\o}\partial_{u}&\gamma_{ij}(u, \o)|+R^2|\partial^2_{u}\gamma_{ij}(u, \o)|\leq \Lambda\end{align}
with $\Lambda$ being a uniform positive constant.

To proceed, we use the below lemma in \cite{AH}:  

\begin{lemma}\label{lemma-variation-Laplacians} (Lemma 3.5 in \cite{AH}) Let $\{\gamma(s)\}$ be a family of metrics and $\{\phi(s)\}$ be a family of functions, both parametrized by $s$. 
Then, 
\begin{equation}\label{eq-variation-Laplace-v1}
(\Delta_{\gamma}\phi)^{\boldsymbol\cdot}=\Delta_\gamma\dot{\phi}-\gamma^{ij}\gamma^{kl}\dot{\gamma}_{jl}\nabla^2_{ik}\phi
-\gamma^{ij}\gamma^{kl}\nabla_i\dot{\gamma}_{jl}\nabla_k\phi+\frac12\gamma^{ij}\gamma^{kl}\nabla_l\dot{\gamma}_{ij}\nabla_k\phi,
\end{equation}
where we denote $\dot{\,}=\f{d\,}{ds}$. 
\end{lemma}

As a corollary, via the same proof as stated in Lemma 3.6 of \cite{AH}, it holds
\begin{lemma}\label{lemma-difference-Laplacians}  
Set $\tp\in C^2(\mathbb S^2)$ to be an arbitrary positive function. Then, we have
\begin{equation}\label{eq-estimate-difference-Laplacians}
|\D_{\gamma}\tp-\D_{\gamma_0}\tp|\le C(1+\tp)(|\nab_{\gamma}^2\tp|+|\nab_{\gamma} \tp|+|\nab_{\gamma}\tp|^2)e^{\tp}a^{-\f13}{\color{black}\O^{\t\d}}\end{equation}
with $C$ a positive constant and $\t\delta$ a small positive constant, independently of $a$ and $\tp$. 
\end{lemma}
\begin{remark}
Note that the additional $a^{-\f13}\O^{\t\delta}$ is gained by applying $$\f{\ub\at}{|u|}=\f{\ub\at\O^2}{|u|\O^2}\lesssim\f{\ub\at\O^2}{\ub a e^{-\tp}}=e^{\tp}a^{-\f12}\O^2\ll e^{\tp}a^{-\f13}\O^{\t\delta}$$
in the proof. Here for the first inequality we use \eqref{ROmega2}. {\color{black}And for the last inequality, if setting $a=1$, the whole elliptic arguments would be developed in the interior region, hence we would employ the property $0<
\O(u,0)\leq \O(u_1, 0)\ll 1$.}
\end{remark}

Recall that, with $R(\o)=(\ub a)^{1-\ms}e^{-(1-\ms)\tp}$ it holds
$$u=-R(\o) \text{ is a} 
\begin{cases}
\text{MOTS} & \text{if $\mathcal S(\tp,\ub)=0$ on }\mathbb{S}^2 ,\\
\text{trapped surface} & \text{if $\mathcal S(\tp,\ub)<0$ on }\mathbb{S}^2,\\
\text{untrapped surface} & \text{if $\mathcal S(\tp,\ub)>0$ on }\mathbb{S}^2.
\end{cases}$$
Back to \eqref{S=0}, to construct the subsolution (untrapped surface), we just set $\underline{\tp}(\o)=\ln(\f32)$. Note $0\leq f\leq 1$. Using \eqref{Ftp} and hyperbolic estimates, we have
$$|F(\tp)|\leq Ce^{\tp}a^{-\f13}{\color{black}\O^{\t\d}}(|\nab_{\gamma}\tp|^2+1).$$ 

{\color{black}For $\O$ being sufficiently small in the interior}, it holds
\begin{equation}\label{subsolution}
S(\ln(\f32), \ub)=1-\f{f}{2}\f32+F(\ln(\f32))\geq 1-\f34+F(\ln(\f32))>0.
\end{equation} 
Hence, the $2$-sphere $\S$ with $u=-(\ub a)^{1-\ms}e^{-(1-\ms)\tp}=(\f23 \ub a)^{1-\ms}$ is an untrapped surface.\\

To construct the supersolution (trapped surface), we need $S(\bar{\t \phi}, \ub)<0$. Here we apply a key property from the hyperbolic parts:
$$\O\omb(u,\ub)\geq \f12\O\omb(\ub,0)=\f{\ms}{8(1-\ms)}\cdot\f{1}{r(u,0)}>0.\footnote{\label{initial omb perturbed data}{\color{black}Note that the key property $\O\omb(\ub, 0)>0$ also holds for perturbed initial data.}}$$
Thus, back to \eqref{S=0} to find the supersolution $\bar{\t \phi}$, it suffices to require 

\begin{equation*} 
\begin{split}
\Delta_{\gamma}\bar{\t \phi}&+1-\f12 f(\ub,\o)e^{\bar{\t \phi}}-(1+\f{R}{2}\O\tr_g\chib)(1-\ms)|\nab_{\gamma}\bar{\t\phi}|^2\\
&+2\gamma^{ij}\eta_i \partial_j\bar{\t\phi}+{\color{black}2(1-\ms)R^{-1}\O\chibh_{kl}\gamma^{ik}\gamma^{jl}\nab_i{\bar{\t \phi}}\nab_j{\bar{\t \phi}}}{\color{black}+[\f{R}{2(1-\ms)} \O^{-1}\tr_g\chi-1+\f12 f(\ub,\o)e^{\bar{\t\phi}}]}<0.
\end{split}
\end{equation*}
Notice that the one-valley supersolution $\bar{\t \phi}$ given in Lemma 3.7 of \cite{AH} and the multi-valley supersolution $\bar{\t \phi}$ constructed in Theorem 7.4 of \cite{AH} fulfill this requirement. Since $-4(1-\ms)\O\omb R|\nab_{\gamma}\bar{\t\phi}|^2\leq 0$, with these $\bar{\t \phi}$ we also have $S(\bar{\t \phi}, \ub)<0$. As listed in Lemma 3.7 of \cite{AH}, for any $\tau>2.2$, the constructed $\bar{\t \phi}$ obeys
\begin{equation}\label{bartildephimax}
\bar{\t \phi}_{min}=-\log(m\epsilon^2)+C(\tau) \quad \mbox{ and } \quad \bar{\t \phi}_{max}=-\log(m\epsilon^{\tau+2})+C(\tau).
\end{equation}
Adopting this $\bar{\t \phi}$, we have that the $2$-sphere $\S$ with $u=-(\ub a)^{1-\ms}e^{-(1-\ms)\bar{\t \phi}}$ is a trapped surface and it holds
$$-(\ub a \cdot m \epsilon^2)^{1-\ms}\leq u \leq -(\ub a\cdot m\epsilon^{\tau+2})^{1-\ms}.$$

In below, for equation \eqref{S=0} we seek the solution $\tp$ bounded between $\underline{\tilde{\phi}}=\ln(\f32)$ and $\bar{\tilde{\phi}}$ constructed above as in \cite{AH}. We require that $\tp$ satisfies
\begin{equation}
0<\ln(\f32)\leq \tp\leq \kappa
\end{equation}
with $\kappa$ being the maximum of the supersolution $\bar{\tilde{\phi}}$. Here $\kappa=\kappa(m, \epsilon)$ depends on constants $m$ and $\epsilon$ stated in Theorem \ref{main thm}. We further define constant $K$ via 
\begin{equation}\label{constant K}
K:=\max\{C(\kappa), |\bar{\tp}|_{C^2(\mathbb S^2)}\}
\end{equation}
with 
\begin{equation}\label{C kappa} 
C(\kappa)=(Ce^{\kappa})^{Ce^{4\Lambda\kappa e^{\kappa}}} 
\end{equation}
being explicitly given in Theorem \ref{theorem-Schauder-phi}. 

Throughout this paper, we will use parameters $a$, function $\O$ and constant $\kappa$, which satisfy
\begin{equation}\label{eq-choice-a-b}
K=\max\{(Ce^{\kappa})^{Ce^{4\Lambda\kappa}}, \|\bar{\tilde{\phi}}\|_{C^2}\}\leq a^{\f14}\cdot\O^{-\t\d/2}\ll a^{\f13}\cdot\O^{-\t\d}.
\end{equation}
Here, $C$ is a uniform positive constant defined in Theorem \ref{theorem-Schauder-phi}, and $\Lambda$ is introduced in \eqref{Lambda-gamma-1}. {\color{black}Note that for the interior region}, we have $a=1$ and $0<\O\lesssim\O(u,0)\leq \O(u_1, 0) \ll 1$. When $\kappa=\kappa(m,\epsilon)$ is fixed, in both regions we choose either parameter $a$ being sufficiently large or $0<\O\leq\O(u,0)$ being sufficiently small, such that \eqref{eq-choice-a-b} holds.

{\color{black}
\begin{remark}
We can set $a=1$ and the above arguments hold. Moreover, the hypersurface $u=-(\ub a\cdot m\epsilon^{\tau+2})^{1-\ms}$ locates within the hyperbolic existence region. 
\end{remark}
}

\section{Regularity Estimates via Moser's Iteration}\label{sec-Schauder-estimates}

{\color{black}
With $u=-R$ and $R=(\ub a)^{1-\ms}e^{-(1-\ms)\tp}$, via \eqref{S=0}, $S_{-R,\ub}$ being a MOTS requires
\begin{equation}\label{MOTS eqn}
\begin{split}
0=&\Delta_{\gamma}\tp+1-\f12 f(\ub,\o)e^{\tp}-4(1-\ms)R\O\omb|\nab_{\gamma}\tp|^2-(1+\f{R}{2}\O\tr_g\chib)(1-\ms)|\nab_{\gamma}\tp|^2+2\gamma^{ij}\eta_i\partial_j\tp\\
&+[\f{1}{2(1-\ms)}R\O^{-1}\tr_g\chi-1+\f12 f(\ub,\o)e^{\tp}]+{\color{black}2(1-\ms)R^{-1}\O\chibh_{kl}\gamma^{ik}\gamma^{jl}\nab_i\tp\nab_j\tp}.
\end{split}
\end{equation}
And we denote
\begin{align}
\label{F1} F_1=&-4(1-\ms)R\O\omb|\nab_{\gamma}\tp|^2,\\
\label{F2}F_2=&-(1+\f{R}{2}\O\tr_g\chib)(1-\ms)|\nab_{\gamma}\tp|^2+2\gamma^{ij}\eta_i\partial_j\tp\\
\nonumber &+[\f{1}{2(1-\ms)}R\O^{-1}\tr_g\chi-1+\f12 f(\ub,\o)e^{\tp}]+{\color{black}2(1-\ms)R^{-1}\O\chibh_{kl}\gamma^{ik}\gamma^{jl}\nab_i\tp\nab_j\tp}.
\end{align}

}

Our aim in this section is to prove Theorem \ref{theorem-Schauder-phi}. We take three steps. In the first step, we control the H\"older norms of the solutions via using the local boundedness and the weak Harnack inequality due to Moser. In the second step, we estimate the H\"older norms of solutions' first derivatives. This step is established based 
on an integral characterization of H\"older continuous functions. 
In the third step, we bound the H\"older norms of solutions second derivatives. 
During the proof, we track the dependence on $\kappa$.

With \eqref{F1} and \eqref{F2}, the above equation \eqref{MOTS eqn} can be rewrite as
\begin{equation}\label{main eqn lambda}
-\D_{\gamma}\tilde{\phi}=1-\f12fe^{\tilde{\phi}}+F_1+F_2
\end{equation}
with $\tp=\tp(\o)$ being the solution and $0\leq f\leq 1$ on $\S$. Note that \eqref{main eqn lambda} is a quasilinear elliptic equation for $\tilde{\phi}$; $\gamma$ depends on $\tp$; $F$ contains the gradient of $\tilde{\phi}$. 

\begin{remark}
For notational simplicity, we set $B_r(p)$ to be the standard ball centered at $p$ with radius $r$ {\color{black}in $\mathbb{R}^2$, which lies in a coordinate chart of $\mathbb{S}^2$. For notational simplicity, when there is no danger of confusion, we also write $B_r(p)\subset \mathbb{S}^2$}. When the center is located at the origin, we often write $B_r$ for short. When emphasizing the radius $r$, we also employ the notation $B_p(r)$.
\end{remark} 

With coordinates, we can rewrite \ref{main eqn lambda} as
\begin{equation}\label{u equation lambda main}
-\partial_j\big(a_{ij}\partial_i\tilde{\phi}\big)=\sqrt{|\gamma|}\Big[1-\f12fe^{\tilde{\phi}}+{\color{black}F_1+F_2}\Big].
\end{equation}
By the hyperbolic part, we have 
\begin{align*}
\Lambda^{-1} I \le \big(a_{ij}(\o, \tilde{\phi})\big)  \le \Lambda I,
\end{align*} 
for some positive constant $\Lambda$, and $\Lambda^{-1}\leq \sqrt{|\gamma|}(\o, \tilde{\phi})\leq \Lambda$. 

Allowing us to abuse notations slightly, for any $(\o, \tilde{\phi}, p) \in B_1
\times \mathbb R \times \mathbb R^n$, applying the hyperbolic estimates, we can also write 
$$F_2(\o, \tilde{\phi}, p)=a^{-\f13}{\color{black}\O^{\t\d}}e^{\tilde{\phi}}\big[c_{ij}p_ip_j+c_ip_i+c_0\big]$$ 
with $c_{ij}, c_i, c_0$ being functions of $(\o,\tilde{\phi})$ on $B_1\times \mathbb R_+$. 
{\color{black}Denoting $p_i:=D_i\tilde{\phi}$ and $D\tilde{\phi}:=\nab\tilde{\phi}$,} the Cauchy-Schwarz inequality then implies 
\begin{equation*}|F_2|\le C a^{-\f13}{\color{black}\O^{\t\d}} e^{\tilde{\phi}}(|D \tilde{\phi}|^2+1)\leq c (|D \tilde{\phi}|^2+1)\end{equation*}
with $c$ a small positive constant. 

Since $F_1=-4(1-\ms)R\O\omb|\nab_{\gamma}\tp|^2$ and $R\O\omb\geq \f{\ms}{8(1-\ms)}>0$, here we only have 
$$|F_1|\leq |D\tilde{\phi}|^2+1.$$
While in \cite{AH}, $\O\omb$ obeys the same upper bound as for $\O\chibh$ and there $|F_1|$ satisfies the same estimate $|F_1|\leq c(|D\tp|^2+1)$ with $c$ being a small positive constant as for $|F_2|$ here. Hence it is negligible in \cite{AH}. 

In this paper, $F_1$ is of size $1$ and serves as the main term. A \underline{key} observation of this paper (also pointed out one line below (4.3) in \cite{AH}) is that the regularity estimates carried out in \cite{AH} also work for the scenario
\begin{equation}\label{F1andF2estimate}
|F_1|+|F_2|\leq c(|D\tp|^2+1) \quad \mbox{ with } \quad c=1.
\end{equation}

In below, we start to conduct these elliptic estimates. We then write \eqref{u equation lambda main} in divergence form as 
\begin{equation}\label{eq-main-nonlinear}\int a_{ij} (\o,\tilde{\phi})D_i \tilde{\phi} D_j \tilde{\psi}\, d\o
= \int \sqrt{|\gamma|}\Big[1-\f12fe^{\tp} +F(\o,\tp, D\tp)\Big]\t \psi\, d\o, 
\end{equation} 
with $\tilde{\psi}$ being any function in $H^1_0(B_1) \cap L^\infty (B_1)$. We also assume $\tilde{\phi}=\tilde{\phi}(\o)\in H^1(B_1)$ and require
\begin{equation}\label{eq-main-bound-u}0\le \tilde{\phi}\le \kappa\quad\text{on }B_1\end{equation} 
with $\kappa$ being some positive constant, {\color{black}which is related to the} anisotropicity {\color{black}via the construction of upper and lower barriers}. \\

Our first goal is to bound the H\"older norm of $\tilde{\phi}$ via Moser's iteration. For the estimates, we pay attention to their explicit dependence on $\kappa$. Following the steps in the proof of Theorem 4.2 in \cite{AH}, we obtain

\begin{proposition}\label{theorem-Holder-solution} 
Let $\tilde{\phi}\in H^1(B_1)$ be a positive solution to
\eqref{eq-main-nonlinear} with \eqref{eq-main-bound-u} satisfied. Then, it holds  
\begin{equation}\label{eq-Holder-semi-norm}
[\tilde{\phi}]_{C^\alpha(B_{1/2})}\le Ce^{\kappa}\end{equation}
with $C$ being a positive constant depending only on $\Lambda$. And $\alpha$ is given by 
\begin{equation}\label{eq-definition-alpha}
\alpha=\epsilon_0e^{-4\Lambda\kappa} \in (0,1).\end{equation}
Here $\epsilon_0>0$ is a small constant depending on $\Lambda$.
\end{proposition}

\begin{proof}  For $s \in (0,1)$, we set
$$M(s){\color{black}}=\max_{B_{s}}{\tilde{\phi}},\quad m(s){\color{black}}=\min_{B_{s}}{\tilde{\phi}}.$$
Writing $m=m(s)$ and $M=M(s)$ for brevity, proceeding the same as in the reasoning for Theorem 4.2 in \cite{AH}, we get

\begin{equation}\label{eq-estimate-v3}
M\big(\frac{s}{2}\big) - m(s) \le  Ce^{4\Lambda\kappa}\Big[m\big(\frac{s}{2}\big)
 -  m(s)\Big]  +  Ce^{4\Lambda\kappa+\kappa} {s}^{2},
\end{equation}

\begin{equation}\label{eq-estimate-w3}
M(s) - m \big(\frac{s}{2}\big) \le Ce^{4\Lambda\kappa}
 \Big[ M(s) - M \big(\frac{s}{2}\big)\Big]
 +  Ce^{4\Lambda\kappa+\kappa} {s}^{2}. 
\end{equation}

For any $s \in (0,1)$, define 
\begin{equation}\label{h=M-m} 
h(s): = M(s) - m(s).
\end{equation} 
Adding \eqref{eq-estimate-v3} and \eqref{eq-estimate-w3}, we then obtain
$$
h(s)+h\big(\frac{s}{2}\big)
\le Ce^{4\Lambda\kappa}\Big[h(s)
-h\big(\frac{s}{2}\big)\Big]+Ce^{4\Lambda\kappa+\kappa} {s}^{2}. 
$$
This implies
$$
h\big(\frac{s}{2}\big) 
\le C_{\gamma} h(s)+e^{\kappa} {s}^{2} \quad \mbox{ with } \quad C_{\gamma} = \frac{Ce^{4\Lambda\kappa}-1}{Ce^{4\Lambda\kappa}+1} < 1.$$
Using this $C_{\gamma}$, we pick up $\mu\in (0,1)$ satisfying 
$$\alpha:=(1-\mu)\log C_{\gamma}/\log(1/2) <2\mu.$$
Further applying Lemma 8.23 in \cite{GT} or Lemma 4.19 in \cite{HL}, for any $s\in (0, 1/2]$, we hence prove
\begin{equation}\label{omega hat 1}
h(s)\le C{s}^{\alpha}\big\{h(1)
+ e^{\kappa} s^{2\mu}\big\}\leq Ce^{\kappa} {s}^{\alpha},
\end{equation}
where  we employ $h(1)\leq M(1)\leq \kappa\leq e^{\kappa}$ for the second inequality. This concludes \eqref{eq-Holder-semi-norm}.\\
\end{proof} 

Next, we derive the gradient estimate for $\tilde{\phi}$.  By the hyperbolic part, we have
\begin{align}\label{eq-assumption-a-ij-Lambda-1}
|D_{\o}a_{kl}(\o,\tilde{\phi})|+|\partial_{\tilde{\phi}} a_{kl}(\o,\tilde{\phi})|\le{\color{black} \Lambda_1}
\end{align} 
with $\Lambda_1$ being some positive constant. The constant in the below theorem depends on $\Lambda_1$. And we have 

\begin{proposition}\label{theorem-Holder-gradient} Set $\tilde{\phi}\in H^1(B_1)$ to be a positive solution to
\eqref{eq-main-nonlinear} with \eqref{eq-main-bound-u} satisfied. Then, it holds 
\begin{equation}\label{eq-estimate-final}
|\tilde{\phi}|_{C^{1,1/3}(B_{1/2})}\le (Ce^{\kappa})^{Ce^{4\Lambda\kappa}}\end{equation}
with $C$ being a positive constant depending only on $\Lambda$ {\color{black} and $\Lambda_1$}. 
\end{proposition} 

\begin{proof} The detailed steps for proving Theorem 4.3 in \cite{AH} can also be carried out here. For any function $\tilde{\psi}$ belonging to $H_0^1(B_{1}) \cap L^\infty(B_{1})$ with $\tilde{\psi} \ge 0$, by \eqref{eq-main-nonlinear} and \eqref{F1andF2estimate}, it holds  
\begin{align*} 
\int a_{ij} D_i\tilde{\phi} D_j \tilde{\psi}\, d\o\le \int \sqrt{|\gamma|}\Big[1-\f12fe^{\tilde{\phi}}+|D\tilde{\phi}|^2+1\Big] \tilde{\psi}\, d\o
\le  \int [4+2|D\tilde{\phi}|^2] \tilde{\psi}\, d\o.
\end{align*}
Taking $\underline{\xi}{\color{black}(\tau)}=2^{-1}\Lambda^{-1}(e^{2\Lambda \tau}-1)$, we have 
\begin{align*} \int a_{ij} D_i(\underline{\xi}(\tilde{\phi}))D_j \tilde{\psi}\, d\o
\le \int 4\underline{\xi}'(\tilde{\phi})\tilde{\psi} \, d\o.
\end{align*}
This is equivalent to
\begin{align*} \int a_{ij} e^{2\Lambda\tilde{\phi}}D_i\tilde{\phi} D_j \tilde{\psi}\, d\o
\le \int 4e^{2\Lambda\tilde{\phi}}\tilde{\psi} \, d\o.
\end{align*}
Taking $\tilde{\psi}=\tilde{\phi} \t\mu^2$  for some $\t\mu\in C^1_0(B_1)$
with $\t\mu=1$ on $B_{3/4}$, via using the Cauchy-Schwarz inequality, we get 
\begin{equation}\label{eq-L2-estimate-gradient}
\int_{B_{3/4}}|D\tilde{\phi}|^2\, d\o\le C\kappa^2 e^{2\Lambda\kappa}\end{equation}
with $C$ being a positive constant depending only on $\Lambda$.

To derive a bound for the H\"older norm of $D\tp$, we need to employ the regularity of $a_{ij}(\o_0, \tilde{\phi}(\o_0))-a_{ij}(\o, \tilde{\phi}(\o))$, which was already obtained in Proposition \ref{theorem-Holder-solution}. Following the procedures of proving Theorem 4.3 in \cite{AH}, for any $0<\rho\le s$ we then obtain
$$\int_{B_\rho(\o_0)} |D\tilde{\phi}|^2 \, d\o\le C\Big[\Big({\frac{\rho}{s}}\Big)^2
+h(s)+[h(s)]^2+s^2\Big]\int_{B_{s}(\o_0)} |D\tilde{\phi}|^2\, d\o +
Ce^{2\kappa}{s}^{4},$$
where 
$$h(s)=h(\o_0;s):=\sup_{B_s(\o_0)}\tp-\inf_{B_{s}(\o)}\tp.$$

{\color{black}Take any $\delta\in(0,1)$. For a small $\epsilon_{\delta}>0$, we choose $s_\delta$ small such that} $B_{{s}_\delta}(\o_0)\subset B_{3/4}$ and 
\begin{equation}\label{eq-condition-delta}
h(s_\delta)+[h(s_\delta)]^2+s_\delta^2<\epsilon_{\delta}. \end{equation}
By employing \eqref{omega hat 1} and $\kappa\gg 1$, if we set
$s_{\delta}=(C_1 e^{\kappa})^{-C_2 e^{4\Lambda\kappa}}$, {\color{black}it satisfies the above requirement.} 

Applying Lemma 3.4 in \cite{HL},  for any $s\le {s}_\delta$ we then get
\begin{equation*}
\int_{B_{s}(\o_0)} |D\tilde{\phi}|^2\, d\o\le C{s}^{2\delta}\Big[{s}_\delta^{-2\delta}\int_{B_{{s}_\delta}(\o_0)} |D\tilde{\phi}|^2\, d\o
+ e^{2\kappa}\Big].
\end{equation*}
By \eqref{eq-L2-estimate-gradient}, this further implies 
\begin{equation}\label{eq-estimates-gradient-decay}
\int_{B_{s}(\o_0)} |D\tilde{\phi}|^2\, d\o\le C{s}^{2\delta}\big[\kappa^2e^{2\Lambda\kappa} {s}_\delta^{-2\delta}
+ e^{2\kappa}\big].
\end{equation} 

Proceeding the same as in \cite{AH}, {\color{black}defining 
$$(D\tilde{\phi})_{\o_0,s}:=\f{1}{|B_s(\o_0)|}\int_{B_s(\o_0)}D\tilde{\phi},$$} 
it also holds
\begin{align*}\int_{B_{s}(\o_0)}|D\tilde{\phi}-(D\tilde{\phi})_{\o_0,s}|^2 \, d\o\le 
C\kappa^{3}e^{3\Lambda\kappa} {s}_\delta^{-3\delta}{s}^{2+2\alpha'} 
\end{align*}
{\color{black}if we further impose} $\delta\in(2/3, 1)$ and $\alpha'=3\delta/2-1$. 

\noindent Applying Theorem 3.1 in \cite{HL}, together with (\ref{eq-L2-estimate-gradient}), we deduce that, for any $\o\in B_{1/2}$ 
\begin{equation}\label{C11}|D\tilde{\phi}(\o)|\le C\big[\kappa^{\f32}e^{\f32\Lambda\kappa} {s}_\delta^{-\f{3\d}{2}}+\kappa e^{\Lambda\kappa}\big]
\leq C \kappa^{\f32}e^{\f32\Lambda\kappa} {s}_\delta^{-\f{3\d}{2}},\end{equation}
and for any $\o_1,\o_2\in B_{1/2}$ satisfying $|\o_1-\o_2|<{s}_\delta$,
$$|D\tilde{\phi}(\o_1)-D\tilde{\phi}(\o_2)|\le C\kappa^{\f32}e^{\f32\Lambda\kappa} {s}_\delta^{-\f{3\d}{2}}|\o_1-\o_2|^{\alpha'}.$$
In addition, for any $\o_1,\o_2\in B_{1/2}$ we also have
\begin{equation}\label{C12}
|D\tilde{\phi}(\o_1)-D\tilde{\phi}(\o_2)|\le C\kappa^{\f32}e^{\f32\Lambda\kappa} {s}_\delta^{-\f{3\d}{2}}{s}_\delta^{-\alpha'}
|\o_1-\o_2|^{\alpha'}.
\end{equation}
Recall that ${s}_\delta=(C_1 e^{\kappa})^{-C_2 e^{4\Lambda \kappa}}$
and we have $0<\alpha'<1/2$. We now choose $\alpha'=1/3$. Utilizing (\ref{C11}) and (\ref{C12}), we hence conclude
\begin{equation}\label{eq-estimate-gradient}
|D\tilde{\phi}|_{C^{1/3}(B_{1/2})}\le (Ce^{\kappa})^{Ce^{2\Lambda\kappa}}.\\
\end{equation}

\end{proof}

We further prove the following Schauder estimates. 
We work under the assumption 
\begin{align}\label{eq-assumption-a-ij-Lambda-2}
|D^2_{\o}a_{kl}(\o,\tilde{\phi})|+|D_{\o}\partial_{\tilde{\phi}} a_{kl}(\o,\tilde{\phi})|+|\partial^2_{\tilde{\phi}} a_{kl}(\o,\tilde{\phi})|\le \Lambda_2
\end{align} 
with $\Lambda_2$ being a positive constant. 

\begin{proposition}\label{theorem-Holder-second derivative} 
Suppose $\tilde{\phi}\in H^1(B_1)$ is a positive solution to
\eqref{eq-main-nonlinear} with \eqref{eq-main-bound-u} satisfied, and $f\in C^{1/3}(B_1)$.  
Then, it holds 
\begin{equation}\label{eq-estimate-second-derivative}
|\tilde{\phi}|_{C^{2,1/3}(B_{1/2})}\le (Ce^{\kappa})^{Ce^{4\Lambda\kappa}}.\end{equation}
Here $C$ is a positive constant depending only on $\Lambda$, {\color{black} $\Lambda_1$, $\Lambda_2$,  
the $C^{1}$-norms of $\sqrt{|\gamma|}$, $c_{ij}$, $c_i$, and $c_0$ on 
$B_1\times \mathbb R_+$,}
and the $C^{1/3}$-norm of $f$ on $B_1$.
\end{proposition}

\begin{proof} 
We have already proved $\tp\in C^{1, 1/3}(B_{1/2})$ and we can write \eqref{eq-main-nonlinear} in nondivergence form
\begin{equation}\label{Theorem 4.7 application}
a_{ij} (\o,\tilde{\phi}) \partial_{ij}\tilde{\phi} = {h}(\omega)  
\end{equation}
with 
\begin{equation*}
{h}(\o)=-\partial_{j}a_{ij}(\o,\tilde{\phi})\partial_i \tilde{\phi}
-\partial_{\tilde{\phi}} a_{ij}(\o,\tilde{\phi})\partial_i \tilde{\phi}\partial_j \tilde{\phi}- \sqrt{|\gamma|}\Big[1-\f12fe^{\tilde{\phi}}+F(\o, \tilde{\phi}, D\tp)\Big].
\end{equation*}
We now treat equation \eqref{Theorem 4.7 application} as a linear equation of $\tilde{\phi}$, and consider ${h}$ as a given term defined in $B_1$. 
Using the expression of $F$ and Proposition \ref{theorem-Holder-gradient}, we get 
$$|{h}|_{C^{1/3}(B_{3/4})}\le (Ce^{\kappa})^{Ce^{4\Lambda\kappa}}.$$
Employing the standard interior Schauder estimate, we then obtain  
\begin{equation*}
|\tilde{\phi}|_{C^{2,1/3}(B_{1/2})}\le C_*\big[|\tilde{\phi}|_{L^\infty(B_{3/4})}+|h|_{C^{1/3}(B_{3/4})}\big]
\le C_*(Ce^{\kappa})^{Ce^{4\Lambda\kappa}}\end{equation*}
with $C_*>0$ a constant depending only on $\Lambda$, and the $C^{1/3}$-norm of $a_{ij}$ on $B_{3/4}$. 
\end{proof} 

Gathering the above three propositions just obtained, we are ready to prove  

\begin{theorem}\label{theorem-Schauder-phi}
With $\mathcal S(\tilde{\phi})$ being defined as in \eqref{S=0} and $\kappa$ being a positive constant, if we assume $\tilde{\phi}$ is a solution to $\mathcal S(\tilde{\phi})=0$ satisfying 
$0\le \tilde{\phi}\le\kappa$ on $\mathbb S^2$,  
then it holds
$$|\tilde{\phi}|_{C^{2, 1/3}(\mathbb{S}^2)}\leq (Ce^{\kappa})^{Ce^{4\Lambda\kappa}},$$
with $\Lambda$ being the ellipticity constant for the metric component $\gamma_{ij}$ and $C$ being a 
positive constant depending only on
$$\gamma, \chibh, \eta, \, \alpha'_1=-(\f12\O\tr_g\chib+\f{1}{R}), \,\a_2=\f12\O^{-1}\tr_g\chi-\f{1-\ms}{R}+\f{\ub a f(\o, \ub)}{2\O^2_0 R^2}(1-\ms)$$  
and their angular derivatives.
\end{theorem}  

\begin{proof}
We first fix an arbitrary ball $B$ on $\mathbb S^2$. In local coordinates within $B$, 
we have \eqref{u equation lambda main} and 
$$a_{ij}(\o,\tilde{\phi})=\sqrt{|\gamma(u, \o)|}\gamma^{ij}(u, \o).$$
Here $u=-R=-(\ub a)^{1-\ms}e^{-(1-\ms)\tilde{\phi}}.$ Recall that {\color{black} $F=F_1+F_2$ with $F_1=-4(1-\ms)R\O\omb|\nab_{\gamma}\tilde{\phi}|^2$} and $F_2$ being a linear combination of $\nab_i\tilde{\phi}\nab_j\tilde{\phi}, \nab_i\tilde{\phi}, 1$ with the below coefficients 
$$\gamma^{ij}, \quad R^{-1}\Omega\chibh_{kl}, \quad \eta_i, \quad R\alpha'_1, \quad R\alpha_2.$$
The coefficients of $F_2$ are functions of $\o$ and $u=-R=-(\ub a)^{1-\ms}e^{-(1-\ms)\tilde{\phi}}$. To control the $C^{1/3}$-norm of $F_2$, we estimate these coefficients and their derivatives with respect to $\o$ and $\tilde{\phi}$. 
With $\partial_{\tilde{\phi}}=R\partial_u$, we hence need to bound
\begin{align*} 
&R^{-1}\Omega\chibh_{kl},\, \eta_k,\, {\color{black}R\alpha'_1},\, R\alpha_2,\\
&R^{-1}\nab_i(\Omega\chibh_{kl}),\, \nab_i\eta_k,\, {\color{black}R\nab_i\alpha'_1},\, R\nab_i\alpha_2,\\
&\partial_u(\Omega\chibh_{kl}),\, R\partial_u\eta_k,\, {\color{black}R^2\partial_u\alpha'_1},\, R^2\partial_u\alpha_2
\end{align*}
with 
\begin{equation}\label{alpha2} 
\alpha'_1:=-(\f12\O\tr_g\chib+\f{1}{R}),\quad \a_2:=\f12\O^{-1}\tr_g\chi-\f{1-\ms}{R}+\f{\ub a f(\ub, \o)}{2\O^2_0 R^2}(1-\ms). 
\end{equation}
By the hyperbolic estimates, all these quantities are bounded by $e^{\tilde{\phi}}a^{-\f13}{\color{black}\O^{\t\d}}$. 
The conditions in Proposition \ref{theorem-Holder-solution}, 
Proposition \ref{theorem-Holder-gradient}, and Proposition \ref{theorem-Holder-second derivative} all hold. Therefore, we derive the desired $C^{2, 1/3}$ estimates for solution $\tp$ to the equation of MOTS.
\end{proof}   

\section{The Linearized Equation}\label{linearization}  
To derive and to analyze the linearized equation, in contrast to \cite{AH} and previous sections, here we recast the equation of MOTS with a different form and carry out the associated linear analysis.

For $u=-R(\o, \ub)$ being the MOTS along $\Hb_{\ub}$, we set $R=(\ub a)^{1-\ms}\cp^{\lambda}$ with $\lambda$ to be determined later. Then, it holds
$$\nab_g R=\lambda(\ub a)^{1-\ms}\cp^{\lambda-1}\nab_g\cp=\lambda R\cp^{-1}\nab_g\cp,$$
which implies
\begin{align*}
\div(\nab_g R)=&\div(\lambda R \cp^{-1}\nab_g\cp)=(\lambda^2-\lambda)R\cp^{{\color{black}-2}}|\nab_g\cp|^2+\lambda R \cp^{-1}\Delta_g\cp.
\end{align*}
Thus, with the aid of \eqref{eq-expression-cal-G-ub}, the equation of MOTS $0=\M G(R,\ub)$ is changed into
\begin{align*}
0=&\Delta_g R+\f12\O\tr_g\chib|\nab_g R|^2+4\O\omb|\nab_g R|^2-\f12\O^{-1}\tr_g\chi+2g(\eta,\nab_g R){\color{black}+2\O\chibh(\nab_g R, \nab_g R)}\\
=&[-\lambda+\f{\ms}{1-\ms}\lambda^2]R\cp^{-2}|\nab_g\cp|^2+\lambda R\cp^{-1}\Delta_g\cp-\f12\O^{-1}\tr_g\chi+2\lambda R\cp^{-1}g(\eta,\nab_g \cp)\\
&{\color{black}+(\f12\O\tr_g\chib+\f{1}{R})\lambda^2 R^2\cp^{-2}|\nab_g\cp|^2+[4\O\omb-\f{\ms}{1-\ms}\cdot\f{1}{R}]\lambda^2 R^2\cp^{-2}|\nab_g\cp|^2+2\O\chibh(\nab_g R, \nab_g R).}
\end{align*}
We now pick 
$$\lambda=\f{1-\ms}{\ms}, \mbox{ which satisfies } -\lambda+\f{\ms}{1-\ms}\lambda^2=0.$$
This renders the above equation to be 
\begin{align*}
0=&\lambda R\cp^{-1}\Delta_g\cp-\f12\O^{-1}\tr_g\chi+2\lambda R\cp^{-1}g(\eta,\nab_g \cp)+(\f12\O\tr_g\chib+\f{1}{R})\lambda^2 R^2\cp^{-2}|\nab_g\cp|^2\\
&{\color{black}+[4\O\omb-\f{\ms}{1-\ms}\cdot\f{1}{R}]\lambda^2 R^2\cp^{-2}|\nab_g\cp|^2+2\O\chibh(\nab_g R, \nab_g R).}
\end{align*}
{\color{black}With Christodoulou's naked-singularity initial data, we note that $4\O\omb(-R,0)=\f{\ms}{1-\ms}\cdot\f{1}{R}$}. Hence we further define $\a_1$ as below and recall $\a_2$ given in \eqref{alpha2}
$$\a_1:=-(\f12\O\tr_g\chib+\f{1}{R}){\color{black}-[4\O\omb-(4\O\omb)_0],}\footnote{\label{alpha1 perturbed data}For perturbed  initial data, in $\alpha_1$ there is an additional term, i.e., $4\O\omb(u,0)-\f{\ms}{1-\ms}\cdot\f{1}{|u|}$. Its absolute value is bounded by $\f{1}{|u|}\cdot\f{\O(u,0)^{\t\d}}{a^{\f13}}$ and in later arguments, this additional term obeys a desired control.}$$
$$\a_2:=\f12\O^{-1}\tr_g\chi-\f{1-\ms}{R}+\f{\ub a f(\o, \ub)}{2\O^2_0 R^2}(1-\ms). $$
Together with $R=(\ub a)^{1-\ms}\cp^{\lambda}$, $\gamma=R^{-2}g$, $\Delta_g\cp=R^{-2}\Delta_\gamma \cp$ and $\O^2_0(R)=\O^2(u,0)=R^{\f{\ms}{1-\ms}}$, the above equation infers
\begin{align*}
&\lambda\cp^{-1}\Delta_\gamma \cp-(1-\ms)+(1-\ms)\f{f(\o, \ub)}{2}\cp^{-\f{\lambda}{1-\ms}}{\color{black}+}2\lambda R^2\cp^{-1}g(\eta,\nab_g \cp)\\
&-\lambda^2 R^3\cp^{-2}|\nab_g \cp|^2\a_1-R\a_2{\color{black}+2R\O\chibh(\nab_g R, \nab_g R)}=0.
\end{align*}
With $\lambda=(1-\ms)/\ms$, we now define 
\begin{equation}\label{Stp}
\begin{split}
S(\cp, \o):=&\Delta_\gamma \cp-\ms \cp+\f{\ms}{2}f(\o, \ub)\cp^{1-\f{1}{\ms}}{\color{black}+}\f{2\ms}{1-\ms}\lambda R^2g(\eta,\nab_g \cp)\\
&-\f{\ms}{1-\ms}\lambda^2 R^3\cp^{{\color{black}-1}}|\nab_g \cp|^2\a_1-\f{\ms}{1-\ms}\cp R\a_2{\color{black}+\f{2\ms}{1-\ms}\cp R\O\chibh(\nab_g R, \nab_g R)}.
\end{split}
\end{equation} 
And for $2$-sphere $S_{-(\ub a)^{1-\ms}\cp^{\f{1-\ms}{\ms}}, \,\ub}$, we have that, if
$$S(\cp, \o) 
\begin{cases}
<0, & \text{it corresponds to an untrapped surface},\\
=0, & \text{it corresponds to a MOTS},\\
>0, & \text{it corresponds to a trapped surface}.
\end{cases}$$ 
 
For the operator $\mathcal S(\cp, \o)$ defined in \eqref{Stp}, it further holds  

\begin{proposition}\label{form of linearization} 
Suppose $\cp$ is a given function on $\mathbb S^2$. 
The linearized operator $\partial_{\cp} \mathcal S(\cp)$ can be expressed as
\begin{equation}\label{eq-linearization-S}
\partial_{\cp} \mathcal S(\cp)[w]
=\Delta_{\gamma}w-\ms w-\f{1-\ms}{2} f(\o)\cp^{-\f{1}{\ms}}w+c^{i}\nab_i w+cw.
\end{equation}
Here the connection is with respect to the metric $\gamma$, and functions
$c=c(\o,\cp,\ub)$ and $c^i=c^i(\o,\cp,\ub)$ satisfy 
$$|c|+|c^i|\leq Ca^{-\f13}{\color{black}\O^{\t\d}} {\color{black} (|\cp^{-1}\nabla^2_{\gamma}\cp|+|\cp^{-1}\nabla_\gamma\cp|^2+1)}.$$
\end{proposition}

\begin{proof} 

We first write
\begin{align}\label{eq-expression-S-I-1-5}\mathcal S(\cp)=I_1+I_2+I_3+I_4+I_5,\end{align}
with
\begin{align}\label{eq-expression-S-I-1}
I_1=\D_{\gamma}\cp-\ms\cp+\f{\ms}{2}f \cp^{1-\f{1}{\ms}},\end{align}
\begin{align}
\nonumber I_2=\f{2\ms}{1-\ms}{\color{black}\lambda^2\cp^{-1}} R^{-1}\Omega\chibh_{kl} \gamma^{ik}\gamma^{jl}\nab_i \cp\nab_j \cp, \quad
&I_3=\f{2\ms}{1-\ms} \lambda\gamma^{ij}\eta_i \nab_j \cp,\\
\label{eq-expression-S-I-2-5} I_4=-\f{\ms}{1-\ms}{\color{black}\lambda^2\cp^{-1}} R\alpha_1 |\nab_\gamma \cp|^2_{\gamma},\quad
&I_5=-\f{\ms}{1-\ms}\cp R\alpha_2. 
\end{align}
Then we take $\partial_{\cp}$ with respect to each $I_i$ with $i=1,2,3,4,5$. Recall $R=(\ub a)^{1-\ms}\cp^{\f{1-\ms}{\ms}}$ and we set
$$R(\epsilon)=(\ub a)^{1-\ms}(\cp+\epsilon w)^{\f{1-\ms}{\ms}}.$$
Denoting $\dot{\,}=\f{d\,}{d\epsilon}\Big|_{\epsilon=0}$, it then holds 
$$\dot{R}=R\cp^{-1}\f{1-\ms}{\ms}w, \quad \dot{u}=-R\cp^{-1}\f{1-\ms}{\ms}w, \quad \dot{\cp}=w.$$ 
We also have
$$\dot{\gamma}^{ij}=-\gamma^{ik}\gamma^{jl}\dot{\gamma}_{kl} \quad \mbox{ and } \quad \dot{\gamma}_{kl}=\dot{u}\partial_u \gamma_{kl}=-\f{1-\ms}{\ms}\cp^{-1}w[(R\O\tr_g\chib+2)\gamma_{kl}+2R^{-1}\O\chibh_{kl}].$$
For $I_1=\Delta_{\gamma}\cp-\ms\cp+\f{\ms}{2}f\cp^{1-\f{1}{\ms}}$, we further apply Lemma \ref{lemma-variation-Laplacians} to $\gamma$. Together with $\dot{\phi}=w$, we obtain 
\begin{equation}\label{eq-variation-Laplace}
\begin{split}
(\Delta_{\gamma}\cp)^{\boldsymbol\cdot}=&\Delta_\gamma w-\gamma^{ik}\gamma^{jl}\dot{\gamma}_{kl}\nabla^2_{ij}\cp
-\gamma^{ij}\gamma^{kl}\nabla_i\dot{\gamma}_{jl}\nabla_k\cp+\frac12\gamma^{ij}\gamma^{kl}\nabla_l\dot{\gamma}_{ij}\nabla_k\cp
\end{split}
\end{equation} 
and 
$$\dot{I}_1=\Delta_{\gamma}w-\ms w-\f{1-\ms}{2} f(\o)\cp^{-\f{1}{\ms}}w.$$
For $I_2=\f{2\ms}{1-\ms}{\color{black}\lambda^2\cp^{-1}} R^{-1}\Omega\chibh_{kl} \gamma^{ik}\gamma^{jl}\nab_i \cp\nab_j \cp$, with $\lambda=\f{1-\ms}{\ms}$ it holds
\begin{align*}
\dot{I}_2=&-2\lambda\cp^{-2}w R^{-1}\O\chibh_{kl}\gamma^{ik}\gamma^{jl}\nab_i\cp\nab_j\cp-2\lambda\cp^{-1}R^{-2}R\cp^{-1}\lambda w\O\chibh_{kl}\gamma^{ik}\gamma^{jl}\nab_i\cp\nab_j\cp\\
&+2\lambda\cp^{-1}R^{-1}\partial_{u}(\O\chibh_{kl})\cdot(-R\cp^{-1}\lambda w)\cdot \gamma^{ik}\gamma^{jl}\nab_i\cp\nab_j\cp+2\lambda\cp^{-1}R^{-1}\O\chibh_{kl}(\dot{\gamma}^{ik}\gamma^{jl}+\gamma^{ik}\dot{\gamma}^{jl})\nab_i\cp\nab_j\cp\\
&+2\lambda \cp^{-1}R^{-1}\O\chibh_{kl}\gamma^{ik}\gamma^{jl}(\nab_i w\nab_j \cp+\nab_i \cp\nab_j w),
\end{align*}
where the coefficient $c$ for $w$ and the coefficient $c_i$ for $\nab_i w$ satisfying 
$$|c|\leq \f{{\color{black}\O^{\t\d}}}{a^{\f13}}|\cp^{-1}\nab_j\cp|^2, \quad |c_i|\leq \f{{\color{black}\O^{\t\d}}}{a^{\f13}}|\cp^{-1}\nab\cp|.$$

\noindent For $I_3=-\f{2\ms}{1-\ms} \gamma^{ij}\eta_i \nab_j \cp$, we get
\begin{align*}
\dot{I}_3=&\f{2\ms}{1-\ms}\gamma^{ik}\gamma^{jl}\bigg(-\f{1-\ms}{\ms}\cp^{-1}w[(R\O\tr_g\chib+2)\gamma_{kl}+2R^{-1}\O\chibh_{kl}] \bigg)\eta_i\nab_i\cp\\
&-\f{2\ms}{1-\ms}\gamma^{ij}\partial_u\eta_i\cdot(-R\cp^{-1}\f{1-\ms}{\ms}w)\nab_j\cp-\f{2\ms}{1-\ms}\gamma^{ij}\eta_i\nab_j w.
\end{align*}
And on the right the coefficient $c$ for $w$ and the coefficient $c_i$ for $\nab_i w$ satisfy 
$$|c|\leq \f{{\color{black}\O^{\t\d}}}{a^{\f13}}|\cp^{-1}\nab_j\cp|, \quad |c_i|\leq \f{{\color{black}\O^{\t\d}}}{a^{\f13}}.$$ 

\noindent For $I_4=-\f{\ms}{1-\ms}{\color{black}\lambda^2\cp^{-1}} R\alpha_1 |\nab_\gamma \cp|^2$, we have the expression
\begin{align*}
\dot{I}_4=&\lambda\cp^{-1}R\partial_u\alpha_1 R\cp^{-1}\lambda w \gamma^{ij}\nab_i\cp\nab_j\cp+\lambda\cp^{-2}w R\alpha_1\gamma^{ij}\nab_i\cp\nab_j\cp\\
&-\lambda\cp^{-1}R\cp^{-1}\lambda w \alpha_1 \gamma^{ij}\nab_i\cp\nab_j\cp-\lambda\cp^{-1}R\alpha_1\dot{\gamma}^{ij}\nab_i\cp\nab_j\cp\\
&-\lambda\cp^{-1}R\alpha_1\gamma^{ij}(\nab_i w\nab_j \cp+\nab_i\cp\nab_j w).
\end{align*}
With precise calculations, we show that the size of the coefficients for $\o$ and $\nab_i\o$ are small. First recall 
$$\partial_u\alpha_1=-\f12\partial_u[\O\tr_g\chib+\f{2}{R}]-4\partial_u[\O\omb-(\O\omb)_0].$$
For the term $-\partial_u[\f12\O\tr_g\chib+\f{1}{R}]$, there are cancellations. By the null structure equation and $\partial_u(2R^{-1})-2R^{-2}=0$, it holds
\begin{equation*}
\begin{split}
&\partial_u(\O\tr\chib+\f{2}{R})+\f12(\O\tr\chib-\f{2}{R})(\O\tr\chib+\f{2}{R})\\
=&-4\O\omb\cdot\O\tr\chib-\O^2|\chibh|^2_g-(\partial_u\phi)^2\\
=&-4[\O\omb-(\O\omb)_0]\O\tr\chib-4(\O\omb)_0[\O\tr\chib-(\O\tr\chib)_0]-\O^2|\chibh|^2_g-[\partial_u\phi-(\partial\phi)_0]\partial_u\phi\\
&-(\partial_u\phi)_0[\partial_u\phi-(\partial_u\phi)_0]-4(\O\omb)_0(\O\tr\chib)_0-(\partial_u\phi)_0(\partial_u\phi)_0.
\end{split}
\end{equation*}
Using the hyperbolic part, we have improved estimates 
$$|\O\omb-(\O\omb)_0|\leq\f{\ub\at}{|u|^2}, \quad |\partial_u\phi-(\partial_u\phi)_0|\leq\f{\ub\at}{|u|^2}.$$
The remaining worrisome terms are 
$$-4(\O\omb)_0(\O\tr\chib)_0-(\partial_u\phi)_0(\partial_u\phi)_0.$$
They could be of size $R^{-2}$ and not small, which is dangerous. However, cancellations of Christodoulou's naked-singularity initial data help us. Along $\ub=0$, we have
\begin{align*}
-4(\O\omb)_0\cdot(\O\tr\chib)_0-(\partial_u\phi)_0^2=&-4\f{\ms}{4(1-\ms)}\cdot\f{1}{R}\cdot(\f{-2}{R})-2\f{\ms}{1-\ms}\cdot\f{1}{R^2}\\
=&\f{2\ms}{1-\ms}\cdot\f{1}{R^2}-\f{2\ms}{1-\ms}\cdot\f{1}{R^2}=0.
\end{align*}
Hence, the worrisome terms disappear.\footnote{\label{initial data cancellation 1}{\color{black}For perturbed initial data, after the cancellation of the leading terms, these worrisome items are of size $\O(-R, 0)^{\t\d}R^{-2}a^{-\f13}$, which are sufficiently small.}}  

The term $-4\partial_u[\O\omb-(\O\omb)_0]$ in $\partial_u\alpha_1$ also encodes this smallness. We have
\begin{align*}
\partial_u[\O\omb-(\O\omb)_0]=&\partial_u[\int_0^{\ub}(\O\nab_4)(\O\omb)(u,\ub',\theta_1, \theta_2)d\ub']=\int_0^{\ub}\partial_u(\O\nab_4)(\O\omb)(u,\ub',\theta_1, \theta_2)d\ub'\\
=&\int_0^{\ub}[\O\nab_3(-\O^2\eta\cdot\etb+\f12\O^2|\eta|^2+\f12\O^2\rho{\color{black}+\f16\O e_3\phi\cdot\O e_4\phi+\f{1}{12}\O^2 e_A\phi e^{A}\phi})](u,\ub',\theta_1, \theta_2)d\ub'.
\end{align*}

In the last line of above equation, except $\O e_3(\O e_3\phi)$ and $\nab_3\eta$, all the other terms can be explicitly expressed and obey desired estimates. To control $\O\nab_3(\O \nab_3\phi)$ term, for $h$ being a scalar function, we recall a fact
$$\O \nab_4(\O \nab_3 h)=\O\nab_3(\O\nab_4 h)-4\O^2\zeta\cdot\nab h.$$
Setting $h=\O e_3\phi$, we then have 
\begin{align*}
\O e_4 (\O e_3)(\O e_3\phi)=&\O e_3 (\O e_4)(\O e_3\phi)-4\O^2\zeta\nab(\O e_3\phi)\\
=&\O e_3[-\f12\O\tr\chib\cdot\O e_4\phi-\O^2\Delta_g\phi-\f12\O\tr\chi\cdot\O e_3\phi+2\O^2\etb^A e_A\phi]-4\O^2\zeta\nab(\O e_3\phi). 
\end{align*}
Using naked-singularity initial data and integrating along $e_4$ direction, we get
$$|\O e_3 (\O e_3\phi)|\lesssim \f{1}{|u|^2},$$
which is a desired estimate for here. To bound $\nab_3\eta$ term, we utilize  
$$\nab_3\eta=-\nab_3\etb+2\nab_3\nab\log\O.$$
There is the null structure equation for $\nab_3\etb$ and we also have
$$\nab_3\nab\log\O=\nab\nab_3\log\O+\f12(\eta+\etb)\nab_3\log\O-\chib\cdot\nab\log\O=-2\nab\omb-(\eta+\etb)\omb-\chib\cdot\nab\log\O.$$
These imply
$$|\nab_3\eta|\leq\f{\ub\at}{|u|^3}\leq \f{1}{|u|^2}\cdot\f{{\color{black}\O^{\t\d}}}{a^{\f13}},$$
which is a desired estimate for here as well. 

With the derived estimate above, a direct check shows that, for $\dot{I}_4$ the coefficient $c$ for $w$ and the coefficient $c_i$ for $\nab_i w$ obey 
$$|c|\leq \f{{\color{black}\O^{\t\d}}}{a^{\f13}}(\cp^{-1}\nab_j\cp)^2, \quad |c_i|\leq \f{{\color{black}\O^{\t\d}}}{a^{\f13}}|\cp^{-1}\nab_{\gamma}\cp|.$$  

We proceed to study $\dot{I}_5$. For $I_5=-\lambda^{-1}\cp R\alpha_2$ with 
\begin{align*}
\alpha_2=&[\f12\O^{-1}\tr_g\chi-\f{1-\ms}{R}+\f{\ub a f(\o, \ub)}{2\O^2_0 R^2}(1-\ms)](u, \ub, \o)\\
=&-\f12\int_0^{\ub}(\f12\tr\chi\tr\chi+|\chih|_g^2+\partial_4\phi\partial_4\phi)(u, \ub', \o)d\ub'\\
&+\f12\int_0^{\ub}\f{\O^{-2}(u,\ub', \o)}{|u|^2} \bigg(|\O\chih|_g^2(-1, \ub', \o)+[\O e_4\phi(-1, \ub', \o)-\O e_4\phi(-1, 0, \o)]^2 \bigg) d\ub'\\
=&-\f12\int_0^{\ub}\f{\O^{-2}}{2}(\O\tr\chi)^2(u,\ub', \o)d\ub'-\f12\int_0^{\ub}\f{\O^{-2}}{|u|^2}[|u\O\chih|^2_g(u, \ub', \o)-|\O\chih|^2_g(-1, \ub', \o)]d\ub'\\
&-\f12\int_0^{\ub}\f{\O^{-2}}{|u|^2}\bigg((u\O e_4\phi)^2(u, \ub', \o)-[\O e_4\phi(-1, \ub, \o)-\O e_4\phi(-1, 0, \o)]^2\bigg)d\ub',
\end{align*}
it holds
\begin{align*}
\dot{I}_5=\lambda^{-1}\cp R\partial_u \alpha_2\cdot R\cp^{-1}\lambda w-\lambda^{-1}w R\alpha_2-\lambda^{-1}\cp R\cp^{-1}\lambda w\alpha_2.
\end{align*}

To bound $\partial_u \alpha_2$, we first control
\begin{equation}\label{partialutrchitrchi}
\begin{split}
&\int_0^{\ub}\partial_u \tr\chi \cdot\tr\chi(u, \ub', \o)d\ub'\\
=&\int_0^{\ub}\bigg(-\f12\O\tr\chib \tr\chi+2\O\omb\tr\chi+2\O\rho-\O\chih\cdot\chibh\bigg)\tr\chi(u, \ub', \o)d\ub'\\
&+\int_0^{\ub}\bigg(2\O\div\eta+2\O|\eta|^2{\color{black}-\f13 \O e_3\phi e_4\phi+\f13 \O e_A\phi e^A\phi} \bigg)\tr\chi(u, \ub', \o)d\ub'
\end{split}
\end{equation}
Applying the improved estimate for $\tr\chi$, we bound the potentially dangerous terms
\begin{align*}
&\int_0^{\ub}(\O\tr\chib, \O\omb, \O e_3\phi)\cdot(\tr\chi, e_4\phi)\cdot\tr\chi(u, \ub', \o)d\ub'\\
\leq&\f{\ub\O^{-2}}{|u|}\cdot\f{\at}{|u|}\cdot(\f{\O}{|u|}+\O^{-1}\f{\ub a}{|u|^2})\\
=&\f{\ub a}{|u|^2\O^2}\cdot\f{a^{-\f12}\O}{|u|}+\f{\ub a}{|u|^2\O^2}\cdot\f{1}{|u|}\cdot\f{\ub\at}{|u|\O}\leq \f{1}{|u|^2}\cdot\f{\O^{\t\delta}}{a^{\f13}}.
\end{align*}
All the other terms in \eqref{partialutrchitrchi} obey even smaller upper bound.

For the remaining terms in $\partial_u\alpha_2$, we need to show that\footnote{\label{initial data cancellation 2}With perturbed initial data, we would get a similar bound.}
\begin{equation}\label{e4 difference}
||u|\O e_4\phi(u, \ub, \o)-[\O e_4\phi(-1, \ub, \o)-\O e_4\phi(-1, 0, \o)]|\leq \at\cdot a^{-\f13}\O^{{\color{black}\t\delta}}. 
\end{equation}
A detailed discussion in Section \ref{BV instability}, in particular \eqref{e4phi difference with initial data}, will give
\begin{equation*}
|u\O e_4\phi(u, \ub, \o)-[\O e_4\phi(-1, \ub, \o)-\O e_4\phi(-1, 0, \o)]|\lesssim \f{\ub a}{|u|}+|u|^{\f{\ms}{1-\ms}}\leq \at\cdot a^{-\f13}\O^{{\color{black}\t\delta}}.
\end{equation*}

With these estimates, we can check that $\partial_u\alpha_2$ obeys
$$|\partial_u \alpha_2|\leq \f{1}{|u|^2}\f{{\color{black}\O^{\t\d}}}{a^{\f13}}.$$
Back to $\dot{I}_5$, we conclude the coefficient $c$ for $w$ and the coefficient $c_i$ for $\nab_i w$ obey 
$$|c|+|c_i|\leq \f{{\color{black}\O^{\t\d}}}{a^{\f13}}.$$ 

\end{proof}

For the next step, we set $$L_0{\color{black}w}:=\Delta_{\gamma}w-\ms w+\f{\ms-1}{2}f(\o, \ub)\cp^{-\f{1}{\ms}}w.$$ 
Integrating $-wL_0 w$ on $\mathbb{S}^2$ and applying integration by parts, we obtain
$$\int_{\mathbb{S}^2}\ms w^2 d\o\leq \int_{\mathbb{S}^2}\bigg(|\nab_{\gamma}w|^2+\ms w^2+\f{1-\ms}{2}f(\o, \ub)\cp^{-\f{1}{\ms}}w^2\bigg)d\o=-\int_{\mathbb{S}^2}w L_0 w d\o.$$
This implies
$$\int_{\mathbb{S}^2}w^2 d\o\leq C\int_{\mathbb{S}^2}|L_0 w|^2 d\o.$$
We further let
$$<w_1, w_2>:=\int_{\mathbb{S}^2}\bigg(\nab_{\gamma}w_1\cdot \nab_{\gamma}w_2+\ms w_1 w_2+\f{1-\ms}{2}f(\o, \ub)\cp^{-\f{1}{\ms}}w_1 w_2\bigg)d\o.$$
Since $0<\ms<1$ and $f\geq 0$, we can use $<w_1, w_2>$ as an inner-product on $H^1(\mathbb{S}^2)$, equivalent to the standard $H^1(\mathbb{S}^2)$ inner product. The first eigenvalue of $-L_0$ can be expressed as
$$\mu_1=\inf\{ \int_{\mathbb{S}^2}[|\nab_{\gamma}w|^2+\ms w^2+\f{1-\ms}{2}f(\o, \ub)\cp^{-\f{1}{\ms}}w^2]d\o; \int_{\mathbb{S}^2}w^2 d\o=1\}.$$
Immediately, we see that $\mu_1\geq \ms$. And via using  Theorem 8.38 in \cite{GT}, we also have

\begin{lemma}\label{lemma-H1-estimates-L0} 
Assume $f$ and $\cp$ to be given nonnegative functions on $\mathbb S^2$ with \eqref{eq-assumption-lower-bound-f} satisfied. Then the equation $L_0w=0$ only admits the trivial solution. Moreover, the first eigenvalue $\mu_1$ of $-L_0$ is positive and simple and it has a positive eigenfunction.
\end{lemma}

With this property, we obtain the crucial invertibility of the operator $\partial_{\cp}S(\cp)$:  
\begin{lemma}\label{lemma-linearized-op-invert}
Let $\beta\in (0,1)$ be a fixed constant. Set 
$f$ and $\cp$ to be given nonnegative functions on $\mathbb S^2$
with \eqref{eq-assumption-lower-bound-f} satisfied and 
$|\cp|_{C^2(\mathbb S^2)}\le K$ with $K$ obeying \eqref{constant K}. Then within the hyperbolic region, it holds that  
$\partial_{\cp}\mathcal S(\cp): C^{2, \beta}(\mathbb{S}^2)\rightarrow C^{{0},\beta}(\mathbb{S}^2)$
is invertible.
\end{lemma}

\begin{proof} Applying Lemma \ref{lemma-H1-estimates-L0},  we first take a positive 
eigenfunction $\psi_1$ associated with the first eigenvalue $\mu_1>0$, i.e.,
$$\D_{\gamma} \psi_1-\ms\psi_1-\f{1-\ms}{2}f \cp^{-\f{1}{\ms}}\psi_1=-\mu_1\psi_1.$$
Using Proposition \ref{form of linearization}, we have   
\begin{align*}
\partial_{\cp} \mathcal S(\cp)[\psi_1]
&=\Delta_{\gamma}\psi_1-\ms\psi_1-\f{1-\ms}{2}f \cp^{-\f{1}{\ms}}\psi_1+c^{i}\nab_i \psi_1+c\psi_1\\
&=-\mu_1\psi_1+c^{i}\nab_i \psi_1+c\psi_1.
\end{align*}
{\color{black}From derived apriori estimates and constructed lower barrier for $\cp$, we have $|\cp^{-1}\nab^2_\gamma\cp|+|\cp^{-1}\nab_\gamma \cp|^2+1$ is independent of $u, \ub$ and it satisfies  
$$|c|+|c^i|\leq Ca^{-\f13}\O^{\tilde{\d}}(|\cp^{-1}\nab^2_\gamma\cp|+|\cp^{-1}\nab_\gamma \cp|^2+1).$$}
{\color{black}With $\O$ being sufficiently small {\color{black}depending on $u, \ub$} in the interior region,} we have  
$|c^i|+|c|\ll 1$. This implies
$\partial_{\cp} \mathcal S(\cp)[\psi_1]<0$ on $\mathbb S^2$.
Employing {\color{black}the calculation in the proof of} Theorem 2.11 in \cite{HL}, for any $w\in C^2(\mathbb S^2)$, {\color{black}we consider an elliptic equation for $\f{w}{\psi_1}$. For notional simplicity, let us rewrite the elliptic operator $\partial_{\cp} \mathcal S(\cp)$ as
$$\partial_{\cp} \mathcal S(\cp)[w]=a_{ij}D_{ij}w+b_i D_iw+cw$$
with $a_{ij}, b_i, c$ being functions. We then calculate and directly check that $\f{w}{\psi_1}$ satisfies 
\begin{equation}\label{Theorem 2.11 in HL}
a_{ij}D_{ij}\f{w}{\psi_1}+(b_i+\f{2}{\psi_1}a_{ij}D_j\psi_1)D_i\f{w}{\psi_1}+\f{\partial_{\cp} \mathcal S(\cp)[\psi_1]}{\psi_1}\f{w}{\psi_1}=\f{\partial_{\cp} \mathcal S(\cp)[w]}{\psi_1}. 
\end{equation}
Since $\partial_{\cp} \mathcal S(\cp)[\psi_1]<0$ and $\psi_1>0$, the above elliptic equation is invertible. For fixed $\psi_1>0$, giving $\partial_{\cp} \mathcal S(\cp)[w]/\psi_1$, we can solve for the unique solution $\f{w}{\psi_1}$. This further implies that $\partial_{\cp}\mathcal S(\cp): C^{2, \beta}(\mathbb{S}^2)\rightarrow C^{{0},\beta}(\mathbb{S}^2)$ is invertible and it holds
\begin{equation}\label{eq-estimate-linearization}
\begin{split}
|\f{w}{\psi_1}|_{C^{2,\beta}(\mathbb S^2)}\le C|\f{\partial_{\cp}\mathcal S(\cp)[w]}{\psi_1}|_{C^{0,\beta}(\mathbb S^2)} \quad \mbox{ and } \quad |w|_{C^{2,\beta}(\mathbb S^2)}\le C|\partial_{\cp}\mathcal S(\cp)[w]|_{C^{0,\beta}(\mathbb S^2)}
\end{split}
\end{equation}
with some positive constant $C$ independent of $w$ {\color{black}and depending on function $\cp$ and its associated eigenfunction $\psi_1$ }.
 }
\end{proof} 
 
In the next lemma, we eliminate such dependence {\color{black}on $\cp$}. 

\begin{lemma}\label{lemma-linearized-apriori-estimates} 
Assume $f$ and $\cp$ to be given nonnegative functions on $\mathbb S^2$ with \eqref{eq-assumption-lower-bound-f} satisfied and 
$|\cp|_{C^2(\mathbb S^2)}\le K$ with $K$ obeying \eqref{constant K}. Then, within the hyperbolic region, for any $w\in C^2(\mathbb S^2)$, we have that
 \begin{equation}\label{eq-estimate-linearization-v2}
\max_{\mathbb S^2}|w|\le C\max_{\mathbb S^2}|\partial_{\cp}\mathcal S(\cp)[w]|.\end{equation}
Here the positive constant $C$ depends only on {\color{black} $m$ and $\e$ in \eqref{eq-assumption-lower-bound-f} and $K$,  
independent of $\cp$ and $\ub\in (0,\delta]$.}
\end{lemma}

\begin{proof} 
We take a smooth nonnegative function $f_0=f_0(\o)$ on $\mathbb S^2$ satisfying $f_0\ge m/2$ on $B_p(\e)$ and $f(\cdot, \ub)\ge f_0$ on $\mathbb S^2$ for all $\ub\in (0,\delta]$.
With the standard spherical metric $\gamma_0$ on $\mathbb S^2$, we first consider 
$$L_*w=\D_{\gamma_0} w{\color{black}-\ms w-\f{1-\ms}{2}f_0 w}.$$
Employing Lemma \ref{lemma-H1-estimates-L0} to $L_*$, we can find a positive 
eigenfunction $\psi_1$  associated with its first eigenvalue $\mu_1>0$, i.e., 
$$\D_{\gamma_0} \psi_1{\color{black}-\ms\psi_1-\f{1-\ms}{2}f\psi_1}=-\mu_1\psi_1.$$
{\color{black} Here $\mu_1$ and $\psi_1$ depend only on $m$, $\e$ and are independent of $\ub\in (0,\delta]$.} With the fact {\color{black}$0<\cp\leq 1$}, {\color{black} $f\ge f_0$}, applying Lemma \ref{lemma-difference-Laplacians}, and Proposition \ref{form of linearization}, we then get 
\begin{align*}
\partial_{\cp} \mathcal S({\color{black}\cp})[\psi_1]
&=\Delta_{\gamma_0}\psi_1-\ms\psi_1-\f{1-\ms}{2}f\cp^{-\f{1}{\ms}}\psi_1+\Delta_{\gamma}\psi_1-\Delta_{\gamma_0}\psi_1+c^{i}\nab_i \psi_1+c\psi_1\\
&\le\Delta_{\gamma_0}\psi_1-\ms\psi_1-\f{1-\ms}{2}f_0\psi_1+\Delta_{\gamma}\psi_1-\Delta_{\gamma_0}\psi_1+c^{i}\nab_i \psi_1+c\psi_1\\
&=-\mu_1\psi_1+\Delta_{\gamma}\psi_1-\Delta_{\gamma_0}\psi_1+c^{i}\nab_i \psi_1+c\psi_1\le -{\color{black} c(m, \e, K)}<0.
\end{align*}
For the last inequality, we choose $a$ being suitably large in the exterior and $\O$ being sufficiently small in the interior. {\color{black}We then repeat the proof of Lemma \ref{lemma-linearized-op-invert} with this $\psi_1$, which is independent of $\cp$.} And we prove the current lemma. \end{proof}

In below we further prove that  the equation $\mathcal S(\cp)=0$ admits a unique solution $\cp$ satisfying 
$|\cp|_{C^2(\mathbb S^2)}\le K$.

\begin{proposition}\label{lemma-comparison}
Set $f$ and $\cp_1, \cp_2$ to be given nonnegative functions on $\mathbb S^2$ 
with \eqref{eq-assumption-lower-bound-f} satisfied. Require
$|\cp_1|_{C^2(\mathbb S^2)}\le K$ and $|\cp_2|_{C^2(\mathbb S^2)}\le K$ with $K$ obeying \eqref{constant K}.
If $\mathcal S(\cp_1)=\mathcal S(\cp_2)$, 
then it holds that $\cp_1\equiv\cp_2$ for all $\o\in\mathbb S^2$. 
\end{proposition}

\begin{proof} For the given $\cp_1$ and $\cp_2$, we define the below $\M L$ and view it as a linear operator of $\cp_1-\cp_2$. 
$$\mathcal L(\cp_1-\cp_2):=\mathcal S(\cp_1)-\mathcal S(\cp_2)=\int_0^1\partial_{\cp} \mathcal S(t\cp_1+(1-t)\cp_2)[\cp_1-\cp_2]dt.$$
The same reasoning as for Lemma 5.5 {\color{black}of \cite{AH}} yields the desired result. 

Here operator $\int_0^1 \partial_{\cp}\mathcal S(t\cp_1+(1-t)\cp_2)dt$ possesses a similar structure to the linearized operator $\partial_{\cp}\mathcal S(\cp)$ in Proposition \ref{form of linearization}. The corresponding Lemmas \ref{lemma-H1-estimates-L0}-\ref{lemma-linearized-apriori-estimates} can be obtained in the same manner. By the assumption, we first use $\mathcal S(\cp_1)\geq\mathcal S(\cp_2)$ and it gives $\mathcal L(\cp_1-\cp_2)\ge 0$. With the same proof as in Lemma \ref{lemma-linearized-op-invert}, for the operator $\M L$, applying Theorem 2.11 in \cite{HL}, the strong maximum principle holds here. Hence it holds either $\cp_1<\cp_2$ or $\cp_1\equiv \cp_2$ on $\mathbb{S}^2$. We then use $\mathcal S(\cp_2)\geq\mathcal S(\cp_1)$. Repeating the argument, we now have either $\cp_2<\cp_1$ or $\cp_1\equiv \cp_2$ on $\mathbb{S}^2$. Combining these together, we have {\color{black}proved} the desired uniqueness result. 
\end{proof}

As a corollary of the above lemma, we also prove that  

\begin{theorem}\label{theorem-existence-solutions}
Take $\ub$ to be a fixed constant in $(0,\delta]$. Let $f$ be a given smooth function on $\mathbb S^2$, 
satisfying \eqref{eq-assumption-f-0-1} and \eqref{eq-assumption-lower-bound-f}, 
and $\underline{\cp}$ and  $\overline{\cp}$ 
be constructed as in Section \ref{sec-Barriers}. Then, there exists a unique smooth solution ${\cp}=\cp(\o,\ub)$ of $\mathcal S(\cp, \ub)=0$, satisfying 
$\underline{\cp}< \cp(\cdot, \ub)<\overline{\cp}$ on $\mathbb S^2$. 
\end{theorem} 

\begin{proof} 

Here, for any constant $\lambda\in [0,1]$ and 
any $C^2$-function $\cp$ on $\mathbb S^2$, we set  
$$\mathcal S_\lambda(\cp)=\Delta_{\gamma}\cp-\ms\cp+\f{\ms}{2}\big[\lambda f+(1-\lambda)]\cp^{1-\f{1}{\ms}}+\lambda(I_2+I_3+I_4+I_5),$$
where $I_i$, with $2\le i\le 5$, are defined in \eqref{eq-expression-S-I-2-5}. {\color{black}With the same line of reasoning as for the detailed proof of Theorem 6.1 in \cite{AH}, by applying the method of continuity for parameter $\lambda\in[0,1]$, together with the utilization of the invertibility of  $\partial_{\cp}\mathcal S_\lambda(\cp)$, we conclude the existence. The uniqueness follows from Proposition \ref{lemma-comparison}.} 
\end{proof}

\section{The Anisotropic Apparent Horizon}\label{sec-existence-solutions}   

Considering $\ub$ as a parameter, in this section we show that our constructed MOTSs form a smooth apparent horizon. Viewing the solution $\cp$ in Theorem \ref{theorem-existence-solutions} as a function of $\o$ and $\ub$, we have

\begin{theorem}\label{thrm-smoothness-tubes}  
With the assumptions in Theorem \ref{theorem-existence-solutions}, 
for each $\ub \in (0,\delta]$, setting $\cp=\cp(\o, \ub)$ to be the solution 
to $\mathcal S(\cp, \ub)=0$ as in Theorem \ref{theorem-existence-solutions}, 
then it holds that the function $\cp=\cp(\o, \ub)$ is smooth with respect to both variables $(\o, \ub)\in \mathbb S^2\times(0,\d]$.
\end{theorem}

\begin{proof} We employ the method of continuity again as in Section \ref{linearization}. This time we take $\ub$ as a new parameter. Given $\cp(\o,\ub)$ being the constructed unique solution to $\mathcal S(\cp,\ub)=0$, we then apply the implicit function theorem to obtain a smooth solution $\cp(\o,\ub')$ to $\mathcal S(\cp,\ub')=0$ with $\ub'$ close to $\ub$. The key is the invertibility of $\partial_{\cp}\mathcal{S}(\cp,\ub): C^{2,\beta}(\mathbb{S}^2)\rightarrow C^{0,\beta}(\mathbb{S}^2)$, which was proved in Lemma \ref{lemma-linearized-op-invert}. Hence, $\cp(\o, \ub)$ is smooth with respect to both $\o$ and $\ub$ variables. 
 \end{proof}

Applying {\color{black}a null comparison principle and an upcoming theorem in \cite{AnHe} by the author and He}, with double null foliations we show that the apparent horizon formed in gravitational collapse is always either null or spacelike. 

In this section, we specify a condition of $f$ to guarantee the spacelikeness of our constructed apparent horizon $\{u=-(\ub a)^{1-\ms}\cp^{\f{1-\ms}{\ms}}\}$. To do so, we derive the equation for $\partial_{\ub}\cp$ {\color{black} and to obtain its estimate.  

\begin{proposition}\label{prop-equation-partial-ub-phi}  
Let $\cp=\cp(\o, \ub)$ be the solution 
to $\mathcal S(\cp, \ub)=0$ as in Theorem \ref{theorem-existence-solutions}. For each $\ub \in (0,\delta]$,  $\partial_{\ub}\cp$ obeys
\begin{equation}\label{eq-equation-partial-ub-phi}
\Delta_{\gamma}(\partial_{\ub}\cp)-\ms\partial_{\ub}\cp-\f{1-\ms}{2}f(\o, \ub)\cp^{-\f{1}{\ms}}\partial_{\ub}\cp+c^{i}\nab_i (\partial_{\ub}\cp)+c\partial_{\ub}\cp=\frac{1}{\ub}h
+\f{\ms}{2}\partial_{\ub}f\cp^{1-\f{1}{\ms}}.
\end{equation}
Here $c=c(\o,\cp,\ub)$, $c^i=c^i(\o,\cp,\ub)$ and $h=h(\o,\cp,\ub)$
are functions satisfying 
\begin{equation}\label{eq-estimates-partial-ub-coefficients}
|c|+|c^i|+|h|\leq Ca^{-\f13}\cdot{\color{black}\O^{\t\d}}\cdot {\color{black}\cp^{-\f{1}{\ms}}}{\color{black} (|\nabla^2_{\gamma}\cp|+|\nabla_\gamma\cp|^2+1)}.\end{equation}
\end{proposition} 

\begin{proof} 
The main part of this proof is similar to the corresponding one in \cite{AH} and hence is omitted. Only noting that, in the estimates for here we employ
$$\alpha_1:=-[\f12\O\tr_g\chib+\f{1}{R}]-4[\O\omb-(\O\omb)_0],$$
\begin{align*}
\alpha_2:=&-\f12\int_0^{\ub}[(\f12\tr\chi\tr\chi+|\chih|_g^2+\partial_4\phi\partial_4\phi)(u, \ub', \o)]d\ub'\\
&+\f12\int_0^{\ub}\f{\O^{-2}(u, \ub', \o)}{|u|^2}\bigg(|\O\chih|_g^2(-1, \ub', \o)+[\O e_4\phi(-1, \ub', \o)-\O e_4\phi(-1, 0, \o)]^2 \bigg)d\ub'.
\end{align*}
The desired estimates for $\partial_u\alpha_1$ and $\partial_u \alpha_2$ were already obtained in Proposition \ref{form of linearization}. To bound $\partial_{\ub}\alpha_1$ and $\partial_{\ub}\alpha_2$, we utilize the explicit formulas for $\partial_{\ub}[\f12\O\tr_g\chib+\f{1}{R}]$ and $\partial_{\ub}(\O\omb)$, the estimates for $\partial_{\ub}\alpha_1$ can be obtained straightforwardly. And for $\partial_{\ub}\alpha_2$, we have
\begin{align*}
\partial_{\ub}\alpha_2(u, \ub, \o)=&-\f12\times(\f12\tr\chi\tr\chi+|\chih|_g^2+\partial_4\phi\partial_4\phi)(u, \ub, \o)\\
&+\f12\times\f{\O^{-2}(u, \ub, \o)}{|u|^2}\bigg(|\O\chih|_g^2(-1, \ub, \o)+[\O e_4\phi(-1, \ub, \o)-\O e_4\phi(-1, 0, \o)]^2 \bigg).
\end{align*}
The desired estimate for above expression can be got in the same fashion as for \eqref{e4 difference}.  
\end{proof}

Proceeding similarly to the argument for Corollary 6.4 in \cite{AH}, we then obtain

\begin{corollary}\label{cor-equation-partial-ub-phi}
For each $\ub \in (0,\delta]$, let $\cp=\cp(\o, \ub)$ be the solution 
to $\mathcal S(\cp, \ub)=0$. Under the assumptions in Theorem \ref{theorem-existence-solutions}, it holds
\begin{equation}\label{eq-estimate-partial-ub-phi}
\max_{\mathbb S^2}|\partial_{\ub}\cp|\le C\big[\ub^{-1}a^{-\f13}{\color{black}\O^{\t\d}} \max_{\mathbb S^2}
({\color{black}|\cp|+|\nab_{\gamma}\cp|+|\nab^2_{ij}\cp|+|\cp^{-1}\nab_{\gamma}\cp\nab_{\gamma}\cp |})+\max_{\mathbb S^2}|\partial_{\ub}f{\color{black}\cp^{1-\f{1}{\ms}}}|\big],
\end{equation}
where $C$ is a positive constant independent of $a$, $\ub$ and $\cp$. 
\end{corollary}

This further implies 

\begin{proposition}\label{thrm-spacelike} 
For each $\ub \in (0,\delta]$, set $\cp=\cp(\o, \ub)$ to be the solution 
of $\mathcal S(\cp, \ub)=0$. Under the assumptions in Theorem \ref{theorem-existence-solutions}, if we further require
\begin{equation}\label{eq-assumption-ub-f}
|\ub\partial_{\ub}f(\o,\ub)|\leq a^{-\f13}{\color{black}\O^{\t\d}(u_1)}, 
\end{equation} 
the apparent horizon $\{(\o, u, \ub) |\, u=-(\ub a)^{1-\ms}\cp^{\lambda}(\o,\ub)\}$ is spacelike. 
\end{proposition}

\begin{proof} 
Using $R=(\ub a)^{1-\ms}\cp^{\lambda}(\o,\ub)$ and $\lambda=\f{1-\ms}{\ms}$, we have $\partial_{\ub}R=\ub^{-1}R[(1-\ms)+\lambda\ub\cp^{-1}\partial_{\ub}\cp]$.
Applying Corollary \ref{cor-equation-partial-ub-phi}, if $|\ub \partial_{\ub}f(\o,\ub)|\leq a^{-\f13}\O^{{\color{black}\t\delta}}$, then we get $\partial_{\ub}R\geq \f{1-\ms}{2}\ub^{-1}R$. 
Assume $g'$ is the induced metric to the apparent horizon
$u=-R(\o, \ub)$. Rewrite $\o\in\mathbb{S}^2$ as $\o=(\theta_1, \theta_2)$. As shown in \cite{An17}, then it holds
\begin{align*}
g'_{\theta_i \theta_j}=g_{\theta_i \theta_j},\quad
g'_{\theta_i \ub}=2\partial_{\theta_i}R,\quad
g'_{\ub\, \ub}=4\partial_{\ub} R.
\end{align*}

Suppose $X=\lambda_1\partial_{\theta_1}+\lambda_2\partial_{\theta_2}+\lambda_3\partial_{\ub}$ to be an arbitrary 
nonzero tangent vector along the apparent horizon with constants $\lambda_1, \lambda_2, \lambda_3$. It holds
\begin{align*}
g'(X,X)
=\lambda_i\lambda_jg_{ij}+4\lambda_i \lambda_3 \partial_{i}R+4\lambda_3^2 \partial_{\ub} R,\end{align*}
{\color{black} where $i,j=1,2$ are summed over.}
The case $\lambda_3=0$ implies $g'(X,X)
=\lambda_i\lambda_jg_{ij}>0$. For the case $\lambda_3\neq0$, using  $g_{ij}=R^2\gamma_{ij}$ and $\partial_i R=\lambda R\cp^{-1}\nab_i\cp$, we get
\begin{equation}\label{gXX}
\begin{split}
g'(X,X)=&\lambda_i \lambda_j g_{ij}+4\lambda_i \lambda_3 \partial_i R+4\lambda_3^2 \partial_{\ub}R \\
\geq& R^2 \lambda_i\lambda_j \gamma_{ij}+4\lambda_i\lambda_3\lambda R \cp^{-1}\nab_i\cp+4\lambda_3^2\cdot\f{1-\ms}{2}\ub^{-1}R\\
\geq&2\lambda_3^2(1-\ms)\ub^{-1}R-4\lambda_3^2\lambda^2(\cp^{-1}\nab_i\cp)^2\\
>&2\lambda_3^2(1-\ms)\ub^{-\ms}(a m \epsilon^{\tau+2})^{1-\ms}-4\lambda_3\lambda^2(\cp^{-1}\nab_i\cp)^2.
\end{split}
\end{equation}
For the last inequality, we use $\bar{\tp}_{max}=-\log(m\epsilon^{\tau+2})+C(\tau)$ in \eqref{bartildephimax} and hence 
$$R\geq (\ub a\cdot m\epsilon^{\tau+2})^{1-\ms} \mbox{ with } \tau>2.2 \mbox{ fixed}.$$
At the same time, by \eqref{subsolution} we have $u=-(\f23 \ub a)^{1-\ms}$ is an untrapped surface.  For the constructed MOTS locating at $u=-R=-(\ub a)^{1-\ms}\cp^{\f{1-\ms}{\ms}}$, it thus holds
$\cp^{\f{1-\ms}{\ms}} \geq (\f{2}{3})^{1-\ms}$. That is $\cp\geq (\f23)^{\ms}$.  
Therefore, back to \eqref{gXX}, {\color{black}when $\ub$ is sufficiently small}, we have $g'(X, X)>0$. This renders our constructed apparent horizon spacelike  
and hence a dynamical horizon.  
\end{proof}  

\section{The Instability Theorems}\label{section instability theorems}
{\color{black}
In this section, we prove nonlinear instability theorems with Christodoulou's naked-singularity initial data prescribed along $\ub=0$.\footnote{\label{instability theorems perturbed initial data}{\color{black}For perturbed initial data, analogous instability theorems can be proved correspondingly.}}
}
\subsection{Instability Subject to {\color{black}Perturbations with Finite BV Norms}}\label{BV instability} 

We first recall equation \eqref{e3e4phi v2}
\begin{equation*}
(\O e_3)(\O e_4\phi)+\f12{\O\tr\chib}\O e_4\phi=\O^2\Delta_g \phi-\f12{\O\tr\chi}\O e_3\phi{\color{black}+2\O^2\eta^A e_A\phi}.
\end{equation*}
With $e_3=\O^{-1}{\partial}/{\partial u}$, this equation is equivalent to
\begin{equation*}
\begin{split}
\f{\partial}{\partial u}(u\O e_4\phi)+\f12(\O\tr\chib+\f{2}{|u|})u\O e_4\phi=u\O^2\Delta_g\phi-\f12 u\O\tr\chi\cdot\O e_3\phi+2u\O^2\eta^A e_A\phi.
\end{split}
\end{equation*}
Via integration, we hence obtain
\begin{equation}\label{e3e4phi integration} 
\begin{split}
u\O e_4\phi(u,\ub, \o)=&u\O e_4\phi(-1,\ub, \o)+\int_{-1}^u -\f12u'\O\tr\chi\cdot\O e_3\phi(u',\ub, \o)du'\\
&+\int_{-1}^u [-\f12(\O\tr\chib+\f{2}{|u|})\cdot u\O e_4\phi+u\O^2\Delta_g\phi+2u\O^2\eta^A e_A\phi](u',\ub, \o)du'.
\end{split}
\end{equation}
For the first two terms on the right of \eqref{e3e4phi integration}, we can rewrite them as
\begin{equation*}
\begin{split}
&u\O e_4\phi(-1,\ub, \o)+\int_{-1}^u -\f12u'\O\tr\chi\cdot\O e_3\phi(u', \ub, \o)du'\\
=&u\O e_4\phi(-1, \ub, \o)-u\O e_4\phi(-1, 0, \o)+\underline{u\O e_4\phi(-1, 0, \o)}-\f12\int_{-1}^u u'\O\tr\chi\cdot [\O e_3\phi-(\O e_3\phi)_0](u', \ub, \o)du'\\
&-\f12\int_{-1}^u u'[\O\tr\chi-(\O\tr\chi)_0]\cdot(\O e_3\phi)_0(u',\ub, \o)du' \underline{-\f12\int_{-1}^u u'(\O\tr\chi)_0\cdot(\O e_3\phi)_0(u', \ub, \o)du'}. 
\end{split}
\end{equation*}
Employing \eqref{e3e4phi integration}, with Christodoulou's initial data along $\ub=0$ we have
$$u\O e_4\phi(u, 0, \o)=u\O e_4\phi(-1, 0,\o)+\int_{-1}^u-\f12u'(\O\tr\chi)_0\cdot(\O e_3\phi)_0(u', 0, \o)du'.$$
Hence, the underlined two terms can be replaced by $u\O e_4\phi(u, 0, \o)$. Back to \eqref{e3e4phi integration} again, we now obtain
\begin{equation*}
\begin{split}
&u\O e_4\phi(u, \ub, \o)-[u\O e_4\phi(-1, \ub, \o)-u\O e_4\phi(-1, 0, \o)]\\
=&u\O e_4\phi(u, 0, \o)-\f12\int_{-1}^u u'\O\tr\chi[\O e_3\phi-(\O e_3\phi)_0](u', \ub, \o)du'\\
&-\f12\int_{-1}^u u'[\O\tr\chi-(\O\tr\chi)_0]\cdot(\O e_3\phi)_0(u', \ub, \o)du'\\
&+\int_{-1}^u [-\f12(\O\tr\chib+\f{2}{|u|})\cdot u\O e_4\phi+u\O^2\Delta_g\phi+2u\O^2\eta^A e_A\phi](u',\ub, \o)du'.
\end{split}
\end{equation*}
Utilizing the fact 
$$u\O e_4\phi(u, 0, \o)=-\sqrt{2}\sqrt{\f{1-\ms}{\ms}}\cdot\f{1-\ms}{4}\cdot|u|^{\f{\ms}{1-\ms}}$$
and the hyperbolic estimates derived in previous sections, we get
\begin{equation}\label{e4phi difference with initial data}
|u\O e_4\phi(u, \ub, \o)-[\O e_4\phi(-1, \ub, \o)-\O e_4\phi(-1, 0, \o)]|\lesssim \f{\ub a}{|u|}+|u|^{\f{\ms}{1-\ms}}.
\end{equation}

Now we consider the key equation that drives the trapped surface formation
$$\O \nab_4(\O^{-1}\tr\chi-(\O^{-1}\tr\chi)_0)=-|\chih|^2-(\nab_4\phi)^2-\f12(\tr\chi)^2.$$ 
With this equation, we have 
\begin{equation}\label{MOTS formation integration}
\begin{split}
&\O^{-1}\tr\chi(u, \ub, \o)-(\O^{-1}\tr\chi)(u, 0, \o)\\
=&\int_0^{\ub}-|\chih|^2(u, \ub', \o)d\ub'+\int_0^{\ub}-|\nab_4\phi|^2(u, \ub', \o)d\ub'+\int_0^{\ub}-\f12(\tr\chi)^2(u, \ub', \o)d\ub'.
\end{split}
\end{equation}
We then find the main terms on the right. By hyperbolic estimates, we first have
$$|\int_0^{\ub}-|\chih|^2(u, \ub', \o)d\ub'-\int_0^{\ub}-\f{|\O\chih|^2(-1, \ub', \o)}{|u|^2\O^2}d\ub'|\leq \f{\ub a}{|u|^2\O^2}\cdot\f{\ub\at}{|u|}.$$
Using \eqref{e4phi difference with initial data}, we also get 
\begin{equation*}
\begin{split}
&|\int_0^{\ub}-|\nab_4\phi|^2(u, \ub', \o)d\ub'+\f{1}{|u|^2\O^2}\int_0^{\ub}[\O e_4\phi(-1, \ub' ,\o)-\O e_4\phi(-1, 0, \o)]^2du'|\\
\leq&\int_0^{\ub}\f{1}{\O^2}\cdot\f{1}{|u|^2}\cdot(\f{\ub' a}{|u|}+|u|^{\f{\ms}{1-\ms}})\cdot\at d\ub' \leq\f{1}{|u|^2\O^2}\cdot(\f{\ub^2 a^{\f32}}{|u|}+|u|^{\f{\ms}{1-\ms}}\cdot\ub\at)\\
\leq& \f{\ub a}{|u|^2\O^2}\cdot(\f{\ub\at}{|u|}+|u|^{\f{\ms}{1-\ms}}a^{-\f12}).
\end{split}
\end{equation*}
For the last term, via the obtained estimate in Proposition \ref{trch.bd}
$$|\O\tr\chi-\f{2(1-\ms)\O^2}{|u|}|\lesssim\f{\ub a}{|u|^2},$$
we deduce
\begin{equation*}
\begin{split}
|\int_0^{\ub}-\f12(\tr\chi)^2(u, \ub', \o)d\ub'|\lesssim&\int_0^{\ub}\f{1}{\O^2}\cdot(\O\tr\chi)^2(u, \ub' ,\o)d\ub'\\
\leq&\int_0^{\ub}\f{1}{\O^2}(\f{2(1-\ms)\O^2}{|u|}+\f{\ub' a}{|u|^2})^2 (u, \ub' ,\o) d\ub'\\
\lesssim& \f{\ub}{|u|^2\O^2}\O^4(u, \ub, \o)+\f{\ub a}{|u|^2\O^2}\cdot\f{\ub^2 a}{|u|^2}(u, \ub, \o).
\end{split}
\end{equation*}
Note that the MOTS forms in the interior region, where $\O(u, \ub, \o)\lesssim \O(u, 0)\leq \O(u_1, 0)\ll1$. With these estimates, we spot the main contribution on the right of \eqref{MOTS formation integration}, that is
$$-\int_0^{\ub}\f{|\O\chih|^2(-1, \ub', \o)}{|u|^2\O^2}d\ub'-\f{1}{|u|^2\O^2}\int_0^{\ub}[\O e_4\phi(-1, \ub' ,\o)-\O e_4\phi(-1, 0, \o)]^2du'.$$
We then define $f(\o, \ub)$ via requiring 
\begin{equation}\label{f definition}
\begin{split}
&\f{\ub a f(\o, \ub)}{|u|^2\O^2(u, 0)}\cdot(1-\ms)\\
:=&\int_0^{\ub}\f{|\O\chih|^2(-1, \ub', \o)+[\O e_4\phi(-1, \ub', \o)-\O e_4\phi(-1, 0, \o)]^2}{|u|^2\O^2(u, 0)} d\ub'.
\end{split}
\end{equation}
With this defined $f(\o, \ub)$, combining all the previous arguments above, we arrive at our first main theorem 

\begin{theorem} \label{main thm section 14} 
Consider the characteristic initial value problem for the Einstein-scalar field system \eqref{1.1}. Assigning Christodoulou's naked-singularity initial data in \cite{Chr.2} along  $\ub=0$ with $-1\leq u \leq 0$ and prescribing perturbed initial data along $u=-1$ satisfying
\begin{equation}\label{upper bound main thm a=1 section 14}
\sum_{i\leq 5, j\leq 3}\ub^{j}a^{-\f12} \|(\partial_{\ub})^j\nab^{i}\chih\|_{L^{\infty}_{\ub}L^2(S_{-1,\ub})}+\ub^{j}a^{-\f12} \|(\partial_{\ub})^j\nab^{i}\nab_4\phi\|_{L^{\infty}_{\ub}L^2(S_{-1,\ub})}\leq B, 
\end{equation}
then for each $B$ there exist a sufficient small $\delta=\delta(B)$ and constant $a=a(B)$\footnote{In particular, we can choose $a=1$.}, and the Einstein-scalar field system admits a unique regular solution in the spacetime region $(u,\ub, \o)$, where $0\leq\ub\leq \delta$ and $\ub\leq |u|\O^{2-\t\delta}(u,0)\leq 1$ with $0<\t\d\ll1$.

Moreover, if we further assume  
\begin{equation}\label{main thm f def}
 \int_0^{\ub}|\chih|^2 (-1, \ub', \o)+[\nab_4\phi(-1, \ub', \o)-\nab_4\phi(-1, 0, \o)]^2 d\ub'=f(\o, \ub)\ub a
\end{equation}
with $f(\o,\ub)$ obeying  
\begin{equation}\label{main thm f lower bound}
0\leq f\leq 1 \mbox{ on } \mathbb S^2\times (0,\delta] \quad \mbox{ and } \quad f(\cdot, \ub)\ge m \mbox{ on } B_{p}(\e)  
\end{equation}
for a point  $p\in \mathbb S^2$  and constants  $m\in (0,1)$,  $\e\in (0,\pi/2)$, then within the solved hyperbolic region, there \textit{exists a unique} MOTS $M_{\ub}$ on each $\Hb_{\ub}$ with $0< \ub \leq \d$, and the collection of MOTS $\{M_{\ub}\}$ emerges and censors the central singularity. And they form an achronal apparent horizon.
\end{theorem}
Set $a=1$, we then get Theorem \ref{main thm}.   \\ 

Note that for $\ub<0$ we require our spacetime to coincide with Christodoulou's naked-singularity solution in \cite{Chr.2}. Along $u=-1$ we consider the possible extensions of $\chih(-1, \ub, \o), \partial_{\ub}\phi(-1, \ub, \o)$ for $\ub>0$. Via the next theorem, we can see that Christodoulou's naked-singularity solution is associated with nonlinear instability of high co-dimensions: 

\begin{theorem}\label{co dimension instability theorem v1}
Consider the space of prescribed initial data for $\chih(-1, \ub, \o)$ and $\partial_{\ub}\phi(-1, \ub, \o)$ along $u=-1$ with $\ub> 0$. With $|\chih|(-1, 0, \o)=0, \,\, \partial_{\ub}\phi(-1, 0, \o)=\sqrt{2}\cdot\sqrt{\f{1-\ms}{\ms}}\cdot\f{1-\ms}{4}$ corresponding to their values of Christodoulou's naked-singularity solution in \cite{Chr.2}, then for any $k\in \mathbb{Z}^+$, we can find smooth functions $f_1, f_2, ... , f_k$ and $g_1, g_2, ... , g_k$ with variables $\ub$ and $\o$, such that the below prescribed initial data for $|\chih|$ and $\partial_{\ub}\phi$ with $\ub>0$
\begin{equation*}
\begin{split}
&\bigg(|\chih|(-1, \ub, \o), \quad \partial_{\ub}\phi(-1, \ub, \o)\bigg)\\
=&\bigg(0+\lambda_1 f_1+\lambda_2 f_2+\cdot\cdot\cdot+\lambda_k f_k, \quad \partial_{\ub}\phi(-1, 0, \o)+\lambda_1' g_1+\lambda_2' g_2+\cdot\cdot\cdot+\lambda_k' g_k\bigg)
\end{split}
\end{equation*}
would lead to the trapped surface and the apparent horizon formation, when 
$$\sum_{1}^k \lambda_i^2+\sum_1^k {\lambda_i'}^2\neq 0 \quad \mbox{ with } \quad \lambda_k, \lambda_k'\in \mathbb{R}.$$ 

Moreover, if it holds
\begin{equation*}
\begin{split}
&\bigg(0+\lambda_1 f_1+\lambda_2 f_2+\cdot\cdot\cdot+\lambda_k f_k, \quad \partial_{\ub}\phi(-1, 0, \o)+\lambda_1' g_1+\lambda_2' g_2+\cdot\cdot\cdot+\lambda_k' g_k\bigg)\\
=&\bigg(0+\t\lambda_1 f_1+\t\lambda_2 f_2+\cdot\cdot\cdot+\t\lambda_k f_k, \quad \partial_{\ub}\phi(-1, 0, \o)+\t\lambda_1' g_1+\t\lambda_2' g_2+\cdot\cdot\cdot+\t\lambda_k' g_k\bigg),
\end{split}
\end{equation*}
then it renders
$$\lambda_i=\t\lambda_i \quad \mbox{ and } \quad \lambda_i'=\t\lambda_i \quad \mbox{ with } \quad 1\leq i \leq k.$$
We may therefore say that, for any $k\in\mathbb{Z}^+$, Christodoulou's naked-singularity solution in \cite{Chr.2} has at least co-dimensional $2k$ nonlinear instability subject to outgoing BV characteristic perturbations at $(-1, 0, \o)$. 
\end{theorem} 
 
\begin{proof}
For any $k\in\mathbb{Z}^+$ and $0\leq\ub \ll 1$,  on the $2$-sphere $S_{-1, \ub}$ we first fix $2k$ disjoint discs: $D_1, D_2, ..., D_k$ and $D_1', D_2', ..., D_k'$. 
We prescribe smooth functions $f_i$ being compactly supported in $D_i$ and smooth functions $g_i$ being compactly supported in $D_i'$. For $1\leq i \leq k$, we further require
$$0\leq f_i\leq \at \mbox{ on } \mathbb S^2\times (0,\delta], \quad f_i(\cdot, \ub)\ge m_i\at \mbox{ on } B_{p_i}(\e_i)\subset D_i ,$$
$$0\leq g_i\leq \at \mbox{ on } \mathbb S^2\times (0,\delta], \quad  g_i(\cdot, \ub)\ge m_i'\at \mbox{ on } B_{p_i'}(\e_i')\subset D_i'$$  
with points  $p_i\in D_i\subset \mathbb S^2$, $p_i'\in D_i'\subset \mathbb S^2$ and constants  $m_i, m_i'\in (0,1)$, $\e_i, \e_i'\in (0,\pi/2)$. With these choices of initial data, conditions \eqref{main thm f def} and \eqref{main thm f lower bound} in Theorem \ref{main thm section 14} are satisfied and the corresponding conclusions for trapped surface and apparent horizon formation hold. And if we have 
\begin{equation*}
\begin{split}
&\bigg(0+\lambda_1 f_1+\lambda_2 f_2+\cdot\cdot\cdot+\lambda_k f_k, \quad \partial_{\ub}\phi(-1, 0, \o)+\lambda_1' g_1+\lambda_2' g_2+\cdot\cdot\cdot+\lambda_k' g_k\bigg)\\
=&\bigg(0+\t\lambda_1 f_1+\t\lambda_2 f_2+\cdot\cdot\cdot+\t\lambda_k f_k, \quad \partial_{\ub}\phi(-1, 0, \o)+\t\lambda_1' g_1+\t\lambda_2' g_2+\cdot\cdot\cdot+\t\lambda_k' g_k\bigg),
\end{split}
\end{equation*}
since all the $f_i$ and $g_i$ with $1\leq i \leq k$ are compacted supported in disjoint discs, it immediately implies 
$$\lambda_i=\t\lambda_i \quad \mbox{ and } \quad \lambda_i'=\t\lambda_i \quad \mbox{ with } \quad 1\leq i \leq k.$$

\end{proof}

\subsection{Instability Subject to {\color{black}Perturbations with Finite $C^0$ Norms}}\label{C1 instability} 

The above two theorems in Section \ref{BV instability} generalize both the results in \cite{AL} by An-Luk and in \cite{AH} by An-Han with the requirement: 
$$f_i(\o, \ub)\ge m_i\at \mbox{ on } B_{p_i}(\e_i)\subset D_i, \quad g_i(\o, \ub)\ge m_i'\at \mbox{ on } B_{p_i'}(\e_i')\subset D_i' \quad \mbox{ for} \quad \ub> 0,$$
and
$$f_i(\o, \ub)=0, \quad  g_i(\o, \ub)=0 \quad \mbox{ for } \quad \ub\leq 0.$$
In this subsection, we extend the instability theorems in Section \ref{BV instability} further by allowing choosing continuous functions $f_i(\o, \ub)$ and $g_i(\o, \ub)$ for the $\ub$ variable.  

We start from reconsidering
\begin{equation*}
\begin{split}
&\O^{-1}\tr\chi(u, \ub, \o)-\f{2(1-\ms)}{|u|}\\
=&\int_0^{\ub}-|\chih|^2(u, \ub', \o)d\ub'+\int_0^{\ub}-|\nab_4\phi|^2(u, \ub', \o)d\ub'+\int_0^{\ub}-\f12(\tr\chi)^2(u, \ub', \o)d\ub'.
\end{split}
\end{equation*}
By the estimates derived above, we have
\begin{equation*}
\begin{split}
&|\O^{-1}\tr\chi(u, \ub, \o)-\f{2(1-\ms)}{|u|}-\int_0^{\ub}\f{|\O\chih|^2(-1, \ub', \o)+[\O e_4\phi(-1, \ub', \o)-\O e_4\phi(-1, 0, \o)]^2}{|u|^2\O^2(u, 0, \o)} d\ub'|\\
\leq&\f{\ub a}{|u|^2\O^2}\cdot\f{\ub\at}{|u|}+\f{\ub a}{|u|^2\O^2}\cdot(\f{\ub\at}{|u|}+|u|^{\f{\ms}{1-\ms}}a^{-\f12})+\f{\ub}{|u|^2\O^2}\O^4(u, \ub, \o)+\f{\ub a}{|u|^2\O^2}\cdot\f{\ub^2 a}{|u|^2}(u, \ub, \o).
\end{split}
\end{equation*}
Together with $\O(u,0)=|u|^{\f{\ms}{2(1-\ms)}}$, this implies
\begin{equation*}
\begin{split}
&|\O^{-1}\tr\chi(u, \ub, \o)-\f{2(1-\ms)}{|u|}-\int_0^{\ub}\f{|\O\chih|^2(-1, \ub', \o)+[\O e_4\phi(-1, \ub', \o)-\O e_4\phi(-1, 0, \o)]^2}{|u|^2\O^2(u, 0, \o)} d\ub'|\\
\leq&\f{\ub a}{|u|^2\O^2}\cdot \O^{\f32} \cdot [\f{\ub\at}{|u|\O^{\f32}}+\f{\O^{\f12}}{\at}+\f{\O^{\f52}}{a}]\ll \f{\ub a}{|u|^2\O^2}\cdot \O^{\f32}.
\end{split}
\end{equation*} 
 
\vfill 
 
For the $f(\o, \ub)$ introduced in \eqref{main thm f def}, with any constant $C_1>2$ and $\Lambda$ defined in Proposition \ref{gamma}, we further set
\begin{equation}\label{f tilde def}
f(\o, \ub)=g(\ub) \t f(\o, \ub) \quad \mbox{ with } \quad g(\ub)=(C_1\Lambda+C_1)^{\f12}\cdot [\ln (\ln\f{1}{\ub})]^{-\f12}
\end{equation}
and $\t f(\o, \ub)$ to satisfy  
\begin{equation}\label{main thm f bound v2}
0\leq \t f \leq 1 \mbox{ on } \mathbb{S}^2\times (0, \delta] \quad \mbox{ and } \quad \t f(\cdot, \ub){\color{black}\geq} m \mbox{ on } B_p(\epsilon)
\end{equation}  
for a point $p\in \mathbb{S}^2$ and constants $m\in (0,1)$, $\epsilon\in(0, \pi/2)$. With this choice of $g(\ub)$, we let
$$\kappa:=\ln[\f{3}{2}g(\ub)^{-1}]=\ln\bigg(\f32\{\ln[(\ln\f{1} {\ub})^{\f{1}{C_1\Lambda+C_1}}]\}^{\f12}\bigg).$$
Now we have $e^{\kappa}=3\{\ln[(\ln{1}/{\ub})^{\f{1}{C_1\Lambda+C_1}}]\}^{\f12}/2$. And when $0\leq\ub\ll 1$, it holds
$$8\Lambda \kappa e^{\kappa}\leq e^{2\kappa}=\ln[(\ln\f{1}{\ub})^{\f{1}{C_1\Lambda+C_1}}]\leq \ln [l\ln(\ub^{-1})] = \ln\ln(\ub^{-l}) \quad \mbox{ with } \quad 0<l\ll1. $$
For the $C(\kappa)$ defined in \eqref{C kappa}, we then obtain 
$$\ln [C(\kappa)]=C e^{4\Lambda \kappa e^{\kappa}}\ln (C e^{\kappa})\leq e^{8\Lambda \kappa e^{\kappa}}\leq \ln(\ub^{-l}).$$
Hence it holds that the $C(\kappa)$ defined in \eqref{C kappa} obeys
$$C(\kappa):=(C e^{\kappa})^{C e^{4\Lambda \kappa e^{\kappa}}}\leq \ub^{-l} \quad \mbox{ with } \quad 0<l\ll1.$$
If we require 
\begin{equation}\label{ub delta l}
\ub^{-l} \leq |u|^{\f{-\t\d\ms}{2(1-\ms)}} \iff \ub\geq |u|^{\f{\t\d \ms}{2l(1-\ms)}},
\end{equation}
we immediately have
\begin{equation}\label{kappa and Omega}
C e^{\kappa}\ll C(\kappa)=(C e^{\kappa})^{C e^{4\Lambda \kappa e^{\kappa}}}\leq \ub^{-l}\leq |u|^{\f{-\t\d \ms}{2(1-\ms)}}=\O(0,u)^{-\t\delta}.
\end{equation} 

With new conditions \eqref{f tilde def} and \eqref{main thm f bound v2} we then revisit our detailed constructions of the anisotropic trapped and untrapped surfaces in Section 3 of \cite{AH}. We replace $f(\o, \ub)$ by $g(\ub) \t f(\o, \ub)$. Consider the below main operator defined in (3.1) of \cite{AH}
$$\mathcal{S}(\phi)=\Delta_{\gamma}\phi+1-\f12 f e^{\phi}+F_1(\phi)+F_2(\phi)$$
\begin{equation*}
\begin{split}
\mbox{with } F_1(\phi):=&-4\O\omb\cdot \ub a e^{-\phi}|\nab_{\gamma}\phi|^2,\\
F_2(\phi):=&{\color{black}2R^{-1}\O\chibh_{kl}\gamma^{ik}\gamma^{jl}\nab_i\phi\nab_j\phi}+2\gamma(\eta, \nab_{\gamma}\phi)+[-\f12\O \tr_g\chib-\f{1}{\ub a e^{-\phi}}]\cdot\ub a e^{-\phi}|\nab_{\gamma}\phi|^2\\
&+\ub a e^{-\phi}\cdot[\f12\O^{-1}\tr_g\chi-\f{1}{\ub a e^{-\phi}}+\f{\ub a f(\o, \ub)}{2(\ub a)^2 e^{-2\phi}}].
\end{split}
\end{equation*}
Replacing $f$ by $g(\ub)\t f$, by the hyperbolic estimates in this article, we have
\begin{equation}\label{operator S phi}
\mathcal{S}(\phi)=\Delta_{\gamma}\phi+1-\f12 g(\ub) \t f e^{\phi}+F_1(\phi)+F_2(\phi) \mbox{ with } |F_2(\phi)|\leq C e^{\phi}a^{-\f13}\O^{\t\d}(|\nab_{\gamma}\phi|^2+1). 
\end{equation} 
Recalling
$$\kappa=\ln[\f{3}{2}g(\ub)^{-1}]=\ln\bigg(\f32\{\ln[(\ln\f{1}{\ub})^{\f{1}{C_1\Lambda+C_1}}]\}^{\f12}\bigg),$$
and using \eqref{operator S phi}, we first check that 
$$\mathcal{S}(\kappa)=1-\f12g(\ub)\cdot\t f\cdot \f32 g(\ub)^{-1}+F_1(\kappa)+F_2(\kappa)\geq 1-\f34-Ce^{\kappa}a^{-\f13}\O(0, u)^{-\t\d}>\f18,$$
where we use $F_1(\kappa)=0$ and \eqref{kappa and Omega}. We also revisit the construction of $\bar{\tilde{\phi}}$ in Section 3 of \cite{AH} and with $f$ replaced by $g(\ub)\t f$ we have that 
$$|\bar{\tilde{\phi}}|_{C^0(\mathbb{S}^2)}\leq -\log(g(\ub)m\epsilon^{\tau+2}), \quad |\nab_{\gamma_0}\bar{\tilde{\phi}}|\leq \epsilon^{-1}, \quad |\nab^2_{\gamma_0}\bar{\tilde{\phi}}|\leq \epsilon^{-2}$$
 with $\tau$ being a fixed number larger than $2.2$. Back to \eqref{eq-choice-a-b}, if requiring \eqref{ub delta l}, we then verify
\begin{equation}\label{K requirement v2}
K=\max\{(Ce^{\kappa})^{Ce^{4\Lambda\kappa}}, \|\bar{\tilde{\phi}}\|_{C^2}\}\leq \ub^{-l} \leq a^{\f13}\cdot\O(u,0)^{-\t\d}=a^{\f13}\cdot|u|^{\f{-\t\d\ms}{2(1-\ms)}}.
\end{equation}

Now we revisit the arguments in Section \ref{sec-Schauder-estimates}, Section \ref{linearization} and Section \ref{sec-existence-solutions} of this article with $f(\o, \ub)$ replaced by  $g(\ub)\t f(\o, \ub)$.  The arguments in Section \ref{sec-Schauder-estimates} stay the same. For the arguments in Lemma \ref{lemma-linearized-op-invert} of Section \ref{linearization}, we now consider the positive eigenfunction $\psi_1$ associated with the corresponding first eigenvalue $\mu_1>0$, i.e.,
\begin{align*}
\partial_{\cp} \mathcal S(\cp)[\psi_1]
&=\Delta_{\gamma}\psi_1-\ms\psi_1-\f{1-\ms}{2}f \cp^{-\f{1}{\ms}}\psi_1+c^{i}\nab_i \psi_1+c\psi_1\\
&=-\mu_1\psi_1+c^{i}\nab_i \psi_1+c\psi_1.
\end{align*} 
With $\O$ being sufficiently small in the interior region, we have $|c^i|+|c|\ll 1$. To make sure that  $\partial_{\cp} \mathcal S(\cp)[\psi_1]<0$ on $\mathbb{S}^2$, here we check the size of $f\cp^{-\f{1}{\ms}}$. Using the facts $f(\o,\ub)=g(\ub)\t f(\o, \ub)$ with $0\leq \t f\leq 1$ and $R=(\ub a)^{1-\ms}\cp^{\f{1-\ms}{\ms}}$ being the MOTS, we have $\ub a \cp^{\f{1}{\ms}}=R^{\f{1}{1-\ms}}\approx \ub a g(\ub)$, {\color{black}which is from the construction of the barriers. Here $\approx$ is independent of $\ub$.} Hence it holds $f\cp^{-\f{1}{\ms}}=g(\ub)\t f(\ub, \o)\cp^{-\f{1}{\ms}}\approx 1$, {\color{black}being independent of $\ub$}. 
Thus, the rest arguments stay the same as in Section \ref{linearization}. We proceed similarly as in Lemma \ref{lemma-linearized-apriori-estimates}. With replacing $f$ and $f_0$ by $g(\ub)\t f$ and $g(\ub) \t f_0$, we get the corresponding desired conclusion.  All the other proofs for the elliptic part can be obtained in the same way as in Section \ref{linearization} and Section \ref{sec-existence-solutions}. In summary, allowing the continuous perturbation for $\chih(-1, \ub, \o)$ and $\nab_4\phi(-1, \ub, \o)$ with respect to the $\ub$ variable, proceeding the same as in Section \ref{BV instability} we then prove 

\begin{theorem} \label{main thm 2 section 14} 
Consider the characteristic initial value problem for the Einstein-scalar field system \eqref{1.1}. Assigning Christodoulou's naked-singularity initial data in \cite{Chr.2} along  $\ub=0$ with $-1\leq u \leq 0$ and prescribing perturbed initial data along $u=-1$ satisfying
\begin{equation}
\sum_{i\leq 5, j\leq 3}\ub^{j}a^{-\f12} \|(\partial_{\ub})^j\nab^{i}\chih\|_{L^{\infty}_{\ub}L^2(S_{-1,\ub})}+\ub^{j}a^{-\f12} \|(\partial_{\ub})^j\nab^{i}\nab_4\phi\|_{L^{\infty}_{\ub}L^2(S_{-1,\ub})}\leq B, 
\end{equation}
then for each $B$ there exist a sufficient small $\delta=\delta(B)$ and constant $a=a(B)$\footnote{In particular, we can choose $a(B)=1$.}, and the Einstein-scalar field system admits a unique regular solution in the spacetime region $(u, \ub, \o)$ satisfying $0\leq\ub\leq \delta$ and $\ub\leq |u|\O^{2-\f{\t\delta}{2}}(u,0)\leq 1$ with $0<\t\d\ll 1$. 

Moreover, if we further assume  
\begin{equation}\label{lower bound for C0 perturbation} 
 \int_0^{\ub}|\chih|^2 (-1, \ub', \o)+[\nab_4\phi(-1, \ub', \o)-\nab_4\phi(-1, 0, \o)]^2 d\ub'=g(\ub)\t f(\o, \ub)\ub a,
\end{equation}
with 
$$g(\ub)=(C_1\Lambda+C_1)^{\f12}\cdot [\ln(\ln\f{1}{\ub})]^{-\f12} $$
where $C_1>2$, $\Lambda$ being a fixed positive constant, and $\t f(\o,\ub)$ obeying  
\begin{equation}
0\leq \t f\leq 1 \mbox{ on } \mathbb S^2\times (0,\delta] \quad \mbox{ and } \quad \t f(\cdot, \ub)\ge m \mbox{ on } B_{p}(\e)  
\end{equation} 
for a point  $p\in \mathbb S^2$  and constants  $m\in (0,1)$, $\e\in (0,\pi/2)$, then within the solved hyperbolic region, there \textit{exists a unique} MOTS $M_{\ub}$ on each $\Hb_{\ub}$ with $0< \ub \leq \d$ and the collection of MOTS $\{M_{\ub}\}$ emerges and censors the central singularity. And they form an achronal apparent horizon.
\end{theorem}

We also have the corresponding nonlinear instability result with high co-dimensions. It is summarized in Theorem \ref{main thm 1.2}. Here we prove it. 

\begin{proof} 
We proceed in the same way as in Theorem \ref{co dimension instability theorem v1}. Only we modify the requirement for $f_i$ and $g_i$ as: For 
$$1\leq i\leq k, \quad C_1>2 \quad \mbox{ and } \quad g(\ub)=(C_1\Lambda+C_1)^{\f12}\cdot [\ln (\ln\f{1}{\ub})]^{-\f12},$$ 
the prescribed $f_i$ and $g_i$ obey
$$0\leq f_i\leq g(\ub)^{\f12}\at \mbox{ on } \mathbb S^2\times (0,\delta], \quad f_i(\cdot, \ub)\ge m_i g(\ub)^{\f12}\at \mbox{ on } B_{p_i}(\e_i)\subset D_i ,$$
$$0\leq g_i\leq g(\ub)^{\f12}\at \mbox{ on } \mathbb S^2\times (0,\delta], \quad  g_i(\cdot, \ub)\ge m_i' g(\ub)^{\f12}\at \mbox{ on } B_{p_i'}(\e_i')\subset D_i'$$  
with points  $p_i\in D_i\subset \mathbb S^2$, $p_i'\in D_i'\subset \mathbb S^2$ and constants  $m_i, m_i'\in (0,1)$, $\e_i, \e_i'\in (0,\pi/2)$. Note that we can set $a=1$. And as $\ub\rightarrow 0^+$, we have $$\lim_{\ub\rightarrow 0^+}g(\ub)^{\f12}=0.$$ 
Therefore, we can choose continuous-in-$\ub$ and smooth-in-$\o$ functions $f_i, g_i, \chih, \partial_{\ub}\phi$ along $u=-1$ to verify all the requirement of this theorem. The rest proof is the same as in Theorem \ref{co dimension instability theorem v1}. 

\end{proof}

{\color{black}
At the end of this section, we check the size of the initial perturbation. We first prove

\begin{proposition}\label{ln ln initial data size} 
For $0\leq \ub \leq \delta$, it holds 
$$\|[\ln(\ln \f{1}{\ub})]^{-\f14}\|_{\dot{H}^{\f12}([0,\delta])}\leq \f{C}{|\ln\f{1}{\delta}|^{\f12}}$$
with $C$ a uniform positive constant.  
\end{proposition}

\begin{proof}
Let $\hat{f}(\ub):=[\ln(\ln \f{1}{\ub})]^{-\f14}$. Using the Sobolev-Slobodeckij space, we have
\begin{equation*}
\begin{split}
\|\hat{f}\|^2_{\dot{H}^{\f12}([0,\delta])}=&\iint_{\substack{0\leq x\leq \delta, \\ 0\leq y\leq \delta}}
\f{|\hat{f}(x)-\hat{f}(y)|^2}{|x-y|^2}dx dy=2\iint_{0\leq x\leq y\leq \delta}\f{|\hat{f}(x)-\hat{f}(y)|^2}{|x-y|^2}dx dy\\
=&2\iint_{\substack{0\leq x\leq y\leq \delta, \\ x\leq y\leq 2x}}\f{|\hat{f}(x)-\hat{f}(y)|^2}{|x-y|^2}dx dy+2\iint_{\substack{0\leq x\leq y\leq \delta, \\ y>2x }}\f{|\hat{f}(x)-\hat{f}(y)|^2}{|x-y|^2}dx dy\\
=&\hat{I}_1+\hat{I}_2. 
\end{split}
\end{equation*} 
For $\hat{I}_1$, applying the intermediate value theorem, with $\theta$ being a constant $\in[0,1]$, we get 
\begin{equation*}
\begin{split}
\hat{I}=&2\iint_{\substack{0\leq x\leq y\leq \delta, \\ x\leq y\leq 2x}}\bigg(\f14[\ln(\ln\f{1}{x+\theta(y-x)})]^{-\f54}\cdot\f{1}{\ln\f{1}{x+\theta(y-x)}}\cdot\f{1}{x+\theta(y-x)} \bigg)^2 dx dy\\
\lesssim&2\iint_{\substack{0\leq x\leq y\leq \delta, \\ x\leq y\leq 2x}}\f{1}{(\ln\f{1}{x})^2}\cdot\f{1}{y^2}dxdy \lesssim \int_0^{\delta}\f{1}{(\ln\f{1}{x})^2}\cdot\f{1}{x}dx \lesssim \f{1}{\ln\f{1}{\delta}}.
\end{split}
\end{equation*} 
To bound $\hat{I}_2$, we set $x:=e^u$ and $y:=e^v$. The requirement $0<2x<y<\delta$ implies
$$-\infty<u \leq -\ln\f{1}{\delta}-\ln 2, \quad -\infty<v \leq -\ln\f{1}{\delta} \quad \mbox{ and } \quad v-u> \ln 2>0.$$
Utilizing $(u, v)$ and $(u, w)$ coordinates with $w=v-u$, we then obtain
\begin{equation*}
\begin{split}
\hat{I}_2=&2\iint_{\substack{0\leq x\leq y\leq \delta, \\ y>2x }}\f{|[\ln(\ln \f1y)]^{-\f14}-[\ln(\ln \f1x)]^{-\f14}|^2}{|y-x|^2}dx dy\\
\lesssim&2\iint_{\substack{-\infty<u, v\leq -\ln\f{1}{\delta}, \\ v-u\geq \ln 2}}e^{-2v}e^u e^v |[\ln(-v)]^{-\f14}-[\ln(-u)]^{-\f14}|^2 dudv\\
\lesssim&2\int_{\ln 2}^{+\infty}\int_{-\infty}^{-\ln\f{1}{\delta}-w}e^{-w} |[\ln(-u-w)]^{-\f14}-[\ln(-u)]^{-\f14}|^2 du dw.
\end{split}
\end{equation*} 
The last line can be bounded as 
\begin{equation*}
\begin{split}
\hat{I}_2\lesssim&2\int_{\ln 2}^{+\infty}\int_{-\infty}^{-\ln\f{1}{\delta}-w}e^{-w}|\int_{-u}^{-u-w}-\f14(\ln r)^{-\f54}\cdot\f{1}{r} dr|^2 du dw\\
 \lesssim& \int_{\ln 2}^{+\infty}\bigg(\int_{-\infty}^{-\ln\f{1}{\delta}-w}e^{-w}\cdot\f{w^2}{|u+w|^2}du\bigg) dw=\int_{\ln 2}^{+\infty} e^{-w} w^2\bigg(\int_{-\infty}^{-\ln\f{1}{\delta}-w}\f{1}{|u+w|^2}du\bigg) dw\\
 =&\int_{\ln 2}^{+\infty} e^{-w} w^2\bigg(\int_{-\infty}^{-\ln\f{1}{\delta}}\f{1}{|\tilde{u}|^2} d\tilde{u}\bigg) dw\lesssim \f{1}{\ln\f{1}{\delta}}.
\end{split}
\end{equation*} 

\end{proof}

Back to Theorem \ref{main thm 2 section 14}, along $H_{-1}^{[0,\delta]}$ we could set the perturbation
\begin{equation}\label{set the C0 perturbation}
|\chih|(-1, \ub, \o)+|\nab_4\phi(-1, \ub, \o)-\nab_4\phi(-1, 0, \o)|=[\ln(\ln\f{1}{\ub})]^{-\f14}\cdot \tilde{f}(\o)^{\f12}a^{\f12}
\end{equation}
with $\tilde{f}(\o)$ being a smooth function of $\o$. Hence, to trigger the naked-singularity censoring, we only need the perturbation of $g$ in scale-critical $\|\cdot\|_{\dot{H}^{\f32}{(H_{-1}^{[0, \delta]}})}$ norm, i.e., 
$$\mbox{the } \|\cdot\|_{\dot{H}^{\f12}{(H_{-1}^{[0, \delta]}})}=\|\cdot\|_{L^2_{\o}\dot{H}^{\f12}_{\ub}([0,\delta])}+\|\cdot\|_{L^2_{\ub}\dot{H}^{\f12}_{\o}(S_{-1,\ub})} \mbox{ norm }$$ 
of the above perturbation in \eqref{set the C0 perturbation} to be of size $(\ln\f{1}{\delta})^{-\f12}a^{\f12}\ll 1$. And as $\delta\rightarrow 0^+$, we allow the size of the scale-critical perturbation shrink to $0$. 
}

\appendix
\section{Construction of Localized-In-Angle Initial Data}\label{Appendix A}

Along $u=-1$, in this appendix we construct the localized-in-angle initial data for $\partial_{\ub}\phi(-1, \ub, \o)$ and  $|\chih|^2_{g}(-1, \ub, \o)$ with $\ub\geq 0$. Since $\partial_{\ub}\phi(-1, \ub, \o)$ is a scalar function, we can prescribe the desired localized-in-angle initial data as we need. Notice that $\chih_{AB}(-1, \ub, \o)$ is a traceless $2$-form. To prescribe its characteristic initial data,  we follow calculations in Chapter 2 of \cite{Chr:book} and in Appendix C of \cite{An17}. With stereographic coordinates $(\theta_1, \theta_2)$, we have
\begin{equation}\label{initial data hat phi}
g|_{S_{-1,\ub}}=(\hat\phi|_{S_{-1,\ub}})^2 \hat{g}|_{S_{-1,\ub}}
\end{equation}
$$\mbox{with} \quad \hat{g}_{AB}(-1,\ub,\theta_1, \theta_2)=\f{1}{1+\f14(\theta_1^2+\theta_2^2)} \,m_{AB}(-1,\ub, \theta_1, \theta_2)$$ 
and $m_{AB}$ can be expressed as $m(-1, \ub, \theta_1, \theta_2)=\exp \Psi(\ub, \theta_1, \theta_2)$. Here $\Psi$ is a symmetric trace-free 2-dimensional matrix. Notice that as pointed out in Chapter 2 of \cite{Chr:book}, $\Psi(\ub, \theta_1, \theta_2)$ is the free data we can prescribe along $u=-1$ and $\hat\phi|_{S_{-1,\ub}}$ in \eqref{initial data hat phi} is determined by $\Psi(\ub, \theta_1, \theta_2)$.\footnote{In the below, we will prescribe $\Psi(\ub, \theta_1, \theta_2)$ in the north polar chart. And  $\Psi'(\ub, \theta_1', \theta_2')$ in the south polar chart can be obtained accordingly.} 

For $\sqrt{\theta^2_1+\theta^2_2}<20$} we set $$\Psi(\ub,\theta_1, \theta_2):=\sqrt2 x(\ub, \theta_1, \theta_2)\cdot\at\cdot\Psi_0(\ub, \theta_1, \theta_2).$$ 
Here $\partial x(\ub, \theta_1, \theta_2)/\partial \ub$ is a free scalar function to be prescribed and we also choose
$$\Psi_0(\ub, \theta_1, \theta_2)=
\begin{bmatrix}
1&0\\
0&-1
\end{bmatrix} \quad \mbox{for} \quad \ub\geq 0.$$
With a similar calculation carried out in Appendix C of \cite{An17}, we deduce
\begin{equation*}
\begin{split}
\f{\partial}{\partial \ub}m(-1, \ub,\theta_1,\theta_2)=&\sqrt2\at \begin{bmatrix}
\exp\bigg(\sqrt2 x(\ub, \theta_1, \theta_2)\cdot\at\bigg)&0\\
0&-\exp\bigg(-\sqrt2 x(\ub, \theta_1, \theta_2)\cdot\at\bigg)
\end{bmatrix}\\
&\times \f{\partial x}{\partial \ub}(\ub, \theta_1, \theta_2)
\end{split}
\end{equation*}
From Chapter 2 in \cite{Chr:book}, we have 
$$\chih_{AB}=\f12\phi^2\f{\partial \hat{g}}{\partial \ub}=\f12\phi^2 W^2(\theta_1, \theta_2)\f{\partial}{\partial \ub}m_{AB}(-1, \ub,\theta_1, \theta_2) \mbox{ and}$$

\begin{equation*}
\begin{split}
\f12 |\chih|^2_{g}=&\f18 (m^{-1})^{CA}(\f{\partial m}{\partial \ub})_{AB} (m^{-1})^{BD} (\f{\partial m}{\partial \ub})_{DC}.
\end{split}
\end{equation*}
For our concrete construction, along $u=-1$ it hence holds
\begin{equation}\label{e value}
\begin{split}
\f12|\chih|^2_g=&\f18\cdot \bigg(\sqrt2\cdot \f{\partial x}{\partial \ub}(\ub, \theta_1, \theta_2)\cdot\at\bigg)^2 \cdot\tr \begin{bmatrix} 1&0\\0&1 \end{bmatrix}=\f12 a\cdot  \bigg(\f{\partial x}{\partial \ub}(\ub, \theta_1, \theta_2)\bigg)^2.
\end{split}
\end{equation}
With $\partial x(\ub, \theta_1, \theta_2)/\partial \ub$ as the freely prescribed scalar function on $S_{-1, \ub}$, we therefore obtain the desired localized-in-angle initial data for $|\chih|^2_g(-1, \ub, \o)$.

\end{document}